\documentclass[11pt,letterpaper]{thesis2}
\pdfoutput=1
\newcommand{\al}{\alpha} 
\newcommand{\aj}{\alpha_j} 
\newcommand{\aN}{\alpha_N} 

\newcommand{\pd}[2]{\frac{\partial #1}{\partial #2}} 
\newcommand{\pdd}[2]{\frac{\partial^2 #1}{\partial #2^2}} 
\newcommand{\pddm}[3]{\frac{\partial^2 #1}{\partial #2 \partial #3}} 
\newcommand{\md}[2]{\frac{D^{#1} #2}{Dt}} 
\newcommand{\mds}[1]{\md{s}{#1}} 
\newcommand{\cd}{\cdot} 
\newcommand{\grad}{\boldsymbol{\nabla}} 
\newcommand{\diver}{\grad \cdot} 
\newcommand{\sumj}{\sum_{j=1}^N} 
\newcommand{\sumjj}{\sum_{j=1}^{N-1}} 
\newcommand{\suma}{\sum_{\alpha}} 
\newcommand{\sumba}{\sum_{\beta \neq \alpha}} 

\newcommand{\pp}[2]{p^{#1_{(#2)}}}
\newcommand{\ppi}[2]{\pi^{#1_{(#2)}}}
\newcommand{\ppbar}[2]{\overline{p}^{#1_{(#2)}}}

\newcommand{\psia}{\psi^{\alpha}} 
\newcommand{\hpsia}{\hat{\psi}^{\alpha}}

\newcommand{\psiaj}{\psi^{\aj}} 

\newcommand{\porosity}{\varepsilon} 
\newcommand{\Sdot}{\dot{S}}
\newcommand{\eps}[1]{\varepsilon^{#1}} 
\newcommand{\epsa}{\varepsilon^{\al}} 
\newcommand{\epsdot}[1]{\dot{\varepsilon}^{#1}} 

\newcommand{\rhoa}{\rho^{\alpha}} 
\newcommand{\rhos}{\rho^s} 
\newcommand{\rhol}{\rho^{l}} 
\newcommand{\rhoaj}{\rho^{\alpha_j}} 

\newcommand{\rholj}{\rho^{l_j}} 
\newcommand{\rhosj}{\rho^{s_j}} 
\newcommand{\muaj}{\mu^{\aj}}

\newcommand{\lamaj}{\lambda^{\aj}}

\newcommand{\lamhataN}{\soten{\boldsymbol{\lambda}}^{\aN}}

\newcommand{\rh}{\varphi}
\newcommand{\bv}[1]{\boldsymbol{v}^{#1}} 
\newcommand{\bq}{\boldsymbol{q}} 

\newcommand{\Cbars}{\overline{\underline{\underline{\boldsymbol{C}}}}^s} 
\newcommand{\Cbarsdot}{\dot{\overline{\underline{\underline{\boldsymbol{C}}}}}^{s}}
\newcommand{\Fbars}{\overline{\underline{\underline{\boldsymbol{F}}}}^{s}}
\newcommand{\Jsdot}{\dot{J}^s} 

\newcommand{\stress}[1]{\underline{\underline{\boldsymbol{t}}}^{#1}} 
\newcommand{\bd}[1]{\underline{\underline{\boldsymbol{d}}}^{#1}} 
\newcommand{\Thatba}{\hat{\boldsymbol{T}}_{\beta}^{\alpha}}
\newcommand{\That}[2]{\hat{\boldsymbol{T}}_{#1}^{#2}}
\newcommand{\Thatbaj}{\hat{\boldsymbol{T}}_{\beta}^{\alpha_j}}

\newcommand{\ehat}[2]{\hat{e}_{#1}^{#2}}
\newcommand{\ehatba}{\hat{e}_{\beta}^{\alpha}}
\newcommand{\ehatbaj}{\hat{e}_{\beta}^{\alpha_j}}

\newcommand{\rhataj}{\hat{r}^{\alpha_j}}

\newcommand{\ihat}{\hat{\boldsymbol{i}}}
\newcommand{\ihataj}{\hat{\boldsymbol{i}}^{\alpha_j}}

\newcommand{\foten}[1]{\boldsymbol{#1}} 
\newcommand{\soten}[1]{\underline{\underline{\boldsymbol{#1}}}} 
\newcommand{\toten}[1]{\underline{\underline{\underline{\boldsymbol{#1}}}}} 
\newcommand{\Foten}[1]{\underline{\underline{\underline{\underline{\boldsymbol{#1}}}}}} 

\usepackage{graphicx}
\usepackage{graphicx}
\usepackage{array}
\usepackage{amsmath}
\usepackage{amsfonts}
\usepackage{mathrsfs}
\usepackage[breaklinks,colorlinks=false]{hyperref}
\usepackage{tikz,pgfplots}
\usepackage{cleveref}
\allowdisplaybreaks
\usepackage{multirow}
\usepackage{rotating}
\usepackage{subfigure,caption}
\usepackage{floatrow}
\usepackage{float}
\newtheorem{theorem}{Theorem}[chapter]

\newtheorem{prop}[theorem]{Proposition}

\newenvironment{proof}{{\bf Proof:} }{\hfill\rule{2.1mm}{2.1mm}}

\newcounter{example}

\def\OnlyTechRpt{0}
\def\OnlyThesis{1}
\newcommand{\TechRpt}[1]{\ifnum\OnlyTechRpt=1 #1 \fi}
\newcommand{\Thesis}[1]{\ifnum\OnlyThesis=1 #1 \fi}

\begin{document}

\pagenumbering{roman}
\pagestyle{empty}

\title{Heat and Moisture Transport in Unsaturated Porous Media: \\ A Coupled Model in Terms of
Chemical Potential}
\author{Eric R.}{Sullivan}
\otherdegrees{B.S., Iowa State University, 1998 \\
M.S., University of Colorado at Colorado Springs, 2007}
\degree{Doctor of Philosophy}{Ph.D., Applied Mathematics}
\degreeyear{2013} \dept{Applied Mathematics}{}
\titlepage

\pagestyle{fancy}

\advisor{Associate Professor}{Lynn Schreyer-Bennethum}
\readerone{Julien Langou}
\readertwo{Jan Mandel}
\readerthree{Richard Naff}
\readerfour{Kathleen Smits}
\approvalpage

\abstractpage{
Transport phenomena in porous media are commonplace in our daily lives. Examples and
applications include heat and moisture transport in soils, baking and drying of food
stuffs, curing of cement, and evaporation of fuels in wild fires.  Of particular interest
to this study are heat and moisture transport in unsaturated soils. Historically,
mathematical models for these processes are derived by coupling classical Darcy's,
Fourier's, and Fick's laws with volume averaged conservation of mass and energy 
and empirically based source and sink terms. Recent experimental and mathematical research
has proposed modifications and suggested limitations in these classical equations. The
primary goal of this thesis is to derive a thermodynamically consistent system of
equations for heat and moisture transport in terms of the chemical potential that
addresses some of these limitations. The physical processes of interest are primarily
diffusive in nature and, for that reason, we focus on using the macroscale chemical
potential to build and simplify the models.  The resulting coupled system of nonlinear
partial differential equations is solved numerically and validated against the classical
equations and against experimental data. It will be shown that under a mixture theoretic
framework, the classical Richards' equation for saturation is supplemented with gradients
in temperature, relative humidity, and the time rate of change of saturation.
Furthermore, it will be shown that restating the water vapor diffusion equation in terms
of chemical potential eliminates the necessity for an empirically based fitting parameter.
}

\clearpage \vspace*{.25in}
\centerheading{DEDICATION}
\begin{block}
\begin{center} To Johnanna \end{center}
\end{block}
\newpage

\clearpage \vspace*{.25in}
\centerheading{ACKNOWLEDGMENT}
\begin{block}
There are several people who have either directly or indirectly made this work possible.
First and foremost, I would like to extend many thanks to Lynn Schreyer-Bennethum. Without
her patience, support, amazing talents as a teacher, mathematical and physical insights,
and constant questions I would likely not have been able to get involved so heavily in
mathematical modeling. I would like to thank Kate Smits for agreeing to join my committee
so late and for sharing her expertise with experimentation.  I would also like to thank
Jan Mandel, Julien Langou, and Rich Naff for their unending support and patience as my
committee members.

There are several people in the UCD community that I would like to thank.  To Keith
Wojciechowski, who was one of the primary reasons that I started working with Lynn.  To
Cannanut Chamsri, Tom Carson, and Mark Mueller, my research group comrades, for putting up
with my ramblings and never-ending slue of equations during our weekly meetings.  To Mark
especially, for his thought provoking questions and constant push to tie the math and
physics together.  To Jeff Larson and Henc Bouwmeester, who have become life-long friends.
I appreciated the comic relief, the deep mathematical conversations, the computer help,
and the fun that we had. Of course, I cannot forget to thank the many other people who
helped me to persevere through graduate school: Jenny Diemunsch, Cathy Erbes, Piezhun Zhu,
Brad Lowry, Tim Morris, Samantha Graffeo, and many other. 

Finally, I would like to thank my family and friends.  Without my {\it support network}
this would have been much tougher.  To my parents, Drew and Donna, who have taught me what
real hard work really is.  To my siblings, Doug, Emily, and Julie, for helping me get to
this point in my life.  To Emily, especially, who went through grad school along with me.
To my friends Fred Hollingworth, Jodie Collins, and Erik Swanson who reminded me that it
was equally important to head up into the mountains and go climbing as it was to work.  To
my {\it teacher friends} Scott Strain, Ginger Anderson, and Krista Bruckner who have
helped to make me a better teacher and better person. Lastly, to Johnanna.  I couldn't
have done this without you.
\end{block}
\newpage

\tableofcontents
\listoffigures
\listoftables

\setcounter{page}{0}
\pagenumbering{arabic}

\chapter{Introduction}\label{ch:intro}
Water flow, water vapor diffusion, and heat transport within variably saturated soils (above the
groundwater and below the soil surface) are important physical processes in evaporation studies, contaminant
transport, and agriculture.  The mathematical models governing these physical processes
are typically combinations of classical empirical law (e.g. Darcy's law) and volume
averaged conservation laws. The resulting equations are valid in many situations, but
recent experimental and mathematical research has suggested modifications and corrections
to these models. 
The primary goal of this thesis is to build a thermodynamically consistent mathematical
model for heat and moisture transport that takes these recent advancements into
consideration.  To realize this goal Hybrid Mixture Theory (HMT) and the macroscale
chemical potential are used as the primary modeling tools.

\section{Previous Work}
In 1856, Henri Darcy published his research on the use of sand filters to clean the water
sources for the fountains in Dijon, France. He found that the flux of water across sand
filters was directly proportional to the gradient of the pressure head.  This simple
observation has become known as Darcy's law and is one of the main modeling tools in
hydrology and soil science \cite{Darcy1856}.  Darcy's law is an example of a
historical rule (or {\it law}) that has perpetuated to the present day. Darcy's law was
originally derived for saturated porous media, but near the turn of the century it was
extended for use in unsaturated soils.  In 1931, L.A. Richards coupled Darcy's law with
liquid mass balance to derive what is now known as Richards' equation.
This equation relies on the assumption that Darcy's law is valid for unsaturated media,
but it also relies on an empirically-derived relationship between pressure and saturation. 

The pressure-saturation relationship is known to be hysteretic in nature
(depends on the direction of wetting), and only recently have researchers been able to
move toward functional relationships that capture this effect \cite{Joekar-Niasar2007}.
Correction terms in Richards' equation have been proposed via HMT that suggest that the rate at
which the capillary pressure is changing may play a role in the overall dynamics of the
saturation \cite{Hassanizadeh2001,Hassanizadeh2002}.  This proposed, third-order, term in
Richards' equation has only recently been studied mathematically and experimentally.  One
possible physical interpretation of this term is that is accounts for how fast the
liquid-gas interfaces are rearranging at the pore scale.

In the 1950s, Philip and deVries published their works on vapor and heat transport in
porous media \cite{deVries1958,deVries1957}.  Their approach accounts for water flow in
both the liquid and fluid phases in response to water content and temperature gradients in
the soil. In order to account for the observation that Fickian diffusion inadequately
describes diffusion in
porous media, Philip and deVries implemented an {\it enhancement} factor, $\eta$, to
adjust the diffusion coefficient.  This factor is fitted to the measured diffusion data
for a particular medium.  In \cite{Cass1984}, Cass et al.\,found that $\eta$ increases
with saturation (with $\eta \approx 1$ for dry soils). Philip and deVries proposed that
thermal gradients and the condensation and evaporation  through ``liquid islands'' are the
pore-scale mechanisms that cause the observed enhancement.  Counter intuitively, the
governing equation holds at the macroscale while these mechanisms are inherently
pore-scale. The Philip and de Vreis model has not been validated in a laboratory setting.

More recent works question the validity of the Philip and deVries model. Shokri et
al.\,\cite{Shokri2009a} suggests that the coupling between the water flow and Fickian
diffusion is the key to estimating the vapor flux. They suggest that under this
consideration there is no need for the enhancement factor.  Webb \cite{Webb1998}, and more
recently Shahareeni et al.\,\cite{Shahraeeni2010}, showed that enhancement can exist in
the absence of thermal gradients. It was initially thought that enhancement couldn't
occurin the absence of thermal gradients and was therefore ignored. Based on the
observtion that enhancement can occur without the need for thermal gradients, it is clear
that the Philip and de Vries model needs modification. Cass \cite{Cass1984} and Campbell
\cite{Campbell1985} give a functional form of the enhancement factor that is commonly used
(e.g.  \cite{Sakai2009,Smits2011}), but relies on an empirical fitting parameter.  In the
present work we take the view of Shokri et al. that there is no need for the enhancement
fact, and wederive a diffusion equation simply based on liquid and vapor flow. The novelty
of the present approach is the use of the chemical potential as the driving force for both
types of flow.

For energy transport, the 1958 deVries model \cite{deVries1958} is still commonly used
(e.g. \cite{Sakai2009,Smits2011}).
Similar to the enhanced diffusion model, deVries built this model so as to account for the
flux of the fluid phases. This is sensible as the fluid phases will certainly transport
heat. More recently, Bennethum et al.\,\cite{Bennethum1999} and Kleinfelter
\cite{Kleinfelter2007} used Hybrid Mixture Theory to derive heat transport equations in
porous media (Bennethum et al.\,studies saturated porous media and Kleinfelter studied
multiscale unsaturated media). They
verified many of the findings by deVries but also proposed several new terms associated
with the physical processes of heat transport.  In this work we extend the Bennethum et
al. and Kleinfelter approachs to unsaturated media.

\section{Hybrid Mixture Theory and Thesis Goals}
To build the models in this work we make extensive use of Hybrid Mixture Theory (HMT).
HMT, statistical upscaling, and homogenization have all been used
as techniques to re-derive, confirm, and extend Darcy's, Fick's, and
Fourier's laws in porous media. For a technical summary of some of these
methods see \cite{Cushman2002}.  HMT is a term for the process of using
volume averaged pore-scale conservation laws along with the second law of thermodynamics
to give thermodynamically consistent constitutive equations in porous media.  The
technique as applied to porous media was developed by several parties, the most notable
being Hassanizadeh and Gray \cite{Gray1977,Gray1979} and Cushman et
al.\,\cite{Bennethum1996,Bennethum1996a,Cushman1985}, but the general principles were
developed by Coleman and Noll \cite{Coleman1963}.  

In the present work we use HMT to derive new extensions to these laws in the case of
unsaturated porous media.  These extensions are then used to derive a model for
total moisture transport in unsaturated soils. While this sort of modeling has been done
in the past, rarely have the three principle physical process (movement of saturation
fronts, vapor diffusion, and thermal conduction) been considered from first principles and
put on the same theoretical footing (HMT in this case).  No known work attempts
to couple these three different effects together with one physical measurement: the
chemical potential. This is one of the unique features of this work.

In the most general sense, the chemical potential is a measure of the tendency of a
substance (thinking particularly of a fluid or species) to diffuse. This diffusion could
be of a species within a mixture (e.g. water vapor diffusing through air) or it could be a
phase diffusing into another (e.g. water into a Darcy-type sand filter).  The fact that
most physical processes in porous media are of this type gives an inspiration for the
potential usefulness of the chemical potential as a modeling tool.  That is, from a broad
point of view, it should be possible to restate the physical processes of moving
saturation fronts and vapor diffusion more naturally by the chemical potential. This
approach has not been thoroughly explored in the past since the chemical potential is not directly
measureable and the theoretical footings of upscaling the chemical potential are
relatively new.  In the
saturated case, extensions to Darcy's law have also been developed via HMT, and the
results indicate that the macroscale chemical potential is a viable modeling tool for
diffusive velocity in saturated porous media
\cite{Murad2000,Schreyer-Bennethum2012,weinstein}. In the present work we give chemical
potential forms of Darcy's and Fick's law as well as presenting simplifications to
Fourier's law based on the chemical potential.  In the case of a pure liquid phase we will
show that the chemical potential form of Darcy's law is no different than the more
traditional pressure formulation.  In the gas phase, on the other hand, we will show that
the pairings of chemical potential forms of Fick's and Darcy's laws gives a new form of
the diffusion coefficient that does not need the {\it enhancement factor} indicated in the
work by Phillip and DeVries \cite{deVries1957}. The chemical potential will finally be
used to derive a novel form of Fourier's law for heat conduction in multiphase media.

Once we have derived new forms of the classical constitutive equations we pair these
equations with volume averaged conservation laws to give a coupled system of partial
differential equations governing heat and moisture transport.  The second law of
thermodynamics is used to suggest additional closure conditions for each of the equations.
Together, the system consists of a nonlinear pseudo-parabolic equation for saturation, a
nonlinear parabolic equation for vapor diffusion, and a nonlinear parabolic-hyperbolic
equation for heat transport. 

In summary, this work serves several purposes: (1) it is a step toward better
understanding the role of the chemical potential in multiphase porous media, (2) it makes
strides toward understanding the phenomenon of {\it enhanced vapor diffusion} in porous
media, and (3) finally we propose a novel coupled system of equations for heat and
moisture transport.

\section{Thesis Outline}
In Chapter \ref{ch:diffusion_comp} we take a step back from porous media and discuss
pore-scale diffusion models.  This is done in an attempt to elucidate the assumptions,
derivations, and models used in various disciplines as there tends to be confusion about
where the miriad of assumptions are valid.  In this chapter we give mathematical and
physical reasons for the many commonly used assumptions as well as proposing an
alternative advection-diffusion model as compared to the popular Bird, Stewart, and
Lightfoot model \cite{Bird2007}.

In Chapter \ref{ch:HMT} we present the necessary background information in order to
understand volume averaging and the exploitation of the entropy inequality.  Much of this
chapter is paraphrased from previous works, such as
\cite{Bennethum1996,Bennethum1996a,Gray1977,Gray1979,weinstein,Wojciechowski2011}. In the
beginning of Chapter \ref{ch:Exploit} we use the tools from Chapter \ref{ch:HMT} to build
and exploit a version of the entropy inequality specific for multiphase media where each
phase consists of multiple species. The remainder of Chapter \ref{ch:Exploit} is dedicated
to the exploitation of the entropy inequality for a novel choice of independent variables
describing these media.  Appendix \ref{app:AbstractEntropy} serves as a companion to this
discussion as it gives the abstract formulation and logic of the entropy inequality.
Throughout Chapter \ref{ch:Exploit}, the goal is to derive new forms of Darcy's, Fick's,
and Fourier's laws and to propose extensions to these laws in terms of the macroscale
chemical potential.  As part of these derivations we arrive at new expressions for the
pressure and wetting potentials in unsaturated media. All of these derivations are done in
a general sense with as few assumptions as possible.  This leaves open the possibilities
of future research.

In Chapter \ref{ch:Transport} we couple the results found from the exploitation of the
entropy inequality (Chapter \ref{ch:Exploit}) with the volume averaged conservation laws
derived in Chapter \ref{ch:HMT}.  In Section \ref{sec:classical_models}, a more in-depth historical perspective of
the classical equations used for heat and moisture transport is given to orient the reader
to the recent research. Fluid transport equations are presented in Section
\ref{sec:mass_balance_equations} along with a discussion of the relationship between mass
transfer and chemical potential.   Considerable effort is put toward deriving a heat transport
equation with the final equation presented in Section \ref{sec:TotalEnergyBalanceEqn}. In
Section \ref{sec:simplifying_assumptions} several simplifying assumptions are presented in
order to close the system of equations. In particular, Sections \ref{sec:liquid_simplifications},
\ref{sec:vapor_diffusion}, and \ref{sec:total_energy_simplifications} give
simplifications, assumptions, and dimensional analysis
for the liquid, gas, and heat equations respectively. In Section
\ref{sec:constitutive_equations} we present the remaining constitutive equations necessary
to close the system of equations.  Since so many assumptions and simplifications are made
throughout the chapter a summary of all of the results is presented in Section
\ref{sec:simplifications_summary}.

In Chapter \ref{ch:Existence_Uniqueness} we examine the proper regularity and
assumptions needed for existence and uniqueness of solutions.  These results are
preliminary and do not constitute a complete existence and uniqueness study for these
equations. 

In Chapter \ref{ch:Transport_Solution} we perform numerical analysis on the equations
derived in Chapter \ref{ch:Transport}. In Sections \ref{sec:numerical_saturation},
\ref{sec:numerical_vapor_diffusion}, and \ref{sec:numerical_coupled_sat_vap} we examine
numerical solutions and parameter sensitivity for the saturation equation, vapor diffusion
equation, and the coupled saturation-vapor diffusion equations respectively.  In Section
\ref{sec:FullyCoupledSolutions} we compare numerical solutions to the fully coupled heat
and moisture transport model to the experimental data collected in \cite{Smits2011}.

In Chapters \ref{ch:Transport} - \ref{ch:Transport_Solution} we work toward building and
analyzing the saturation, vapor diffusion, and heat equations.  The flow of thought for
these chapters is to apply each set of new assumptions or simplifications to each of the
three equations before moving to the next set of assumptions.  That is, if a set of
assumptions are proposed then the subsequent sections will apply those assumptions to the
saturation, vapor diffusion, and heat equations in turn.  Only then will the next set of
assumptions be discussed. This is done so that each set of assumptions are only stated
once and since many of the assumptions create interleaving effects between the equations.

Finally, as an aid to the reader there are several appendices. Appendix
\ref{app:pore-scale_nomenclature} contains a nomenclature index for the pore-scale
diffusion processes considered in Chapter
\ref{ch:diffusion_comp}. Appendix \ref{app:nomenclature} contains a nomenclature index for
the macroscale results in the remaining chapters.  There is some overlap between the
nomenclature for these distinct parts, and effort has been made to not create any
excessive notational confusions (even though this work is necessarily notation heavy).
Appendix \ref{app:UpscaledDefinitions} gives a list (in alphabetical order) of the
upscaled definitions of variables defined in chapters \ref{ch:HMT} and \ref{ch:Exploit}.
As mentioned previously, Appendix \ref{app:AbstractEntropy} gives an abstract view of the
entropy inequality in an effort to make the exploitation process more clear to the
interested reader.  Appendix \ref{app:EntropyResults} gives a summary of the results
extracted from the entropy inequality in Chapter \ref{ch:Exploit}.  This is done for ease
of reference mostly on the author's part, but it is also done to provide an index of these
results for use in future research. Finally, Appendix \ref{app:dimensional_quantities}
gives several tables of dimensional quantities used throughout.  

It is suggested that the detail-oriented reader have Appendix
\ref{app:pore-scale_nomenclature} at hand when reading Chapter \ref{ch:diffusion_comp}
and Appendix \ref{app:nomenclature} at hand when reading chapters \ref{ch:HMT} -
\ref{ch:Transport}. There are some minor abuses of notation, but effort was made to bring
them to the reader's attention whenever possible and to use notation that didn't confuse
the immediate discussion.

\newpage
\chapter{Fick's Law and Microscale Advection Diffusion Models}\label{ch:diffusion_comp}
This chapter consists of a short technical note related to pore-scale diffusion problems.
Vapor diffusion in macroscale porous media is an important phenomenon with many
applications (e.g.\,evaporation from soils, moisture transport through filters, and CO$_2$
sequestration).  In order to better understand macroscale diffusion it behooves the
researcher to first understand pore-scale mechanics and models.  This chapter attemplts to
elucidate the models and assumptions used for diffusion at the pore-scale so that when we
turn our attention to macroscale diffusion we are firmly grounded. A secondary goal of
this chapter is to give a thorough discussion of the diffusion coefficient used in Fick's
law.  This is necessary since this coefficient is typically wrongly assumed constant for
all choices of dependent variables (mass concentration, molar concentration, chemical
potential, etc.). 

To make matters simpler, we focus our pore-scale discussion on the, so called, Stefan
diffusion tube problem.  This is a well-studied problem that models the diffusion of a
species through an ideal binary gas mixture above a liquid-gas interface
\cite{Benkhalifa1995,Bird2007,Camassel2005,Dormieux2006,Kerkhof1997,Whitaker1991a,Whitaker1991}.
This is an idealization of the juxtaposition of phases in a capillary tube geometry, and a
capillary tube geometry is an idealization of geometry of pore-scale porous media.  To
derive a mathematical model for the time evolution of the evaporating (or condensing)
species, one typically couples Fick's first law with the mass balance equation. 

In Section \ref{sec:FicksForms} we discuss the various forms of Fick's law and briefly
discuss the relationships between the diffusion coefficients.  In Section
\ref{sec:FicksTransientDerivation} we derive the transient diffusion equations associated
with Fick's law and compare with the associated equation of Bird, Stewart, and Lightfoot
\cite{Bird2007} (henceforth referred to as BSL).

\section{Comparison of Fick's Laws}\label{sec:FicksForms}
For a system consisting of an ideal mixture of water vapor,
$g_v$, and inert air, $g_a$, Fick's law can be written in terms of molar concentration,
mass concentration, or the chemical potential. This is potentially confusing since there
are inherently different diffusion coefficients for the different forms of Fick's law.
The purpose of this subsection is to clarify the relationships between these coefficients.
In porous media it is common to use mass flux for Fick's law, but in chemistry (and
related fields) it is more common to use molar flux. As such, we will make most of our
comparisons between mass and molar flux. 

According to BSL \cite{Bird2007},
the mass and molar forms of Fick's law are given by equations \eqref{eqn:Ficks_mass_flux}
and \eqref{eqn:Ficks_molar_flux}
(modified from BSL Table 17.8-2).  
In Table \ref{tab:FicksForms}, $\bv{g_j,g} = \bv{g_j} - \bv{g}$ is the diffusive velocity
relative to a mass weighted velocity, $\bv{g_j,g}_c = \bv{g_j} - \bv{g}_c$ is the
diffusive velocity relative to a mole weighted velocity, $\rho^{g_j}$ is the mass density
of species $j$, $C^{g_j} = \rho^{g_j} / \rho^g$ is the mass concentration of species $j$
in the mixture, $c^{g_j} = mol(g_j)/vol(g)$ is the molar density of species $j$, and
$x^{g_j} = c^{g_j}/c^g$ is the molar concentration of species $j$ in the mixture.
\def\vs{-1.1}
\begin{table}
    \centering
    \begin{tabular}{|c|c|m{7cm}|}
        \hline
        Flux Type & Flux Expression & \hspace{2.5cm} Fick's Law \\ \hline \hline 
        mass flux [$M L^{-2} T$] & $\foten{J}^{g_j}_\rho = \rho^{g_j} \bv{g_j,g}$ &
        \vspace{\vs cm} \begin{flalign} \foten{J}^{g_j}_\rho = - \rho^g D_\rho \grad
            C^{g_j} \label{eqn:Ficks_mass_flux} \end{flalign} \vspace{\vs cm} \\ \hline
        molar flux [$(mol) L^{-2} T$] & $\foten{J}^{g_j}_c = c^{g_j} \bv{g_j,g}_c$ &
        \vspace{\vs cm} \begin{flalign} \foten{J}^{g_j}_c = - c^g D_c \grad x^{g_j}
            \label{eqn:Ficks_molar_flux}
        \end{flalign} \vspace{\vs cm} \\ \hline 
        mass flux [$M L^{-2} T$] & $\foten{J}^{g_j}_{\mu,\rho} = \rho^{g_j} \bv{g_j,g}$ &
        \vspace{\vs cm} \begin{flalign} \foten{J}^{g_j}_{\mu,\rho} = - D_{\mu,\rho} \left( \frac{\rho^{g_j}}{R^{g_j} T}
            \right) \grad \mu^{g_j}_\rho \label{eqn:Ficks_mass_flux_chempot} \end{flalign} \vspace{\vs cm} \\ \hline
        mole flux [$(mol) L^{-2} T$] & $\foten{J}^{g_j}_{\mu,c} = c^{g_j} \bv{g_j,g}_c$ &
        \vspace{\vs cm} \begin{flalign} \foten{J}^{g_j}_{\mu,c} = - D_{\mu,c} \left( \frac{c^{g_j}}{R T}
            \right) \grad \mu^{g_j}_c \label{eqn:Ficks_molar_flux_chempot} \end{flalign} \vspace{\vs cm} \\ \hline
    \end{tabular}
    \caption{Mass and molar flux forms of Fick's law}
    \label{tab:FicksForms}
\end{table}

The chemical potential forms of Fick's law can be given in terms of two different types of
chemical potential: mass weighted (equation \eqref{eqn:Ficks_mass_flux_chempot}) or mole
weighted (equation \eqref{eqn:Ficks_molar_flux_chempot}).  In
physical chemistry and thermodynamics \cite{Callen1985,McQuarrie1997} the chemical
potential is known as the {\it tendency for a species to diffuse}, and for this reason it
is a natural candidate for the statement of Fick's law (an exact
thermodynamic definition will be presented in subsequent chapters).    In equations
\eqref{eqn:Ficks_mass_flux_chempot} and \eqref{eqn:Ficks_molar_flux_chempot}, $\mu^{g_j}_c$ is the mole weighted
chemical potential [$J (mol)^{-1}$] and $\mu^{g_j}_\rho$ is the mass weighted chemical
potential [$J M^{-1}$].

The reader should first note that the two fluxes are measured with respect to different
velocities.  The mass weighted velocity is $\rho^g \bv{g} = \sum_{j=v,a} \rho^{g_j} \bv{g_j}$ and the
mole weighted velocity is $c^g \bv{g}_c = \sum_{j=v,a} c^{g_j} \bv{g_j}$ where $\bv{g_j}$ is the
velocity of the species relative to a fixed coordinate system.  This means that that there
cannot be a direct comparison between the two different types of flux without considering
them relative to the same frame of reference. Using these definitions of $\bv{g}$ and
$\bv{g}_c$ we see that the difference between the two bulk velocities, $\bv{g} - \bv{g}_c$
is
\begin{flalign}
    \bv{g} - \bv{g}_c &= \sumj \left( \frac{\rho^{g_j} - x^{g_j}
\rho^g}{\rho^g} \right) \bv{g_j}. \label{eqn:diff_bulk_velocities}
\end{flalign}
In \eqref{eqn:diff_bulk_velocities}, the summation over $j$ indicates that this is an
accumulation over the $N$ species in the gas mixture.  In future work we will be
interested in the diffusion of water vapor ($j=v$) and will consider the gas mixture as
binary: $j \in \{ v,a \}$, where $j=a$ represents the mixture of all species that are not
water vapor. Therefore we can write \eqref{eqn:diff_bulk_velocities} as
\begin{flalign}
    \notag \bv{g} - \bv{g}_c &= \sum_{j=v,a} \left( \frac{\rho^{g_j} - x^{g_j}
\rho^g}{\rho^g} \right) \bv{g_j} \\
    &= \left( x^{g_a} C^{g_v} - x^{g_v} C^{g_a} \right) \bv{g_v} + \left( x^{g_v} C^{g_a}
    - x^{g_a} C^{g_v} \right) \bv{g_a}.
\end{flalign}
Converting to a mass weighted velocity we see that $\bv{g_j,g}_c = \bv{g_j,g}
+(\bv{g}-\bv{g}_c)$, and therefore the molar flux is $\foten{J}^{g_j}_c = c^{g_j}
\bv{g_j,g}_c = c^{g_j}\bv{g_j,g} + c^{g_j}(\bv{g}-\bv{g}_c)$.  Assuming that $c^{g_j}
x^{g_k} C^{g_l} \ll 1$ we see that the difference between the frame of reference is
potentially quite small.

Next note that the diffusion coefficients are (initially) assumed to be different for each
choice of independent variable as indicated by the subscripts.  To compare $D_\rho$ and $D_c$ we note that
$\foten{J}^{g_j}_\rho = m^{g_j}\foten{J}^{g_j}_c$
and
\begin{flalign}
    \grad C^{g_j} &= \left( \frac{m^{g_j} m^{g_k}}{\left( x^{g_j} m^{g_j} + x^{g_k}
    m^{g_k} \right)^2}
    \right) \grad x^{g_j} \label{eqn:mass_molar_diffusion_coeff_comparison_v1}
\end{flalign}
to conclude that 
\begin{flalign}
    \left( \frac{C^{g_k}}{x^{g_k}} \right) D_\rho = D_c,
    \label{eqn:mass_molar_diffusion_coeff_comparison}
\end{flalign}
where $m^{g_j}$ is the molar mass of species $j$ and the minuscule, $k$, represents the
{\it other} species. If the molar density form of the diffusion coefficient, $D_c$, is
assumed to be constant (at constant temperature) we conclude that the mass density version of
the diffusion coefficient is not constant (and visa versa). The fraction,
$C^{g_k}/x^{g_k}$ can be interpreted as the ratio of the molar mass of species $k$ to the
molar mass of the mixture.  If $j=v$ is the water vapor in an air-water mixture then $k=a$
is the inert air and the scaling factor between the diffusion coefficients is the ratio of
molar mass of the air to the molar mass of the mixture.  For sufficiently dilute systems
(where the amount of water vapor is small) the ratio is approximately 1 and the diffusion
coefficients can be considered as approximately equal.  

For ideal air-water mixtures, the densities are related through $\rho^g = \rho^{g_v} +
\rho^{g_a}$ and the water vapor density is related to the relative humidity through
$\rho^{g_v} = \rho_{sat}^{g_v} \rh$.  Here we are taking $\rho_{sat}^{g_v}$ as the
saturated vapor density and $\rh$ as the relative humidity.  At standard temperature and
pressure we note that  $\rho_{sat}^{g_v} \approx 0.02 kg/m^3$ and $\rho^g \approx 1
kg/m^3$.  This indicates that at standard temperature and pressure we can likely assume
that the mixture is always {\it sufficiently dilute}. Therefore, in the systems under
consideration we can assume that the diffusion coefficients are approximately equal.

For the diffusion coefficients associated with the chemical potential forms of Fick's law
we first observe that if we multiply and divide the right-hand side of the molar form by
the molar mass of species $j$ then
\begin{flalign}
    \foten{J}^{g_j}_{\mu,c} = - D_{\mu,c} \left( \frac{c^{g_j}}{R T} \right) \left(
    \frac{m^{g_j}}{m^{g_j}} \right) \grad \mu^{g_j}_c = - \left( \frac{D_{\mu,c}}{m^{g_j}}
    \right) \left( \frac{\rho^{g_j}}{R^{g_j} T} \right) \grad \mu^{g_j}.
    \label{eqn:flux_mu_c}
\end{flalign}
Here, $R^{g_j} = R/m^{g_j}$ is the specific gas constant, and we have used $\mu^{g_j} =
\mu^{g_j}_c / m^{g_j}$. 
Again noting that $\foten{J}^{g_j}_{\mu,\rho} = m^{g_j} \foten{J}^{g_j}_{\mu,c}$ we
conclude that $D_{\mu,c} = D_{\mu,\rho}$.

It remains to compare $D_\rho$ to $D_{\mu,\rho}$ and $D_c$ to $D_{\mu,c}$. We focus here
on the mass fluxes without loss of generality. If the mass fluxes are equal, then in
particular
\[ \rho^g D_\rho \grad C^{g_v} = D_\mu \left( \frac{\rho^{g_v}}{R^{g_v} T} \right) \grad
    \mu^{g_v}. \]
Rearranging, it can be seen that 
\[ D_\rho \grad C^{g_v} = C^{g_v} D_\mu \grad \left( \frac{\mu^{g_v}}{R^{g_v} T}
    \right) \]
(assuming constant temperature).  From physical chemistry \cite{McQuarrie1997}, 
recall that the chemical potential is related to a reference chemical potential
($\mu^{g_v}_*$) and the ratio of partial pressure, $p^{g_v}$, to bulk pressure, $p^g$, via
\begin{flalign}
    \mu^{g_v} = \mu^{g_v}_* + R^{g_v} T \ln\left( \frac{p^{g_v}}{p^g} \right).
    \label{eqn:chempot_pore_defn}
\end{flalign}
Therefore, 
\begin{flalign}
    \grad \left( \frac{\mu^{g_v}}{R^{g_v} T} \right) = \grad \left( \ln \left(
    \frac{p^{g_v}}{p^g} \right) \right) = \left( \frac{p^g}{p^{g_v}} \right) \grad \left(
    \frac{p^{g_v}}{p^g} \right).
\end{flalign}
Using Dalton's law for ideal gases, $p^g = p^{g_v} + p^{g_a}$, and using the specific gas
constants we note that the partial pressure of species $j$ can be written as $p^{g_j} =
R^{g_j} T \rho^{g_j}$.  Therefore,
\begin{flalign}
    p^g = R^{g_v} T \rho^{g_v} + R^{g_a} T \rho^{g_a},
\end{flalign}
and, after simplifying,
\begin{flalign}
    \frac{p^{g_v}}{p^g} = \frac{\rho^{g_v}}{\rho^{g_v} + \left( \frac{R^{g_a}}{R^{g_v}}
    \right) \rho^{g_a}}.
    \label{eqn:pgv_over_pg}
\end{flalign}
From the values found in Appendix \ref{app:dimensional_quantities} we see that $R^{g_a} /
R^{g_v} \approx 0.6$ and therefore equation \eqref{eqn:pgv_over_pg} is similar, but not
equal to, the mass concentration, $C^{g_v} = \rho^{g_v} / (\rho^{g_v} + \rho^{g_a})$.
Defining $\mathcal{C}^{g_v} = p^{g_v} / p^g$ we see that 
\begin{flalign}
    \frac{D_\rho}{D_\mu} = \left( \frac{C^{g_v}}{\mathcal{C}^{g_v}} \right) \frac{\grad
        \mathcal{C}^{g_v}}{\grad C^{g_v}},
        \label{eqn:Drho_Dmu}
\end{flalign}
where division is understood component wise (that is, equation \eqref{eqn:Drho_Dmu}
represents three equations when the gradient is understood in three spatial dimensions).
The right-hand side of equation \eqref{eqn:Drho_Dmu} is not constant at 1 for all
densities, but the variation in the right-hand side depends mostly on the variation in
$\rho^{g_a}$ in the gas mixture.  Fortunately, the water vapor density is much smaller
than the air-species density, and hence $\rho^{g_a}$ is approximately constant.  In one
spatial dimension, the right-hand side of \eqref{eqn:Drho_Dmu} can therefore be
approximated by
\[ d(x) := \frac{\left( \frac{x}{x+1} \right) }{\left( \frac{x}{x+0.6} \right) }
    \frac{\frac{d}{dx} \left( \frac{x}{x+0.6} \right)}{\frac{d}{dx} \left( \frac{x}{x+1}
\right)} \]
(where $x = \rho^{g_v}$ and $\rho^{g_a} \approx 1$).  It is easy to show that $0.98 < d(x)
< 1$ for $0 \le x \le \rho_{sat}^{g_v}$.  Furthermore, for {\it sufficiently dilute}
mixtures, $d(x) \approx 1$, and we therefore conclude that $D_\rho \approx D_\mu$.

The conclusion from this subsection is that while the diffusion coefficients for the molar and
mass flux forms of Fick's law are not the same, for dilute mixtures they can be
approximated as equal.

\section{Transient Diffusion Models}\label{sec:FicksTransientDerivation}
In porous media there is a phenomenon known as {\it enhanced vapor diffusion}
\cite{deVries1957}.  This phenomenon states that vapor diffusion in porous media occurs
{\it faster} than as predicted by Fickian diffusion models. This is merely a statement
about the observed imbalance between Fickian diffusion and experimental measure.  Since
the ultimate goal of this work is to develop macroscale advection diffusion models, we
seek to understand the pore-scale diffusion models so that in subsequent chapters we can
tackle the enhanced diffusion problem.  

To build a transient model for molecular diffusion we couple Fick's law with the
appropriate form of the mass balance equation.  In the previous subsection we showed that
the various forms of Fick's law are approximately equal (for sufficiently dilute
mixtures), so the results stated here will only be in terms of the mass flux form of
Fick's law (equation \eqref{eqn:Ficks_mass_flux}).
 
The mass balance equation for species $j$ in the gas phase can be written as
\begin{flalign}
    \pd{\rho^{g_j}}{t} + \diver \left( \rho^{g_j} \bv{g_j} \right) =
    \hat{r}^{g_j}
    \label{eqn:mass_balance}
\end{flalign}
where $\bv{g_j}$ is the velocity of species $j$ within the gas mixture relative to a fixed
frame of reference, and $\hat{r}^{g_j}$ is a mass exchange term accounting for chemical
reactions between species \cite{weinstein,Whitaker1991}.  In the present work we assume that
no chemical reactions occur, and therefore $\hat{r}^{g_j} = 0$. The combination of the
mass balance equation with the mass flux form of Fick's law (for $j=v$) gives
a transport equation for the mass of water vapor via advective, $\rho^{g_v} \bv{g}$, and
diffusive, $\foten{J}^{g_v}$, fluxes: 
\begin{flalign}
    \pd{\rho^{g_v}}{t} + \diver \left( \foten{J}^{g_v} + \rho^{g_v} \bv{g} \right) = 0.
\end{flalign}
Substituting the mass flux form of Fick's law\footnote{The
    subscripts on the flux and the diffusion coefficient have been dropped since all of
    the versions presented in Table \ref{tab:FicksForms} are approximately equal} (from equation
\eqref{eqn:Ficks_mass_flux}) we get 
\begin{flalign}
    \pd{\rho^{g_v}}{t} + \diver \left(
    \rho^{g_v} \bv{g} \right) =D \diver \left( \rho^g \grad C^{g_v} \right).
    \label{eqn:pore_mass_balance1}
\end{flalign}
Notice here that the diffusion coefficient has been factored out of the divergence
operator.  This is only valid in constant temperature environments. If the gas-phase
density were constant in space then we would arrive at the traditional advection
diffusion equation (by dividing equation \eqref{eqn:pore_mass_balance1} by $\rho^g$ of by
rewriting the diffusion term as $D \diver \grad \rho^{g_v}$) and would need an expression
for the bulk velocity in terms of density (or concentration) to close the equation.
Unfortunately, if the density of the water vapor is allowed to vary then the density of
the gas varies.  Again, for {\it sufficiently dilute} mixtures the variation in gas-phase
density is very small and the nonlinear diffusion on the right-hand side can be
approximated by the linear diffusion term $D \diver \grad \rho^{g_v}$. It should be noted
here that this later case is what is typically thought of as ``Fick's law'' and is what
leads to the traditional linear diffusion equation (when the advection term is neglected)
\cite{Crank1979}.

A different form of equation \eqref{eqn:pore_mass_balance1}, suggested in BSL
\cite{Bird2007}, is derived by considering the mass weighted bulk velocity.  In a binary
system, 
\begin{flalign*}
    \rho^{g_v} \bv{g_v} = \rho^{g_v} \bv{g_v,g} + \rho^{g_v} \bv{g} 
    = \rho^{g_v} \bv{g_v,g} + \rho^{g_v} \left( C^{g_v} \bv{g_v} + C^{g_a}
    \bv{g_a}\right).
\end{flalign*}
Solving for $\rho^{g_v} \bv{g_v}$
\begin{flalign}
    \rho^{g_v} \bv{g_v} &= \left( \frac{\rho^{g_v}}{1-C^{g_v}}
    \right) \bv{g_v,g} + \rho^{g_v} \bv{g_a}.
    \label{eqn:bird_mass_flux}
\end{flalign}
Using Fick's law for $\rho^{g_v} \bv{g_v,g}$, and eliminating $\rho^{g_v} \bv{g_v}$ in
(\ref{eqn:mass_balance}) with (\ref{eqn:bird_mass_flux}) gives
\begin{flalign}
    \pd{\rho^{g_v}}{t} + \diver \left( \rho^{g_v} \bv{g_a} \right) = \diver
    \left( \frac{D \rho^g}{1-C^{g_v}} \grad C^{g_v} \right).
    \label{eqn:Bird_concentration}
\end{flalign}
Whitaker \cite{Whitaker1991} suggested that ``one can develop convincing
arguments in favor of \dots'' neglecting the air-species flux term.  Certainly at steady
state we can assume (as is done in BSL) that $\bv{g_a}$ is approximately zero at the
interface since ``there is no net motion of [water vapor] away from the interface''
\cite{Bird2007}, but in the transient case this would constitute a change of frame of
reference. This new frame of reference would be such that the inert air molecules are
viewed as stationary with the water vapor diffusing through them. 

In either equation \eqref{eqn:pore_mass_balance1} or \eqref{eqn:Bird_concentration} one
must find appropriate conditions or equations to either neglect or rewrite the advective
term, $\rho^{g_v} \bv{g}$ or $\rho^{g_v} \bv{g_a}$ respectively. Typically this term is
neglected in a pure diffusion problem. As these are two different simplifications of the
same equation one must have different reasons for neglecting the advective term. The
easiest fix for this issue is to couple with either the bulk gas mass balance equation 
or the air-species mass balance equation 
and to use the mass-weighted velocity: $\rho^g \bv{g} = \sumj \rho^{g_j} \bv{g_j}$. The
point being that one cannot simply neglect the advection term in the transient case of
either equation without proper consideration of the implications: a fixed bulk velocity or
a changing frame of reference respectively.

A final comment can be make regarding equations  \eqref{eqn:pore_mass_balance1} and
\eqref{eqn:Bird_concentration}.  The bulk density term, $\rho^g$, on the right-hand side
of these equations is often factored out of the divergence operator.  This is an error
committed by several researchers \cite{Bird2007,Camassel2005,Kerkhof1997}.  The reasoning
for assuming that the density is constant (and hence returning to a linear diffusion model
in the absence of advection) is that in an ideal gas, $p^{g} = \rho^{g} R^{g} T$.  Under
constant temperature conditions, and if the pressure is assumed constant, then the density
is assumed be constant.  There are two possible mistakes here. (1): If the species
densities are allowed to vary then the bulk density must vary. (2): The value of $R^g$
will vary with the changing composition of the mixture (since the molar mass of the
mixture changes). The effect of this is that, while the pressure may remain constant, the
component parts are not necessarily constant and therefore cannot be factored from the
divergence operator.

\section{Conclusion}
In this chapter we have compared various forms of Fick's law for molecular
diffusion.  We have shown that, while the diffusion coefficients are indeed different,
under certain common circumstances the diffusion coefficients can be considered as
approximately equal.  It is common to take the diffusion coefficient as constant (or 
only a function of temperature), and in many cases it is safe to assume the
same diffusion coefficient may be used in the common forms of Fick's law.  

In the transient
case there are two natural formulations for (the mass flux form of) Fick's second law.  In
either case, the natural governing equation is a nonlinear advection diffusion equation
that must be closed with the use of another mass balance equation. When considering the
advection term, it is the author's opinion that equation \eqref{eqn:pore_mass_balance1} is
the more natural choice.  The reason for this is that the bulk velocity, $\bv{g}$, is
likely more
naturally measured as compared to that of the species velocity. This chapter concludes our
discussion on pore-scale modeling.  We now turn our attention to building macroscale
models, but in doing so we keep in mind the diffusion models at the pore scale and use
cues from this scale to help make proper assumptions about the larger scale.

\newpage
\chapter{Hybrid Mixture Theory}\label{ch:HMT}
In this chapter we use a combination of classical mixture theory and rational
thermodynamics (henceforth called Hybrid Mixture Theory (HMT)) to study novel extensions
to Darcy's law, Fick's law, and Fourier's law in variably saturated porous media. This
approach was pioneered by Hassanizadeh and Gray in the 70's and 80's
\cite{Gray1977,Gray1979,Hassanizadeh1986,Hassanizadeh1986a} and later extended by
Bennethum, Cushman, Gray, Hassanizadeh, and many others
\cite{Cushman1985,Cushman2002,Gray1998,weinstein} to model multi-phase, multi-component,
and multi-scale media.  HMT involves volume averaging, or upscaling, pore-scale balance
laws to obtain macroscale analogues.  The second law of thermodynamics is then used to
derive constitutive restrictions on these macroscale balance laws.  Constitutive relations
are particular to the medium being studied, and hence depend on a judicious choice of
independent variables for the energy of each phase in the medium.  There are many
excellent resources for the curious reader to gain a more thorough understanding of HMT
(eg \cite{Cushman2002,weinstein}).  For that reason we will not derive every identity
along the way.  Instead partial derivations of the identities necessary to understand the
present application of HMT are presented.

To begin this overview we consider the upscaling of pore-scale balance laws (conservation
laws) via a mixture theoretic approach. The subsequent sections in this and the next
chapter introduce the entropy inequality and it's exploitation to derive constitutive
laws.  A judicious choice of independent variables for the energy of each phase in the medium is chosen and is used
to derive novel versions of Darcy's, Fick's, and Fourier's laws.  These constitutive
equations will be used in subsequent chapters to develop models for moisture transport in
variably saturated porous media.


\section{The Averaging Procedure}
When considering a porous medium one cannot avoid discussing the various scales involved.
This particular work deals with two principal scales: the microscale and the macroscale.
At the microscale the phases are separate and distinguishable. Typical microscale porous
media will have pores that measure on the order of microns to millimeters (depending on
the type of solid).  At the macroscale the phases are indistinguishable and the typical
measurements range from millimeters to meters.  The macroscale is where most physical
measurements are made, and as such, we seek to derive governing equations that hold at
this scale. The microscale structure may vary dramatically for different media depending
on the type of solid phase and the microscale behavior of the fluid phases. As such, the
microscale geometry can have a dramatic influence on flow and phase interaction.

For any given
phase at the pore scale the mass, linear momentum, angular momentum, and energy balance
laws must hold.  The problem is that it is difficult to obtain geometric information
everywhere at this scale For this reason we seek to average (or {\it upscale}) the microscale balance laws to the macroscale.

There are many methods for mathematically averaging balance laws.  Here we choose the
simplest method of weighted integration.  Before introducing the technical details of the
weighted integration we must first introduce the concept of a Representative Elementary
Volume and local geometry in a porous medium. This elementary volume will become our basic
unit of volume throughout this research. The following discussions closely follow and
paraphrase those of Bear \cite{Bear1988}, Bennethum \cite{Bennethum1996a,Bennethum1996b},
Hassanizadeh and Gray \cite{Gray1977,Gray1979}, Weinstein \cite{weinstein}, and
Wojciechowski \cite{Wojciechowski2011}.


\subsection{The REV and Averaging}
In this work we consider unsaturated porous media. Characteristic to these media is the
juxtaposition of liquid, solid, and gas phases within the pore matrix.  We make the
assumption that a representative elementary volume (REV), in the sense of Bear
\cite{Bear1988}, exists at every point in space.  To properly define the REV we first
define the porosity. 

Consider a sequence of small volumes within a porous medium, $(\delta
V)_k$, each with centroid $\foten{x} \in \mathbb{R}^3$.  For each $k$, let
$(\delta V_{void})_k$ be the volume of the void space within $(\delta
V)_k$.  The porosity for the $k^{th}$ volume is given as the ratio
\begin{flalign}
    \phi_k &= \frac{(\delta V_{void})_k}{(\delta V)_k}.
    \label{eqn:porosity_defn}
\end{flalign}
Generate the sequence, $\{ \phi_k \}$, by gradually shrinking $(\delta
V)_k$ about $\foten{x}$ such that $(\delta V)_1 > (\delta V)_2 > (\delta
V)_3 > \cdots$. As $k$ increases, the porosity will certainly fluctuate due
to heterogeneities in the medium.  As $(\delta V)_k$ shrinks there will be
a certain value, $k=k^*$, such that for $k>k^*$ the fluctuations in
porosity become small and are only due to fluctuations in the arrangement
of the solid matrix.  If $(\delta V)_k$ is reduced well beyond $(\delta V)_{k^*}$ the
sequence of volumes will eventually converge to $\foten{x}$.  The point,
$\foten{x}$ only lies within one phase, so the limit of the sequence of
porosities will either be 0 or 1 (completely in the void space or
completely in the solid).  This indicates that there will be some other
intermediate volume, $(\delta V)_{k^{**}} < (\delta V)_{k^*}$, where the
sequence of porosities begins to fluctuate again as $k$ gets larger.  We
define the REV, $\delta V$, as any particular volume $(\delta V)_{k^{**}} <
\delta V < (\delta V)_{k^*}$. Without loss of generality we can simply
choose $\delta V \equiv (\delta V)_{k^{**}}$.  
\Cref{fig:REV_figure}
illustrates two typical sequences of porosities, $\phi_k$, as the volume is
decreased (right to left).

\linespread{1.0}
\begin{figure}[ht]
    \begin{center}
        \begin{tikzpicture}
            \draw[thick,->] (-0.1,0) -- (6,0);
            \draw[thick,->] (0,-0.1) -- (0,3);
            \draw (6,0) node[anchor=west]{Volume};
            \draw (-0.1,1.5) node[anchor=east]{Porosity};
            \draw (0,0) node[anchor=north east]{$0$};
            \draw (0,2.5) node[anchor=east]{$1$};
            \draw (-0.1,2.5) -- (0.1,2.5);
            \draw[dashed] (2.5,-0.1)  -- (2.5,2.75);
            \draw[->] (2,-0.5) node[anchor=north]{\small{$k=k^{**}$}} --
            (2.45,-0.05);
            \draw[dashed] (5.75,-0.1) -- (5.75,2.75);
            \draw[->] (5.25,-0.5) node[anchor=north]{\small{$k=k^{*}$}} --
            (5.7,-0.05);
            \draw[thick,color=red,smooth,domain=0:6] plot
            (\x,{exp(-1.5*\x)*cos(550*\x-0.7)+1.5});
            \draw[thick,color=blue,smooth,domain=0:6] plot
            (\x,{-1*exp(-1.3*\x)*cos(430*\x)+1.5});
            \draw[thick,->] (3.5,2) node[anchor=south]{REV} -- (2.5,1.5);
            \draw (0.5,0.7) node[anchor=west]{\tiny{Domain of}};
            \draw (0.5,0.5) node[anchor=west]{\tiny{heterogeneity}};
            \draw (3,0.7) node[anchor=west]{\tiny{Domain of}};
            \draw (3,0.5) node[anchor=west]{\tiny{homogeneity}};
        \end{tikzpicture}
    \end{center}
    \caption{Illustration of the definition of the REV via a sequence of
    porosities corresponding to a sequence of shrinking volumes. (Image similar to Figures
    1.3.1 and 1.3.2 in Bear \cite{Bear1988})}
    \label{fig:REV_figure}
\end{figure}
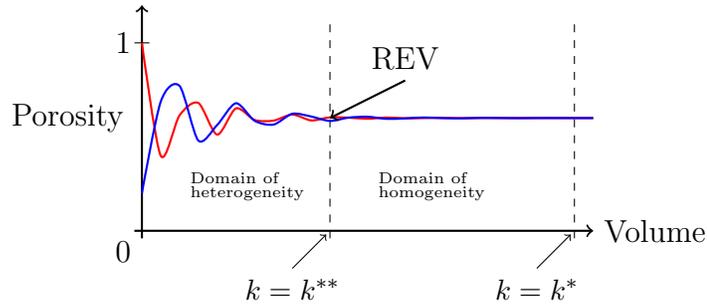
\renewcommand{\baselinestretch}{\normalspace}

\linespread{1.0}
\begin{figure}[h!]
    \begin{center}
        \includegraphics[trim=2.5cm 0cm 0cm 0cm,width=0.95\columnwidth]{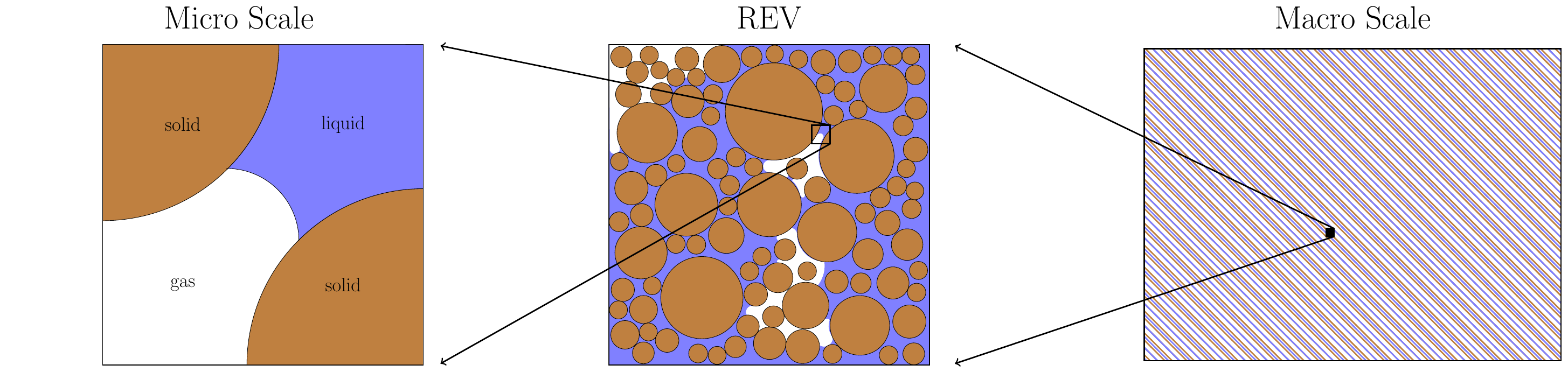}
    \end{center}
    \caption{Cartoon of the microscale, REV, and macroscale in a granular soil. The
    right-hand plot depicts the mixture of all phases.}
    \label{fig:MicroMacro}
\end{figure}
\renewcommand{\baselinestretch}{\normalspace}

Consider now a coordinate system superimposed on the porous medium.  Let $\foten{x}$ be
the centroid of the REV, and let $\foten{r}$ be some other vector inside the REV.  Define
the vector, $\boldsymbol{\xi}$, as a vector originating from the centroid of the REV such
that 
\begin{flalign}
    \foten{r} = \foten{x} + \boldsymbol{\xi}.
    \label{eqn:local_spatial_coord}
\end{flalign}
We can now view $\boldsymbol{\xi}$ as a local coordinate in the REV as in Figure
\ref{fig:REV_local_coordinates}.

\linespread{1.0}
\begin{figure}[ht]
    \begin{center}
        \begin{tikzpicture}
            \draw[->] (-0.1,0) -- (3.5,0);
            \draw[->] (0,-0.1) -- (0,3.5);
            \draw[->] (0,0) -- (-2,-1);
            \draw (2,1.6) circle(1cm);
            \draw[->] (0,0) -- (2,1.6);
            \draw[->] (2,1.6) -- (2.7,1.2);
            \draw[->] (0,0) -- (2.7,1.2);
            \draw (1.8,2) node{$\delta V$};
            \draw (0.9,1) node{$\foten{x}$};
            \draw (1.2,0.4) node{$\foten{r}$};
            \draw (2.3,1.6) node{$\boldsymbol{\xi}$};
        \end{tikzpicture}
    \end{center}
    \caption{Local coordinates in and REV.}
    \label{fig:REV_local_coordinates}
\end{figure}
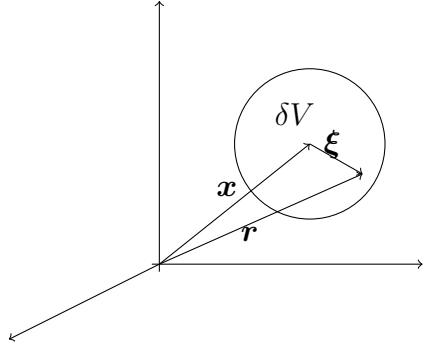
\renewcommand{\baselinestretch}{\normalspace}

Define the {\it phase indicator function} as 
\begin{flalign}
    \gamma_\al(\foten{r},t) = \left\{ \begin{array}{ll} 1, & \foten{r} \in
        \al\text{-phase} \\ 0, & \foten{r} \not \in \al\text{-phase}
    \end{array} \right.,
    \label{eqn:phase_indicator}
\end{flalign}
where $\foten{r}$ is a position vector as indicated in Figure
\ref{fig:REV_local_coordinates}. The averaging technique involved multiplying a micro
scale quantity (such as density) by $\gamma_\al$ and integrate over the REV.  This effectively smears out
the phases.  A consequence of this is that the averaged value may not accurately represent
the actual values being measured at the pore scale.  A further mathematical complication
arises since the integrations may not make sense in the traditional (Riemannian) sense.
Therefore, we must understand all of the following mathematics in the distributional sense
(integrals are understood to be Lebesgue, and derivatives are understood to be generalized
derivatives).  For more specifics on these mathematical tools see standard graduate texts
on functional analysis (eg.  \cite{Oden2009}).

To find the volume of the $\al$ phase in the REV we simply integrate
$\gamma_\al$ over the REV.  Define this volume as $|\delta V_\al|$:
\begin{flalign}
    |\delta V_\al| &= \int_{\delta V} \gamma_\al \left( \foten{r},t
    \right) dv = \int_{\delta V} \gamma_\al \left( \foten{x} + \boldsymbol{\xi}, t
    \right) dv(\boldsymbol{\xi}).
    \label{eqn:deltava}
\end{flalign}
The $\al$-phase volume fraction, $\eps{\al}$, is defined as \footnote{Note: the notation
``$\epsa$'' for the volume fraction is not necessarily standard. Some authors use
``$\phi^\al$'', ``$\theta^\al$'', or ``$n^\al$''. Furthermore, the
superscript notation is sometimes replaced by subscripts. The present notation is chosen
to be consistent with the primary references for Hybrid Mixture Theory mentioned in the
introduction to this chapter.}
\begin{flalign}
    \epsa\left( \foten{x},t \right) &= \frac{|\delta V_\al|}{|\delta V|}.
    \label{eqn:defn_vol_frac}
\end{flalign}
Since $0 \le |\delta V_\al| \le |\delta V|$ it is clear that $0 \le \epsa
\le 1$.  Furthermore, since the REV is made up of all of the phases, 
\begin{flalign}
    \suma \epsa = 1.
    \label{eqn:vol_frac_1}
\end{flalign}
The volume fraction is the first example of a macroscale variable. That is, it is a
variable that describes a pore-scale property but is upscaled to the larger, more
measurable, scale.

It is useful to note that there are two main types of averaging that will be used: mass
averaging and volume averaging \cite{Gray1979,Gray1998,weinstein}.  Let $\psi^j$ be the
$j^{th}$ constituent of some quantity of interest. To volume average $\psi^j$ we define
\begin{flalign}
    \left< \psi^j \right>^{\al} &= \frac{1}{|\delta V_\al|} \int_{\delta V}
    \psi^j(\foten{r},t) \gamma_\al(\foten{r},t) dv(\boldsymbol{\xi}),
    \label{eqn:general_vol_average}
\end{flalign}
and to mass average $\psi^j$ we define
\begin{flalign}
    \overline{\psi^j}^{\al} &= \frac{1}{\left< \rho^j \right>^{\al} |\delta
    V_\al|} \int_{\delta V} \rho^j(\foten{r},t) \psi^j(\foten{r},t)
    \gamma_\al(\foten{r},t) dv(\boldsymbol{\xi}).
    \label{eqn:general_mass_average}
\end{flalign}
Implicitly in (\ref{eqn:general_mass_average}) we see that
density is volume averaged. That is,
\begin{flalign}
    \left< \rho^j \right>^{\al} &= \frac{1}{|\delta V_\al|} \int_{\delta V}
    \rho^j(\foten{r},t) \gamma_\al(\foten{r},t) dv(\boldsymbol{\xi}).
    \label{eqn:density_volume_averaged}
\end{flalign}
(Note: some authors use the mass averaged notation on density even though
it is technically volume averaged.  Given the definition of a mass averaged
quantity there is usually little confusion.)

The basic rules of thumb for deciding whether to volume or mass average
were originally proposed by Hassanizadeh and Grey in 1979 \cite{Gray1979}.
They propose four criteria, listed below, for making this decision. In these criteria it
is emphasized that the microscale quantities correspond to small scale pre-averaged
quantities, while macroscale quantities are defined via the averaging process.
\begin{enumerate}
    \item ``When an averaging operation involves integration, the integrand
        multiplied by the infinitesimal element of integration must be an
        additive quantity.  For example, the internal energy density
        function, $E$, is not additive, but the total internal energy, $\rho E
        dv$ is additive and an average defined in terms of this quantity
        will be physically meaningful.''
    \item ``The macroscopic quantities should exactly account for the total
        corresponding microscopic quantity. For example, total macroscopic
        momentum fluxes through a given boundary must be equal to the total
        microscopic momentum fluxes through that boundary.''
    \item ``The primitive concept of a physical quantity, as first
        introduced into the classical continuum mechanics must be preserved by
        proper definition of the macroscopic quantity.  For instance, heat
        is a mode of transfer of energy through a boundary different from
        work.  The definition of macroscopic heat flux must also be a mode
        of energy transfer different from macroscopic work.''
    \item ``The averaged value of a microscopic quantity must be the same
        function that is most widely observed and measured in a field
        situation or in laboratory practice.  For example, velocities
        measured in the field are usually mass averaged quantities;
        therefore, the macroscopic velocity should be a mass averaged
        quantity.''   This ensures applicability of the resulting
        equations.
\end{enumerate}

In the upscaling procedure to follow we wish to apply a weighted
integration to a pore-scale balance law (a partial differential equation).
This will involve terms such as 
\[ \int_{\delta V} \pd{\rho^j}{t} \gamma_\al dv, \]
and to ensure that we properly define the macroscale variables as either
volume or mass averaged quantities, we need a theorem that allows for the
interchange of integration and differentiation. This theorem is due to Whitaker and
Slattery \cite{Slatery1967,Whitaker1967,Whitaker1969} and a generalization of this theorem is due to
Cushman \cite{Cushman1985}.

\begin{theorem}[Averaging Theorem]
    If $\foten{w}_{\al \beta}$ is the microscopic velocity of the $\al
    \beta$ interface and $\foten{n}^{\al}$ is the outward unit normal
    vector of $\delta V_\al$ indicating that the integrand should be
    evaluated in the limit as the $\al \beta$ interface is approached from
    the $\al$ side, then
    \begin{subequations}
        \begin{flalign}
            \notag & \frac{1}{|\delta V|} \int_{\delta V} \pd{f}{t}
            \gamma_\al dv(\boldsymbol{\xi}) \\
            & \qquad = \pd{ }{t} \left[ \frac{1}{|\delta V|} \int_{\delta
            V} f \gamma_\al dv(\boldsymbol{\xi}) \right] - \sumba
            \frac{1}{|\delta V|} \int_{A_{\al \beta}} f \foten{w}_{\al
            \beta} \cd \foten{n}^\al da 
            \label{eqn:averaging_theorem_time_deriv} \\
            \notag & \frac{1}{|\delta V|} \int_{\delta V} \grad f
            \gamma_\al dv(\boldsymbol{\xi}) \\
            & \qquad = \grad \left[ \frac{1}{|\delta V|} \int_{\delta
            V} f \gamma_\al dv(\boldsymbol{\xi}) \right] + \sumba
            \frac{1}{|\delta V|} \int_{A_{\al \beta}} f \foten{n}^\al da, 
            \label{eqn:averaging_theorem_grad} 
        \end{flalign}
        \label{eqn:averaging_theorem_equations}
    \end{subequations}
    \label{thm:averaging_theorem}
where $f$ is the quantity to be averaged.  
\end{theorem}
Keep in mind that $f$ could be a scalar or a vector quantity.  In the latter case, the
symbol $f$ is replaced with $\boldsymbol{f}$ and appropriate tensor contractions are
inserted. 

The averaging procedure is now carried out in the following steps:
\begin{enumerate}
    \item State the pore-scale balance law for a particular species (or
        phase).
    \item Multiply the equation by $\gamma_\al$.
    \item Average each term over the REV.
    \item Apply Theorem \ref{thm:averaging_theorem} to arrive at terms
        representing macroscale quantities.
    \item Define physically meaningful macroscopic quantities.
\end{enumerate}

We now turn our attention to averaging pore-scale balance laws in the sense
listed above. In the following discussion, 
$\bv{j}$ is the microscopic
velocity of constituent $j$, $\foten{w}_{\al \beta_j}$ is the velocity of
the $j^{th}$ constituent in the $\al \beta$ interface, and
$\foten{n}^{\al}$ is the outward unit normal vector of $\delta V_\al$. A full nomenclature
index can be found in Appendix \ref{app:nomenclature}.


\section{Macroscale Balance Laws}
As it is the simplest balance law, let us first consider the mass balance equation for a
single constituent:
\begin{flalign}
    \md{ }{\rho^j} + \rho^j \diver \bv{j} &= \rho^j \hat{r}^j.
    \label{eqn:pore_scale_mass_balance}
\end{flalign}
Here, $\rho^j$ is the density of the constituent, $\bv{j}$ is the velocity of the
constituent, and any source of mass from chemical reactions between the constituents is
given as $\hat{r}^j$.  Recall that the material (Lagrangian) derivative is 
\begin{flalign}
    \md{j}{(\cd)} &= \pd{(\cd)}{t} + \bv{j} \cd \grad (\cd)
\end{flalign}
This derivative contains the usual Eulerian derivative along with an advective term.
Written in terms of the Eulerian time derivative, (\ref{eqn:pore_scale_mass_balance}) is
\begin{flalign}
    \pd{\rho^j}{t} + \diver \left( \rho^j \bv{j} \right) &= \rho^j
    \hat{r}^j.
    \label{eqn:pore_scale_mass_balance_eulerian}
\end{flalign}
While this is specifically the mass balance equation, it takes the prototypical form of
all balance laws: a time derivative plus a flux is equal to any source.

The constituent momentum balance can be written in a similar manner:
\begin{flalign}
    \pd{\left( \rho^j \bv{j} \right)}{t} + \diver \left( \rho^j \bv{j} \otimes
    \bv{j} - \stress{j}\right) &= \rho^j \left( \foten{g} +
    \ihat^j + \hat{r}^j \bv{j} \right).
    \label{eqn:pore_scale_mom_balance_eulerian}
\end{flalign}
Here, $\stress{j}$ is the Cauchy stress tensor on species $j$, and the sources on the
right-hand side are gravity, momentum transfer from other constituents, and momentum
gained from chemical reactions respectively.  These equations describe the change in mass
and momentum over time and space within a specific constituent.  They are sufficient field
equations for modeling systems composed of a single phase gas, liquid, or solid, but
equations (\ref{eqn:pore_scale_mass_balance_eulerian}) and
(\ref{eqn:pore_scale_mom_balance_eulerian}) are insufficient for modeling multiphase and
multiconstituent systems as a mixture because the interactions between the phases and
constituents are not present.

In this work we consider a porous medium consisting of a solid phase and two fluid phases
with multiple constituents within each phase. The phases will be denoted as $\al = l, g,$
and $s$ for liquid, gas, and solid respectively.  The constituents will be enumerated
$j=1:N$ (using MATLAB-style notation to indicate $j=1,2,\dots,N$).  The following
derivations follow similar derivations given by Gray \cite{Gray1979}, Weinstein
\cite{weinstein} and Wojciechowski \cite{Wojciechowski2011}. 


\subsection{Macroscale Mass Balance}
To obtain macroscale equations in multiphase and multi-constituent media we multiply a
constituent balance equations by the phase indicator function, $\gamma_\al$, integrate
over $\delta V$, and divide by $|\delta V|$.  Applying the averaging theorem
(\ref{thm:averaging_theorem}) to the appropriate terms in equation
(\ref{eqn:pore_scale_mass_balance_eulerian}) we have
\begin{subequations}
    \begin{flalign}
        \notag \frac{1}{|\delta V|} \int_{\delta V} \pd{\rho^j}{t} \gamma_\al dv &= \pd{
        }{t} \left[ \frac{|\delta V_\al|}{|\delta V||\delta V_\al|} \int_{\delta V} \rho^j
            \gamma_\al dv \right] \\
        \notag & \quad - \sumba \frac{1}{|\delta V|} \int_{\delta A_{\al
            \beta}} \rho^j \foten{w}_{\al \beta_j} \cd \foten{n}^{\al} da \\
        &= \pd{ }{t} \left( \eps{\al} \overline{\rho^j}^{\al} \right) - \sumba
        \frac{1}{|\delta V|} \int_{\delta A_{\al \beta}} \rho^j \foten{w}_{\al \beta_j}
        \cd \foten{n}^{\al} da, \\
        \notag \frac{1}{|\delta V|} \int_{\delta V} \diver \left( \rho^j \bv{j} \right)
        \gamma_\al dv &= \diver \left[ \frac{|\delta V_\al|}{|\delta V||\delta V_\al|} \int_{\delta V} \rho^j \bv{j}
            \gamma_\al dv \right] \\
        \notag & \quad + \sumba \frac{1}{|\delta V|} \int_{\delta A_{\al
            \beta}} \rho^j \bv{j} \cd \foten{n}^\al da \\
        &= \diver \left( \epsa \overline{\rho^j}^{\al} \overline{\bv{j}}^{\al} \right) +
        \sumba \frac{1}{|\delta V|} \int_{\delta A_{\al \beta}} \rho^j \bv{j} \cd
        \foten{n}^\al da, \text{   and} \\
        \frac{|\delta V_\al|}{|\delta V||\delta V_\al|} \int_{\delta V} \rho^j \hat{r}^j dv &= \epsa
        \overline{\rho^j}^{\al} \overline{\hat{r}^j}^{\al}.
    \end{flalign}
    \label{eqn:mass_balance_averaged_terms}
\end{subequations}
Substituting equations (\ref{eqn:mass_balance_averaged_terms})(a)-(c) into
equation (\ref{eqn:pore_scale_mass_balance_eulerian}), recognizing the
volume fraction terms, and recognizing the averaged mass and velocity gives
the upscaled mass balance equation:
\begin{flalign}
    \notag & \pd{\left( \epsa \overline{\rho^j}^{\al} \right)}{t} + \diver
    \left( \epsa \overline{\rho^j}^{\al} \overline{\bv{j}}^{\al} \right) \\
    & \qquad = \sumba \frac{1}{|\delta V|} \int_{\delta A_{\al \beta}}
    \rho^j \left( \foten{w}_{\al \beta_j} - \bv{j} \right) \cd
    \foten{n}^{\al} + \epsa \overline{\rho^j}^{\al}
    \overline{\hat{r}^j}^{\al}.
    \label{eqn:upscaled_cons_mass_v1}
\end{flalign}
Rewriting equation \eqref{eqn:upscaled_cons_mass_v1} in terms of the material time
derivative and defining
\begin{flalign}
    \rhoaj &\equiv \overline{\rho^j}^{\al}, \\
    \bv{\aj} &\equiv \overline{\bv{j}}^{\al}, \\
    \ehatbaj &\equiv \frac{\epsa}{|\delta V_\al|} \int_{\delta A_{\al
    \beta}} \rho^j \left( \foten{w}_{\al \beta_j} - \bv{j} \right) \cd
    \foten{n}^{\al} da, \text{  and } \\
    \hat{r}^{\aj} &\equiv \epsa \rhoaj \overline{\hat{r}^j}^{\al}
\end{flalign}
respectively to be the averaged mass over $\delta V_\al$, the mass averaged
velocity, the net rate of mass gained by constituent $j$ in phase $\al$
from phase $\beta$, and the rate of mass gain due to interaction with other
species within phase $\al$, we get
\begin{flalign}
    \md{\aj}{\left( \epsa \rho^{\aj} \right)} + \epsa \rho^{\aj}
    \diver \bv{\aj} &= \sumba \ehatbaj + \hat{r}^{\aj}.
    \label{eqn:upscaled_cons_mass_v2}
\end{flalign}

Next we define the corresponding bulk phase variables so that the macroscale equations are
consistent with experimentally measured terms as much as possible.  Define:
\begin{flalign}
    \rhoa &\equiv \sumj \rhoaj, \text{  and }
    \label{eqn:bulk_density} \\
    C^{\aj} &\equiv \frac{\rhoaj}{\rhoa}
    \label{eqn:mass_concentration}
\end{flalign}
respectively to be the mass density of the $\al$ phase, and the mass
concentration of the $j^{th}$ constituent in the $\al$ phase.  If
equation (\ref{eqn:upscaled_cons_mass_v2}) is rewritten as
\[ \pd{\left( \epsa \rhoa C^{\aj} \right)}{t} + \diver \left( \epsa \rhoa
C^{\aj} \bv{\aj}
\right) = \sumba \ehatbaj + \hat{r}^{\aj} \]
and then summed over $j=1:N$ we obtain a mass balance equation for the
$\al$ phase:
\begin{flalign}
    \pd{\left( \epsa \rhoa \right)}{t} + \diver \left( \epsa \rhoa \bv{\al}
    \right) = \sumba \ehatba.
    \label{eqn:upscaled_cons_mass_summed}
\end{flalign}
Here we have used the fact that $\ehatba = \sumj \ehatbaj$; the rate of
mass transfer to the $\al$ phase from the $\beta$ phase is the sum of the
rates of mass transfer to each individual constituent in the $\al$ phase
from the $\beta$ phase.

Now use the definition of the material time derivative to write the mass
balance equation for the $\al$ phase as
\begin{flalign}
    \md{\al}{\left( \epsa \rhoa \right)} + \epsa \rhoa \diver \bv{\al} =
    \sumba \ehatba,
    \label{eqn:upscaled_mass_balance_phase}
\end{flalign}
where the following restrictions have been applied:
\begin{flalign}
    \sumj \hat{r}^{\aj} &= 0, \quad \forall \al, \text{ and} 
    \label{eqn:restriction_rhat} \\
    \suma \sumba \ehatbaj &= 0, \quad j=1:N.
    \label{eqn:restriction_ehatbaj}
\end{flalign}
Restriction (\ref{eqn:restriction_rhat}) states that the rate of net gain of mass within
species $\al$ from chemical reactions alone must be zero.  \Cref{eqn:restriction_ehatbaj}
states that the rate of mass gained by phase $\al$ from phase $\beta$ is equal to the rate
of mass gained by phase $\beta$ from phase $\al$.


\subsection{Macroscale Momentum Balance}
We now turn our attention to the momentum balance equation,
(\ref{eqn:pore_scale_mom_balance_eulerian}).  We can apply the same principles to upscale
this equation (for full details see \cite{weinstein}).  The macroscopic linear momentum
balance equation for constituent $j$ in the $\al$ phase is
\begin{flalign}
    \epsa \rhoaj \md{\aj}{\bv{\aj}} - \diver \left( \epsa \stress{\aj}
    \right) - \epsa \rhoaj \foten{g}^{\aj} = \ihat^{\aj} + \sumba \Thatbaj,
    \label{eqn:upscaled_mom_balance_species}
\end{flalign}
and the macroscopic linear momentum balance equation for the $\al$ phase is
\begin{flalign}
    \epsa \rhoa \md{\al}{\bv{\al}} - \diver \left( \epsa \stress{\al}
    \right) - \epsa \rhoa \foten{g} = \sumba \Thatba,
    \label{eqn:upscaled_mom_balance_phase}
\end{flalign}
where $\stress{\aj}$ and $\stress{\al}$ are the Cauchy stress tensors for the species and
the phase and $\Thatbaj$ and $\Thatba$ are momentum transfer terms. Most specifically, for
the momentum transfer terms, the former represents momentum transfered to constituent $j$
in the $\al$ phase through mechanical interactions from phase $\beta$, and the latter
represents the momentum transfered to phase $\al$ through mechanical interactions from
phase $\beta$.  Also notable is the $\ihat^{\aj}$ term.  This term represents the rate of
momentum gain due to mechanical interactions with other species within the same phase. 

In the processes of deriving these equations the following restrictions
were enforced:
\begin{flalign}
    \sumj \left( \ihat^{\aj} + \hat{r}^{\aj} \bv{\aj,\al} \right) &= 0
    \qquad \forall \al, \text{ and} 
    \label{eqn:restriction_ihat} \\
    \suma \sumba \left( \Thatbaj + \bv{\aj} \ehatbaj \right) &= 0 \qquad
    j=1:N.
    \label{eqn:restriction_Thatbaj}
\end{flalign}
Restriction (\ref{eqn:restriction_ihat}) states that linear momentum can
only be lost due to interactions with other phases (not within the
species), and restriction (\ref{eqn:restriction_Thatbaj}) states that the
interface can hold no linear momentum.  The comma in the superscript of
(\ref{eqn:restriction_ihat}) indicates a relative term: $\bv{\aj,\al} =
\bv{\aj} - \bv{\al}$.  For a complete list of notation see Appendix
\ref{app:nomenclature}.

Lastly, to tie the momentum transfer and stress tensor for the $\al$ phase
to those of the species we note two identities that were used in the
derivation:
\begin{flalign}
    \stress{\al} &= \sumj \left( \stress{\aj} - \rhoaj \bv{\aj,\al} \otimes
    \bv{\aj,\al} \right) 
    \label{eqn:stress_phase_species} \\
    \Thatba &= \sumj \left( \Thatbaj + \ehatbaj \bv{\aj,\al} \right).
    \label{eqn:mom_trans_phase_species}
\end{flalign}
These identities will be used later and so are presented here for
conciseness.


\subsection{Macroscale Energy Balance}
The derivations for the macroscale angular momentum and energy balance laws
are more algebraically complicated.  The angular momentum equation will not
be used in this work since we assume that we're dealing with granular-type
media where the angular momentum balance results in the solid phase Cauchy
stress tensor being symmetric \cite{Murad2000,Gray1979,weinstein}.  The energy balance equation, on the other
hand, will allow us to derive a novel form of the heat equation in porous
media.  For this reason we state the full equation here.  

Applying the same routine as in the mass and linear momentum equations we
arrive (after significant simplification) at a balance law for the energy
in species $j$:
\begin{flalign}
    \eps{\al} \rhoaj \md{\aj}{(e^{\aj})} - \diver \left( \eps{\al}
    \bq^{\aj} \right) - \eps{\al} \stress{\aj} : \grad \bv{\aj} - \eps{\al}
    \rho^{\aj} h^{\aj} = \hat{Q}^{\aj} + \hat{Q}_{\beta}^{\aj}
    \label{eqn:energy_balance_species}
\end{flalign}
(see \cite{Bennethum1994,weinstein} for details on the derivation).
Here, $h^{\aj}$ is the external supply of energy, $e^{\aj}$ is the energy
density, $\bq^{\aj}$ is the partial heat flux vector for the $j^{th}$
component of the $\al$ phase, $\hat{Q}^{\aj}$ is the rate of energy gain
due to interaction with other species within the $\al$ phase, and
$\hat{Q}_{\beta}^{\aj}$ is the rate of energy transfer from the $\beta$
phase to the $\alpha$ phase not due to mass or momentum transfer. 

Again, following the derivation of \cite{weinstein}, the bulk phase energy
equation is 
\begin{flalign}
    \epsa \rhoa \md{\al}{e^\al} - \diver \left( \epsa \bq^{\al} \right) -
    \epsa \stress{\al} : \grad \bv{\al} - \epsa \rhoa h^a = \sumba
    \hat{Q}^\al_\beta
    \label{eqn:energy_phase}
\end{flalign}
where

To arrive at this form of the energy equation we enforced the following
restrictions:
\begin{subequations}
    \begin{flalign}
        \sumj \left[ \hat{Q}^{\aj} + \hat{\foten{i}}^{\aj} \bv{\aj,\al} +
        \hat{r}^{\aj} \left( e^{\aj} + \frac{1}{2} \left( \bv{\aj,\al}
        \right)^2 \right) \right] &= 0 \quad \forall \al, \text{ and }
        \label{eqn:energy_restriction1} \\
        \suma \sumba \left[ \hat{Q}^{\aj}_\beta + \Thatbaj \cd \bv{\aj} +
        \ehatbaj \left( e^{\aj} + \frac{1}{2} \left( \bv{\aj,\al} \right)^2
        \right) \right] &= 0 \quad j=1:N.
        \label{eqn:energy_restriction2}
    \end{flalign}
\end{subequations}
Restriction \eqref{eqn:energy_restriction1} states that energy gained or lost due to
species interactions within the $\al$ phase must be gained or lost due to interactions
with other phases.
Restriction \eqref{eqn:energy_restriction2} states that the rate of energy gained or lost
by one component in one phase must go to another component or phase.  That is, this second
restriction states that the interface retains no energy.

A system of equations governed by mass, momentum, and energy balance requires each of the
upscaled equations listed.  A count of the variables indicates that there are far
more variables than equations.  It is at this point where we need a method for deriving
constitutive equations for these remaining variables. The method chosen for this work uses
another macroscale balance law based on the second law of thermodynamics.


\section{The Entropy Inequality}
The development of constitutive laws is central to the modeling process.  As we mentioned
previously, this has historically been a process of fitting mathematical models to
empirical evidence.  The construct of Hybrid Mixture Theory (HMT) couples the averaging
theorems discussed in the previous section and the second law of thermodynamics to provide
us with restriction on the form of the constitutive relations; hence narrowing down the
experiments required to those that are thermodynamically admissible.  It is then up to the
experimentalists to verify and refine these models.  Both theoretical and experimental
directions of study have their merits, but putting the constitutive equations on a firm
theoretical footing is ultimately preferred whether it is before or after the experiments
are run. In this section we give a brief derivation of the upscaled entropy inequality,
and we then use this inequality, along with a judicious choice of variables, to derive
constitutive equations for unsaturated porous media.


\subsection{A Brief Derivation of the Entropy Inequality}
The second law of thermodynamics states that entropy will never decrease as
a system evolves toward equilibrium \cite{Atkins2010,Callen1985}. The
microscale entropy {\it balance} equation that describes this
phenomenon is
\begin{flalign}
    \rho^j \md{j}{\eta^j} + \diver \boldsymbol{\phi}^j - \rho^j
    b^j = \hat{\eta}^j + \hat{r}^j \bv{j} + \hat{\Lambda},
    \label{eqn:micro_entropy}
\end{flalign}
where $\eta^j$ is the entropy density of constituent $j$,
$\boldsymbol{\phi}^j$ is the entropy flux, $b^j$ is the external supply of
entropy, $\hat{\eta}^j$ is entropy gained from other constituents, and
$\hat{\Lambda}$ is the entropy production. Since the second law of
thermodynamics must hold we know that $\hat{\Lambda} \ge 0$ for all time. 

Applying Theorem \ref{thm:averaging_theorem} to equation
(\ref{eqn:micro_entropy}) and defining appropriate macroscale definitions
of the variables gives the upscaled entropy balance equation:
\begin{flalign}
    \epsa \rhoaj \md{\aj}{\eta^{\aj}} - \diver \left( \epsa
    \boldsymbol{\phi}^{\aj} \right) - \epsa \rhoaj b^{\aj} &= \sumba
    \hat{\Phi}^{\aj}_{\beta} + \hat{\eta}^{\aj} + \hat{\Lambda}^{\aj},
    \label{eqn:entropy_upscaled_v1}
\end{flalign}
where the terms on the right-hand side of equation \eqref{eqn:entropy_upscaled_v1}
represent transfer of entropy through mechanical interaction, entropy gained due to
interactions with other species, and the rate of entropy generation respectively.

Next, assume that the material we are modeling is {\it simple} in the
sense of Coleman and Noll \cite{Coleman1963}. This means that we assume
that the entropy flux and external supply are due to heat fluxes and
sources respectively.  To remove the dependence on external heat sources we add
($1/T$ times) the upscaled conservation of energy equation
(\ref{eqn:energy_balance_species}),
\begin{flalign*}
    \epsa \rhoaj \md{\aj}{e^{\aj}} - \diver \left( \epsa \bq^{\aj}
    \right) - \epsa \stress{\aj} : \grad \bv{\aj} - \epsa \rhoaj h^{\aj}
    = \hat{Q}^{\aj} + \sumba \hat{Q}^{\aj}_{\beta}.
\end{flalign*}

At this point we perform a Legendre transformation in order to convert
convert internal energy, $e^{\aj}$, to Helmholtz potential, $\psi^{\aj}$
(see any thermodynamics text, eg \cite{Callen1985}):
\begin{flalign}
    \psiaj = e^{\aj} - \eta^{\aj} T.
    \label{eqn:Legendre}
\end{flalign}
This is done because internal energy has entropy as a natural independent
variable, and entropy is difficult to measure experimentally.  It should be noted that the
Helmholtz potential is only one choice of thermodynamic potential we could have made.
This is done primarily for historical reasons, but the Gibbs potential and possibly the
Grand Canonical potential could have also been viable choices. The appeal of the Helmholtz
potential is that it naturally has independent variables of temperature and volume (or,
in intensive variables, density). 

To arrive at a simplified entropy inequality for the total production of entropy (across
all constituents and phases) we now solve for $\hat{\Lambda}^{\aj}$ and then sum over
$\al=l,g,s$ and $j=1:N$.  This step requires significant algebra so the details of the
derivation are omitted for brevity sake. After much simplification, the entropy inequality
becomes
\begin{flalign}
    \notag 0 \le \hat{\Lambda} &= \sum_{\alpha} \left\{ -\frac{\eps{\al}
    \rhoa}{T} \left( \md{\al}{\psia} + \eta^{\al} \mds{T} \right)  \right.
    \\
    \notag & \hspace{1.3cm} + \frac{\eps{\al}}{T} \left( \sumj \stress{\aj}
    \right) : \bd{\al} \\
    \notag & \hspace{1.3cm} + \frac{\eps{\al} \grad{T}}{T^2} \cd \left\{
    \bq^{\al} - \sumj \left( \stress{\aj} \cd \bv{\aj,\al} - \rhoaj
    \bv{\aj,\al} \left( \psiaj + \frac{1}{2} \bv{\aj,\al} \cd
    \bv{\aj,\al} \right) \right) \right\} \\
    \notag & \hspace{1.3cm} - \frac{1}{T} \sumj \left\{ \left( \sumba
    \Thatbaj \right) + \hat{{\bf i}}^{\aj} + \grad \left( \eps{\al}
    \rho^{\aj} \psiaj \right) \right\} \cd \bv{\aj,\al} \\
    \notag & \hspace{1.3cm} + \frac{\eps{\al}}{T} \sumj \left( \stress{\aj}
    - \rho^{\aj} \psi^{\aj} \soten{I} \right) : \grad \bv{\aj,\al} \\
    \notag & \hspace{1.3cm} - \frac{1}{T} \sumba \left\{ \Thatba +
    \eps{\al} \rho^{\al} \eta^{\al} \grad T \right\} \cd \bv{\al,s} \\
    \notag & \hspace{1.3cm} - \frac{1}{2T} \sumj \left\{ \left( \sumba
    \ehatbaj \right) + \rhataj \right\} \bv{\aj,\al} \cd \bv{\aj,\al}
    \\
    & \hspace{1.3cm} \left. - \frac{1}{T} \sumba \left\{ \ehatba \left(
    \psia + \frac{1}{2} \bv{\al,s} \cd \bv{\al,s} \right) \right\}
    \right\}, 
    \label{eqn:entropy}
\end{flalign}
where $\hat{\Lambda}$ is the rate of entropy generation.

Several new terms have appeared in (\ref{eqn:entropy}).  First, $\bd{\al} =
(\grad \bv{\al})_{sym}$ is the rate of deformation tensor (also known as
the strain rate). As before, terms with a comma in the superscript are
relative terms: $\bv{a,b} = \bv{a} - \bv{b}$.

Several identities were needed to derive (\ref{eqn:entropy}). A complete list of these
identities has been included in Appendix \ref{app:EntropyIdentities}.  The next step is to
expand the Helmholtz potential in terms of constitutive independent variables that
describe our system. This allows freedom to make choices about which variables control
behavior of the system.  The choice of these variables is generally non-trivial so in the
next section we discuss motivations for the choice of variables.


\newpage
\chapter{New Independent Variables and Exploitation of the Entropy
Inequality}\label{ch:Exploit}
Now that we have an expression for the entropy inequality we must choose a set of
independent variables that describes our system of interest.  We seek to describe a
multiphase system where the solid phase may undergo finite deformation, where the relative
saturations of the two fluid phases vary in time and space, and where phase changes
between the fluids possibly occurs throughout the porous medium.  Hassanizadeh and Gray
have modeled similar media in the past \cite{Gray1998,Hassanizadeh1986,Hassanizadeh1986a}.
These models include effects from common interfaces, common lines (where three phases
meet), and common points (where four phases meet).  These models are very thorough and
follow the same HMT approach.  The down sides to their models, in the author's opinion,
are three fold: (1) the complexity of the resulting equations is such that in order to use
these equations a host of simplifying assumptions must be made, (2) the thermodynamics of
the common points and lines make sense physically but are likely negligible relative to
other effects, and (3) constitutive equations must be derived for transfer rates between
interfaces, common lines, common points, and phases.  This final drawback indicates that a
detailed knowledge of the pore-scale physics must be somehow upscaled. Approaches have
been taken recently to do just this, but the proposed theories have not yet gained
widespread acceptance. Examples of such work include those of Gray et al.
\cite{Gray2007,Gray2007a}

In the present approach we choose not to directly model interfaces and instead strive to
eventually write our governing equations in terms of the macroscale chemical potential.
The chemical potential is known from physical chemistry and thermodynamics as a {\it
generalized driving force} that is a function of pressure and temperature.  It is well
known that mass transfer from liquid to gas states is driven by gradients of chemical
potential \cite{McQuarrie1997}, so if we can write constitutive equations (such as
Darcy's, Fick's, and Fourier's laws) in terms of this potential we can possibly couple the
relevant effects into much simpler governing equations; for example, equations that track
changes in chemical potential instead of pressure or concentration.  The immediate
drawback to the present modeling efforts is that the recent work by Hassanizadeh et al.\
seems to indicate that saturation and capillary pressure are linked to the amount of
interfacial area between phases within the medium \cite{Bottero2011,Joekar-Niasar2007}.
In the present work we will not directly model the fluid-fluid and fluid-solid interfaces.
We proceed with the present modeling effort despite the results proposed by Hassanizadeh
et al.  We will discuss this drawback as we run up against it in future sections and
chapters.

\section{A Choice of Independent Variables}
In this section we present a choice of independent variables for the Helmholtz free energy
(potential) so as to expand the entropy inequality and to derive the relevant
forms of Darcy's law, Fick's law, and Fourier's law.  These variables are known as {\it
constitutive} independent variables as they represent a postulation of the variables that
control the energy in the system.  ``Deriving physically meaningful results depends on our
ability to relate thermodynamically defined variables to physically interpretable
quantities'' \cite{weinstein}. To that end, we use our a priori knowledge of
thermodynamics to choose some of the variables.  For the remainder of this work we
restrict our attention to a three-phase system consisting of an elastic solid, a viscous
liquid phase, and a gas phase.  To begin the modeling process we assume that each of these
phases consists of $N$ constituents (also called species or components), and all
interfacial effects are neglected. Examples of the constituents include dissolved minerals
in the liquid, species evaporated into the gas, or precipitated minerals associated with the
solid phase.

The motivation for choosing some of the variables is relatively trivial.  For example, to
allow for a heat conducting medium, temperature, $T$, and the gradient of temperature,
$\grad T$, are included in the list of independent variables. The pore space is expected
to be variably saturated with the two fluid phases so the volume fractions, $\eps{l}$ and
$\eps{g}$, must be included in the set of variables.  The fact that $\suma \epsa = 1$
precludes us from using all three volume fractions since they are not independent of each
other. In future chapters we will further restrict this assumption since for a
rigid solid phase the sum of the fluid phase volume fractions is equal to the fixed
porosity
\begin{flalign}
    \eps{l} + \eps{g} = \porosity. \label{eqn:porosity_defn_v2}
\end{flalign}
The reason for not making this assumption
initially is that it allows us to develop models for deformable media as well as for media
with a rigid solid phase (hence, a more general model may be derived from these assumptions
later if necessary).

Recall from thermodynamics that the change in extensive Helmholtz potential, $A$, with respect
to volume is minus the pressure: $\partial A / \partial V = -p$.  In terms of intensive
variables this means that $\rho^2 \partial \psi / \partial \rho = -p$.  To remain
consistent with the extensive definition of the Helmholtz potential, the densities must
then be included in the set of independent variables.  Given the fact that there are $N$
constituents in each phase, this could be done in two different ways: (1) we could include
the mass concentrations, $C^{\aj}$, for $j=1:N-1$ along with the phase density, or (2) we
could include all of the constituent densities, $\rhoaj$ for $j=1:N$.  Bennethum, Murad,
and Cushman \cite{Murad2000}, and also Weinstein \cite{weinstein} took the first of these
options when using HMT to derive constitutive relations involving chemical potentials.
The trouble with this approach is that the mass concentration of the $N^{th}$ constituent
is dependent on the mass concentrations of the previous $N-1$ constituents (since the
concentrations sum to 1). These results indicate that the behavior of the constituents
depends on how they are labeled instead of simply being independent. Various techniques were successfully developed in
\cite{Murad2000} to deal with this complication.  To avoid these complications 
we choose
the second option and include the species densities, $\rho^{\aj}$ for $j=1:N$.  Since each constituent is free to move within each phase, the spatial
gradients of the species densities, $\grad \rho^{l_j}$ and $\grad \rho^{g_j}$, are also
included.

Darcy's law and Fick's law are classical empirical expressions for creeping flow and
constitutive diffusion. Darcy's law is a statement about the relative velocity of a fluid
phase in a porous medium, and Fick's law is a statement about the relative diffusive
velocity of a species within a phase.  Since we seek novel forms of these two laws we
include $\bv{\al,s}$ and $\bv{\aj,\al}$ for $\al=l,g,s$ in the list of independent
variables.  It should be noted that neither of these variables is {\it objective} in the
sense that they are not frame invariant.  This poses a problem since any governing
equation should not depend on an observer's frame of reference.  In \cite{Eringen2003},
Eringen proposed a modification to Darcy's law that creates a frame invariant relative
velocity.  The new terms needed for this new relative velocity are second order and are
assumed to be negligible in Darcy flow.  A similar argument can be used for Fick's law.

The reasoning given in the previous few paragraphs leads us to the set of independent
variables for $\psi^{\al}$ to include: 
\[ T, \grad T, \eps{l}, \eps{g}, \rho^{l_j}, \rho^{g_j}, \grad \rho^{l_j}, \, \grad
    \rho^{g_j}, \bv{l,s}, \bv{g,s}, \bv{l_j,l}, \bv{g_j,g}, \text{ and } \bv{s_j,s} \]
where $j=1:N$.  It is apparent, now, that solid-phase terms corresponding to the density
and gradient of density are missing. The principle of equipresence, from constitutive
theory in continuum mechanics, states that ``all constitutive variables are a function of
the same set of independent variables'' \cite{Schreyer-Bennethum2012}.  To give symmetry
between the phases we include $\rho^s$ and $\grad \rho^s$. The Stokes assumption for the
Cauchy stress tensor in a viscous fluid states that stress is the sum of the fluid
pressure and the strain rate.  For this reason we include the strain rate (also known as
the rate of deformation tensor) for the fluid phases: $\bd{l}$ and $\bd{g}$.  The theory
of equipresence also states that if we include strain rate in the fluid phases then we
must include a comparable term in the solid phase.  

A natural choice of variables for the solid phase are the solid phase volume fraction,
density, and the (averaged) strain.  Weinstein \cite{weinstein} pointed out that these
three variables are not independent, as explained below, and used a modified set of independent variables for
the solid phase.  The same modified set will be used here, so the
following simply states Weinstein's results with brief derivations.

Let $J^s$ be the Jacobian of the solid phase given by $J^s=$det$\left(
(\soten{F}^s)^T \cd \soten{F}^s \right)$, where $\soten{F}^s$ is the
deformation gradient
\begin{flalign}
    \soten{F}^s &= \pd{x^s_k}{X^s_K},
\end{flalign}
$\foten{x}^s$ is the Eulerian coordinate, and $\foten{X}^s$ is the Lagrangian coordinate.
Using standard identities from Continuum Mechanics, the Jacobian can be rewritten as
\begin{flalign}
    J^s=\text{det}\left( 2 \soten{E}^s + \soten{I} \right). \label{eqn:jacobian}
\end{flalign}
Furthermore, through the conservation of mass, the Jacobian is also a scaling factor for
volumetric changes, $J^s = (\eps{s}_0 \rho^s_0)/(\eps{s} \rho^s)$.  This clearly shows the
dependence of the three variables. To mitigate this issue, Weinstein \cite{weinstein}
adopted ideas from solid mechanics and 
considered a ``multiplicative decomposition'' of the deformation gradient, $\soten{F}^s$,
and the Green's deformation tensor, $\soten{C}^s$, as
\begin{flalign}
    \soten{C}^s &= \left( J^s \right)^{2/3} \Cbars,\label{eqn:CauchyGreen} \\
    \soten{F}^s &= \left( J^s \right)^{1/3} \Fbars, \label{eqn:DeformationGradient}
\end{flalign}
where $(J^s)^{1/3} \soten{I}$ and $(J^s)^{2/3} \soten{I}$ represent volumetric
deformation, and ``$\Fbars$ and $\Cbars$ are the modified deformation gradient and the
modified right Cauchy-Green tensor, respectively.'' With this modification to the solid
strain, the solid phase variables we consider here are $J^s, \Cbars, C^{s_k}$ and $\grad
C^{s_k}$ where $k=1:N-1$. We note here that in order to get physically meaningful results
for phase change, we include the same components, so that $C^{s_j}, C^{l_j},$ and
$C^{g_j}$ all refer to the same component. Pairing the mass concentrations and the Jacobian gives a description
of the density of the solid phase, and the modified Cauchy-Green tensor is used in place
of strain.

The principle of equipresence states that all of the constitutive variables must be a
function of the same set of the postulated independent variables.  In particular, we
postulate that the Helmholtz potential for each phase is a function of the following set
of variables:
\begin{flalign}
    \left\{ T, \grad T, \eps{l}, \eps{g}, \rho^{l_j}, \rho^{g_j},
    \grad \rho^{l_j}, \grad \rho^{g_j}, \bv{\al,s}, \bd{l}, \bd{g}, \bv{\aj,\al}, J^s,
    \Cbars, C^{s_k}, \grad C^{s_k} \right\},
    \label{eqn:variables}
\end{flalign}
where $\al=l,g,s$; $j=1:N$ and $k=1:N-1$. We postulate that a three phase
porous medium with an elastic solid phase and $N$ constituents per phase
can be modeled by set \eqref{eqn:variables}.

%
\subsection{The Expanded Entropy Inequality}\label{sec:entropy_expanded}
Consider now that the first line of the entropy inequality, (\ref{eqn:entropy}), contains
a material time derivative of the Helmholtz potential for the $\al$ phase.  Using the
identity
\begin{flalign}
    \md{\al}{(\cd)} = \md{s}{(\cd)} + \bv{\al,s} \cd \grad (\cd),
\end{flalign}
and applying the chain rule, the entropy inequality can be expanded to include each of our
constitutive independent variables.  The central idea to the exploitation of the second
law of thermodynamics is that no term in the entropy inequality can take values such that
entropy generation is negative.  A close examination of the expanded entropy inequality
reveals that there are many terms that show up linearly.  In these linear terms we notice
some that are neither independent nor constitutive.  Examples of such coefficients are
$\grad T, \grad C^{s_j}, \grad \rho^{l_j},
\grad \rho^{g_j}, \bd{l}, \bd{g}, \bv{l,s}, \bv{g,s}, \bv{s_j,s}, \bv{l_j,l}$, 
$\bv{g_j,g}$, $\dot{T}$, $\dot{\rho}^{\aj}$, and $\grad \bv{\aj,\al}$ (where the dot
notation (e.g. $\dot{T}$)
indicates a material time derivative).
Loosely speaking, we have no control over these variables and they could take values that
violate the second law. For example, take as a thought experiment a process where all of
these variables except $\dot{T}$ are zero. From Bennethum \cite{bennethum},
\begin{quote}
    ``Since none of the other terms in the entropy inequality are a function of $\dot{T}$,
    by varying the value of $\dot{T}$ we can make the left-hand side of the entropy
    inequality as large positive or as large negative as we want - hence violating the
    entropy inequality. Since the entropy inequality must hold for all processes
    (including those for which $T$ is any value), the entropy inequality can be violated
    unless the coefficient of $\dot{T}$ is zero.''
\end{quote}
In order not to violate the inequality in (\ref{eqn:entropysimplified}), the
coefficients of all of these factors must be zero.  
This implies that terms such as $\suma \left( \epsa \rhoa
\pd{\psia}{\grad T} \right)$ are zero and will therefore be left out of the expansion of
(\ref{eqn:entropy}) for brevity. The time rates of change of volume fractions are not this
type of variable since they are constitutive; that is, we assume a rule for the time rates
of change of volume fractions that depends on the specific medium of interest.

With this simplification in mind,
(\ref{eqn:entropy}) becomes
\begin{flalign}
    \notag T \Lambda = &\suma \eps{\al} \rho^{\al} \left(
    \pd{\psi^{\al}}{T} + \eta^{\al} \right) \dot{T} - \suma \left(
    \eps{\al} \rho^{\al} \pd{\psi^{\al}}{\eps{l}} \right) \epsdot{l} -
    \suma \left( \eps{\al} \rho^{\al} \pd{\psi^{\al}}{\eps{g}} \right)
    \epsdot{g} \\
    %
    %
    \notag &- \suma \sumjj \left( \eps{\al} \rho^{\al}
    \pd{\psi^{\al}}{C^{s_j}} \right) \dot{C}^{s_j} \\
    \notag &- \suma \sumj \left(
    \eps{\al} \rhoa \pd{\psi^{\al}}{\rholj} \right) \dot{\rho}^{l_j} -
    \suma \sumj \left( \eps{\al} \rhoa \pd{\psi^{\al}}{\rho^{g_j}} \right)
    \dot{\rho}^{g_j} \\
    \notag &- \suma \left( \eps{\al} \rhoa \pd{\psi^{\al}}{J^s} \right)
    \dot{J}^{s} - \suma \left( \eps{\al} \rhoa \pd{\psi^{\al}}{\left(
    \Cbars \right)} \right) : \Cbarsdot \\
%
%
    \notag &- \left[ \eps{l} \rhol \left\{ \left( \pd{\psi^{l}}{T} + \eta^{l} \right)
        \grad T + \pd{\psi^{l}}{\eps{l}} \grad \eps{l} + \pd{\psi^{l}}{\eps{g}} \grad
        \eps{g} \right.  \right. \\
    \notag & \qquad \qquad + \sumjj \pd{\psi^{l}}{C^{s_j}} \grad C^{s_j} + \sumj
    \pd{\psi^{l}}{\rho^{l_j}} \grad \rho^{l_j}+ \sumj
    \pd{\psi^{l}}{\rho^{g_j}} \grad \rho^{g_j} \\
    \notag &\left. \left. \qquad \qquad + \pd{\psi^{l}}{J^s} \grad J^s +
    \pd{\psi^{l}}{\Cbars} : \grad \left( \Cbars \right) 
     \right\} + \That{s}{l} +
    \That{g}{l} \right] \cd \bv{l,s} \\
    %
    %
    %
    \notag &- \left[ \eps{g} \rho^{g} \left\{ \left( \pd{\psi^{g}}{T} + \eta^{g} \right)
        \grad T + \pd{\psi^{g}}{\eps{l}} \grad \eps{l} + \pd{\psi^{g}}{\eps{g}} \grad
        \eps{g} \right.  \right. \\
    \notag & \qquad \qquad + \sumjj \pd{\psi^{g}}{C^{s_j}} \grad C^{s_j} + \sumj
    \pd{\psi^{g}}{\rho^{l_j}} \grad \rho^{l_j}+ \sumj
    \pd{\psi^{g}}{\rho^{g_j}} \grad \rho^{g_j} \\
    \notag & \left. \left. \qquad \qquad + \pd{\psi^{g}}{J^s} \grad J^s +
    \pd{\psi^{g}}{\Cbars} : \grad \left( \Cbars \right) 
     \right\} + \That{s}{g} + \That{l}{g} \right] \cd
    \bv{g,s} \\
    \notag &+ \suma \left[ \eps{\al} \left( \sumj \stress{\al_j} \right) :
    \bd{\al} \right] \\
    \notag &+ \suma \left[ \frac{\eps{\al} \grad T}{T} \cd \left\{
    \bq^{\al} - \sumj \left( \stress{\aj} \cd \bv{\aj,\al} - \rho^{\aj}
    \bv{\aj,\al} \left( \psi^{\aj} + \frac{1}{2} \left( \bv{\aj,\al}
    \right)^2 \right) \right) \right\} \right] \\
    \notag &- \suma \sumj \left\{ \left( \sumba \Thatbaj \right) + \hat{
    {\bf i}}^{\aj} + \grad \left( \eps{\al} \rho^{\aj} \psi^{\aj} \right)
    \right\} \cd \bv{\aj,\al} \\
    \notag &+ \suma \left\{ \eps{\al} \sumj \left( \stress{\aj} -
    \rho^{\aj} \psi^{\aj} \soten{I} \right) : \grad \bv{\aj,\al} \right\}
    \\
    \notag &- \frac{1}{2} \suma \sumj \left\{ \left( \sumba \ehatbaj
    \right) + \rhataj \right\} \left( \bv{\aj,\al} \right)^2 \\
    &- \suma \sumba \left\{ \ehatba \left( \psi^{\al} +
    \frac{1}{2} \left( \bv{\al,s} \right)^2 \right) \right\} \geq 0
    \label{eqn:entropy2}
\end{flalign}

The next step is to enforce two additional relationships using Lagrange multipliers.  In
doing so, the Lagrange multipliers become unknowns of the system. We will see in
subsequent sections that the Lagrange multipliers are associated with partial pressures
and chemical potentials of species in the fluid phases.  The first relationship
considered is the dependence of the diffusive velocities:
\begin{flalign}
    \sumj \left( C^{\aj} \bv{\aj,\al} \right) = \foten{0}.
\end{flalign}
One can see this since
\begin{flalign}
    \sumj C^{\aj} \bv{\aj,\al} = \sumj C^{\aj} \bv{\aj} - \sumj C^{\aj}
    \bv{\al} = \bv{\al} - \bv{\al} = \foten{0}.
\end{flalign} 
The implication is that if we know the concentrations and diffusive velocities of the
first $N-1$ constituents, then we would know the concentration and diffusive velocity of
the $N^{th}$ constituent.  Multiplying by the density, taking the gradient, and using the
product rule gives the following relationship:
\begin{flalign}
    \grad \left( \sumj \rhoaj \bv{\aj,\al} \right) = \sumj \left( \rhoaj
    \grad \bv{\aj,\al} + \bv{\aj,\al} \grad \rhoaj \right) = \foten{0}.
    \label{eqn:Nth_term_dep_grad}
\end{flalign}
Following Bennethum, Murad, and Cushman \cite{Murad2000}, we enforce this
relationship with a Lagrange multiplier so as to account for the $N^{th}$
term dependence.

The second relationship to be enforced with Lagrange multipliers is the
mass balance equation for each of the constituents
(\ref{eqn:upscaled_cons_mass_v2}):
\[ \md{\aj}{\left( \epsa \rhoaj \right)} + \epsa \rhoaj \diver \bv{\aj} =
\sumba \ehatbaj + \hat{r}^{\aj}. \]

Let $\Lambda_{old}$ denote $\Lambda$ from equation (\ref{eqn:entropy2}),
and let $\lamaj$ and $\lamhataN$ be the Lagrange multipliers for the mass
balance and $N^{th}$ term dependencies, \eqref{eqn:Nth_term_dep_grad}, respectively. The entropy inequality
is rewritten as follows:
\begin{flalign}
    \notag \Lambda_{new} = \Lambda_{old} &+ \suma \sumj
    \frac{\lamaj}{T}\left[ \md{\aj}{\left( \epsa \rhoaj \right)} + \epsa
    \rhoaj \diver \bv{\aj} - \left( \sumba \ehatbaj + \hat{r}^{\aj} \right)
    \right] \\
    &+ \suma \sumj \frac{\epsa}{T} \lamhataN : \grad \left(
    \rhoaj \bv{\aj,\al} \right).
    \label{eqn:entropy_modified}
\end{flalign}
After a significant amount of algebraic simplification (with no additional physical
assumptions), this yields the following
form of the entropy inequality:
\begin{flalign}
    \notag T \Lambda = &\suma \left\{ \eps{\al} \rho^{\al} \left(
    \pd{\psi^{\al}}{T} + \eta^{\al} \right) \right\} \dot{T} \\
    \notag &- \sum_{\beta = l,g} \left\{ \left[ \suma \left( \eps{\al}
    \rho^{\al} \pd{\psi^{\al}}{\eps{\beta}} \right) - \sumj \left(
    \lambda^{\beta_j} \rho^{\beta_j} \right) \right] \epsdot{\beta}
    \right\} \\
    %
%
    \notag &- \sumjj \left[ \left( \suma \eps{\al} \rho^{\al}
    \pd{\psi^{\al}}{C^{s_j}} \right) - \lambda^{s_j} \eps{s} \rhos \right]
    \dot{C}^{s_j} \\
    \notag &- \sum_{\beta=l, g} \left\{ \sumj \left[ \left( \suma
    \eps{\al} \rho^{\al} \pd{\psi^{\al}}{\rho^{\beta_j}} \right) -
    \lambda^{\beta_j} \eps{\beta} \right] \dot{\rho}^{\beta_j} \right\} \\
%
    \notag &- \left[ \suma \left( \eps{\al} \rhoa \pd{\psi^{\al}}{J^s}
    \right) - \frac{1}{3} \frac{\eps{s}}{J^s} \left( \sumj tr\left(
    \stress{s_j} \right) \right) \right] \dot{J}^{s} \\
    \notag &- \left[ \suma \left( \eps{\al} \rhoa \pd{\psi^{\al}}{\left(
    \Cbars \right)} \right) - \frac{\eps{s}}{2} \left( \left( \Fbars
    \right)^{-1} \cd \left( \sumj \stress{s_j} \right) \cd \left( \Fbars
    \right)^{-T} \right) \right. \\
    \notag & \qquad \left. - \frac{\eps{s}}{2} \left( \left( \Fbars
    \right)^{-1} \cd \left( \Fbars \right)^{-T} \right) \left( \sumj
    \lambda^{s_j} \rhosj \right) \right] : \Cbarsdot \\
%
%
    \notag &- \mathop{\sum_{\beta = l,g}}_{\gamma \neq \beta} \left\{
    \left[ \eps{\beta} \rho^{\beta} \left\{ \left( \pd{\psi^{\beta}}{T} +
    \eta^{\beta} \right) \grad T + \sumjj \pd{\psi^{\beta}}{C^{s_j}} \grad
    C^{s_j}\right. \right. \right. \\
    \notag & \qquad \qquad + \left( \pd{\psi^{\beta}}{\eps{\beta}} -
    \frac{1}{\eps{\beta} \rho^{\beta}} \sumj \lambda^{\beta_j}
    \rho^{\beta_j} \right) \grad \eps{\beta} +
    \pd{\psi^{\beta}}{\eps{\gamma}} \grad \eps{\gamma} \\
    %
    %
    \notag & \qquad \qquad  +  \sumj \left( \pd{\psi^{\beta}}{\rho^{\beta_j}} -
    \frac{\lambda^{\beta_j}}{\rho^{\beta}} \right) \grad \rho^{\beta_j} +
    \sumj \pd{\psi^{\beta}}{\rho^{\gamma_j}} \grad \rho^{\gamma_j} \\
    \notag & \qquad \qquad \left. \left. \left.+ \pd{\psi^{\beta}}{J^s}
    \grad J^s + \pd{\psi^{\beta}}{\Cbars} : \grad \left( \Cbars \right)
    \right\} + \That{s}{\beta} + \That{\gamma}{\beta} \right] \cd
    \bv{\beta,s} \right\} \\
%
%
    \notag &+ \sum_{\beta=l,g} \left\{ \eps{\beta} \sumj \left(
    \stress{\beta_j} + \lambda^{\beta_j} \rho^{\beta_j} \soten{I} \right) :
    \bd{\beta} \right\} \\
%
    \notag &+ \suma \left[ \frac{\eps{\al} \grad T}{T} \cd \left\{
    \bq^{\al} - \sumj \left( \stress{\aj} \cd \bv{\aj,\al}
    \phantom{\frac{1}{2}} \right. \right.
    \right. \\
    \notag & \hspace{2in} - \left. \left. \left. \rho^{\aj} \bv{\aj,\al}
    \left( \psi^{\aj} + \frac{1}{2} \bv{\aj,\al} \cd \bv{\aj,\al}
    \right) \right) \right\} \right] \\
    %
    %
    \notag &- \suma \sumj \left\{ \left( \sumba \Thatbaj \right) +
    \ihat^{\aj} + \grad \left( \epsa \rhoaj \psiaj \right) \right. \\
    \notag & \qquad \qquad \qquad \left. - \lamaj \grad \left( \epsa \rhoaj
    \right) - \epsa \lamhataN \cd \grad \rhoaj \right\} \cd \bv{\aj,\al} \\
%
    %
    \notag &+ \suma \sumj \left\{ \epsa \stress{\aj} + \epsa \rhoaj \left(
    \lamaj - \psiaj \right) \soten{I} + \epsa \rhoaj \lamhataN\right\} :
    \grad \bv{\aj,\al} \\
%
    %
    \notag &- \suma \sumj \left\{ \rhataj \left( \lambda^{\aj} +
    \frac{1}{2} \left( \bv{\aj,\al} \right)^2 \right) \right\} \\
    \notag &- \sumj \ehat{g}{l_j}\left\{  \left( \lambda^{l_j} +
    \psi^{l} \right) - \left( \lambda^{g_j} + \psi^{g} \right) \right. \\
    \notag & \qquad \quad \left. +
    \frac{1}{2} \left( \left( \bv{l,s} \right)^2 - \left( \bv{g,s}
    \right)^2 \right) + \frac{1}{2} \left( \left( \bv{l_j,l} \right)^2 -
    \left( \bv{g_j,g} \right)^2 \right)\right\} \\
    \notag &- \sumj \ehat{l}{s_j}\left\{  \left( \lambda^{s_j} + \psi^{s}
    \right) - \left( \lambda^{l_j} + \psi^{l} \right) - \frac{1}{2}
    \left( \bv{l,s} \right)^2  + \frac{1}{2} \left( \left( \bv{s_j,s}
    \right)^2 - \left( \bv{l_j,l} \right)^2 \right)\right\} \\
    &- \sumj \ehat{s}{g_j}\left\{  \left( \lambda^{g_j} + \psi^{g} \right)
    - \left( \lambda^{s_j} + \psi^{s} \right) + \frac{1}{2} \left( \bv{g,s}
    \right)^2 + \frac{1}{2} \left( \left( \bv{g_j,g} \right)^2 - \left(
    \bv{s_j,s} \right)^2 \right)\right\} \ge 0
    \label{eqn:entropysimplified}
\end{flalign}

The exploitation of equation \eqref{eqn:entropysimplified} will be the source of all of the constitutive
relations for the remainder of this work. The next section outlines the details of this
exploitation to form constitutive relations specific to multiphase media governed by our
choice of constitutive independent variables, \eqref{eqn:variables}.


\section{Exploiting the Entropy Inequality}
In this section we exploit the entropy inequality, (\ref{eqn:entropysimplified}), in the
sense of Colman and Noll \cite{Coleman1963}. The basic principle here is that, according
to the second law of thermodynamics, entropy is always non-decreasing as time evolves.
This fact is used to extract constitutive relationships from the entropy inequality.  Not
every result from this exploitation is relevant to the current study, so we only 
present the more notable and useful results in the next subsections. Furthermore, we
exploit equation \eqref{eqn:entropysimplified} with an eye toward deformable, multiphase,
media.  The assumption of deformable media will be removed in the future, but
this leaves open the possibility of returning to these results for future work. For an
abstract summary of how the exploitation of the entropy inequality works, along with
subtle but important assumptions, see Appendix \ref{app:AbstractEntropy}.


\subsection{Results That Hold For All Time}
As mentioned in Section \ref{sec:entropy_expanded}, several of the terms that appear
linearly in the entropy inequality have factors that are neither independent nor
constitutive.  We now use this fact to derive relationships that must hold for all time in
order to not violate the second law of thermodynamics.
To illustrate this point consider the coefficient of $\dot{T}$. If
this coefficient is set to zero we recover with the thermodynamic constraint
that temperature and entropy are conjugate variables,
\begin{flalign}
    \pd{\psia}{T} = -\eta^{\al}.
    \label{eqn:Helmholtz_entropy_conjugate}
\end{flalign}
This is a classical result known from thermodynamics.

\subsubsection{Fluid Lagrange Multipliers}
For the gas and liquid phases, the definitions of the Lagrange multipliers stem from the
coefficient of $\dot{\rho}^{\aj}$ and $\grad \bv{\aj,\al}$.  Setting the coefficient of
$\dot{\rho}^{\aj}$ to zero gives the definition of the Lagrange multiplier for the mass
balance equations:
\begin{flalign}
    \lambda^{\beta_j} = \suma \frac{\epsa \rhoa}{\eps{\beta}}
    \pd{\psia}{\rho^{\beta_j}}.
    \label{eqn:lagrange_betaj}
\end{flalign}
Setting the coefficient of $\grad \bv{\aj,\al}$ to zero, summing over
$j=1:N$, and solving for $\lamhataN$ yields an expression for the other
Lagrange multiplier:
\begin{flalign}
    \lamhataN &= -\frac{1}{\rhoa} \sumj \left[ \stress{\aj} + \left( \rhoaj
    \lamaj \right) \soten{I} \right] + \psia \soten{I}.
    \label{eqn:lagrange_lamhataN}
\end{flalign}

\subsubsection{Solid Phase Identities}
Several identities for the solid phase can be derived from the terms associated with the
time derivatives of the solid phase Jacobian,
$\Jsdot$, and the modified Cauchy-Green, $\Cbarsdot$ terms. 
From the $\Jsdot$ term we see that 
\begin{flalign}
    \frac{1}{3} \sumj tr \left( \stress{s_j} \right) &= \frac{J^s}{\eps{s}} \suma \left(
    \eps{\al} \rho^{\al} \pd{\psi^{\al}}{J^s} \right).
    \label{eqn:solid_species_stress}
\end{flalign}
Next, consider the identity
\begin{flalign}
    \stress{\al} = \sumj \left( \stress{\aj} + \rho^{\aj} \bv{\aj,\al} \otimes
    \bv{\aj,\al} \right)
    \label{eqn:solid_stress_id}
    \end{flalign}
resulting from upscaling the momentum balance equation.  Taking the trace of
\eqref{eqn:solid_stress_id},  neglecting the diffusive terms, and substituting this into
\eqref{eqn:solid_species_stress} gives a definition for the solid phase pressure: 
\begin{flalign}
    p^s := -\frac{1}{3} tr \left( \stress{s} \right) &= -\frac{J^s}{\eps{s}} \suma \left(
    \eps{\al} \rho^{\al} \pd{\psi^{\al}}{J^s} \right).
    \label{eqn:solid_pressure}
\end{flalign}
This is a generalization of the solid phase pressure found by Weinstein for saturated
porous media in \cite{weinstein}.  

The coefficient of the $\Cbarsdot$ term gives a relationship for the stress in the solid
phase.  This will give a generalization of the solid phase stress
\cite{Bennethum1997,Bennethum2007} and closely follows the derivations of
Bennethum \cite{Bennethum2007} and Weinstein \cite{weinstein}. Setting the coefficient of
the $\Cbarsdot$ term to zero, left multiplying by the modified deformation gradient,
$\Fbars$, and right multiplying by the transpose of the deformation gradient gives a
relationship that defines the Lagrange multiplier for the solid phase, $\lambda^{s_j}$:
%
\begin{flalign}
    \sumj \stress{s_j} + \sumj \lambda^{s_j} \rho^{s_j} \soten{I} &= \frac{2}{\eps{s}}
    \Fbars \cd \left[ \suma \left( \eps{\al} \rho^{\al} \pd{\psi^{\al}}{\Cbars} \right)
        \right] \cd \left( \Fbars \right)^T.\label{eqn:solid_temp}
\end{flalign}
Using identity \eqref{eqn:solid_stress_id} in the stress term of \eqref{eqn:solid_temp},
neglecting the diffusive velocities, taking one-third the trace of the result,
and using equation \eqref{eqn:solid_pressure} for the solid-phase pressure yields a relationship for the solid phase
Lagrange multiplier:
%
\begin{flalign}
    \sumj \lambda^{s_j} \rho^{s_j} &= p^s + \frac{2}{3\eps{s}} \suma \left( \eps{\al}
    \rho^{\al} \pd{\psi^{\al}}{\Cbars} \right) : \Cbars.  \label{eqn:solid_temp2}
\end{flalign}
Substituting \eqref{eqn:solid_temp2} back into \eqref{eqn:solid_temp} gives the following
relation for the solid phase stress:
\begin{flalign}
    \stress{s} &= -p^s \soten{I} + \frac{2}{\eps{s}} \Fbars \cd \left[ \suma \left(
        \eps{\al} \rho^{\al} \pd{\psi^{\al}}{\Cbars} \right) \right] \cd \left( \Fbars
        \right)^T - \frac{2}{3\eps{s}} \suma \left( \eps{\al} \rho^{\al}
        \pd{\psi^{\al}}{\Cbars} \right) : \Cbars \soten{I}.
    \label{eqn:solid_stress_decomposition}
\end{flalign}
This can be rewritten as
\begin{flalign}
    \stress{s} = -p^s \soten{I} + \stress{s}_{e} + \frac{\eps{l}}{\eps{s}}
    \stress{l}_{h} + \frac{\eps{g}}{\eps{s}} \stress{g}_{h} 
    \label{eqn:solid_stress}
\end{flalign}
where
\begin{subequations}
    \begin{flalign}
        \stress{s}_{e} &= 2\left( \rho^{s} \Fbars \cd \pd{\psi^{s}}{\Cbars} \cd \left(
        \Fbars \right)^T -\frac{1}{3} \rho^{s} \pd{\psi^{s}}{\Cbars} : \Cbars \soten{I}
        \right), \\
        \stress{l}_{h} &= 2\left( \rho^{l} \Fbars \cd \pd{\psi^{l}}{\Cbars} \cd \left(
        \Fbars \right)^T -\frac{1}{3} \rho^{l} \pd{\psi^{l}}{\Cbars} : \Cbars \soten{I}
        \right), \\
        \stress{g}_{h} &= 2\left( \rho^{g} \Fbars \cd \pd{\psi^{g}}{\Cbars} \cd \left(
        \Fbars \right)^T -\frac{1}{3} \rho^{g} \pd{\psi^{g}}{\Cbars} : \Cbars \soten{I}
        \right).
    \end{flalign}
\end{subequations}
The stresses above are termed the effective stress, hydrating stress for the liquid phase, and
hydrating stress for the gas phase respectively. Equation
\eqref{eqn:solid_stress_decomposition} states that the stress in the solid phase can be
decomposed into the solid pressure and stresses felt due to the presence of the fluid
phases.  It is here that the modifications of the deformation gradient and Cauchy-Green
tensors become clear. If we take the trace of the stress tensor then
we see that $(1/3)tr(\stress{s}) = -p^s$; which is how the solid phase pressure is
measured.  Therefore, this thermodynamic definition of $p^s$ is consistent with experimental measure.  Furthermore, the effective
stress and hydrating stresses are terms associated with the interaction between the solid
and the fluids. For saturated porous media, Bennethum \cite{Bennethum1997} states the following:
\begin{quote}
    ``The effective stress tensor is the stress of the solid phase due to the strain of the
    porous matrix, and the hydrating stress tensor is the stress the liquid phase supports due
    to the strain of the solid matrix (which would be negligible if the liquid and solid phase
    were not interactive, but which becomes significant for swelling porous materials).''
\end{quote}

One final note on the solid phase stress is that the total stress in the porous medium is
related to the pressures in all three phases.  This can be seen by taking the weighted sum
of the stresses in each phase:
\begin{flalign}
    \stress{ } = \suma \epsa \stress{\al} &= -\eps{s} p^s \soten{I} + \eps{s} \stress{s}_e
    + \eps{l} \left(\stress{l}_h + \stress{l} \right) + \eps{g} \left( \stress{g}_h +
    \stress{g} \right).
\end{flalign}
Taking one-third the trace of the total stress, and recalling that the effective and hydrating
stresses are trace free, gives
\begin{flalign}
    \left( \frac{1}{3} \right) tr \left( \stress{ } \right) &= - \eps{s} p^s +
    \frac{\eps{l}}{3} tr \left( \stress{l} \right) + \frac{\eps{g}}{3} tr \left(
    \stress{g} \right).
    \label{eqn:total_stress}
\end{flalign}
In order to fully understand the stresses in the fluid phases we must continue our
examination of the results coming from the entropy inequality. Equation
\eqref{eqn:total_stress} is similar to the Terzaghi stress principle; suggesting that the
fluid phases help to support the pore space in the medium.


\subsection{Equilibrium Results}
There are several more relationships that we can extract from the entropy inequality.  In
particular, we now seek relationships between the Lagrange multipliers and the fluid-phase
pressures. We also seek relationships for the momentum and energy exchange terms.  At
equilibrium the production of entropy is minimized.  Since this is a minimum, the gradient
of $\hat{\Lambda}$ with respect to the set of independent variables \eqref{eqn:variables}
is zero.  This indicates that the coefficients of the independent variables that appear
linearly in the entropy inequality are zero at equilibrium.  In the case of a three phase
porous medium of this nature, we define equilibrium to be when a subset of the independent
variables are zero. In particular, equilibrium is defined when no heat conduction occurs,
$\grad T = \foten{0}$, the strain rates in the fluid phases are zero, $\bd{\beta} =
\soten{0}$, and all relative velocities are zero, $\bv{\aj,\al} = \foten{0}$ and
$\bv{\al,s} = \foten{0}$.  This {\it definition} of equilibrium is particular to this type
of media and is chosen as it gives physically relevant and meaningful results. Another way
to look at this is to say that equilibrium is exactly the state when all of these
variables are zero.

\subsubsection{Fluid Stress Tensor}
The first notable equilibrium result comes from the coefficient of the rate of deformation
tensor, $\bd{\beta}$. Setting the coefficient of $\bd{\beta}$ to zero, eliminating the sum
of the constituent stress tensors using the identity
\begin{flalign}
    \stress{\beta} = \sumj \left[ \stress{\beta_j} + \rho^{\beta_j} \bv{\beta_j,\beta}
        \otimes \bv{\beta_j,\beta} \right],
\end{flalign}
and noting that at equilibrium the diffusive velocity is zero, yields
\begin{flalign}
    \sumj \lambda^{\beta_j} \rho^{\beta_j} &= -\frac{1}{3} tr \left( \stress{\beta}
    \right) = p^{\beta}.
    \label{eqn:liquid_vapor_lagrange2}
\end{flalign}
\Cref{eqn:liquid_vapor_lagrange2} links the Lagrange multipliers to the
equilibrium pressure of the fluid phases.  This is the classical definition
of pressure in a fluid: minus one-third the trace of the stress tensor.
Using equation (\ref{eqn:lagrange_betaj}) the pressure in the fluid phases
can now be written as
\begin{flalign}
    p^{\beta} = -\frac{1}{3} tr \left( \stress{\beta} \right) &= \sumj
    \suma \left( \frac{\eps{\al} \rho^{\al} \rho^{\beta_j}}{\eps{\beta}}
    \pd{\psi^{\al}}{\rho^{\beta_j}} \right).
    \label{eqn:generalpressure}
\end{flalign}

With the definition of pressure in equation \eqref{eqn:generalpressure} we note that the
coefficient of $\epsdot{\beta}$ can now be rewritten as
\begin{flalign}
    p^\beta - \suma \left( \eps{\al} \rhoa \pd{\psia}{\eps{\beta}} \right).
    \label{eqn:coeff_epsdot_beta}
\end{flalign}
The second term is the change in energy with respect to volumetric changes, and is
therefore interpreted as the relative affinity for one phase to another.  That is, this
term is related to the wetability of the $\al$ and solid phases by the $\beta$ fluid
phase. The time rate of change of volume fluid phase volume fraction is an equation of
state (that is not yet known), but rewriting the coefficient as in
\eqref{eqn:coeff_epsdot_beta} hints at the fact that the equation of state is related to
the pressure and the wettability of the phases.  Furthermore, pressure, wettability, and
surface tension are related to capillary pressure; hence indicating that the equation of
state for the time rate of change of volume fraction is related to capillary pressure.  It
is here that we note the drawback to the present modeling effort.  Recall that in the
present expansion of the entropy inequality we do not include interfacial effects.  If we
were to include these effects then a surface tension term would appear here (as shown in
Hassanizadeh and Gray \cite{Gray1998}) and these terms together would more readily be
associated with capillary pressure.  More discussion will be dedicated to the exact
equation of state for the time rate of change of volume fraction after a discussion on
cross coupling pressures in Section \ref{sec:pressure} and capillary pressure in Chapters
\ref{ch:Transport} and \ref{ch:Transport_Solution}.

\subsubsection{Momentum Transfer Between Phases}
The next notable equilibrium result we can extract from \eqref{eqn:entropysimplified}
comes from the coefficient of the fluid phase relative velocities, $\bv{\beta,s}$.
Setting this coefficient to zero, recalling that $\grad T = \foten{0}$ at equilibrium, using the definition of the fluid
phase pressure, (\ref{eqn:generalpressure}), the definition of the fluid phase Lagrange
multipliers, (\ref{eqn:lagrange_betaj}), and solving for the momentum transfer terms gives
\begin{flalign}
    \notag -\left( \That{s}{\beta} + \That{\gamma}{\beta} \right) &= \left( \eps{\beta}
    \rho^{\beta} \pd{\psi^{\beta}}{\eps{\beta}} - p^{\beta} \right) \grad \eps{\beta} +
    \eps{\beta} \rho^{\beta} \pd{\psi^{\beta}}{\eps{\gamma}} \grad \eps{\gamma} + \eps{\beta} \rho^{\beta} \sumjj \pd{\psi^{\beta}}{C^{s_j}} \grad
    C^{s_j} \\
    %
    %
    \notag & \quad - \sumj \left[ \left( \eps{\gamma} \rho^{\gamma}
        \pd{\psi^{\gamma}}{\rho^{\beta_j}} + \eps{s} \rho^s \pd{\psi^s}{\rho^{\beta_j}}
        \right) \grad \rho^{\beta_j} \right]+ \eps{\beta} \rho^{\beta} \sumj
        \pd{\psi^{\beta}}{\rho^{\gamma_j}} \grad \rho^{\gamma_j} \\ 
    & \quad + \eps{\beta} \rho^{\beta} \pd{\psi^{\beta}}{J^s} \grad J^s + \eps{\beta}
    \rho^{\beta} \pd{\psi^{\beta}}{\Cbars} : \grad \left( \Cbars \right), 
    \label{eqn:phase_mom_transfer_equilib}
\end{flalign}
where $\gamma$ is the {\it other} fluid phase not equal to $\beta$. This
particular result will be coupled with the conservation of momentum to
yield novel forms of Darcy's law in Section \ref{sec:Darcy_pressure}.

\subsubsection{Momentum Transfer Between Species}
Another notable equilibrium results comes from the coefficient of the diffusive velocity,
$\bv{\aj,\al}$. Equation (\ref{eqn:liquid_vapor_lagrange2}) indicates that at equilibrium
the definition of the Lagrange multiplier, equation (\ref{eqn:lagrange_lamhataN}),
simplifies to
\begin{flalign}
    \lamhataN = \psia \soten{I}.
    \label{eqn:lamhataN_equilib}
\end{flalign}
This implies that, at equilibrium, the stress tensor for constituent $j$ (from the
coefficient of $\grad \bv{\aj,\al}$)
can be written as
\begin{flalign}
    \stress{\aj} &= -\epsa \rhoaj \left( \lamaj - \psiaj \right) \soten{I}
    - \epsa \rhoaj \psia \soten{I}.
\end{flalign}

Consider the diffusive velocity term in the entropy inequality:
\begin{flalign*}
    & - \suma \sumj \left\{ \left( \sumba \Thatbaj \right) + \ihat^{\aj} + \grad \left( \epsa
    \rhoaj \psiaj \right) \right. \\
    & \quad \quad \left. \phantom{\left( \sumba \right)}
    - \lambda^{\aj} \grad \left( \epsa \rhoaj \right) - \epsa \lamhataN \cd \grad \rhoaj
\right\} \cd \bv{\aj,\al}.
\end{flalign*}
Add $-\suma \sumj (\rhoaj
\lamhataN \cd \grad \epsa) \cd \bv{\aj,\al} = 0$ and simplify to get
\begin{flalign*}
    & - \suma \sumj \left\{ \left( \sumba \Thatbaj \right) + \ihat^{\aj} + \grad \left(
\epsa \rhoaj \psiaj \right) \right. \\
& \quad \quad \left. \phantom{\left( \sumba \right)} - \lamaj \grad \left( \epsa \rhoaj \right) - \lamhataN : \grad
\left( \epsa \rhoaj \right) \right\} \cd \bv{\aj,\al}.
\end{flalign*}
At equilibrium the diffusive velocity is assumed to be zero.  By the logic used herein for
the exploiting the entropy inequality, and given the fact that $\lamhataN = \psia
\soten{I}$ at equilibrium, we observe that for each $j$,
\begin{flalign}
    \sumba \Thatbaj + \ihat^{\aj} &= -\grad \left( \epsa \rhoaj \psiaj
    \right) + \lamaj \grad \left( \epsa \rhoaj \right) + \psia
    \grad \left( \epsa \rhoaj \right).
    \label{eqn:species_mom_transfer_equilib}
\end{flalign}
This is an expression for the momentum transfer for species $j$ in the $\al$ phase.
This result will be coupled with the constituent conservation of momentum
equation to derive a form of Fick's law in Section \ref{sec:Ficks}.

\subsubsection{Partial Heat Flux}
To conclude the equilibrium results we examine the $\grad T$ term in the entropy
inequality.  We have assumed that
$\grad T = \foten{0}$ and $\bv{\aj,\al} = \foten{0}$ at equilibrium, so by the logic used
above we see that the coefficient of $\grad T$ must be zero and  
\begin{flalign}
    \suma \epsa \bq^\al = \foten{0}
    \label{eqn:eq_partial_heat_flux}
\end{flalign}
at equilibrium.  This is the partial
heat flux of the entire porous media and will be used in Section
\ref{sec:heat_flux} to derive a generalized Fourier's law.

\subsection{Near Equilibrium Results}
The next step in exploiting the entropy inequality is to derive {\it near equilibrium}
results.  These results arise by linearizing the equilibrium results about the equilibrium
state.  The linearization process is simply the first-order terms of the Taylor series,
but one must keep in mind that each of the derivatives is a function of all of the
constitutive independent variables that are not zero at equilibrium. For example, if $f =
f_{eq}$ at equilibrium, then near equilibrium, $f \approx f_{near}= f_{eq} + (\partial f /
\partial (\grad T))\cd \grad T + \cdots + (\partial f / \partial \bv{l,s}) \cd \bv{l,s}$.
This full expansion may yield terms that are not readily physically interpretable. For
this reason, considerable efforts must be made to relate the linearization constants to
measurable parameters. For a thorough explanation of the linearization process with the
entropy inequality see Appendix \ref{app:AbstractEntropy}. 

For the momentum transfer in the fluid phases, the linearization of equation
\eqref{eqn:phase_mom_transfer_equilib} can be simply written as
\begin{flalign}
    \left( \sum_{\al \neq \beta} \That{\al}{\beta} \right)_{near} =
    \left( \sum_{\al \neq \beta} \That{\al}{\beta} \right)_{eq} -
    \left( \eps{\beta}\right)^2 \soten{R}^{\beta} \cd \bv{\beta,s}.
    \label{eqn:mom_near_eq}
\end{flalign}
The linearization constant, $\soten{R}^\beta$, is related to
the resistivity of a porous medium; the inverse of the hydraulic conductivity. It should
be noted that we have only expanded about one of the possible variables: $\bv{\beta,s}$.
Strictly speaking this is incorrect and we should expand about all other variables which
are zero at equilibrium.  A more thorough expansion is
\begin{flalign}
    \notag \left( \sum_{\al \neq \beta} \That{\al}{\beta} \right)_{n.eq} &= \left(
    \sum_{\al \neq \beta} \That{\al}{\beta} \right)_{eq} - \left( \eps{\beta}\right)^2
    \soten{R}^{\beta} \cd \bv{\beta,s} \\
    &\quad + \soten{H}^\beta \cd
    \grad T + \soten{J}^\beta \cd \bv{\beta_j,\beta} + \toten{L}^{\beta} : \bd{\beta} +
    \cdots,
    \label{eqn:mom_near_eq_full}
\end{flalign}
where $\soten{H}^\beta$ and $\soten{J}^\beta$ are second-order tensors and
$\toten{L}^\beta$ is a third-order tensor.  The ellipses at the end of this equation
indicates that there are higher order terms that are not being written explicitly. The
left-hand side of \eqref{eqn:mom_near_eq_full} is the rate of momentum transfer due to mechanical
means.  It is reasonable to think that this transfer term might be a function of fluid
velocity, but the effects due to thermal gradients, diffusive velocity, and velocity
gradients are likely small in comparison.  To be completely correct we would have to
include these terms in the modeling problems to follow.  The trouble is that each of the
coefficients needs to be associated with a physical parameter. We will see that
$\soten{R}^\beta$ is physically associated with a material parameter of the porous medium,
but it is presently unclear what the physical interpretations are for the other
coefficients. Neglecting these terms simply leaves the door open for future modeling
research.

Proceeding in a similar manner, the linearized constituent momentum transfer from equation
\eqref{eqn:species_mom_transfer_equilib} is
\begin{flalign}
    \left( \sumba \Thatbaj + \ihat^{\aj} \right)_{near} = \left( \sumba
    \Thatbaj + \ihat^{\aj} \right)_{eq} - \epsa \rhoaj \soten{R}^{\aj} \cd
    \bv{\aj,\al}.
    \label{eqn:mom_species_near_eq}
\end{flalign}
The linearization constant is related to the inverse of the diffusion tensor.
The linearized partial heat flux from equation \ref{eqn:eq_partial_heat_flux} is
\begin{flalign}
    \left( \suma \epsa \bq^\al \right) &= \soten{K} \cd \grad T
    \label{eqn:heat_near_eq}
\end{flalign}
(recalling that the partial heat flux is zero at equilibrium), and the linearization
constant is related to the thermal conductivity.

In each of these linearization results, the factors of volume fraction and density are chosen so
that the linearization constants better match experimentally measured
coefficients. The signs are chosen so that the entropy inequality is not violated.

Several of the relationships resulting from the entropy inequality rely on proper
definitions of the partial derivatives of the energy with respect to particular
independent variables.  The pressure is one such quantity, but there are several others
that appear in the preceding results.  For this reason, we now turn our attention to the
exact definitions of pressure and chemical potential under our choice of independent
variables. This will help to simplify and to attach physical meaning to the terms
appearing in each of the linearized results.  In saturated swelling porous material,
Bennethum and Weinstein \cite{Bennethum2004} showed that there are three pressures acting
on the system.  These results are extended in the next section to
media with multiple fluid phases.


\section{Pressures in Multiphase Porous Media}\label{sec:pressure}
We will see in this subsection that the three pressures defined in \cite{Bennethum2004}
can be extended to broader definitions in multiphase media. These definitions will
help to simplify and attach physical meaning to the terms appearing in each of the
linearized results discussed in the previous subsection.  We will also define several new
{\it pressures} acting as coupling terms between the phases in multiphase media. It will
be shown that we can return to the three pressure relationship of Bennethum and Weinstein
if we simplify these results to a single fluid phase.

Recall from the entropy inequality that the equilibrium pressure in multiphase media can
be written as an accumulation of cross effects as follows: 
\begin{flalign}
    p^{\beta} = \suma \sumj \left( \frac{\eps{\al} \rho^{\al}
    \rho^{\beta_j}}{\eps{\beta}} \pd{\psia}{\rho^{\beta_j}} \right).
\label{eqn:multiphase_pressure}
\end{flalign}
The partial derivative is taken while holding $\eps{\al}, \rho^{\al_k}, \eps{\beta},$
and $\rho^{\beta_m}$ fixed where $k=1:N$ and $m=1:N, \, m\neq j$. Define a cross-coupling
pressure as
\begin{flalign}
    \pp{\al}{\beta} := \sumj \left( \frac{\eps{\al} \rho^{\al}
    \rho^{\beta_j}}{\eps{\beta}} \left. \pd{\psia}{\rho^{\beta_j}}
    \right|_{\eps{\al}, \rho^{\al_k}, \eps{\beta}, \rho^{\beta_m}} \right)
    \label{eqn:pressure_al_beta}
\end{flalign}
so that the $\beta$-phase pressure can be simply written as the sum of these
cross-coupling pressures
\begin{flalign}
    p^{\beta} = \suma \pp{\al}{\beta} \quad \text{for} \quad \beta \in \{
    l, g \} \quad \text{and} \quad \al \in \{ l,g,s\}. 
    \label{eqn:pressure_sum_al_beta}
\end{flalign}

Now we derive an identity that is analogous to the three pressure relationship derived by
Bennethum and Weinstein \cite{Bennethum2004}.  To that end, consider the Helmoltz
potential as a function of two sets of independent variables where there is a one-to-one
relationship between the two sets.
%
\begin{flalign}
    \psi^{\al} = \psi^{\al}\left( \eps{\al}, \eps{\al} \rho^{\aj},
    \eps{\beta}, \eps{\beta} \rho^{\beta_j} \right) \quad \text{and} \quad
    \hat{\psi}^{\al} = \hat{\psi}^{\al}\left( \eps{\al}, \rho^{\aj},
    \eps{\beta}, \rho^{\beta_j} \right),
\end{flalign}
where $\al \in \{ l,g\}$ and $\beta \neq \al,s$. The Helmholtz potential is actually a
function of several other variables, but these are suppressed here to make the notation
more readable.  Since $\psia$ and $\hat{\psi}^{\al}$ are
functions of an equivalent set of variables, the total differentials must be equal to each other.
Setting $d\hat{\psi}^{\al} = d\psia$ gives
\begin{flalign}
    \notag d\hpsia &= \left. \pd{\psia}{\eps{\al}} \right|_{\eps{\al}
    \rho^{\al_k}, \eps{\beta}, \eps{\beta} \rho^{\beta_k}} d\eps{\al} +
    \sumj \left.  \pd{\psia}{(\eps{\al} \rho^{\aj})} \right|_{\eps{\al},
    \eps{\al} \rho^{\al_m}, \eps{\beta}, \eps{\beta} \rho^{\beta_k}}
    d(\eps{\al} \rho^{\aj}) \\
    & \quad + \left. \pd{\psia}{\eps{\beta}} \right|_{\eps{\al}, \eps{\al}
    \rho^{\al_k}, \eps{\beta} \rho^{\beta_k}} d\eps{\beta} + \sumj \left.
    \pd{\psia}{(\eps{\beta} \rho^{\beta_j})} \right|_{\eps{\al}, \eps{\al}
    \rho^{\al_k}, \eps{\beta}, \eps{\beta} \rho^{\beta_m}} d( \eps{\beta}
    \rho^{\beta_j})
\end{flalign}
where in each case we are taking $m=1:N$ such that $m \neq j$, and $k=1:N$.
Now take the partial derivative with respect to $\eps{\beta}$ while holding
$\eps{\al}, \rho^{\al_k},$ and $\rho^{\beta_k}$ fixed.  In this case, the
$d\eps{\al}$ and $d(\eps{\al} \rho^{\aj})$ terms will be zero.  This leaves
us with: 
\begin{flalign}
    \notag \left. \pd{\hpsia}{\eps{\beta}} \right|_{\eps{\al}, \rho^{\al_k},
    \rho^{\beta_k}} &= \left. \pd{\psia}{\eps{\beta}} \right|_{\eps{\al}
    \rho^{\al_k}, \eps{\beta}, \eps{\beta} \rho^{\beta_k}} \left.
    \pd{\eps{\beta}}{\eps{\beta}} \right|_{\eps{\al},\rho^{\al_k},
    \rho^{\beta_k}} \\
    \notag &\quad + \sumj \left.  \pd{\psia}{(\eps{\beta} \rho^{\beta_j})}
    \right|_{\eps{\al}, \eps{\al} \rho^{\al_k}, \eps{\beta}, \eps{\beta}
    \rho^{\beta_m}} \left. \pd{(\eps{\beta} \rho^{\beta_j}) }{\eps{\beta}}
    \right|_{\eps{\al}, \rho^{\al_k}, \rho^{\beta_k}} \\
    &= \left. \pd{\psia}{\eps{\beta}} \right|_{\eps{\al} \rho^{\al_k},
    \eps{\beta}, \eps{\beta} \rho^{\beta_k}} + \sumj \left.
    \frac{\rho^{\beta_j}}{\eps{\beta}}  \pd{\psia}{\rho^{\beta_j}}
    \right|_{\eps{\al}, \eps{\al} \rho^{\al_k}, \eps{\beta}, \eps{\beta}
    \rho^{\beta_m}}.
\end{flalign}
Now multiply by $-\eps{\al} \rho^{\al}$ to get
\begin{flalign}
    -\left. \eps{\al} \rho^{\al} \pd{\hpsia}{\eps{\beta}}
    \right|_{\eps{\al}, \rho^{\al_k}, \rho^{\beta_k}} &= -\left. \eps{\al}
    \rho^{\al} \pd{\psia}{\eps{\beta}} \right|_{\eps{\al} \rho^{\al_k},
    \eps{\beta}, \eps{\beta} \rho^{\beta_k}} - \sumj \left. \frac{\eps{\al}
    \rho^{\al} \rho^{\beta_j}}{\eps{\beta}} \pd{\psia}{\rho^{\beta_j}}
    \right|_{\eps{\al}, \eps{\al} \rho^{\al_k}, \eps{\beta}, \eps{\beta}
    \rho^{\beta_m}}.
\end{flalign}
Notice that the third term is $\pp{\al}{\beta}$ from equation
(\ref{eqn:pressure_al_beta}).  Define the following new terms:
\begin{flalign}
    \ppbar{\al}{\beta} &:= -\eps{\al} \rho^{\al} \left.
    \pd{\psia}{\eps{\beta}} \right|_{\eps{\al}, \eps{\al} \rho^{\al_k},
    \eps{\beta} \rho^{\beta_k}}
    \label{eqn:pbar_al_beta} \\
    \ppi{\al}{\beta} &:= \eps{\al} \rho^{\al} \left.
    \pd{\psia}{\eps{\beta}} \right|_{\eps{\al}, \rho^{\al_k},
    \rho^{\beta_k}}
    \label{eqn:pi_al_beta}
\end{flalign}
to get the relationship
\begin{flalign}
    \pp{\al}{\beta} = \ppbar{\al}{\beta} + \ppi{\al}{\beta}.
    \label{eqn:threepressuresab}
\end{flalign}
Note that the new definitions only hold if $\eps{\beta} \neq 0$.  This can
be seen if one returns back to the Lagrange multiplier equation (at the
beginning of this section) for the pressure.  Furthermore, this
relationship holds if we had taken the derivative with respect to
$\eps{\al}$ instead of $\eps{\beta}$.  

For completeness sake we define $\overline{p}^{\beta}$ and $\pi^{\beta}$ so
that our definitions are consistent with \cite{Bennethum2004}:
\begin{flalign}
    \overline{p}^{\beta} &:= \suma \ppbar{\al}{\beta}
    \label{eqn:pbar_total} \\
    \pi^{\beta} &:= \suma \ppi{\al}{\beta}.
    \label{eqn:pi_total}
\end{flalign}
With these definitions we recover the three pressure relationship derived
by Bennethum and Weinstein
\begin{flalign}
    p^{\beta} = \overline{p}^{\beta} + \pi^{\beta}. 
    \label{eqn:three_pressures}
\end{flalign}

The physical meaning of $\overline{p}^\beta$ is the change in energy with respect to
changes in volume while holding mass fixed.  In terms of extensive variables this is the
same definition as pressure encountered in classical thermodynamics for a single phase.  For this reason we call
$\overline{p}^\beta$ the {\it thermodynamic} pressure.  The physical meaning of
$\pi^{\beta}$ is the change in energy with respect to changes in saturation while holding
the densities fixed.  This {\it pressure} (or swelling potential as it is called in
\cite{Bennethum2004}) relates the deviation between the {\it classical pressure},
$p^{\beta}$, and the thermodynamic pressure.  It can be seen as a preferential wetting
function that measures the affinity for one phase over another. With these physical
considerations in mind we now return to the coefficient of the $\epsdot{\beta}$ terms in
the entropy inequality.  With the present definitions, the coefficient is
\[ p^\beta - \pi^\beta. \]
Using \eqref{eqn:three_pressures} this is clearly $\overline{p}^\beta$. Since
the time rate of change of volume fraction is taken as a constitutive variable, the
linearization result for this term can now be stated as
\begin{flalign}
    \overline{p}^\beta \Big|_{n.eq.} =  \overline{p}^\beta \Big|_{eq.} + \tau
    \epsdot{\beta}.
    \label{eqn:epsdot_near_eq}
\end{flalign}
The coefficient $\tau$ arose from linearization and is formally defined as
\[ \tau = \pd{\overline{p}^\beta}{\epsdot{\beta}} \Big|_{eq.} \]
Equation \eqref{eqn:epsdot_near_eq} does little to make clear the exact meaning of this
equation.  The exact meaning will become clear in Chapter \ref{ch:Transport} under the
assumption that the fluid-phase volume fractions are not independent.

The definitions of the three {\it pressures} allow us to attach more physical meaning (and
more convenient notation) to the results found when building constitutive equations in the
next sections.  Before building these equations we define the upscaled chemical potential
for a multiphase system, and after this point we will have all of the tools necessary to
derive the new constitutive equations.


\section{Chemical Potential in Multiphase Porous Media}
Chemical potential is defined thermodynamically as the change in energy with respect to
changes in the number of molecules in the system \cite{Atkins2010,Callen1985}. This
classical definition has the following characteristics \cite{Murad2000}: (1) it is a
scalar and measures the energy required to insert a particle into the system, (2) its
gradient is the driving force for diffusive flow (Fick's law), and (3) it is constant for
a single constituent in two phases at equilibrium.  In \cite{Schreyer-Bennethum2012}, Bennethum
proposed a definition for chemical potential in saturated porous media that satisfies all
three of these criteria:
\begin{flalign}
    \mu^{\aj} &= \psia + \rhoa \left. \pd{\psia}{\rhoaj} \right|_{\epsa,
        \rho^{\al_m}} = \left. \pd{\left( \rhoa \psia \right)}{\rhoaj} \right|_{\epsa,
            \rho^{\al_m}} = \left. \pd{\left( \epsa \rhoa \psia \right)}{\left( \epsa
            \rhoaj\right)} \right|_{\epsa, \rho^{\al_m}},
    \label{eqn:chem_pot_bennethum}
\end{flalign}
for $m=1:N$ and $m\neq j.$
In saturated media, if the changes in energy in the solid phase due to changes in liquid
density are assumed to be zero, then the numerator of the right-hand side of 
(\ref{eqn:chem_pot_bennethum}) can be seen as the total energy in a saturated system.
Under this assumption, the chemical potential can be rewritten as
\begin{flalign}
    \mu^{\aj} &= \left. \pd{\psi_T}{\left( \eps{\al} \rhoaj \right)} \right|_{\epsa,
        \rho^{\al_m}}
\end{flalign}
This indicates that in a saturated porous medium, the chemical potential of the $j^{th}$
constituent in the $\al-$phase is the change in total energy with respect to changes in
mass of constituent $j$. We now extend this definition to multiphase unsaturated systems.

Extending this idea to multiphase and multiconstituent media, we define chemical potential
to be the change in total energy with respect to changes in mass in the
constituent.  In multiphase media we cannot make the assumption that the energy in one
phase is not effected by changes in other phases.  With this in mind, we recall that the
total energy  can be given by $\psi_T = \suma \epsa \rhoa \psia$.  Therefore, the present
definition of chemical potential is
\begin{flalign}
    \mu^{\beta_j} &= \left. \pd{\psi_T}{\left( \eps{\beta} \rho^{\beta_j}
    \right)} \right|_{\epsa, \eps{\beta}, \rho^{\al_k}, \rho^{\beta_m}} =
    \left. \suma \pd{\left( \epsa \rhoa \psia \right)}{\left( \eps{\beta}
    \rho^{\beta_j} \right)} \right|_{\epsa, \eps{\beta}, \rho^{\al_k},
    \rho^{\beta_m}},
    \label{eqn:chem_pot_defn}
\end{flalign}
where again, $k=1:N$ and $m=1:N, \, m\neq j$. Notice that if $\partial
\psia / \partial \rho^{\beta_j} = 0$ for $\al \neq \beta$ then this
definition collapses to equation (\ref{eqn:chem_pot_bennethum}).
Furthermore, recalling the definition of the Lagrange multiplier, $\lamaj$,
from (\ref{eqn:lagrange_betaj}), equation (\ref{eqn:chem_pot_defn}) can be
rewritten as
\begin{flalign}
    \mu^{\beta_j} = \psi^{\beta} + \suma \left( \frac{\epsa
    \rhoa}{\eps{\beta}} \left. \pd{\psia}{\rho^{\beta_j}} \right|_{\epsa,
    \eps{\beta}, \rho^{\al_k}, \rho^{\beta_m}} \right) = \psi^{\beta} +
    \lambda^{\beta_j}.
    \label{eqn:chem_pot_defn2}
\end{flalign}
\Cref{eqn:chem_pot_defn2} only holds for $\beta = l$ and $\beta =g$.  A definition
of the solid phase chemical potential is beyond the scope of this work.

As a result of this definition of chemical potential we observe an immediate effect on the
rate of mass transfer terms in the entropy inequality.  Using equation
(\ref{eqn:chem_pot_defn2}), the last three terms in the entropy inequality,
(\ref{eqn:entropysimplified}), can be rewritten as
\begin{flalign}
    \notag &- \sumj \ehat{g}{l_j}\left\{  \left( \lambda^{l_j} + \psi^{l} \right) - \left(
    \lambda^{g_j} + \psi^{g} \right) \right. \\
    \notag & \qquad \quad \left. + \frac{1}{2} \left( \left( \bv{l,s} \right)^2 - \left(
    \bv{g,s} \right)^2 \right) + \frac{1}{2} \left( \left( \bv{l_j,l} \right)^2 - \left(
    \bv{g_j,g} \right)^2 \right)\right\} \\
    \notag &- \sumj \ehat{l}{s_j}\left\{  \left( \lambda^{s_j} + \psi^{s} \right) - \left(
    \lambda^{l_j} + \psi^{l} \right) - \frac{1}{2} \left( \bv{l,s} \right)^2  +
    \frac{1}{2} \left( \left( \bv{s_j,s} \right)^2 - \left( \bv{l_j,l} \right)^2
\right)\right\} \\
    \notag &- \sumj \ehat{s}{g_j}\left\{  \left( \lambda^{g_j} + \psi^{g} \right) - \left(
    \lambda^{s_j} + \psi^{s} \right) + \frac{1}{2} \left( \bv{g,s} \right)^2 + \frac{1}{2}
    \left( \left( \bv{g_j,g} \right)^2 - \left( \bv{s_j,s} \right)^2 \right)\right\} \\
    \notag = &- \sumj \ehat{g}{l_j}\left\{  \mu^{l_j} - \mu^{g_j} + \frac{1}{2} \left(
    \left( \bv{l,s} \right)^2 - \left( \bv{g,s} \right)^2 \right) + \frac{1}{2} \left(
    \left( \bv{l_j,l} \right)^2 - \left( \bv{g_j,g} \right)^2 \right)\right\} \\
    \notag &- \sumj \ehat{l}{s_j}\left\{  \left( \lambda^{s_j} + \psi^{s} \right) -
    \mu^{l_j} - \frac{1}{2} \left( \bv{l,s} \right)^2  + \frac{1}{2} \left( \left(
    \bv{s_j,s} \right)^2 - \left( \bv{l_j,l} \right)^2 \right)\right\} \\
    &- \sumj \ehat{s}{g_j}\left\{  \mu^{g_j} - \left( \lambda^{s_j} + \psi^{s} \right) +
    \frac{1}{2} \left( \bv{g,s} \right)^2 + \frac{1}{2} \left( \left( \bv{g_j,g} \right)^2
    - \left( \bv{s_j,s} \right)^2 \right)\right\}.
\end{flalign}
The square of the relative velocities are likely zero
since these models are designed with creeping flow in mind. With these simplifications,
the mass transfer terms from the entropy inequality are rewritten as
\begin{flalign}
    - \sumj \ehat{g}{l_j}\left\{  \mu^{l_j} - \mu^{g_j}\right\} - \sumj \ehat{l}{s_j}\left\{  \left( \lambda^{s_j} + \psi^{s} \right) -
    \mu^{l_j}\right\} - \sumj \ehat{s}{g_j}\left\{  \mu^{g_j} - \left( \lambda^{s_j} + \psi^{s} \right)\right\}.
\end{flalign}

At equilibrium we assume that the mass transfer between phases is zero. Take note that
this is an assumption about how the constitutive variable behaves at equilibrium and not
an assumption about the equilibrium state itself. This fine point is made since in several
works this assumption is made as part of the definition of equilibrium (for example,
\cite{Bennethum2004}).  In the author's opinion this is a subtle mistake.  The assumption
that $\ehat{g}{l_j} = 0$ at equilibrium implies a
final equilibrium relationship; the mass transfer between the fluid phases is proportional
to the difference in chemical potentials
\begin{flalign}
    \notag \ehat{g}{l_j}|_{n.eq} &= \ehat{g}{l_j}|_{eq} + \left[ \left( \rho^{l_j} - \rho^{g_j}
        \right) M \right] \left( \mu^{l_j} - \mu^{g_j} \right) \\
    &= \left[ \left( \rho^{l_j} - \rho^{g_j} \right) M \right] \left(
    \mu^{l_j} - \mu^{g_j} \right).
    \label{eqn:mass_transfer_linearized}
\end{flalign}
This helps to verify our choice of upscaled chemical potential by satisfying the third
criteria set forth at the beginning of this subsection.  Furthermore, this suggests a
natural coupling between the liquid and gas phase mass balance equations. The mass
transfer coefficient is chosen to have a factor of the difference in densities so as to
better match experimental measures \cite{Smits2011}.  Given that the units of the rate of
mass transfer are [$M L^{-3} t^{-1}$] we see that the units of the linearization constant
are 
\[ \left[ M \right] = \frac{t}{L^2} = \frac{1}{\left( L^2/t \right)}. \]

A further verification that we have properly defined the multiphase chemical potential
correctly can be seen through the Gibbs-Duhem relationship from
thermodynamics \cite{Callen1985}. Simply stated, the Gibbs-Duhem relationship states that
the Gibbs potential of the $\al$ phase is the weighted sum of the chemical potentials:
\begin{flalign}
    \Gamma^{\al} = \sumj C^{\aj} \mu^{\aj}.
    \label{eqn:Gibbs_Duhem}
\end{flalign}
Equation \eqref{eqn:Gibbs_Duhem} specifies the relationship between the Gibbs potential,
$\Gamma^\al$, and the chemical potential.  Substituting (\ref{eqn:chem_pot_defn2}) into the
right-hand side of (\ref{eqn:Gibbs_Duhem}), carrying out the summation, and applying the
definition of pressure, (\ref{eqn:liquid_vapor_lagrange2}), gives the equation
\begin{flalign}
    \Gamma^{\al} = \psia + \frac{p^\al}{\rhoa},
\end{flalign}
which is the standard thermodynamic relationship between the Helmholtz potential and
the Gibbs potential.  This clearly demonstrates that the definition of multiphase chemical
potential used here is consistent with the classical thermodynamic definition.

At this point we turn our attention toward using the relationships derived from the
entropy inequality to develop novel expressions for Darcy's, Fick's, and Fourier's laws of
flow, diffusion, and heat conduction. For a concise summary of all of the results derived
in this chapter, see Appendix \ref{app:EntropyResults}.


\section{Derivations Constitutive Equations}
In this section we derive general forms of Darcy's, Fick's, and Fourier's laws based on the HMT
results in the previous sections. These equations will be coupled with mass and energy
balance equations to form a macroscale model for heat and moisture transport for
unsaturated media. The results derived in this section extend the classical forms of each
of these laws. These extensions suggest terms that, in the author's knowledge, are
previously unreported. Also, we propose new forms of these laws
in terms of the macroscale chemical potential. This suggests that the chemical potential
is a generalized driving force for flow, diffusion, and heat transport.


\subsection{Darcy's Law}\label{sec:Darcy_pressure}
In 1856, Henri Darcy proposed his empirical law governing flow through saturated porous
media \cite{Darcy1856}.  This was derived through experimentation on sand filters used to
purify the water in the fountains of Dijon, France. In its simplest form, Darcy's law
states that the averaged fluid flux is proportional to the gradient of hydraulic head (or fluid
pressure)
\begin{flalign}
    \eps{l} \bv{l,s} = - k \grad h.
    \label{eqn:classical_Darcy}
\end{flalign}
Under the construct of Hybrid Mixture Theory, Darcy's law is obtained by coupling the
momentum balance equation for a fluid phase, (\ref{eqn:upscaled_mom_balance_phase}), with
the linearized constitutive equation for the momentum transfer from other phases.  This
has been illustrated by several authors (some examples include
\cite{Bennethum1996a,Bennethum1996b,Gray1979,weinstein}), and depending on the set of
independent variables postulated for the Helmholtz potential, the momentum transfer term
can suggest different forms of Darcy's law. 

In the present case, we recall from equation \eqref{eqn:mom_near_eq} that the linearized
momentum transfer terms can be written as
\begin{flalign}
    \notag \left( \That{s}{\beta} + \That{\gamma}{\beta} \right) &= -\eps{\beta}
    \soten{R}^{\beta}\cd \left( \eps{\beta} \bv{\beta,s} \right) - \left(
    \ppi{\beta}{\beta} - p^{\beta} \right) \grad \eps{\beta} - \ppi{\beta}{\gamma} \grad
    \eps{\gamma} \\
    \notag & \quad - \eps{\beta} \rho^{\beta} \sumjj \pd{\psi^{\beta}}{C^{s_j}} \grad
    C^{s_j}+ \sumj \left[ \left( \eps{\gamma} \rho^{\gamma}
        \pd{\psi^{\gamma}}{\rho^{\beta_j}} + \eps{s} \rho^s \pd{\psi^s}{\rho^{\beta_j}}
        \right) \grad \rho^{\beta_j} \right] \\  
    & \quad - \eps{\beta} \rho^{\beta} \sumj \pd{\psi^{\beta}}{\rho^{\gamma_j}} \grad
    \rho^{\gamma_j} - \eps{\beta} \rho^{\beta} \pd{\psi^{\beta}}{J^s} \grad J^s -
    \eps{\beta} \rho^{\beta} \pd{\psi^{\beta}}{\Cbars} : \grad \left( \Cbars \right),
    \label{eqn:darcy_linearized_mom_transfer}
\end{flalign}
where we recall that $\soten{R}^{\al}$ is related to the resistivity of the medium and
arose from the linearization process.

Linearization of the stress-pressure relationship for the fluid phases gives an expression
for the stress near equilibrium: 
\begin{flalign}
    \stress{\beta} = - p^{\beta} \soten{I} + \Foten{\nu}^{\beta} : \bd{\beta}
    \label{eqn:linearized_stress}.
\end{flalign}
The fourth-order tensor multiplying the rate of deformation tensor can be simplified, in
most cases, to correspond to the viscosity of the medium (see any text on continuum
mechanics).  Ignoring the acceleration terms in the momentum balance equation
\eqref{eqn:upscaled_mom_balance_phase}, and substituting equation
(\ref{eqn:darcy_linearized_mom_transfer}) for momentum transfer and
(\ref{eqn:linearized_stress}) for the stress tensor gives the following generalization of
Darcy's law:
\begin{flalign}
    \notag \eps{\beta} \soten{R}^{\beta} \cd \left( \eps{\beta}
    \bv{\beta,s} \right) &= -\eps{\beta} \grad p^{\beta} -
    \ppi{\beta}{\beta} \grad \eps{\beta} - \ppi{\beta}{\gamma}\grad
    \eps{\gamma} + \eps{\beta} \rho^{\beta} \foten{g} \\
    \notag & \quad + \sumj \left[ \left( \eps{\gamma} \rho^{\gamma}
    \pd{\psi^{\gamma}}{\rho^{\beta_j}} + \eps{s} \rho^s
    \pd{\psi^s}{\rho^{\beta_j}} \right) \grad \rho^{\beta_j} \right] -
    \eps{\beta} \rho^{\beta} \sumj \pd{\psi^{\beta}}{\rho^{\gamma_j}} \grad
    \rho^{\gamma_j} \\ 
    & \quad - \eps{\beta} \rho^{\beta} \pd{\psi^{\beta}}{J^s} \grad J^s -
    \eps{\beta} \rho^{\beta} \pd{\psi^{\beta}}{\Cbars} : \grad \left(
    \Cbars \right) + \diver \left( \Foten{\nu}^{\beta} : \bd{\beta}
    \right).
    \label{eqn:Darcy}
\end{flalign}
To arrive at this form of Darcy's law we have also assume that $\grad C^{s_j} \approx
\foten{0}$ since it assumed that concentration gradients in the solid phase do not affect
flow. The first term indicates that flow is primarily due to pressure
gradients, as expected.  The eighth and ninth terms were previously reported by Weinstein
in \cite{weinstein}. Note that the {\it extra} factor of $\eps{\beta}$ on the left-hand
side of the equation can be moved to the right.  If all but the first term on the
right-hand side are then ignored we arrive at the classical Darcy's Law
\begin{flalign}
    \soten{R}^{\beta} \cd \left( \eps{\beta} \bv{\beta,s} \right) &= -
    \grad p^{\beta},
    \label{eqn:classical_Darcy_v2}
\end{flalign}
where $\bq^{\beta} = \eps{\beta} \bv{\beta,s}$ is known as the {\it Darcy
Flux}.

The linearization constant, $\soten{R}^\beta$, is related to the resistivity
of the porous medium, the inverse of which is assumed to exist, and we
define $\soten{K}^\beta = \left( \soten{R}^\beta \right)^{-1}$.  The tensor
$\soten{K}^\beta$ is related to the hydraulic conductivity. To determine the exact meaning
of the linearization constant we consider the units of the simplest terms:
\[ \bq^\beta \approx - \soten{K}^\beta \cd \grad p^\beta. \]
The units of the Darcy flux are length per time [$L/t$], and the units of the pressure are
mass per length per time squared [$M/(L \cd
t^2)$].  This indicates that the linearization constant has units [$(L^3 \cd t)/M$], which
can be rewritten as [$(L^2) / (M/(L \cd t))$].  The numerator of this fraction has
units of permeability, $\kappa$, and the denominator has units of dynamic viscosity,
$\mu_\beta$. This suggests that
\begin{flalign}
    \soten{K}^\beta &= \frac{\soten{\kappa}}{\mu_\beta},
    \label{eqn:conductivity_permeability}
\end{flalign}
and this relationship is confirmed in equations (11.4) and (11.5) of Pinder et al.\
\cite{Pinder2006}.  The hydraulic conductivity of a porous medium is defined as
\begin{flalign}
    \soten{k}_c &= \frac{\rho^\beta g \soten{\kappa}}{\mu_\beta} = \frac{g
        \soten{\kappa}}{\nu_\beta},
        \label{eqn:hydraulic_cond_nu}
\end{flalign}
where $\nu_\beta$ is the kinematic viscosity of the fluid. This indicates that
$\soten{K}^\beta$ can also be defined as 
\begin{flalign}
    \soten{K}^\beta = \frac{\soten{k}_c}{\rho^\beta g}.
\end{flalign}
It is clear from these relationships that $\soten{K}^\beta$ is a function of both the type
of fluid and the geometry of the porous medium. This coefficient ``describes, in some
sense, the ability of the porous medium to transmit fluid'' \cite{Pinder2006}.  In
saturated porous media, the permeability is typically assumed only to be a function of
geometry.  Under Hybrid Mixture Theory we must note that the permeability is a function of
any variable which is not necessary zero at equilibrium.  Typically it is assumed that the
permeability of an unsaturated medium is a function of the volume fractions
\cite{Bear1988,Pinder2006}. The tensorial notation may be dropped in isotropic media, but
for anisotropic media it is assumed that the permeability may depend on the direction of
flow.  We will expand upon this idea in later chapters when building a macroscale mass
balance model.

\subsection{Darcy's Law In Terms of Chemical Potential}\label{sec:Darcy_chem_pot}
Equation \eqref{eqn:Darcy} couples all of the physical processes that we wished to model
at the outset; multiphase flow with constituents in each phase and a deformable solid. In order to build reasonable models
based on this constitutive equation, functional forms for the
wetting potentials, $\ppi{\beta}{\beta}$ and $\ppi{\beta}{\gamma}$, and the changes in
energy with respect to density are needed.  The solid-phase terms are likely negligible for
non-deformable media, but dealing with the remaining terms represent a significant
modeling task.  The goal of this subsection is to greatly simplify this model while
maintaining the physical interpretation.  This is done by switching
thermodynamic potentials.

Recall from thermodynamics that the Gibbs potential, $\Gamma^\al$, can be written in terms of the
Helmholtz potential, $\psia$, as
\begin{flalign}
    \Gamma^{\al} = \psia + \frac{p^{\al}}{\rhoa}.
    \label{eqn:Gibbs_Helmholtz}
\end{flalign}
Taking the gradient of the Gibbs potential, expanding the resulting gradient of Helmholtz
potential in terms of the constitutive independent variables, and multiplying by
$-\eps{\beta} \rho^{\beta}$ yields the equation:
\begin{flalign}
    \notag & - \ppi{\beta}{\beta} \grad \eps{\beta} - \ppi{\beta}{\gamma}
    \grad \eps{\gamma} - \eps{\beta} \grad p^{\beta} \\
    \notag & \quad = \sumj \eps{\beta}\rho^{\beta}
    \pd{\psi^{\beta}}{\rho^{\beta_j}} \grad \rho^{\beta_j} + \sumj
    \eps{\beta} \rho^{\beta} \pd{\psi^{\beta}}{\rho^{\gamma_j}} \grad
    \rho^{\gamma_j} \\
    & \qquad - \eps{\beta} \rho^{\beta} \grad \Gamma^{\beta} -
    \frac{\eps{\beta} p^{\beta}}{\rho^{\beta}} \grad \rho^{\beta} -
    \eps{\beta} \rho^{\beta} \eta^{\beta} \grad T.
    \label{eqn:grad_gibbs}
\end{flalign}
Matching the common terms between (\ref{eqn:Darcy}) and (\ref{eqn:grad_gibbs}),
and recognizing the resulting chemical potential terms yields
\begin{flalign}
    \notag & \eps{\beta} \soten{R} \cd \left( \eps{\beta} \bv{\beta,s} \right) \\
    \notag & \quad = -\eps{\beta} \rho^{\beta} \grad \Gamma^{\beta} -
    \frac{\eps{\beta} p^{\beta}}{\rho^{\beta}} \grad \rho^{\beta} -
    \eps{\beta} \rho^{\beta} \eta^{\beta} \grad T + \eps{\beta}
    \rho^{\beta} \foten{g} \\
    & \qquad + \sumj \left[ \eps{\beta} \left( \mu^{\beta_j}
    - \psi^{\beta} \right) \grad \rho^{\beta_j} \right] + \diver \left(
    \Foten{\nu}^{\beta} : \bd{\beta} \right).
    \label{eqn:Darcy_Gibbs2}
\end{flalign}
Observe that the summation in \cref{eqn:Darcy_Gibbs2} can be simplified to
\begin{flalign*}
    \sumj \left[ \eps{\beta} \left( \mu^{\beta_j} - \psi^{\beta} \right)
    \grad \rho^{\beta_j} \right] &= \eps{\beta} \left( \Gamma^{\beta} -
    \psi^{\beta} \right) \grad \rho^{\beta} + \eps{\beta} \rho^{\beta}
    \sumj \left( \mu^{\beta_j} \grad C^{\beta_j} \right)
\end{flalign*}
by expanding the gradient of $\rho^{\beta_j}$ and using the Gibbs-Duhem
equation (\ref{eqn:Gibbs_Duhem}). Substituting this back into
(\ref{eqn:Darcy_Gibbs2}) and simplifying yields the chemical potential form
of Darcy's law:
\begin{flalign}
    \notag & \eps{\beta} \soten{R} \cd \left( \eps{\beta} \bv{\beta,s} \right) \\
    & \quad = -\eps{\beta} \rho^{\beta} \sumj \left( C^{\beta_j} \grad
    \mu^{\beta_j} \right) - \eps{\beta} \rho^{\beta} \eta^{\beta} \grad T +
    \eps{\beta} \rho^{\beta} \foten{g} + \diver \left( \Foten{\nu}^{\beta}
    : \bd{\beta} \right) 
\end{flalign}
Cancelling the factor of $\eps{\beta}$ from the left-hand side, rewriting the coefficient
of the $\grad \mu^{\beta_j}$ term, and multiplying by the inverse of $\soten{R}^\beta$
gives
\begin{flalign}
    \eps{\beta} \bv{\beta,s} = - \soten{K}^\beta \cd \left[ \sumj \left(
        \rho^{\beta_j} \grad \mu^{\beta_j} \right) + \rho^{\beta} \eta^{\beta} \grad T -
        \rho^{\beta} \foten{g} - \frac{1}{\eps{\beta}} \diver \left( \Foten{\nu}^{\beta} :
        \bd{\beta} \right) \right].
    \label{eqn:Darcy_ChemPot_with_viscosity}
\end{flalign}

Equation \eqref{eqn:Darcy_ChemPot_with_viscosity} states an amazing fact: the flow of phase $\beta$ is
due only to gradients in chemical potential, temperature, gravity, and viscous forces.
The viscous forces are often neglected in creeping flow.  This gives
\begin{flalign}
    \eps{\beta} \bv{\beta,s} = - \soten{K}^\beta \cd \left[ \sumj \left(
        \rho^{\beta_j} \grad \mu^{\beta_j} \right) + \rho^{\beta} \eta^{\beta} \grad T -
        \rho^{\beta} \foten{g} \right].
    \label{eqn:Darcy_ChemPot}
\end{flalign}
It should be emphasized that no additional assumptions were made to arrive at this
equation.  That is, we still assume multiphase flow with a possibly swelling solid phase.
All of the actions, interrelations, and cross coupling effects are tied up within the
chemical potential term.  This further indicates that the chemical potential is a {\it
generalized force} that, in effect, incorporates several driving forces.  

A final simplification is to consider a pure fluid phase where only one constituent is
present.  In this case, Darcy's law is rewritten as
\begin{flalign}
    \eps{\beta} \bv{\beta,s} &= - \soten{K}^\beta \cd \left[ \rho^\beta \grad \Gamma^\beta
        + \rho^\beta \eta^\beta \grad T - \rho^\beta \foten{g}\right]
    \label{eqn:darcy_gibbs_form}
\end{flalign}
where $\Gamma^\beta$ is the macroscale Gibbs potential. The entropy coefficient of the gradient
of temperature poses a significant modeling issue as the entropy is not readily
measurable.  The fact stated by equation \eqref{eqn:darcy_gibbs_form} is that the Darcy
flux of a pure species is truly controlled by gradients in temperature and Gibbs
potential.  This is a generalization of the classical pressure formulation that captures a
wider range of physical effects.

\subsection{Fick's Law}\label{sec:Ficks}
We now turn our attention to diffusion and Fick's law.  In 1855, Adolf Fick published the
first mathematical treatment of diffusion \cite{Fick1855}. The empirically based equation
simply states that the diffusive flux of a species through a mixture is proportional to
the gradient in concentration of the species.  This has since been generalized through
thermodynamics and physical chemistry \cite{Callen1985,McQuarrie1997} to state that the
diffusive flux is proportional to the gradient in chemical potential of the species. In
this subsection we apply the Hybrid Mixture Theory construct to derive a version of Fick's
law for multiphase porous media. It should be noted here that the classical chemical potential from
the thermodynamic definitions of Fick's law for diffusion in a liquid not in a porous
medium is the not the {\it same} chemical potential
as that defined for the porous media.  In mixture theory we view the porous medium as a
mixture of phases (and species), but the classical thermodynamic definition considers one phase with
a mixture of species.  With this difference in mind, it is not immediately clear that the
multiphase version of Fick's law will be the same.

To derive the present version of Fick's law we first consider a linearization of the
coefficient of the $\grad \bv{\aj,\al}$ term in the entropy inequality.  The gradient of
the diffusive velocity, $\grad \bv{\aj,\al}$, is taken to be zero at equilibrium so this
coefficient is zero (since entropy generation is minimized at equilibrium). Therefore,
\[ \epsa \stress{\aj} + \epsa \rhoaj \left( \lamaj - \psiaj \right) \soten{I} + \epsa
    \rhoaj \lamhataN  = \soten{0} \quad \text{for all } j. \]
Using equation \eqref{eqn:lamhataN_equilib} for the definition of the Lagrange multiplier,
$\lamhataN$ at equilibrium and
linearizing the coefficient of $\grad \bv{\aj,\al}$ about $\grad \bv{\aj,\al}$ gives
\begin{flalign}
    \epsa \stress{\aj} &= \epsa \Foten{S}^{\aj} : \grad \bv{\aj,\al} +
    \left( -\epsa \rhoaj \lamaj + \epsa \rhoaj \psiaj - \epsa \rhoaj \psia
    \right) \soten{I},
    \label{eqn:species_stress_near_eq}
\end{flalign}
where $\Foten{S}^{\aj}$ is a fourth-order tensor that arises from linearization.  Now consider the species conservation
of momentum equation (\ref{eqn:upscaled_mom_balance_species}).  We ignore the inertial
terms since diffusion is assumed to be slow (this is discussed in some detail in Chapter
\ref{ch:diffusion_comp}).  Now eliminate the momentum transfer terms using the linearized
momentum transfer derived from the entropy inequality, (\ref{eqn:mom_species_near_eq}),
use \eqref{eqn:species_stress_near_eq} for the stress tensor, and using the fact that
$\mu^{\aj} = \lamaj + \psia$ gives a generalized form of Fick's law:
\begin{flalign}
    \notag & \left( \epsa \right)^2 \rhoaj \soten{R}^{\aj} \cd \bv{\aj,\al}
    = \\
    & \quad -\epsa \rhoaj \grad \muaj + \diver
    \left( \epsa \Foten{S}^{\aj} : \grad \bv{\aj,\al} \right) + \epsa
    \rhoaj \foten{g}.
    \label{eqn:Fick_new}
\end{flalign}
The term containing the gradient of diffusive velocity is likely negligible as it is
second order. If not, we would have to relate the fourth-order tensor,
$\Foten{S}^{\aj}$, with some physical
process (similar to viscosity for fluid flow).  If we neglect this term then Fick's law
can be written as
\begin{flalign}
    \epsa \rhoaj \soten{R}^{\aj} \cd \bv{\aj,\al} = - \rhoaj \grad
    \muaj + \rhoaj \foten{g}.
    \label{eqn:Ficks_Law}
\end{flalign}
Despite the novel choice of variables for this work, this form of Fick's law is identical
to that found by Bennethum and Murad \cite{Murad2000} and Weinstein \cite{weinstein}.  

The linearization coefficient in Fick's law has a similar meaning to that of the
resistivity tensor in Darcy's law.  In this case, though, we wish to associate the inverse
of this tensor with the diffusivity tensor from classical Fick's law. Assuming that the
inverse exists we have 
\begin{flalign}
    \epsa \rhoaj \bv{\aj,\al} &= - \rhoaj \soten{D}^{\aj} \cd \left[ \grad \mu^{\aj} -
        \foten{g} \right].
    \label{eqn:Ficks_v1}
\end{flalign}
The units of the left-hand side are [$M L^{-2} t^{-1}$], and the left-hand side term is
commonly known as {\it flux}.  Therefore, the units of $\soten{D}^{\aj}$ is simply time
[$t$].  Typically the diffusivity constant in a gas is measured as [$L^2 t^{-1}$], so we
correlate $\soten{D}^{g_j}$ to the diffusion coefficient for that phase, $\soten{D}^g$,
via the relationship
\begin{flalign}
    \soten{D}^{g_j} = \left( \frac{1}{R^{g_j} T} \right) \soten{D}^g,
    \label{eqn:diffusion_coeff_conversion}
\end{flalign}
where $R^{g_j}$ is the specific gas constant for constituent $j$. The units of $R^{g_j} T$
are [$L^2 t^{-2}$] and the units of $\soten{D}^g$ are [$L^2 t^{-1}$], hence making the
units of $\soten{D}^{g_j}$ [$t$]. This is consistent with the forms of Fick's law from
thermodynamics and physical chemistry \cite{Callen1985,McQuarrie1997}.  Hence, the gas
phase form of Fick's law is
\begin{flalign}
    \eps{g} \rho^{g_j} \bv{g_j,g} = - \left( \frac{\rho^{g_j}}{R^{g_j} T} \right)
    \soten{D}^g \cd \left[ \grad \mu^{g_j} - g \right]. 
    \label{eqn:Ficks_gas}
\end{flalign}

To close this subsection we finally recall from our discussion of pore-scale diffusion
(see Chapter \ref{ch:diffusion_comp}) that the diffusive velocities are related via
\begin{flalign}
    \sumj \rhoaj \bv{\aj,\al} = \foten{0}.
\end{flalign}
Multiplying by the volume fraction and recognizing the left-hand side of Fick's law
indicates that 
\begin{flalign}
    \sumj \left\{ \left( \frac{\rhoaj}{R^{g_j} T} \right) \soten{D}^\al \cd \left[ \grad
        \mu^{\aj} - \foten{g} \right] \right\} = \foten{0}
    \label{eqn:diffusion_ficks_sum}
\end{flalign}
near equilibrium.
Equation \eqref{eqn:diffusion_ficks_sum} simply states that the gradients in chemical
potential are not independent of each other.  This fact will be used in future chapters as
part of a moisture transport model.

\subsection{Fourier's Law}\label{sec:heat_flux}
The final result in this chapter is an extension to Fourier's Law for heat
conduction.  Notice that in the chemical potential form of Darcy's Law,
(\ref{eqn:Darcy_ChemPot}), there is a term that involves the gradient of temperature. That
is, the Darcy flux is partially driven by a gradient in temperature. This means that Darcy
flow is naturally driven by gradients in temperature as well as gradients in chemical
potential.  To properly handle this coupling we can either assume that the gradient of
temperature is zero (constant temperature) or consider the energy balance equation and track
temperature as well as chemical potential. To move toward a closed system of equations, we derive a version of
Fourier's Law from the entropy inequality so that we have an expression of heat flux in
the energy balance equation. At the outset we first recall that in the
entropy inequality we've assumed only one temperature for the entire porous medium.  This
implies that we've assumed that the separate phases are in thermal equilibrium.  For this
reason, we will develop an analogue to Fourier's Law that holds for the entire (bulk) medium.

Following Bennethum and Cushman's work on heat transport in porous media
\cite{Bennethum1999} we observe that if we sum the energy equation
(\ref{eqn:energy_phase}) over $\al$ we obtain the bulk energy balance equation
\begin{flalign}
    \rho \md{ }{e} - \stress{}: \grad \bv{} - \diver \bq - \rho h = 0,
    \label{eqn:summed_energy}
\end{flalign}
where
\begin{subequations}
    \begin{flalign}
        \rho &= \suma \epsa \rhoa \\
        \rho \bv{} &= \suma \epsa \rhoa \bv{\al} \\
        \foten{u}^{\al} &= \bv{\al} - \bv{} \\
        \stress{} &= \suma \left[ \epsa \stress{\al} - \epsa \rhoa
        \foten{u}^{\al} \otimes \foten{u}^{\al} \right] \\
        \rho e &= \suma \left[ \epsa \rhoa e^\al + \epsa \rhoa
        \foten{u}^{\al} \cd \foten{u}^{\al} \right] \\
        \bq &= \suma \left[ \epsa \bq^\al + \stress{\al} \cd
        \foten{u}^{\al} - \rhoa \foten{u}^{\al} \left( e^\al +
        \frac{1}{2} \foten{u}^{\al} \cd \foten{u}^{\al} \right) \right]
        \label{eqn:total_heat_flux} \\
        \rho h &= \suma \epsa \rhoa h^a.
    \end{flalign}
\end{subequations}
Given identities (a) - (g), the derivation of \eqref{eqn:summed_energy} follows after some
significant algebra. Define the medium velocity, $\bv{ }$, as the weighted velocity of the
medium, and the relative velocity, $\foten{u}^{\al}= \bv{\al} - \bv{}$, is the $\al-$phase
velocity relative to the medium. Note that $\bv{\al,s} = \bv{\al} - \bv{s} - \bv{ } + \bv{
} = \foten{u}^{\al} - \foten{u}^s = \foten{u}^{\al,s}$.  In the case where the velocity of
the solid phase relative to the medium is zero ($\foten{u}^s = \foten{0}$) we immediately
see that $\foten{u}^{\al} = \bv{\al,s}$.  This assumption along with equation (\ref{eqn:total_heat_flux}) indicates
that there is naturally a coupling between the relative velocities, $\bv{\al,s}$, and the
total heat flux, $\bq{}$. 

Using the near equilibrium result, $\suma \epsa \bq^\al = \soten{K} \cd \grad T$
(equation (\ref{eqn:heat_near_eq})), we can write the total heat flux as
\begin{flalign}
    \bq{} &= \soten{K} \cd \grad T + \suma \left[ \stress{\al} \cd
    \bv{\al,s} - \rhoa \bv{\al,s} \left( e^\al + \frac{1}{2} \bv{\al,s} \cd
    \bv{\al,s} \right) \right].
    \label{eqn:Fourier_Darcy_energy}
\end{flalign}
If we were to (wrongly) neglect all of the terms in the summation we would arrive at
Fourier's Law for heat conduction.  The trouble here is that the terms in the summation
are not negligible, and therefore the total heat flux in a porous medium must be a
function of the gradient of temperature, the relative velocities, the stress in the fluid
phases, and the internal energy.

Since $\bv{s,s} = \bv{s} - \bv{s} = \foten{0}$ the right-hand side of equation
(\ref{eqn:Fourier_Darcy_energy}) is only a function of the fluid velocities relative to
the solid phase.  Neglecting viscous terms we recall that the fluid-phase stress tensors
can be rewritten as $\stress{\al} = - p^{\al} \soten{I}$.  Neglecting the second-order
term, $\bv{\al,s} \cd \bv{\al,s}$, the total heat flux is now written as
\begin{flalign}
    \bq{} &= \soten{K} \cd \grad T - \sum_{\al = l,g} \left[ \left(
    p^{\al} + \rhoa e^\al \right) \bv{\al,s} \right].
\end{flalign}

At this point we replace the internal energy term with Gibbs energy in
hopes of deriving an extended Fourier's Law in terms of the chemical
potential.  Recall from thermodynamics that the Gibbs potential and
internal energy are related through
\begin{flalign}
    e^\al = \Gamma^\al + T \eta^\al - \frac{p^\al}{\rhoa}.
\end{flalign}
Therefore, the total heat flux can be written in terms of the Gibbs
potential as
\begin{flalign}
    \bq{} &= \soten{K} \cd \grad T - \sum_{\al=l,g} \left[ \rho^{\al}
    \left( \Gamma^\al + T \eta^\al \right) \bv{\al,s} \right].
    \label{eqn:Fourier_Gibbs}
\end{flalign}
Using the Gibbs-Duhem relationship, (\ref{eqn:Gibbs_Duhem}), this can be
rewritten in terms of the chemical potential as 
\begin{flalign}
    \bq{} &= \soten{K} \cd \grad T - \sum_{\al=l,g} \left[ \left( \sumj
    \left( \rho^{\aj} \mu^{\aj} \right) + \rhoa T \eta^\al \right)
    \bv{\al,s} \right].
    \label{eqn:Fourier_chem_pot}
\end{flalign}
The trouble with both (\ref{eqn:Fourier_Gibbs}) and (\ref{eqn:Fourier_chem_pot}) is that
they both rely on measurements of entropy. One way to work around this issue is to assume
that the entropy is only a function of temperature, and then to recall that the specific
heat is defined as
\[ c_p^\al = T \pd{\eta^\al}{T} = T \frac{d \eta^\al}{d T}. \]
Solving this separable ordinary differential equation (under the assumption that the
variation of specific heat with temperature negligible) gives
\begin{flalign}
    \eta^\al(T) = c_p \ln \left( \frac{T}{T_0} \right) + \eta^\al_0
    \label{eqn:entropy_temperature}
\end{flalign}
where $T_0$ is a reference temperature, and $\eta_0^\al$ is a reference entropy. While
this is only an approximation it does allow us to move forward without direct measurements
of entropy.

The {\it extended} Fourier's Law \eqref{eqn:Fourier_Gibbs} presented here frames the equations presented in
\cite{Bennethum1999} in terms of the Gibbs potential.  This will allow for easier coupling
with the chemical potential forms of Fick's and Darcy's Laws presented in the previous
subsections. The caveat is that the equation for total energy balance,
\eqref{eqn:summed_energy}, is not particularly useful since we do not have constitutive
relations for the total stress and total energy.  For that reason, we will not use
equation \eqref{eqn:Fourier_Gibbs} or \eqref{eqn:Fourier_chem_pot}
for Fourier's law in the energy equation.  Instead we will use the linearized partial
heat flux and the constitutive relations for the phase stresses and relative velocities to
derive a generalized heat equation.

\section{Conclusion}
In this chapter we have shown that a novel and judicious choice of independent variables
for the Helmholtz Free Energy can be used to derive forms of Darcy's, Fick's, and
Fourier's Laws for multiphase porous media.  These equations are similar to those found in
\cite{Bennethum1996,Bennethum1999,Murad2000,weinstein}. Each equation can be written with
an eye toward the macroscale chemical potential, and in each case the chemical potential
form is more {\it mathematically appealing} in the sense that there are fewer terms and
many of the physical processes are manifested in the chemical potentials.  This
illustrates the usefulness of the chemical potential as a modeling tool. Furthermore,
since the chemical potential appears naturally in each of these equations we have set the
stage for a more natural method of coupling the fluid flow, diffusion, and heat transport.
In Chapters \ref{ch:Transport} and \ref{ch:Transport_Solution} we will couple these
equations with the upscaled mass, momentum, and energy balance equations to yield a system
of equations that will govern total moisture transport and heat flux in unsaturated porous
media.

\newpage
\chapter{Coupled Heat and Moisture Transport Model}\label{ch:Transport}
To form governing equations for heat and  moisture transport in porous media we pair the
constitutive equations derived in Chapter \ref{ch:Exploit} with upscaled mass, momentum,
and energy balance equations derived in Chapter \ref{ch:HMT}. There are several existing
models for each physical process of interest (fluid flow, diffusion, and heat transport)
and recent research indicates a need to understand the fully coupled system of equations
as it relates to moisture transport, evaporation, heat transport, and other physical
phenomena. In this chapter we derive a model for coupled heat and moisture transport using
Hybrid Mixture Theory and knowledge of pore-scale effects. To begin this modeling task we
first investigate the classical models used within the past century in Section
\ref{sec:classical_models}. In Sections \ref{sec:Assumptions} and
\ref{sec:heat_and_mass_derivation} we pair our constitutive equations from
Chapter \ref{ch:Exploit} with upscaled balance laws from Chapter \ref{ch:HMT}, perform a
dimensional analysis, and discuss forms of the linearization coefficients arising from
HMT. This is done in an effort to generate a closed system of governing equations. Several
simplifying assumptions are made to close the system in Section
\ref{sec:simplifying_assumptions}. The solution(s) to the closed system will be discussed
in Chapter \ref{ch:Transport_Solution}.

\section{Introduction and Historical Work}\label{sec:classical_models}
To give the reader a better understanding of the work from the past century, we present
three classical models here with some discussion on their advantages and disadvantages. First we discuss
Richards' equation for unsaturated fluid flow in Section \ref{sec:richards}, second we
discuss Phillip and De Vries enhanced diffusion model in Section
\ref{sec:phillip_devries}, and lastly we discuss De Vries' heat transport model in Section
\ref{sec:devries_heat}.

\subsection{Richards' Equation for Fluid Flow}\label{sec:richards}
The classical equation for fluid flow in unsaturated media is known as Richards' equation
(also called the saturation equation).  This equation was first derived in 1931 by L.A.
Richards at Cornell University \cite{Richards1931}.  It takes a postulated form of the
mass balance equation (similar to equation \eqref{eqn:upscaled_mass_balance_phase}) and
replaces the flux term with Darcy's law. The gradient of pressure is rewritten in terms of
pressure head ($h = p/(\rho g)$), and then a constitutive relation is assumed for the
pressure head as a function of saturation (or volume fraction). Another constitutive
relation relating the relative permeability of the medium to saturation is assumed. There
are several versions of the constitutive relations, but one of the more {\it popular} in
recent research are those of van Genuchten \cite{VanGenuchten1980,Pinder2006}. Another
more recently investigated relationship is the Fayer-Simmons model
\cite{Fayer1995,Sakai2009,Smits2011}, which is an extension of the van Genuchten model to
cover the case of very low saturations. 

The result of the assumption and substitutions in the mass balance equation is a nonlinear
diffusion equation where the primary unknown is the percent saturation of the medium
\begin{flalign}
    \pd{S}{t} = \diver \left[ D(S) \grad S - K(S) \foten{z} \right], 
    \label{eqn:richards}
\end{flalign}
where $K(S)$ is the hydraulic conductivity function and $D(S)$ is the product of $K(S)$
and the derivative of capillary pressure with saturation.
Recall that saturation is defined as
\begin{flalign}
    S = \frac{\eps{l}}{1-\eps{s}} = \frac{\eps{l}}{\eps{g} + \eps{l}} =
    \frac{\eps{l}}{\porosity} = \frac{\text{volume of liquid}}{\text{volume of pore
    space}} 
    \label{eqn:defn_saturation_v1}
\end{flalign}
and is understood as the volume of liquid per volume of pore space.

This model has been effectively used for several decades, but there are a few
disadvantages of note.  First of all, this equation does not allow for phase change
between the liquid and gas.  The original model was proposed for systems with immiscible
fluids, where phase changes likely don't occur, but it is also used for unsaturated soils
where phase change is possible and air is always availabe.. A second
disadvantage is that humidity and temperature gradients are not considered.  A third disadvantage is that
the pressure head - saturation curve is hysteretic (depends on the history of
flow).  The constitutive laws for pressure head don't account for this hysteretic behavior
directly.  Instead, it is often assumed that fitting parameters change with changing
direction of flow.  This leads to the final disadvantage: the use of the van Genuchten capillary
pressure - saturation relation.  This is a widely used relationship, but relies heavily on
two fitting parameters. The measurement of these fitting parameters is difficult, and they
are typically found by fitting numerical solutions of Richards' equation to experimental
data.

Several extensions and modifications to Richards' equation have been made recently, the
most notable of which is that of Hassanizadeh et al.
\cite{Joekar-Niasar2010,Hassanizadeh2002}. In these papers, they propose a dynamic relationship between
capillary pressure and saturation based on Hybrid Mixture Theory with interfaces.  They
also propose that the hysteretic effect observed in the capillary pressure - saturation
curves is due to the (postulated) fact that the capillary pressure, saturation, and
interfacial area density, $\eps{lg}$, form a unique surface. This partially explains hysteretic effects by
seeing them as a projection of this surface onto the capillary pressures - saturation
plane in the $p_c-S-\eps{lg}$ space. This model is gaining in popularity, but is far from
widespread acceptance. Some of the relevant publications are
\cite{Joekar-Niasar2007,Joekar-Niasar2010,Joekar-Niasar2012,Hassanizadeh2001,Hassanizadeh2002,Nordbotten2008}.

In the present chapter we present a modification to the Richards' equation that
incorporates the dynamic capillary pressure relationship of Hassanizadeh et al.
The major differences between the present derivations and their work are: (1) 
modeling in terms of chemical potential, (2) allowing for phase transition, and (3)
allowing for humidity and temperature gradients. Our present modeling effort will account for all
of these effects, and hence, constitutes a generalization of the existing model.

\subsection{Phillip and De Vries' Diffusion Model}\label{sec:phillip_devries}
In 1957, Phillip and de Vries published their comprehensive work on diffusion of
water vapor in porous media \cite{deVries1957}. In their model they postulate
an {\it enhanced} Fick's law, 
\begin{flalign}
    \bq^{g_v} = - \rho^g D \eta \grad C^{g_v},
    \label{eqn:deVries_original}
\end{flalign}
where $\bq^{g_v}$ is the water vapor flux and $\eta$ is an {\it enhancement} factor that
is a function of the toruosity, volume fraction of air, and a ``mass-flow factor''. The
mass-flow factor is then postulated as a function of pore-scale gradients in saturation
and temperature.  This model has successfully been applied to several diffusion and
evaporation problems (e.g.  \cite{Smits2011}), but the trouble is that the exact form of
the enhancement is based on empirical evidence. Furthermore, this model
has come under recent scrutiny due to the fact that the proposed factors affecting $\eta$
are pore-scale effects and are therefore difficult to accurately measure
\cite{Celia2009,Shahraeeni2012a,Shahraeeni2012b,Shahraeeni2010,Shokri2009a,Silverman1999,Smits2011,Webb1998}.
Many of these works use x-ray tomography to attempt to measure these pore-scale effects
directly.

In the work by Cass et al.\,\cite{Cass1984}, an empirical form of the enhancement factor
was proposed. In this work, a fitting parameter is used in the enhancement factor to
arrive at good agreement with experimental data.  This model has been used in more
recent works (e.g. \cite{Sakai2009,Smits2011}) in conjunction with a mass balance equation
for the water vapor in the gas phase. The resulting model is a nonlinear diffusion
equation for concentration of water vapor that deviates from the more classical de Vries
model. Aside from the empirical fitting parameter, the mass transfer between phases also
relies on a fitting parameter and an empirically-derived functional form. 

In the present chapter we build a model for diffusion based on using the chemical
potential as a primary unknown and the Hybrid Mixture Theory construct.
The enhanced diffusion is not incorporated into these models, and the mass transfer is
modeled by the difference in chemical potentials; a more physically natural formulation. A comparison will
be made to the model of Cass et al.

\subsection{De Vries' Heat Transport Model}\label{sec:devries_heat}
In 1958, de Vries published a second paper coupling heat and moisture transport in porous
media \cite{deVries1958}.  In this research, he proposed an extended heat transport model
for porous media that is still used today.  Neither his diffusion nor his heat
transport model were thermodynamically derived. Instead, he began each derivation with a
postulation of the forms of diffusive and heat flux.  For the heat transport equation he
included terms similar to the classical Fourier's law, but also proposed that heat
transport was due to advective transport in the fluid phases.  This model is still
popularly used today to couple heat and mass transport in unsaturated
media \cite{Bear1988,Smits2010a,Smits2011}. That being said, the
effects included in this equations are based solely on de Vries' supposition of the
factors affecting heat flow.  

In 1999 Bennethum and Cushman published (to the author's knowledge) the first work using
Hybrid Mixture Theory to derive an extended de Vries model for heat transport in swelling
saturated porous media \cite{Bennethum1999}.  In the present chapter we take a similar
approach using HMT to derive a thermodynamically consistent model for heat transport in
non-swelling unsaturated media.  This is done with an eye toward using gradients in
temperature as the thermal diffusion process and the chemical potential to describe the
secondary processes such as advection. 

\section{Assumptions}\label{sec:Assumptions}
In this section we state the baseline assumptions that will be used throughout the
remainder of this work.  These assumptions are meant to make minimal limitations on the
applicability of the resulting models, but at the same time they are meant to keep the
mathematics tractable. Possible relaxations to these assumptions (and the source of
possible avenues of future research) will be stated as they are encountered.

The simple set of baseline assumptions are as follows:
\newcounter{AssNum}\stepcounter{AssNum}
\begin{description}
    \item[Assumption \#\theAssNum:\stepcounter{AssNum}] The solid phase is rigid,
        incompressible, and inert.
    \item[Assumption \#\theAssNum:\stepcounter{AssNum}] The liquid and gas phases are each
        made up of $N$ constituents.
    \item[Assumption \#\theAssNum:\stepcounter{AssNum}] No chemical reactions take place
        in any of the phases.
\end{description}

The first assumption is the most restrictive.  Mathematically it corresponds to setting
the Lagrangian derivatives of both density and volume fraction for the solid phase to
zero.  Assuming that the solid is inert simply means that no mass will precipitate onto,
or dissolve away from, the solid phase. With these assumptions, the solid phase mass
balance equation (from equation \eqref{eqn:upscaled_mass_balance_phase}) becomes
\begin{flalign}
    \diver \bv{s} = 0.
    \label{eqn:solid_mass_balance}
\end{flalign}
If a deformable solid is considered where the solid-phase volume fraction can change, then this assumption would need to be relaxed. One
particular relaxation of this assumption is to allow for incompressibility and inertness of the solid
phase but relax the rigidity assumption.  Under this relaxation, the solid phase mass
balance equation becomes
\begin{flalign}
    \md{s}{\eps{s}} - \eps{s} \diver \bv{s} = 0.
\end{flalign}

A consequence of fixing the solid phase volume is that $\eps{l} + \eps{g} = 1-\eps{s} :=
\porosity$, where $\porosity$ is known as the porosity of the porous medium. A further
consequence is that the liquid and gas phase volume fractions are no longer independent of
each other.  Note that we could have made this assumption up front and exploited the
entropy inequality with this assumption (this is done in
\cite{Hassanizadeh1986,Hassanizadeh1986a} for a
different set of independent variables), but proceeding in this order allows us to return
to the present entropy inequality results and consider a deformable solid in the future. Since the
fluid-phase volume fractions are no longer independent we can replace them by saturation
as defined by
\begin{flalign}
    S = \frac{\eps{l}}{\eps{l} + \eps{g}} = \frac{\eps{l}}{\porosity}.
    \label{eqn:defn_saturation}
\end{flalign}
This implies that the volume fractions are related via $\eps{l} = \porosity S$ and $\eps{g} = \porosity (1-S)$.

Assumption \#2 is a byproduct of the principle of equipresence and will be relaxed
later for simplicity. In the most general sense, this assumption states that every species that
exists in one fluid phase also exists in the other. In reality this is likely not
true. For example, if a constituent is present in the liquid phase it is possible that evaporated
particles of the constituent are not be present in the gas phase.  Another example would
be if we were to extend this model to an oil-water system.  The two fluids in this case
are immiscible and it is unlikely that every species in the water phase is present in the
oil phase (and visa versa). We take this into account by setting the appropriate
concentrations to zero after the constitutive equations have been derived. 

Assumption \#3 indicates that the rate of mass exchange due to chemical reactions,
$\hat{r}^{\aj}$, is zero for all phases. The consequence of this is that the rate of mass generation
of a constituent in a phase only occurs between two phases. This is true for
some porous media, but chemical reactions can occur in some specific cases such as
remediation 
problems.  Under this assumption these cases are henceforth eliminated from
the discussion.

Other simplifying assumptions exist for many media, but the three presented herein
constitute a set that leads to several mathematical simplifications with as
few physical restrictions as possible.

\section{Derivation of Heat and Moisture Transport
Model}\label{sec:heat_and_mass_derivation}
In the remainder of this chapter we focus on using the results from Chapters \ref{ch:HMT}
and \ref{ch:Exploit}, along with the assumptions from Section \ref{sec:Assumptions}, to
derive a closed system of equations for heat and mass transport in unsaturated porous
media. This will be done with an eye toward using the chemical potential as the driving
force for these processes.  We will show that under certain additional simplifying
assumptions that a closed system can be derived.

\subsection{Mass Balance Equations}\label{sec:mass_balance_equations}
We first build generalized mass balance equations in terms of the chemical potential under
assumptions \#1 - \#3.  Recall from Chapter \ref{ch:HMT} that the mass balance equation for
the $j^{th}$ constituent in the $\al-$phase is (from equation
\eqref{eqn:upscaled_cons_mass_v2})
\begin{flalign}
    \md{\aj}{\left( \epsa \rhoaj \right)} + \epsa \rhoaj \diver \bv{\aj} = \sumba \ehatbaj
    + \hat{r}^{\aj}.
\end{flalign}
The last term can be dropped under assumption \#3 in Section \ref{sec:Assumptions}.
Because of the form of the constitutive equation, and to adhere to the principle of frame
invariance, it is convenient to rewrite this equation relative to the solid phase.  To do
so we recall the identities
\begin{subequations}
    \begin{flalign}
        \md{\aj}{(\cd)} &= \md{\al}{(\cd)} + \bv{\aj,\al} \cd \grad (\cd) \\
        \md{\al}{(\cd)} &= \md{s}{(\cd)} + \bv{\al,s} \cd \grad (\cd)
    \end{flalign}
\end{subequations}
and expand the Lagrangian time derivatives accordingly to get
\begin{flalign}
    \md{s}{\left( \epsa \rhoaj \right)} + \bv{\aj,\al} \cd \grad \left( \epsa \rhoaj
    \right) + \bv{\al,s} \cd \grad \left( \epsa \rhoaj \right) + \epsa \rhoaj \diver
    \bv{\aj} = \sumba \ehatbaj.
\end{flalign}
Taking the definition of the Lagrangian time derivative,
\[ \md{s}{(\cd)} = \pd{(\cd)}{t} + \bv{s} \cd \grad (\cd), \]
adding and subtracting $\epsa \rhoaj \diver \bv{\al}$, and subtracting $\epsa \rhoaj \diver
\bv{s} = \foten{0}$ gives
\begin{flalign}
    \pd{\left( \epsa \rhoaj \right)}{t} + \diver \left( \epsa \rhoaj \bv{\aj,\al} \right) +
        \diver \left( \epsa \rhoaj \bv{\al,s} \right) = \sumba \ehatbaj.
    \label{eqn:general_mass_balance}
\end{flalign}
Notice the use of Assumption \#1 in the last step, and observe that if Assumption \#1 is
relaxed then the mass balance equation would involve a time derivative of the solid-phase
volume fraction (at least).

Equation \eqref{eqn:general_mass_balance} is the general mass balance equation for both of
the fluid phases.  Notice that we are not replacing the volume fractions with saturation
here since we don't know if $\al$ is the liquid or gas phase.  Substituting Fick's law for
the diffusive flux and Darcy's law for the Darcy flux gives the chemical potential form of
the full mass balance equation for species $j$ in phase $\al$:
\begin{flalign}
    \notag \pd{\left( \epsa \rhoaj \right)}{t} &- \diver \left\{ \rhoaj
        \soten{D}^{\aj} \cd \left[ \grad \mu^{\aj} - \foten{g} \right] \right\} \\
    \notag & - \diver \left\{ \rhoaj \soten{K}^\al \cd
    \left[ \sum_{k=1}^N \left( \rho^{\al_k} \grad \mu^{\al_k} \right) + \rhoa \eta^\al
        \grad T - \rho^\al \foten{g} \right] \right\} \\
    &= \sumba \ehatbaj.
    \label{eqn:general_mass_balance_fluid_species}
\end{flalign}
It should be noted here that the Eulerian and Lagrangian time derivatives are equal under
the assumption that the solid-phase velocity is zero (Assumption \#1). Also note that if we sum over all
constituents then we arrive at the mass balance equation for the phase (where we have used
$\sumj \rhoaj \bv{\aj,\al} = \foten{0}$)
\begin{flalign}
    \pd{\left( \epsa \rhoa \right)}{t} - \diver \left\{ \rhoa
    \soten{K}^\al \cd \left[ \sum_{k=1}^N \left( \rho^{\al_k} \grad \mu^{\al_k}
        \right) + \rhoa \eta^\al \grad T - \rho^\al \foten{g} \right] \right\}  = \sumba
        \ehatba.
    \label{eqn:general_mass_balance_fluid_phase}
\end{flalign}
The chemical potential form of the mass balance equation is only one form.  We could have
used the pressure formulation for Darcy's law and arrived at a pressure - chemical
potential form of the mass balance equation.

The rate of mass transfer term on the right-hand side of the mass balance equation can be
rewritten in terms of a linearized result from the entropy inequality.  Recall from
equation \eqref{eqn:mass_transfer_linearized} that the mass transfer term can be written
as
\begin{flalign}
    \ehatbaj &= \left[ \left( \rhoaj - \rho^{\beta_j} \right) M \right] \left( \mu^{\aj} -
    \mu^{\beta_j} \right), 
\end{flalign}
where the coefficient $(\rho^{\aj} - \rho^{\beta_j})$ is chosen to be consistent with
equation (9) of \cite{Smits2011}.  Also recall that since the interface is assumed to
contain no mass we must have that the rate of mass gained from the $\beta$ phase to the
$j^{th}$ species in the $\al$ phase must be equal to the rate of mass lost from the $\al$
phase to the $j^{th}$ species of the $\beta$ phase:
\[ \ehatbaj = -\ehat{\al}{\beta_j}. \]
If the chemical potential
of the liquid phase is larger than the chemical potential of the water vapor then mass
will transfer from liquid to gas and $\ehat{g}{l} <0$.  Similarly, if the chemical
potential of the liquid phase is smaller than that of the water vapor then mass will
transfer from gas to liquid and $\ehat{g}{l} >0$. 
Recall from the discussion adjacent to equation \eqref{eqn:mass_transfer_linearized} that
the units of $M$ are the reciprocal of flux.

There are clearly more unknowns than equations in the $2N$ fluid equations since we must
account for the densities, temperature, volume fractions, and entropies as well as the
chemical potentials.  Certain sets of simplifying assumptions can be used to reduce the
number of unknowns (e.g. incompressibility of a fluid phase). These will be discussed
in Section \ref{sec:simplifying_assumptions}. Instead of making these assumptions up front
we now turn our attention to deriving a generalized energy balance equation to account for
the temperature.  This will give one more equation but will add no more unknowns to the
system of equations.

\subsection{Energy Balance Equation}
As another step toward developing a closed system of equation equations for heat and
moisture transport we next examine the energy balance equation.  This will give an
equation in terms of temperature, chemical potentials, saturation (volume fractions),
entropy, and densities;  increasing the equation count but not
increasing the variable count. Since we assumed at the outset that all of the phases are in thermal
equilibrium we will only have one equation for energy balance.  This will be derived by
considering the sum of each of the phase energy balance equations. Counter-intuitively, we
will not use the form of Fourier's Law (equation \eqref{eqn:total_heat_flux} or
\eqref{eqn:Fourier_chem_pot}) derived for the total heat flux since the energy equation
derived in that section is more cumbersome to work with than the individual phase energy
equations. Instead we will use the partial heat flux for each phase as derived from
linearization about equilibrium \eqref{eqn:heat_near_eq}.

From equation \eqref{eqn:energy_phase}, the volume averaged energy balance equation is
\begin{flalign}
    \epsa \rhoa \md{\al}{e^{\al}} - \epsa \stress{\al} : \bd{\al} - \diver \left( \epsa
    \bq^{\al} \right) + \epsa \rhoa h^{\al} = \sumba \hat{Q}_{\beta}^{\al}.
\end{flalign}
Using the identity $\md{\al}{(\cd)} = \mds{(\cd)} +
\bv{\al,s} \cd \grad (\cd)$ and using {\it dot} notation for material time derivatives
allows us to rewrite the energy equation as
\begin{flalign}
    \epsa \rhoa \dot{e}^{\al} + \epsa \rhoa \bv{\al,s} \cd \grad e^{\al} - \epsa
    \stress{\al} : \bd{\al} - \diver \left( \epsa \bq^{\al} \right) + \epsa \rhoa h^{\al}
    = \sumba \hat{Q}_{\beta}^{\al}. 
    \label{eqn:energy_temporary1}
\end{flalign}
The trouble with \eqref{eqn:energy_temporary1} is that the first and second terms contain
the interal energy density, $e^\al$.  To tie this equation back to the HMT framework we've
used throughout (and to give the equation a more natural set of dependent variables) we
perform a Legendre transformation to change the energy term into the Helmholtz potential
via the thermodynamic identity $e^{\al} = \psia + T \eta^{\al}$. The energy equation is
now written as
\begin{flalign}
    \notag \sumba \hat{Q}_{\beta}^{\al} =& \epsa \rhoa \mds{\psia} + \epsa \rhoa
    \bv{\al,s} \cd \grad \psia + \epsa \rhoa T \dot{\eta}^{\al} + \epsa \rhoa T \bv{\al,s}
    \cd \grad \eta^{\al} \\
    &+ \epsa \rhoa \eta^{\al} \dot{T} + \epsa \rhoa \eta^{\al} \bv{\al,s} \cd \grad T -
    \epsa \stress{\al} : \bd{\al} - \diver \left( \epsa \bq^{\al} \right) + \epsa \rhoa
    h^{\al}.
    \label{eqn:energy_temporary2}
\end{flalign}

Next we seek to remove the Helmholtz potential and entropy terms from the energy equation.
To do this we recall that the Helmholtz potential is a function of all of the variables
listed in \eqref{eqn:variables}. Under the assumptions listed in Section
\ref{sec:Assumptions} we drop the solid phase terms from this list.  
Furthermore, we know that under these
conditions the volume fractions are not independent so we could replace both $\eps{l}$
and $\eps{g}$ by saturation, $S$. This is not done (yet) as the entropy inequality was
exploited while assuming that they are independent.  The switch can be made at any point
later. Therefore, under the present assumptions,
\[ \psia = \psia(\eps{l}, \eps{g}, \rho^{l_j}, \rho^{g_j}, T)
    \quad \text{for} j=1:N. \]
Entropy, $\eta^\al$, is assumed to be a function of the same set of variables (since $\eta^\al =
-\partial \psia / \partial T$). Using the chain rule to expand all of the derivatives of
$\psia$ and $\eta^\al$ in equation \eqref{eqn:energy_temporary2} we arrive at an expanded
form of the energy equation:
\begin{flalign}
    \notag \sumba \hat{Q}_{\beta}^{\al} =& \epsa \rhoa \left( \pd{\psia}{T} + \eta^{\al} + T
    \pd{\eta^{\al}}{T} \right) \dot{T} \\
    \notag &  + \epsa \rhoa \left( \left[ \pd{\psia}{\eps{l}} + T \pd{\eta^{\al}}{\eps{l}}
    \right] \epsdot{l} + \left[ \pd{\psia}{\eps{g}} + T \pd{\eta^{\al}}{\eps{g}}
    \right] \epsdot{g} \right. \\
    %
    %
    \notag &  \qquad \qquad \left. + \sumj \left[ \pd{\psia}{\rho^{l_j}} + T
    \pd{\eta^{\al}}{\rho^{l_j}} \right] \dot{\rho}^{l_j} + \sumj \left[
    \pd{\psia}{\rho^{g_j}} + T \pd{\eta^{\al}}{\rho^{g_j}} \right]
    \dot{\rho}^{g_j} \right) \\
    \notag &  + \epsa \rhoa \left( \left[ \pd{\psia}{T} + \eta^{\al} + T \pd{\eta^{\al}}{T}
    \right] \grad T \right. \\
    \notag &  \qquad \qquad + \left[ \pd{\psia}{\eps{l}} + T \pd{\eta^{\al}}{\eps{l}}
    \right] \grad \eps{l} + \left[ \pd{\psia}{\eps{g}} + T
    \pd{\eta^{\al}}{\eps{g}} \right] \grad \eps{g} \\
    %
    %
    \notag &  \qquad \qquad \left. + \sumj \left[ \pd{\psia}{\rho^{l_j}} + T
        \pd{\eta^{\al}}{\rho^{l_j}} \right] \grad \rho^{l_j} + \sumj \left[
            \pd{\psia}{\rho^{g_j}} + T \pd{\eta^{\al}}{\rho^{g_j}} \right] \grad
            \rho^{g_j} \right) \cd \bv{\al,s} \\
    &   - \epsa \stress{\al} : \bd{\al} - \diver \left( \epsa \bq^{\al}
    \right) + \epsa \rhoa h^{\al}.
    \label{eqn:energy_temporary3}
\end{flalign}

From the entropy inequality we know that the temperature and entropy are conjugate variables.
For this reason we can cancel these terms from the $\dot{T}$ and $\grad T$ coefficients.

Equation \eqref{eqn:energy_temporary3} is an expression of energy balance for phase $\al$,
but since we are working under the assumption that the phases are in thermal equilibrium
we now sum over all of the phases to form one energy balance equation for the entire
porous medium. The sum is:
\begin{flalign}
    \suma \left\{ \sumba \hat{Q}_{\beta}^{\al} \right\} &= \suma \left\{ \epsa \rhoa T
        \pd{\eta^{\al}}{T} \right\} \dot{T} \\
    \notag & + \suma \left\{ \epsa \rhoa \left[ \pd{\psia}{\eps{l}} + T
    \pd{\eta^{\al}}{\eps{l}} \right] \right\} \epsdot{l} + \suma \left\{ \epsa \rhoa
    \left[ \pd{\psia}{\eps{g}} + T \pd{\eta^{\al}}{\eps{g}} \right] \right\}
    \epsdot{g} \\
    %
    %
    \notag &  + \suma \left\{ \epsa \rhoa \sumj \left[ \pd{\psia}{\rho^{l_j}} + T
    \pd{\eta^{\al}}{\rho^{l_j}} \right] \dot{\rho}^{l_j} \right\} \\
    \notag & + \suma \left\{ \epsa \rhoa \sumj \left[ \pd{\psia}{\rho^{g_j}} + T
        \pd{\eta^{\al}}{\rho^{g_j}} \right] \dot{\rho}^{g_j} \right\} \\
    \notag &  + \sum_{\beta = l,g} \left\{ \eps{\beta} \rho^{\beta} \left( \left[ T
        \pd{\eta^{\beta}}{T} \right] \grad T \right. \right. \\
    %
    \notag & \qquad \qquad \qquad + \left[ \pd{\psi^{\beta}}{\eps{l}} + T
        \pd{\eta^{\beta}}{\eps{l}} \right] \grad \eps{l} + \left[
            \pd{\psi^{\beta}}{\eps{g}} + T \pd{\eta^{\beta}}{\eps{g}} \right] \grad
            \eps{g} \\
    %
    \notag & \qquad\qquad \qquad + \sumj \left[ \pd{\psi^{\beta}}{\rho^{l_j}} + T
        \pd{\eta^{\beta}}{\rho^{l_j}} \right] \grad \rho^{l_j} \\
    \notag &  \qquad\qquad \qquad \left. \left. + \sumj \left[
        \pd{\psi^{\beta}}{\rho^{g_j}} + T \pd{\eta^{\beta}}{\rho^{g_j}} \right] \grad
        \rho^{g_j} \right) \cd \bv{\beta,s} \right\} \\
    %
    %
    &   - \suma \left\{ \epsa \stress{\al} : \bd{\al} \right\} - \diver \left( \suma
    \left\{ \epsa \bq^{\al} \right\} \right) + \suma \left\{ \epsa \rhoa h^{\al} \right\} 
\end{flalign}

The $\dot{T}$ term can be rewritten as $\suma \left\{ \epsa \rhoa T \pd{\eta^{\al}}{T}
\right\} \dot{T} = \rho c_p \dot{T}$, and in doing so we implicitly define the volumetric 
heat capacity of the entire medium:
\[ \rho c_p = \suma \left\{ \epsa \rhoa T \pd{\eta^\al}{T} \right\}. \]
Next we recall from
equation \eqref{eqn:heat_near_eq} that the partial heat flux can be written as $\suma
\left\{ \epsa \bq^{\al} \right\} = \soten{K} \cd \grad T$ (more will be said about the
functional form of $\soten{K}$ in future sections).  The heat source term can be
rewritten as $\suma \left\{ \epsa \rhoa h^{\al} \right\} = \rho h$, where $h$ is any
internal source or sink of heat on the entire medium (i.e. heat sources that are not
boundary conditions).  

Notice that several of the gradient terms are the same as those in the linearized
constitutive equation for the momentum transfer, \eqref{eqn:phase_mom_transfer_equilib} and
\eqref{eqn:mom_near_eq}. Replacing these terms with the remainder of the momentum balance
terms and simplifying gives
\begin{flalign}
    \notag \suma \left\{ \sumba \hat{Q}_{\beta}^{\al} \right\} &= \rho c_p \dot{T} - \diver
    \left( \soten{K} \cd \grad T \right) + \rho h - \suma \left\{ \epsa \stress{\al} :
    \bd{\al} \right\} \\
    \notag & \quad + \suma \left\{ \epsa \rhoa \left[ \pd{\psia}{\eps{l}} + T
    \pd{\eta^{\al}}{\eps{l}} \right] \right\} \epsdot{l} + \suma \left\{ \epsa \rhoa
    \left[ \pd{\psia}{\eps{g}} + T \pd{\eta^{\al}}{\eps{g}} \right] \right\}
    \epsdot{g} \\
    \notag & \quad + \suma \left\{ \epsa \rhoa \sumj \left[ \pd{\psia}{\rho^{l_j}} + T
        \pd{\eta^{\al}}{\rho^{l_j}} \right] \dot{\rho}^{l_j} \right\} \\
    \notag & \quad + \suma \left\{ \epsa
    \rhoa \sumj \left[ \pd{\psia}{\rho^{g_j}} + T \pd{\eta^{\al}}{\rho^{g_j}} \right]
    \dot{\rho}^{g_j} \right\} \\
    \notag & \quad +\sum_{\beta = l,g} \left\{ \left[ -\sum_{\gamma \ne \beta} \left(
        \That{\gamma}{\beta} \right) + p^{\beta} \grad \eps{\beta} \right. \right. \\
    \notag & \qquad \qquad \qquad + \sumj \left[ \left( \suma \left( \eps{\al} \rho^\al
        \pd{\psi^\al}{\rho^{\beta_j}} \right) + \eps{\beta} \rho^\beta T
        \pd{\eta^\beta}{\rho^{\beta_j}} \right) \grad \rho^{\beta_j} \right] \\ 
    \notag & \quad \qquad \qquad + \eps{\beta} \rho^{\beta} \left\{ \left[ T \pd{\eta^{\beta}}{T}
        \right] \grad T + \left[ T \pd{\eta^{\beta}}{\eps{l}} \right] \grad \eps{l} +
        \left[ T \pd{\eta^{\beta}}{\eps{g}} \right] \grad \eps{g} \right. \\
    %
    %
    %
    & \quad \qquad \qquad \qquad \qquad \left. \left. \left. + \sumj \left( \left[ T
        \pd{\eta^{\beta}}{\rho^{\gamma_j}} \right] \grad \rho^{\gamma_j} \right) \right\} \right] \cd
        \bv{\beta,s} \right\}. 
    \label{eqn:energy_temporary4}
\end{flalign}

Equation \eqref{eqn:energy_temporary4} expresses the energy balance for the bulk porous
medium.  Several of the terms can be simplified at this point.  Toward this goal, we will
\begin{enumerate}
    \item derive a relation for the energy transfer terms: $\suma \sumba
        \hat{Q}_{\beta}^{\al}$
    \item rewrite the stress term, $\suma \epsa \stress{\al} : \bd{\al}$, using
        constitutive relationships for $\stress{\al}$
    \item rewrite the momentum transfer terms, $\That{\gamma}{\beta}$, using the
        linearized momentum transfer from the entropy inequality,
        \eqref{eqn:mom_near_eq}
    \item rewrite the advective terms, $\bv{\beta,s}$, using Darcy's law, and
    \item relate the changes in entropy, $\pd{\eta^\al}{(\cd)}$, to material coefficients.
\end{enumerate}
The first two of these are discussed in the following two subsections. The third and
fourth come as a consequence of the first two, and the fifth will be discussed under
proper simplifications in future sections.

\subsubsection{Energy Transfer in the Total Energy Equation}
Consider the energy transfer and stress terms: $\hat{Q}_{\beta}^\al, \hat{Q}^{\aj}$,and $\stress{\al}$. From equations
\eqref{eqn:energy_restriction1} and \eqref{eqn:energy_restriction2} we recall that the
restrictions on the interface are
\begin{subequations}
    \begin{flalign}
        \sumj \left[ \hat{Q}^{\aj} + \hat{\foten{i}}^{\aj} \cd \bv{\aj,\al} + \hat{r}^{\aj}
            \left( e^{\aj} + \frac{1}{2} \bv{\aj,\al} \cd \bv{\aj,\al} \right) \right] &= 0
            \quad \forall \al, \\
            \suma \sumba \left[ \hat{Q}_{\beta}^{\aj} + \Thatbaj \cd \bv{\aj} + \ehatbaj
                \left( e^{\aj} + \frac{1}{2} \bv{\aj} \cd \bv{\aj} \right) \right] &= 0 \quad
                j=1:N.
            \end{flalign}
\end{subequations}
We also note the identity  
\begin{flalign}
    \hat{Q}_{\beta}^{\al} = \sumj \left[ \hat{Q}_{\beta}^{\aj} + \Thatbaj \cd \bv{\aj,\al}
    + \ehatbaj \left( e^{\aj,\al} + \frac{1}{2} \bv{\aj,\al} \cd \bv{\aj,\al}
    \right)\right]
\end{flalign}
(see Appendix A.2 of \cite{weinstein}).
With these three identities, the sum of the energy transfer terms can be written as
\begin{flalign}
    \notag \suma \sumba \hat{Q}_{\beta}^{\al} &= \suma \sumba \sumj \left[ \hat{Q}_{\beta}^{\aj}
        + \Thatbaj \cd \bv{\aj,\al} + \ehatbaj \left( e^{\aj,\al} + \frac{1}{2}
        \bv{\aj,\al} \cd \bv{\aj,\al} \right)\right] \\
    %
    %
    \notag &= \sumj \left\{ \suma \sumba \left[ \hat{Q}_{\beta}^{\aj} \right] \right. \\
    \notag & \qquad \qquad + \left. \suma \sumba \left[ \Thatbaj \cd \bv{\aj,\al} + \ehatbaj
        \left( e^{\aj,\al} + \frac{1}{2} \bv{\aj,\al} \cd \bv{\aj,\al} \right) \right]
    \right\} \\
        \notag &= \sumj \left\{ -\suma \sumba \left[ \Thatbaj \cd \bv{\aj} + \ehatbaj \left(
            e^{\aj} + \frac{1}{2} \bv{\aj}\cd \bv{\aj}
        \right)\right]  \right. \\
        \notag & \qquad \quad \left. + \suma \sumba \left[ \Thatbaj \cd \bv{\aj,\al} + \ehatbaj
            \left( e^{\aj,\al} + \frac{1}{2} \bv{\aj,\al} \cd \bv{\aj,\al} \right) \right]
        \right\} \\
    \notag &= -\sumj \left\{ \suma \sumba \left[ \Thatbaj \cd \bv{\al} + \ehatbaj e^{\al}
        \phantom{\frac{1}{2}} \right. \right. \\
    & \qquad \qquad \qquad \qquad \quad \left. \left. - \frac{1}{2} \ehatbaj \left(
    \bv{\aj,\al} \cd \bv{\aj,\al} - \bv{\aj}\cd \bv{\aj} \right) \right] \right\}. 
    \label{eqn:energy_transfer_temporary1}
\end{flalign}
%

Next we examine the momentum transfer term appearing in equation
\eqref{eqn:energy_transfer_temporary1}. Recall from equation
\eqref{eqn:mom_trans_phase_species} that 
\begin{flalign}
    \Thatba = \sumj \left[ \Thatbaj + \ehatbaj \bv{\aj,\al} \right].
\end{flalign}
Rearranging this identity and multiplying by the $\al-$phase velocity we see that  
\begin{flalign}
    \sumj \Thatbaj \cd \bv{\al} = \Thatba \cd \bv{\al} - \sumj \left[ \ehatbaj
        \bv{\aj,\al} \cd \bv{\al} \right].
    \label{eqn:momentum_transfer_temporary1}
\end{flalign}
Substituting \eqref{eqn:momentum_transfer_temporary1} into
\eqref{eqn:energy_transfer_temporary1}, simplifying, and neglecting the second-order terms
in velocity we see that 
\begin{flalign}
    \suma \sumba \left\{ \hat{Q}_{\beta}^\al \right\} &= -\sum_{\beta \ne l} \left\{
        \That{\beta}{l} \cd \bv{l,s} \right\} - \sum_{\beta \ne g} \left\{ \That{\beta}{g}
        \cd \bv{g,s} \right\} - \ehat{g}{l}\left( e^l - e^g \right).
        \label{eqn:energy_sum_energy_transfer}
\end{flalign}
Notice from this simplified version that we have eliminated the energy transfer in favor
of the mass and momentum transfer terms after summing over $\al$ (and neglecting
second-order effects).

\subsubsection{Stress in the Total Energy Equation}
We next derive the proper form of the stress term in equation
\eqref{eqn:energy_temporary4}.  The $\al-$phase stress near equilibrium is given by $\stress{\al} = -p^{\al} \soten{I} +
\Foten{\nu}^{\al} : \bd{\al}$ from the linearization of the fluid phase stress tensors
about equilibrium.  For the solid phase stress tensor, on the other hand, we will not use
constitutive relations for $\stress{s}$ but keep in mind that it is the sum of 
effective and
hydrating stresses (see equation \eqref{eqn:solid_stress}).  Therefore,
\begin{flalign}
    \notag \suma \left\{ \epsa \stress{\al} : \bd{\al} \right\} &= \sum_{\al = l,g}
    \left\{ \epsa \left( -p^{\al} \soten{I} + \Foten{\nu}^{\al} : \bd{\al} \right) :
    \bd{\al} \right\} + \eps{s} \stress{s} : \bd{s} \\
    &= -\sum_{\al = l,g} \left\{ \epsa p^{\al} \soten{I} : \bd{\al} \right\} + \sum_{\al =
    l,g} \left\{ \Foten{\nu}^{\al} : \bd{\al} : \bd{\al} \right\} + \eps{s} \stress{s} :
    \bd{s}.
\end{flalign}
The second term is likely negligible as the viscous terms typically play little role in
creeping flow. This means that $\suma \left\{ \epsa \stress{\al} : \bd{\al} \right\}$ can
be approximated by
\begin{flalign}
    \suma \left\{ \epsa \stress{\al} : \bd{\al} \right\} = - \sum_{\al=l,g} \left\{ \epsa
    p^{\al} \soten{I} : \bd{\al} \right\} + \eps{s} \stress{s} : \bd{s}.
\end{flalign}
Using indicial notation we note that for the fluid phases, $\soten{I} : \bd{\al} =
\soten{I} : (\grad \bv{\al})_{sym} = \delta_{ij} v^{\al}_{j,i} = v^{\al}_{i,i} = \diver
\bv{\al},$ and therefore the stress tensor terms can be simplified to
\begin{flalign}
    \suma \left\{ \epsa \stress{\al} : \bd{\al} \right\} = - \sum_{\al=l,g} \left\{ \epsa
    p^{\al} \diver \bv{\al} \right\} + \eps{s} \stress{s} : \bd{s}.
\end{flalign}

The solid phase rate-of-deformation tensor is related to the strain rate of the solid
phase.  Assuming that the strain rate is zero (for a rigid and incompressible solid), we
can neglect this term.  This implies that the stress tensor term in
\eqref{eqn:energy_temporary4} can be approximated by 
\[ \suma \left\{ \epsa \stress{\al} : \bd{\al} \right\} = - \sum_{\al = l,g} \left\{ \epsa
    p^\al \diver \bv{\al} \right\}. \]
Using Assumption \# 1 from Section \ref{sec:Assumptions} for the divergence of the
solid-phase velocity ($\diver \bv{s} = 0$), we finally conclude that the stress term in
\eqref{eqn:energy_temporary4} can be simplified to  
\begin{flalign}
    \suma \left\{ \epsa \stress{\al} : \bd{\al} \right\} = - \eps{l} p^{l} \diver \bv{l,s}
    - \eps{g} p^g \diver \bv{g,s}.
    \label{eqn:energy_sum_stress}
\end{flalign}
Not surprisingly, this states that the stress is related to the fluid pressures.

\subsubsection{Total Energy Balance Equation}\label{sec:TotalEnergyBalanceEqn}
In this subsection we use equations \eqref{eqn:energy_sum_energy_transfer} and
\eqref{eqn:energy_sum_stress} to simplify the energy balance equation,
\eqref{eqn:energy_temporary4}. Substituting these into \eqref{eqn:energy_temporary4} and canceling the momentum
transfer terms gives
\begin{flalign}
    \notag 0=& \rho c_p \dot{T} - \diver \left( \soten{K} \cd \grad T \right) + \rho h + \eps{l}
    p^l \diver \bv{l,s} + \eps{g} p^g \diver \bv{g,s} + \ehat{g}{l}\left( e^l - e^g \right) \\
    \notag &  + \suma \left\{ \epsa \rhoa \left[ \pd{\psia}{\eps{l}} + T
    \pd{\eta^{\al}}{\eps{l}} \right] \right\} \epsdot{l} + \suma \left\{ \epsa \rhoa
    \left[ \pd{\psia}{\eps{g}} + T \pd{\eta^{\al}}{\eps{g}} \right] \right\}
    \epsdot{g} \\
    \notag & + \suma \left\{ \epsa \rhoa \sumj \left[ \pd{\psia}{\rho^{l_j}} + T
        \pd{\eta^{\al}}{\rho^{l_j}} \right] \dot{\rho}^{l_j} \right\} \\
    \notag & + \suma \left\{ \epsa
    \rhoa \sumj \left[ \pd{\psia}{\rho^{g_j}} + T \pd{\eta^{\al}}{\rho^{g_j}} \right]
    \dot{\rho}^{g_j} \right\} \\
    \notag & + \sum_{\beta = l,g} \left\{ \left[ p^{\beta} \grad \eps{\beta} + \sumj
        \left[ \left( \suma \left( \eps{\al} \rho^\al \pd{\psi^\al}{\rho^{\beta_j}} \right) +
            \eps{\beta} \rho^\beta T \pd{\eta^\beta}{\rho^{\beta_j}} \right) \grad
            \rho^{\beta_j} \right] \right. \right. \\ 
    \notag & \qquad \qquad + \eps{\beta} \rho^{\beta} \left\{ \left[ T
        \pd{\eta^{\beta}}{T} \right] \grad T + \left[ T \pd{\eta^{\beta}}{\eps{l}} \right]
        \grad \eps{l} + \left[ T \pd{\eta^{\beta}}{\eps{g}} \right] \grad \eps{g} \right.
        \\
    %
    %
    %
    & \qquad \qquad \qquad \qquad \left. \left. \left. + \sumj \left( \left[ T
        \pd{\eta^{\beta}}{\rho^{\gamma_j}} \right] \grad \rho^{\gamma_j} \right) \right\}
        \right] \cd \bv{\beta,s} \right\}. 
\end{flalign}

Next we discuss the $\eps{\al} p^{\al} \diver \bv{\al,s}$ and $p^\al \grad \left(
\eps{\al} \right) \cd \bv{\al,s}$ terms.  Using the product rule it is clear that the sum
of these two terms gives $p^{\al} \diver \left( \eps{\al} \bv{\al,s} \right)$.
A choice is made here to remove these terms in lieu of mass transfer terms.  To do so, we recall from the mass balance
equation that
\[ \mds{\left( \epsa \rhoa \right)} + \diver \left( \epsa \rhoa \bv{\al,s} \right) =
    \sumba \ehatba, \]
and solve for $\diver \left( \epsa \bv{\al,s} \right)$: 
\begin{flalign*}
    \rhoa \diver \left( \epsa \bv{\al,s} \right) &= -\rhoa \epsdot{\al} - \eps{\al}
    \dot{\rho}^{\al} - \epsa \bv{\al,s} \cd \grad \rho^{\al} + \ehatba. 
\end{flalign*}
We have dropped the summation on the mass transfer term since we are assuming that the
solid phase is inert and that there are only two fluid phases.  Multiplying by $(p^\al /
\rhoa)$ gives an expression for $p^\al \diver \left( \epsa \bv{\al,s} \right)$:
\begin{flalign}
    p^\al \diver \left( \epsa \bv{\al,s} \right) &= -p^\al \epsdot{\al} - \left(
    \frac{\epsa p^\al}{\rhoa} \right) \dot{\rho}^\al - \left( \frac{\epsa p^\al}{\rhoa}
    \right) \bv{\al,s} \cd \grad \rhoa + \left( \frac{p^\al}{\rhoa} \right)
    \ehatba.
\end{flalign}

Substituting this into the energy equation gives
\begin{flalign}
    \notag 0=& \rho c_p \dot{T} - \diver \left( \soten{K} \cd \grad T \right) + \rho h +
    \left( \left( \frac{p^l}{\rho^l} + e^l \right) - \left(
    \frac{p^g}{\rho^g} + e^g \right) \right) \ehat{g}{l} \\
    \notag & + \left\{ -p^l +  \suma \left\{ \epsa \rhoa \left[ \pd{\psia}{\eps{l}} + T
        \pd{\eta^{\al}}{\eps{l}} \right] \right\} \right\} \epsdot{l} \\
    \notag & + \left\{ -p^g +  \suma \left\{ \epsa \rhoa \left[ \pd{\psia}{\eps{g}} + T
        \pd{\eta^{\al}}{\eps{g}} \right] \right\} \right\} \epsdot{g} \\
    \notag & + \sumj \left( \left[ - \left( \frac{\eps{l} p^l}{\rho^l} \right) +
        \suma \left\{ \epsa \rhoa \left[ \pd{\psia}{\rho^{l_j}} + T \pd{\eta^{\al}}{\rho^{l_j}}
        \right] \right\} \right] \dot{\rho}^{l_j} \right) \\
    \notag & + \sumj \left( \left[ - \left( \frac{\eps{g} p^g}{\rho^g} \right) +
        \suma \left\{ \epsa \rhoa \left[ \pd{\psia}{\rho^{g_j}} + T \pd{\eta^{\al}}{\rho^{g_j}}
        \right] \right\} \right] \dot{\rho}^{g_j} \right) \\
    %
    %
    \notag & + \sum_{\beta = l,g} \left\{ \left[ \sumj \left( \left[ - \left(
        \frac{\eps{\beta} p^\beta}{\rho^\beta} \right) + \suma \left\{ \eps{\al}
        \rho^\al \left[ \pd{\psi^\al}{\rho^{\beta_j}} + T
            \pd{\eta^\beta}{\rho^{\beta_j}} \right] \right\} \right] \grad \rho^{\beta_j}
            \right) \right. \right. \\ 
    \notag & \qquad \qquad + \eps{\beta} \rho^{\beta} \left\{ \left[ T
        \pd{\eta^{\beta}}{T} \right] \grad T + \left[ T \pd{\eta^{\beta}}{\eps{l}} \right]
        \grad \eps{l} + \left[ T \pd{\eta^{\beta}}{\eps{g}} \right] \grad \eps{g} \right.
        \\
    %
    %
    %
    & \qquad \qquad \qquad \qquad \left. \left. \left. + \sumj \left( \left[ T
        \pd{\eta^{\beta}}{\rho^{\gamma_j}} \right] \grad \rho^{\gamma_j} \right) \right\}
        \right] \cd \bv{\beta,s} \right\}. 
    \label{eqn:energy_temporary5}
\end{flalign}

There are several more simplifications that can be made. To help with these
simplifications recall the following definitions for enthalpy, pressure, wetting
potential, chemical potential, and entropy respectively:
\begin{subequations}
    \begin{flalign}
        H^\al &= \frac{p^\al}{\rhoa} + e^\al \\
        p^\beta &= \suma \sumj \left( \frac{\epsa \rhoa \rho^{\beta_j}}{\eps{\beta}}
        \pd{\psia}{\rho^{\beta_j}} \right) \\
        \pi^\beta &= \suma \left( \epsa \rhoa \pd{\psia}{\eps{\beta}} \right) \\
        \mu^{\beta_j} &= \psi^\beta + \suma \left( \frac{\epsa \rhoa}{\eps{\beta}}
        \pd{\psia}{\rho^{\beta_j}} \right) \\
        \eta^\al &= -\pd{\psia}{T}.
    \end{flalign}
\end{subequations}
With these identities in mind we make the following four simplifications:
\begin{enumerate}
    \item coefficient of the mass transfer term:
        \[ \left( \left( \frac{p^l}{\rho^l} + e^l \right) - \left(
            \frac{p^g}{\rho^g} + e^g \right) \right) \ehat{g}{l} = \left(H^l - H^g \right)
            \ehat{g}{l} := L \ehat{g}{l} \]
            Recalling that $\ehat{g}{l}$ is the rate of mass transfer between the fluid
            phases, $L$ is understood as the latent heat of evaporation since this
            represents the heat lost or gained due to phase exchanged between the fluids.
            This is consistent with the chemist's definition of latent heat as the change
            in enthalpy.

    \item coefficient of the time rates of change of volume fractions:
        \begin{flalign*}
            & -p^\beta + \suma \left\{ \epsa \rhoa \left[ \pd{\psia}{\eps{\beta}} + T
            \pd{\eta^\al}{\eps{\beta}} \right] \right\} \\
            & \quad = -p^\beta +
            \pi^\beta + T \pd{\pi^\beta}{T} \\
            & \quad = -\overline{p}^\beta - T \pd{\pi^\beta}{T},
        \end{flalign*}
        where we recall that $\overline{p}^\beta$ is thermodynamic pressure as defined in
        Chapter \ref{ch:Exploit}
        \[ p^\beta = \overline{p}^\beta + \pi^\beta. \]
        At this point we can exchange the time rates of change of volume fractions for
        time rates of change of saturation.  That is, recall $\epsdot{l} = \porosity
        \Sdot$ and $\epsdot{g} = -\porosity \Sdot$.  The sum of the two associated terms
        is
        \begin{flalign*}
            & \left( -\overline{p}^l - T \pd{\pi^l}{T} \right) \porosity \Sdot - \left(
            -\overline{p}^g - T \pd{\pi^g}{T} \right) \porosity \Sdot \\
            & \quad = \left[ \left(
                \overline{p}^g - \overline{p}^l \right) + T \left( \pd{\pi^g}{T} -
                \pd{\pi^l}{T} \right) \right] \porosity \Sdot.
        \end{flalign*}

        From the near equilibrium results from the entropy inequality we now recall (from
        equation \eqref{eqn:epsdot_near_eq}) that
        \begin{flalign}
            \overline{p}^\beta \Big|_{n.eq.} = \overline{p}^\beta \Big|_{eq.} - \tau
            \epsdot{\beta}
        \end{flalign}
        Therefore, the $\Sdot$ term becomes
        \begin{flalign}
            \left[ \left( \overline{p}^g \Big|_{eq.} - \overline{p}^l \Big|_{eq.} \right)
                + 2\tau \porosity \Sdot + T \pd{ }{T} \left( \pi^g - \pi^l \right) \right]
            \porosity \Sdot.
            \label{eqn:Sdot_energy_term}
        \end{flalign}
        The first set of parenthesis in \eqref{eqn:Sdot_energy_term} (approximately)
        represents the capillary pressure as measured at equilibrium,
        \[ \left( \overline{p}^g \Big|_{eq.} - \overline{p}^l \Big|_{eq.} \right) =
            p_c. \]
        This will be discussed in more detail in Section
        \ref{sec:cap_pressure_and_dyn_cap_pressure}. 
        The middle term in \eqref{eqn:Sdot_energy_term} is an effect of the dynamic
        pressure-saturation relationship (equation \eqref{eqn:epsdot_near_eq}). 
        The temperature derivative can be interpreted as the effect of temperature on the relative
        wetting potential.  That is, how much does temperature affect the relative
        affinity for one phase over the other. It is likely that a constitutive equation
        is needed for this relationship.

    \item coefficient of time rates of change of densities: \\ We wish to rewrite these
        coefficients in terms of enthalpy and chemical potential since it provides a
        mathematically simpler expression.
        \begin{flalign*}
            & \sumj \left( \left[ - \left( \frac{\eps{\beta} p^\beta}{\rho^\beta} \right)
                + \suma \left\{ \epsa \rhoa \left[ \pd{\psia}{\rho^{\beta_j}} + T
                    \pd{\eta^{\al}}{\rho^{\beta_j}} \right] \right\} \right]
                    \dot{\rho}^{\beta_j} \right) \\
            & \quad = \sumj \left( \left[ -\left( \frac{\eps{\beta} p^\beta}{\rho^\beta}
                \right) + \eps{\beta} \left( \mu^{\beta_j} - \psi^\beta \right) -
                \eps{\beta} T \pd{ }{T} \left( \mu^{\beta_j} - \psi^\beta \right) \right]
                \dot{\rho}^{\beta_j} \right) \\
            & \quad = -\eps{\beta} \left[ \frac{p^\beta}{\rho^\beta} + \psi^\beta + T
                \eta^\beta \right] \dot{\rho}^\beta + \eps{\beta} \sumj \left[ \left(
                    \mu^{\beta_j} - T \pd{\mu^{\beta_j}}{T} \right) \dot{\rho}^{\beta_j}
                    \right] \\
            & \quad = -\eps{\beta} H^{\beta} \dot{\rho}^\beta + \eps{\beta} \sumj \left[
                \left( \mu^{\beta_j} - T \pd{\mu^{\beta_j}}{T} \right)
                \dot{\rho}^{\beta_j} \right]
        \end{flalign*}
        where we recall that $H^\beta$ is the enthalpy of phase $\beta$. 

    \item coefficient of relative velocity: \\ For this coefficient we again use the
        definitions of pressure, chemical potential, and entropy.  We also rely on the
        Gibbs-Duhem relationship \eqref{eqn:Gibbs_Duhem}.
        \begin{flalign*}
            & \sumj \left( \left[ - \left(
                \frac{\eps{\beta} p^\beta}{\rho^\beta} \right) + \suma \left\{ \eps{\al}
                \rho^\al \left[ \pd{\psi^\al}{\rho^{\beta_j}} + T
                    \pd{\eta^\beta}{\rho^{\beta_j}} \right] \right\} \right] \grad
                    \rho^{\beta_j} \right) \\ 
            & \qquad + \eps{\beta} \rho^{\beta} \left\{ \left[ T
                \pd{\eta^{\beta}}{T} \right] \grad T + \left[ T \pd{\eta^{\beta}}{\eps{l}}
                    \right] \grad \eps{l} + \left[ T \pd{\eta^{\beta}}{\eps{g}} \right]
                    \grad \eps{g} \right.  \\
    %
    %
    %
            & \qquad \qquad \qquad \left. + \sumj \left( \left[ T
                \pd{\eta^{\beta}}{\rho^{\gamma_j}} \right] \grad \rho^{\gamma_j} \right)
            \right\} \\
            &= - \left( \frac{\eps{\beta} p^\beta}{\rho^\beta} \right) \grad \rho^\beta +
            \sumj \left[ \eps{\beta} \left( \mu^{\beta_j} - \psi^\beta \right) \grad
                \rho^{\beta_j} \right] + \eps{\beta} \rho^{\beta} c_p^\beta \grad T \\
            & \quad + \eps{\beta} \rho^{\beta} T \pd{ }{T} \left\{
                \pd{\psi^\beta}{\eps{l}} \grad \eps{l} + \pd{\psi^\beta}{\eps{g}} \grad
                \eps{g}  \phantom{\sumj} \right. \\
            & \qquad \qquad \qquad \quad \left. + \sumj \left(
            \pd{\psi^\beta}{\rho^{\beta_j}} \grad \rho^{\beta_j} \right) + \sumj \left(
            \pd{\psi^\beta}{\rho^{\gamma_j}} \grad \rho^{\gamma_j} \right) \right\} \\
            &= - \eps{\beta} \Gamma^\beta  \grad
            \rho^\beta + \sumj \left[ \eps{\beta} \mu^{\beta_j} \grad \rho^{\beta_j}
                \right] + \eps{\beta} \rho^{\beta} c_p^\beta \grad T \\
            & \quad + T \pd{ }{T} \left\{ - \left( \That{s}{\beta} + \That{\gamma}{\beta}
            \right) + p^\beta \grad \eps{\beta} + \suma \sumj \left[ \epsa \rhoa
                \pd{\psia}{\rho^{\beta_j}} \grad \rho^{\beta_j} \right] \right\} \\
            &= - \eps{\beta} \Gamma^\beta  \grad
            \rho^\beta + \sumj \left[ \eps{\beta} \mu^{\beta_j} \grad \rho^{\beta_j}
                \right] + \eps{\beta} \rho^{\beta} c_p^\beta \grad T \\
            & \quad + T \pd{ }{T} \left\{ \left( \eps{\beta} \right)^2 \soten{R}^\beta
            \cd \bv{\beta,s} + p^\beta \grad \eps{\beta} + \sumj \left[ \eps{\beta}
                \left( \mu^{\beta_j} - \psi^\beta \right) \grad \rho^{\beta_j} \right]
            \right\} \\
            &= - \eps{\beta} \Gamma^\beta  \grad
            \rho^\beta + \sumj \left[ \eps{\beta} \mu^{\beta_j} \grad \rho^{\beta_j}
                \right] + \eps{\beta} \rho^{\beta} c_p^\beta \grad T \\
            & \quad + T \pd{ }{T} \left\{ \left( \eps{\beta} \right)^2 \soten{R}^\beta \cd
            \bv{\beta,s} + p^\beta \grad \eps{\beta} - \eps{\beta} \psi^\beta \grad
            \rho^\beta + \sumj \left[ \eps{\beta} \mu^{\beta_j} \grad \rho^{\beta_j}
            \right] \right\} \\
            &= - \eps{\beta} \Gamma^\beta  \grad
            \rho^\beta + \sumj \left[ \eps{\beta} \mu^{\beta_j} \grad \rho^{\beta_j}
                \right] + \eps{\beta} \rho^{\beta} c_p^\beta \grad T \\
            & \quad + T \pd{ }{T} \left\{ \left( \eps{\beta} \right)^2 \soten{R}^\beta \cd
            \bv{\beta,s} - \eps{\beta} \Gamma^\beta \grad \rho^\beta + \left(
            \frac{p^\beta}{\rho^\beta} \right) \grad \left( \eps{\beta} \rho^\beta \right)
            + \sumj \left[ \eps{\beta} \mu^{\beta_j} \grad \rho^{\beta_j} \right] \right\}
        \end{flalign*}
        Since this coefficient is contracted with the relative velocity, $\bv{\beta,s}$,
        we can likely neglect the relative velocity term in the temperature derivative as
        it will result in second-order effects. This simplifies the coefficient of the
        relative velocity to
        \begin{flalign}
            \notag & - \eps{\beta} \Gamma^\beta  \grad
            \rho^\beta + \sumj \left[ \eps{\beta} \mu^{\beta_j} \grad \rho^{\beta_j}
                \right] + \eps{\beta} \rho^{\beta} c_p^\beta \grad T \\
            & \quad + T \pd{ }{T} \left\{ - \eps{\beta} \Gamma^\beta \grad \rho^\beta +
            \sumj \left[ \eps{\beta} \mu^{\beta_j} \grad \rho^{\beta_j} \right] + \left(
            \frac{p^\beta}{\rho^\beta} \right) \grad \left( \eps{\beta} \rho^\beta
        \right)\right\}.
        \end{flalign}
\end{enumerate}

After these four simplifications and rearrangements, equation \eqref{eqn:energy_temporary5} is now rewritten as
\begin{flalign}
    \notag 0 &= \rho c_p \dot{T} - \diver \left( \soten{K} \cd \grad T \right) + \rho h + L
    \ehat{g}{l} \\
    \notag & \quad + \left[ \left( \overline{p}^g\Big|_{eq.} - \overline{p}^l\Big|_{eq.}
        \right) + 2\tau \porosity \Sdot + T \pd{ }{T} \left( \pi^g - \pi^l \right)
        \right] \porosity \Sdot \\
    \notag & \quad - \porosity S H^l \dot{\rho}^l + \porosity S \sumj \left[ \left(
        \mu^{l_j} - T \pd{\mu^{l_j}}{T} \right) \dot{\rho}^{l_j} \right] \\
    \notag & \quad - \porosity (1-S) H^g \dot{\rho}^g + \porosity (1-S) \sumj \left[ \left(
        \mu^{g_j} - T \pd{\mu^{g_j}}{T} \right) \dot{\rho}^{g_j} \right] \\
    \notag & \quad + \sum_{\beta = l,g} \left\{ \left[ - \eps{\beta} \Gamma^\beta  \grad
            \rho^\beta + \sumj \left[ \eps{\beta} \mu^{\beta_j} \grad \rho^{\beta_j}
                \right] + \eps{\beta} \rho^{\beta} c_p^\beta \grad T \right. \right. \\
    & \qquad \qquad \quad \left. \left.  + T \pd{ }{T} \left\{ - \eps{\beta} \Gamma^\beta \grad
    \rho^\beta + \sumj \left[ \eps{\beta} \mu^{\beta_j} \grad \rho^{\beta_j} \right] +
    \left( \frac{p^\beta}{\rho^\beta} \right) \grad \left( \eps{\beta} \rho^\beta
\right)\right\} \right] \cd \bv{\beta,s} \right\}.
\label{eqn:energy_simplified}
\end{flalign}

Equation \eqref{eqn:energy_simplified} depends on temperature, wetting potentials,
enthalpies, chemical potentials, Gibbs potentials, saturation, densities, pressures, and
relative velocities. Since the Gibbs potentials are functions of densities and chemical
potentials this does not add more unknowns to the system of equations.  The pressures and
relative velocities can be paired with forms of Darcy's law, and constitutive equations
are needed for the enthalpies and wetting potentials. We now turn our attention to the
coupling of the fluid-phase mass balance equations and the present energy equation.

\section{Simplifying Assumptions -- A Closed System}\label{sec:simplifying_assumptions}
A host of simplifying assumptions can be made on the system consisting of equations
\eqref{eqn:general_mass_balance_fluid_species} (for $\al = l,g$) and
\eqref{eqn:energy_simplified}.  These are made to reduce the number of unknowns and
equations to a count that is more easily handled by numerical solvers.  This is also done
to avoid having to model any secondary (possibly second-order) physical processes
(examples of which include very slow processes such as those on the order of
$(\bv{l,s})^2$ or $(\bv{\aj,\al})^2$).  These assumptions are in addition to Assumptions
\#1 - \#3 made in Section \ref{sec:Assumptions}.
\begin{description}
    \item[Assumption \#\theAssNum:\stepcounter{AssNum}] Assume that the liquid phase is
        composed of a pure fluid with no additional species.
        Strictly speaking this is not realistic since the water in field measurements
        contains contaminants, dissolved solids, charged ions (such as sodium), and other
        impurities.  The consequence of this assumption is that the diffusive terms within
        the liquid mass balance equation are zero
        \[ \bv{l_j,l} = \bv{l,l} = \foten{0}. \]

    \item[Assumption \#\theAssNum:\stepcounter{AssNum}] The liquid phase is assumed to be
        incompressible. This assumption is valid under moderate pressures and allows us to
        remove the liquid phase material time derivative of density from the liquid mass balance
        equation
        \[ \md{l}{\rho^l} = 0. \]
        In isothermal conditions the density of the liquid phase can be assumed
        constant in space and time.  In the presence of thermal gradients, on the other
        hand, we presume that the density of the liquid phase is a function only of
        temperature given by the empirical model
        \begin{flalign}
            \rho^l(T) = 10^3 \left( 1 - 7.37 \times 10^{-6} \left( T - 277.15 \right)^2 +
            3.79 \times 10^{-8} \left( T - 277.15 \right)^3\right)
            \label{eqn:liquid_density_temperature}
        \end{flalign}
        measured in $kg/m^3$ (and where [$T$]=$K$). See Figure
        \ref{fig:DensityTemperaturePlotLiquid}.

    \item[Assumption \#\theAssNum:\stepcounter{AssNum}] The gas phase is assumed to be an
        ideal binary mixture of water vapor and inert {\it air}.  There are most certainly
        more than two species in most practical gas mixtures, but here we are 
        concerned with with the diffusion, evaporation, and condensation of water vapor
        within the gas mixture.  The other species are assumed to be non-reactive and are
        therefore all grouped together into the {\it air} species.  We choose the mixture
        to be ideal so that we can take advantage of the ideal gas law.  This is valid
        since (a) the gas pressures under most experimental considerations are close to
        atmospheric, (b) under Richards' assumption \cite{Pinder2006,Richards1931}, the
        bulk gas pressure doesn't vary much under most experimental considerations, and
        (c) the temperatures under consideration aren't {\it far} from standard room
        temperature.  The use of an ideal gas mixture will break down under higher
        pressures, higher temperatures, and possibly under high variations in temperature.  

    \item[Assumption \#\theAssNum:\stepcounter{AssNum}] The gas-phase chemical potentials
        and densities are only functions of the relative humidity and temperature
        \begin{flalign}
            \mu^{g_j} = \mu^{g_j}(\rh,T) \quad \rho^{g_j} = \rho^{g_j}(\rh,T).
            \label{eqn:mu_rho_functional_forms}
        \end{flalign}
        We make this
        assumption based on the fact that at the pore scale we can easily convert between
        the chemical potential, the density, and the relative humidity. Furthermore, this
        allows for us to tie the gas-phase mass balance equation to experimentally
        measurable quantities such as the relative humidity.  
        
        Just as at the pore scale, we define the macroscale relative humidity, $\rh$, via
        the saturated vapor density, $\rho_{sat}$, and the density of the water vapor in
        the mixture:
        \begin{flalign}
            \rho^{g_v} = \rho_{sat} \rh,
            \label{eqn:density_rel_hum_macro}
        \end{flalign}
        where $\rho_{sat} = \rho_{sat}(T)$ can be expressed through the empirical equation
        \begin{flalign}
            \rho_{sat} = \frac{\text{exp}\left( 31.37 - 6014.79/T - 7.92 \times 10^{-3} T
        \right)}{T} \times 10^{-3}.
            \label{eqn:rhosat_temp_macro}
        \end{flalign}
        (see Figure \ref{fig:DensityTemperaturePlotGas}).

        The chemical potential of the water vapor is defined through the ideal gas law as
        \begin{flalign}
            \mu^{g_v} = \mu^{g_v}_* + R^{g_v} T \ln \left( \lambda \rh \right)
            \label{eqn:water_vapor_chem_pot_macro}
        \end{flalign}
        where $\lambda = p_{sat} / p_*$ is a function of temperature from
        \eqref{eqn:rhosat_temp_macro} and $p_*$ is atmospheric pressure.

        The reason we are calling this an ``assumption'' is that, strictly speaking, these
        relationships hold for the pore-scale chemical potentials and pressures.  We are
        dealing with averaged (upscaled) quantities so we make the assumption that these
        quantities follow the same functional forms.  It is known that the upscaled
        pressure, density, and chemical potential are not the same as the pore-scale
        pressure, so in effect we are defining the upscaled relative humidity through
        these relationships.

\end{description}


\linespread{1.0}
\begin{figure}[H]
    \centering
    \subfigure[Liquid density vs. temperature]{
        \includegraphics[width=0.45\textwidth]{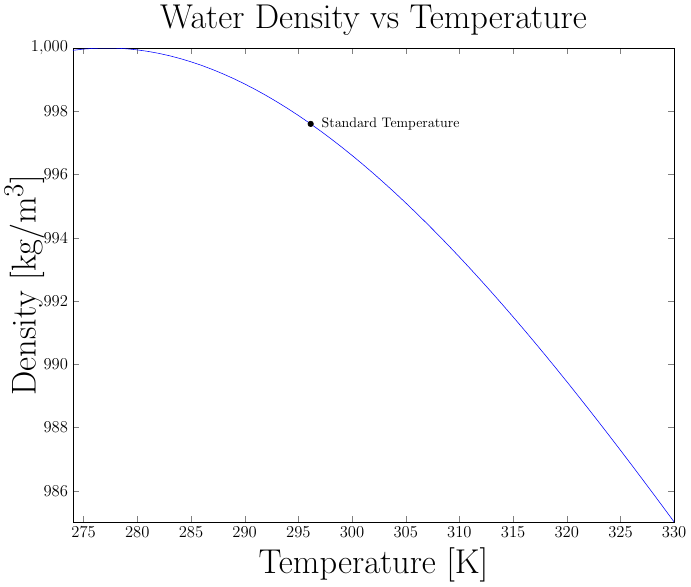}
        \label{fig:DensityTemperaturePlotLiquid}
    }
    \subfigure[Saturated vapor density vs. temperature]{
        \includegraphics[width=0.45\textwidth]{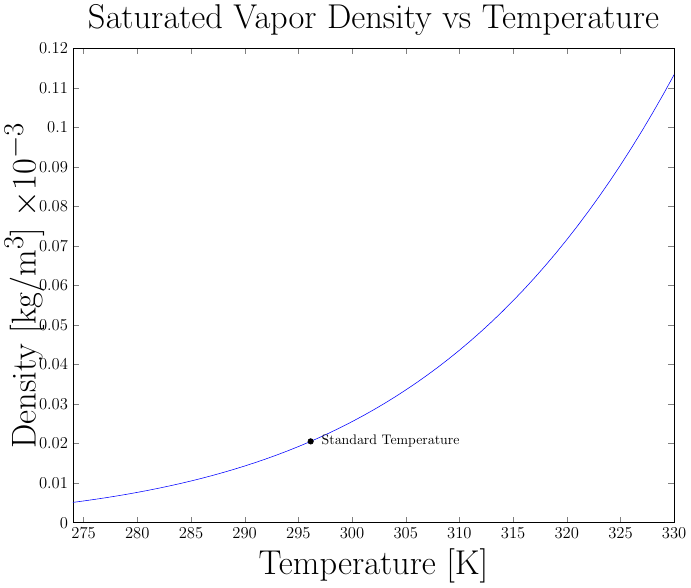}
        \label{fig:DensityTemperaturePlotGas}
    }
    \caption{Densities as functions of temperature}
        \label{fig:DensityTemperaturePlot}
\end{figure}
\renewcommand{\baselinestretch}{\normalspace}

Under assumptions 4 and 5 on the liquid phase we reflect now on the choice of the form of
Darcy's law for the liquid phase.  In the absence of species it may not be reasonable to
use the chemical potential form and instead revert to the pressure form. Recall from
equation \eqref{eqn:darcy_gibbs_form} that the Darcy flux for a fluid with one species is
driven by gradients in Gibbs potential and temperature.  Recall also that the coefficient
of the temperature gradient is the macroscale entropy. To side step the necessity of
modeling the liquid phase entropy and Gibbs potential directly we use the
pressure form of the Darcy flux: equation \eqref{eqn:Darcy}.  Given one liquid species, a
rigid solid phase, two gas species (see assumption \#6), and the assumption that the gas
densities are functions of temperature and relative humidity (see assumption \#7), the
Darcy flux for the liquid phase can be written as
\begin{flalign}
    \notag \eps{l} \soten{R}^l \cd \left( \eps{l} \bv{l,s} \right) &= - \eps{l} \grad p^l -
    \ppi{l}{l} \grad \eps{l} - \ppi{l}{g} \grad \eps{g} + \eps{l} \rho^l \foten{g} \\
    \notag & \quad + \left( \eps{g} \rho^g \pd{\psi^g}{\rho^l} + \eps{s} \rho^s
    \pd{\psi^s}{\rho^l} \right) \grad \rho^l - \eps{l} \rho^l \sum_{j=v,a} \left(
    \pd{\psi^l}{\rho^{g_j}} \grad \rho^{g_j}
    \right) \\
    \notag &= -\eps{l} \grad p^l - \porosity \left( \ppi{l}{l} - \ppi{l}{g} \right) \grad S +
    \eps{l} \rho^l \foten{g} \\
    \notag & \quad + \left( \eps{g} \rho^g \pd{\psi^g}{\rho^l} + \eps{s} \rho^s
    \pd{\psi^s}{\rho^l} \right) \pd{\rho^l}{T} \grad T \\
    \notag & \quad - \eps{l} \rho^l \sum_{j=v,a} \left( \pd{\psi^l}{\rho^{g_j}} \left[
        \pd{\rho^{g_j}}{T} \grad T + \pd{\rho^{g_j}}{\rh} \grad \rh \right] \right) \\
    &= -\eps{l} \grad p^l - \porosity \left( \ppi{l}{l} - \ppi{l}{g} \right) \grad S +
    \eps{l} \rho^l \foten{g} - \eps{l} C_T^l \grad T - \eps{l} C_\rh^l \grad \rh.
    \label{eqn:Darcy_pressure_one_liquid}
\end{flalign}
The functions $C_T^l$ and $C_\rh^l$ are implicitly defined by equations
\eqref{eqn:Darcy_pressure_one_liquid} and may be functions of any variable(s) from the set
of independent variables for the Helmholtz Potential.  
The coefficient of the saturation gradient can be rewritten as
\begin{flalign}
    \notag \porosity \left( \ppi{l}{l} - \ppi{l}{g} \right) &= \porosity \eps{l} \rho^l \left(
    \pd{\psi^l}{\eps{l}} - \pd{\psi^l}{\eps{g}} \right) \\
    \notag &= \porosity \eps{l} \rho^l \left( \pd{\psi^l}{\eps{l}} - \pd{\psi^l}{\eps{g}} \right)
    \\
    &= 2 \eps{l} \rho^l \pd{\psi^l}{S} := \eps{l} C_S^l. \label{eqn:CSl_defn}
\end{flalign}
This coefficient function measures the changes in liquid energy due to changes in
saturation while holding density fixed.  The notation chosen for these
coefficients is meant to be descriptive; the subscript indicates the associated gradient
and the superscript indicates the phase.

Dividing both sides of \eqref{eqn:Darcy_pressure_one_liquid} by $\eps{l}$ gives the
simplified pressure, saturation, temperature, and relative humidity formulation of the
liquid Darcy flux
\begin{flalign}
    \soten{R}^l \cd \left( \eps{l} \bv{l,s} \right) &= -\grad p^l  + \rho^l \foten{g} - C_S^l \grad S - C_T^l
    \grad T - C_\rh^l \grad \rh.
    \label{eqn:Darcy_p_S_T_rh}
\end{flalign}
The first two terms on the right-hand side are the classical Darcy terms, and the
functions $C_S^l, C_T^l,$ and $C_\rh^l$ are, as of yet, unknown.  All of these new functions
measure cross coupling effects due to the presence of other phases. Thought experiments
used to make sense of these new terms will be presented in Section
\ref{sec:cap_pressure_and_dyn_cap_pressure} after a deeper discussion of capillary
pressure. 


Under assumptions \#1 - \#7, the heat and mass transport system can now be written as:
\begin{subequations}
    \begin{flalign}
        \notag & \porosity \pd{S}{t} - \diver \left\{ \soten{K}^l \cd \left[ \grad p^l +
            C_S^l \grad S + C_T^l \grad T  + C_\rh^l \grad \rh - \rho^l \foten{g} \right]
        \right\} \\
        &\quad = \left( \frac{M \left( \rho^l - \rho^{g_v} \right) }{\rho^l} \right) \left(
        \mu^l - \mu^{g_v} \right)
        \label{eqn:liquid_mass_balance_pressure} \\
        \notag & \porosity (1-S) \pd{\rho^{g_v}}{t} - \porosity \rho^{g_v}  \pd{S}{t} \\
        \notag & \quad - \diver \left\{ \rho^{g_v}
        \soten{D}^{g_v} \cd \left[ \grad \mu^{g_v} - \foten{g} \right] \right\} \\
        \notag & \quad - \diver \left\{ \rho^{g_v}
        \soten{K}^g \cd \left[ \rho^{g_v} \grad \mu^{g_v} + \rho^{g_a} \grad \mu^{g_a}
            + \rho^g \eta^g \grad T - \rho^g \foten{g} \right] \right\} \\
        &\quad = -M \left( \rho^l - \rho^{g_v} \right) \left( \mu^l -
        \mu^{g_v}  \right)
        \label{eqn:gas_mass_balance_binary_ideal} \\
        \notag 0 &= \rho c_p \dot{T} - \diver \left( \soten{K} \cd \grad T \right) + \rho h + L
        \ehat{g}{l} \\
        \notag & \quad + \left[ \left( \overline{p}^g\Big|_{eq.} - \overline{p}^l\Big|_{eq.}
            \right) + 2\tau \porosity \Sdot + T \pd{ }{T} \left( \pi^g - \pi^l \right)
            \right] \porosity \Sdot \\
        \notag & \quad - \porosity (1-S) H^g \dot{\rho}^g + \porosity (1-S) \sum_{j=v,a}
        \left[ \left( \mu^{g_j} - T \pd{\mu^{g_j}}{T} \right) \dot{\rho}^{g_j} \right] \\
        \notag & \quad + \left[ \left( \rho^l c_p^l + e^l \frac{d \rho^l}{dT} \right)
            \grad T + \frac{T}{\eps{l}} \pd{p^l}{T} \grad \eps{l} \right] \cd \left( \eps{l} \bv{l,s}
            \right) \\
            \notag & \quad + \left[ - \Gamma^g  \grad \rho^g + \sum_{j=v,a}
            \left[ \mu^{g_j} \grad \rho^{g_j} \right] +
            \rho^{g} c_p^g \grad T \right. \\
        & \qquad \quad \left. + T \pd{ }{T} \left\{ - \Gamma^g \grad
        \rho^g + \sum_{j=v,a} \left[ \mu^{g_j} \grad \rho^{g_j} \right] +
        \left( \frac{p^g}{\eps{g} \rho^g} \right) \grad \left( \eps{g} \rho^g
    \right)\right\} \right] \cd \left( \eps{g} \bv{g,s} \right).
        \label{eqn:thermal_post_assumptions}
    \end{flalign}
    \label{eqn:system_post_assumptions}
\end{subequations}
This system of equations originated from mass, momentum, and energy conservation and was
supplemented with constitutive forms of the rates of mass, momentum, and energy transfer.
We used the incompressibility of the liquid phase to arrive at the fourth line
of the energy equation. In the gas phase, the change in pressure with temperature is
given via the ideal gas law:
\begin{flalign}
    p^g &= \left( \frac{\rho^g R}{M^g} \right) T, \label{eqn:ideal_gas_law} \\
 \pd{p^g}{T} &= \left( \frac{\rho^g R}{M^g} \right) + \left( \frac{T R}{M^g} \right)
    \pd{\rho^g}{T} 
\end{flalign}
where $M^g$ is the molar mass of the gas mixture and $R$ is the universal gas constant. In
the liquid phase, the change in pressure with temperature is the ratio of isobaric and
isothermal compressibilities of liquid water
\[ \pd{p^l}{T} = -\left( \frac{1}{V^l} \pd{V^l}{T} \right) \Big/ \left( -\frac{1}{V^l}
    \pd{V^l}{p^l} \right) = \frac{\al^l}{\beta^l}.  \] 

Recall that $\rho^l = \rho^l(T)$, $\rho^{g_v} = \rho^{g_v}(\rh,T)$,
$\mu^{g_j} = \mu^{g_j}(\rh,T)$, $\eps{l} = \porosity S$, $\eps{g} = \porosity (1-S)$,
and $\eta^\al = \eta^\al(T)$. Furthermore, $\eps{\al}
\bv{\al,s}$ is the Darcy flux associated with the $\al-$phase (see equation
\eqref{eqn:Darcy_ChemPot}) and the latent heat, $L$, is an empirically based function of
temperature. Therefore, assuming that the enthalpy, internal energy, and the linearization
coefficients are known functions of these same variables, equations
\eqref{eqn:liquid_mass_balance_pressure} - \eqref{eqn:thermal_post_assumptions} can be
seen as a closed system of equations in saturation ($S$), relative humidity ($\rh$), and
temperature ($T$). It remains to find relationships for the linearization coefficients,
the cross coupling Darcy terms, the gas-phase entropy,
the enthalpy, and the chemical potentials. In the next subsections we discuss
dimensional analysis, functional forms of the coefficients, and further simplifications
for each equation one at a time.

\subsection{Saturation Equation}\label{sec:liquid_simplifications}
In the liquid phase, the linearization constant, $\soten{K}^l$, is a function of the ease in
which fluid flows through the medium.  This is known as the hydraulic conductivity of the
medium.  The hydraulic conductivity is also known to be a function of the permeability of
the medium.  In saturated (rigid) media this is considered constant (or at least a
tensor), but in unsaturated media they are typically taken as functions of saturation. In
the present case, a careful inspection of the units indicate that 
\begin{flalign}
    \soten{K}^l = \frac{\soten{\kappa}}{\mu_l} = \frac{\soten{k}_c}{\rho^l g},
\end{flalign}
where $\soten{\kappa}$ is the permeability tensor of the medium, $\soten{k}_c$ is the
hydraulic conductivity tensor, and $\mu_l$ is the dynamic viscosity \cite{Bear1988,Pinder2006}.
Notationally ``$\mu^\al$'' (with a superscript) will denote chemical potential,
and ``$\mu_\al$'' (with a subscript) will denote dynamic viscosity.

The permeability, $\soten{\kappa}$, is typically separated into a saturated permeability,
$\soten{\kappa}_s$, and a relative permeability, $\kappa_{r\al}$. The relative
permeability is assumed to be a function of
saturation and depends on whether $\al$ is the wetting or non-wetting phase
\cite{Pinder2006}.  There are several functional forms of $\kappa_{r\al}$, but one of the
more commonly used is that of van Genuchten \cite{VanGenuchten1980},
\begin{subequations}
    \begin{flalign}
        \kappa_{rl} = \kappa_{rw} &= \left( S_{e} \right)^{1/2} \left\{ 1 - \left[ 1 -
            \left( S_{e} \right)^{1/m} \right]^m \right\}^2 
            \label{eqn:vanGenuchten_krw} \\
        \kappa_{rg} = \kappa_{rnw} &= \left[ 1 - \left( S_{e} \right) \right]^{1/3}
        \left[ 1 - \left( S_{e} \right)^{1/m} \right]^{2m},
        \label{eqn:vanGenuchten_krnw}
    \end{flalign}
\end{subequations}
where $m$ is a fitting parameter, and $S_{e}$ is the effective saturation defined by
\begin{flalign}
    S_{e} = \frac{S - S_{min}}{S_{max} - S_{min}} \quad S_{e} \in [0,1].
    \label{eqn:eff_saturation}
\end{flalign}
Typical values of $m$ are less than $1$ where $m=2/3$ is commonly used as a starting point
for fitting numerical models to experimental data.  Typical relative permeability curves
are shown in Figure \ref{fig:relative_permeabilities_vanGenuchten}.  The reader is to keep
in mind that there are several such models in the literature \cite{Bear1988,Pinder2006}.
The van Genuchten model simply constitutes a widely used relative permeability model. Note
that there is not a symmetry in $k_{rnw}$ and $k_{rw}$ in the sense that $k_{rnw}(S_{e})
\ne k_{rw}(1-S_{e})$ as would naively be assumed.  This is a manifestation of the fact
that unsaturated media {\it behave} differently during imbibition and drainage. The value
of $\soten{\kappa}_s$ is chosen based on the type of medium. If the medium is isotropic
then the tensorial notation can be dropped and values from Table
\ref{tab:typical_permeability} can be used.

\linespread{1.0}
\begin{figure}[ht!]
    \begin{center}
        \includegraphics[width=0.75\columnwidth]{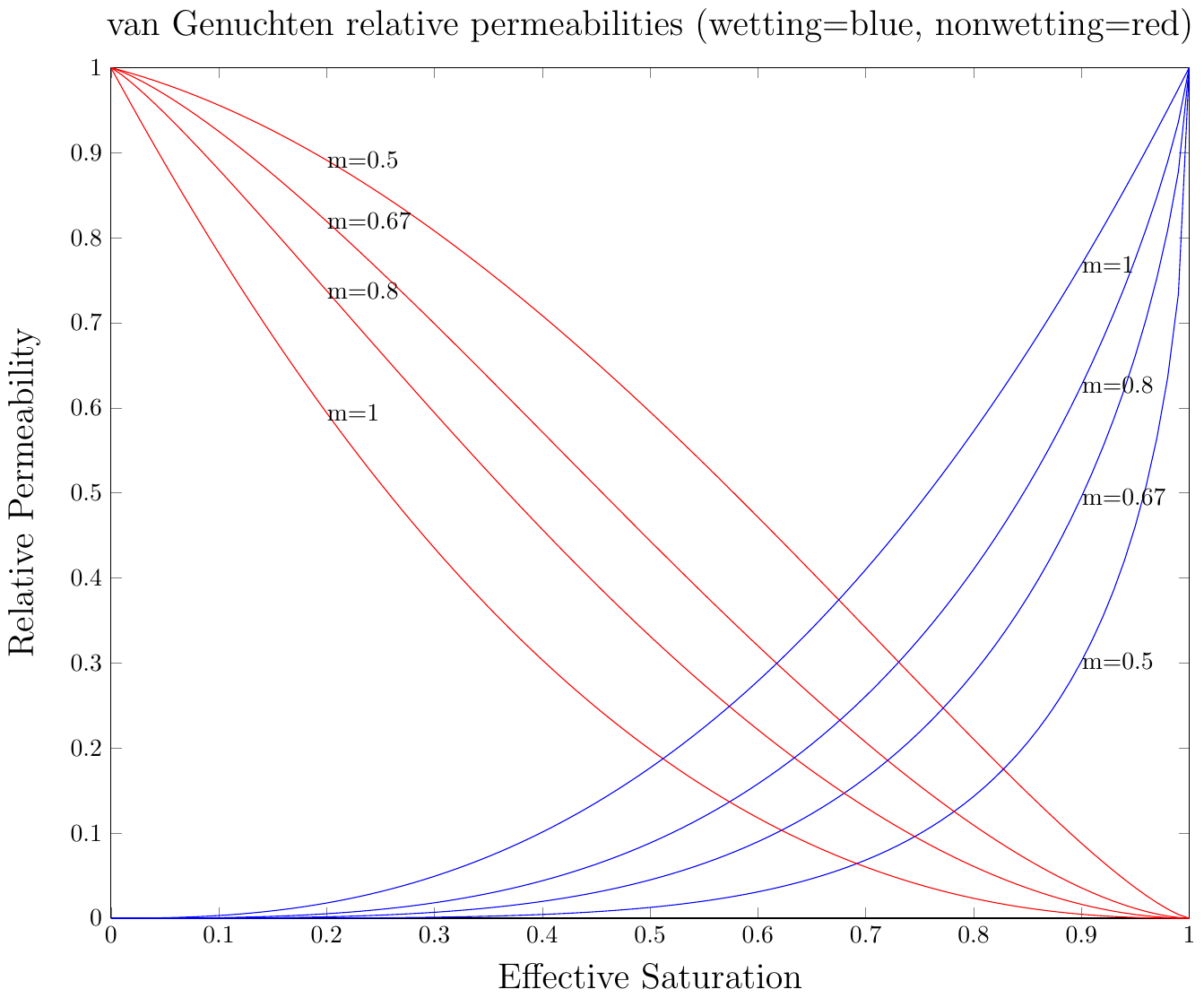}
    \end{center}
    \caption{van Genuchten relative permeability curves. The red curve shows the
        non-wetting phase, $\kappa_{rnw}(S_{e})$, and the blue curves show the wetting
        phase, $\kappa_{rw}(S_{e})$, each for $m=0.5, 0.67, 0.8,$ and $1$.}
    \label{fig:relative_permeabilities_vanGenuchten}
\end{figure}
\renewcommand{\baselinestretch}{\normalspace}

\subsubsection{Capillary Pressure and Dynamic Capillary Pressure}\label{sec:cap_pressure_and_dyn_cap_pressure}
The capillary pressure, $p_c$, is typically defined as the difference between the
non-wetting (gas) and wetting (liquid) phase pressures  when measured in a tube at
equilibrium
\begin{flalign}
    p_c = p_{non-wetting} - p_{wetting}. \label{eqn:cap_pressure_wet_nonwet} 
\end{flalign}
At the microsale, the difference is related to the surface tension of the fluid, the
contact angle, and the effective radius through the Young-Laplace equation (see Figure
\ref{fig:cap_tube})
\begin{flalign}
    p_c = \frac{2 \gamma \cos \theta}{r}.
    \label{eqn:Young_Laplace}
\end{flalign}
\linespread{1.0}
\begin{figure}[ht*]
    \begin{center}
        \begin{tikzpicture}
            \draw[fill=blue!40] (0,0) arc (220:320:1.3cm) -- (2,-2) --
            (0,-2) -- cycle;
            \draw[thick] (0,2) -- (0,-2);
            \draw[thick] (2,2) -- (2,-2);
            \draw[dashed] (-0.5,1) -- (0.5,-1);
            \draw (0,0) -- (1,1);
            \draw (0,0) -- (1,0);
            \draw (0.5,0.5) node[anchor=south]{$1/\kappa$};
            \draw (0.5,0) node[anchor=north]{$r$};
            \draw[dotted] (1,1) -- (1,0);
            \draw (-0.10,-0.25) node[anchor=north west]{$\theta$};
            \draw (0.1,0.25) node[anchor=south east]{$\theta$};
            \draw (0.1,0.1) node[anchor=west]{$\theta$};
        \end{tikzpicture}
    \end{center}
    \caption{Contact angle and effective radius in a capillary tube geometry. $\theta$ is
    the contact angle, $r$ is the effective radius, and $\kappa$ is the radius of
curvature of the interface.}
    \label{fig:cap_tube}
\end{figure}
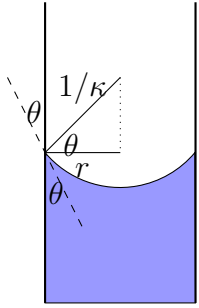
\renewcommand{\baselinestretch}{\normalspace}
The question is which {\it pressure} (thermodynamic, classical, or wetting (see Section
\ref{sec:pressure})) represents the non-wetting and wetting pressures in equation
\eqref{eqn:cap_pressure_wet_nonwet}. The capillary pressure is measured with a force
transducer in the same manner that the {\it classical} pressure is measured.  For this
reason we define the capillary pressure as
\begin{flalign}
    p_c = p^g - p^l.
    \label{eqn:capillary_pressure_defn}
\end{flalign}

Now that we understand which {\it pressure} is associated with the capillary pressure we
turn to the entropy inequality to derive a constitutive equation equation for the time
rate of change of saturation.  In Richards' equation it is standard practice (as mentioned
in Section \ref{sec:richards}) to take the capillary pressure as a function of saturation.
These relations are reasonable for an equilibrium relationships. In the present modeling
effort we look toward the entropy inequality to determine an appropriate form of $p_c$
away from equilibrium. In the entropy inequality (equation \eqref{eqn:entropysimplified}) there are two terms associated with the
time rate of change of saturation:
\[ - \overline{p}^l \epsdot{l} \quad \text{and} \quad - \overline{p}^g \epsdot{g}. \]
Since $\epsdot{g} = - \epsdot{l}$ these terms can be combined to give $
(\overline{p}^g - \overline{p}^l) \epsdot{l}$. The time rate of change of volume
fraction is a constitutive variable so the associated linearized equation is
\begin{flalign}
    \left( \overline{p}^g - \overline{p}^l \right) \Big|_{n.eq} &= \left( \overline{p}^g -
    \overline{p}^l \right) \Big|_{eq} - \tau \epsdot{l},
    \label{eqn:epsdotl_near_eq}
\end{flalign}
where the equilibrium state is not necessarily zero and the minus sign is chosen to be
consistent with the entropy inequality.  From the three pressures relationship,
\eqref{eqn:three_pressures}, the classical pressure is given as
\[ p^\al = \overline{p}^\al + \pi^\al \]
where $\overline{p}^\al$ is the {\it thermodynamic} pressure and $\pi^\al$ is a wetting
potential.  The difference in thermodynamic pressures is therefore rewritten as 
\begin{flalign*}
    \overline{p}_c := \overline{p}^g - \overline{p}^l &= \left( p^g - \pi^g \right) - \left( p^l - \pi^l
    \right) = p_c - \left( \pi^g - \pi^l \right) = p_c - \pi_c
\end{flalign*}
and equation \eqref{eqn:epsdotl_near_eq} becomes
\begin{flalign}
    \left( p_c - \pi_c \right) \Big|_{n.eq} &=  \left( p_c - \pi_c \right) \Big|_{eq} - \tau \epsdot{l}.
    \label{eqn:epsdotl_near_eq_with_wetting}
\end{flalign}
Rewriting we get
\begin{flalign}
    p_c \Big|_{n.eq} = p_c \Big|_{eq} + \left( \pi_c\Big|_{n.eq} - \pi_c\Big|_{eq}
    \right) - \tau \epsdot{l}
    \label{eqn:pc_neareq_epsdot}
\end{flalign}
We assume that the effect of the solid phase on the capillary pressure is completely
captured by the preferential wetting, $\pi_c$. Without the solid phase, the normal
pressures of the liquid and gas phases are zero (this is the case with a flat interface).
With this assumption the thermodynamic pressures are equal across the phases at
equilibrium. Therefore, $\overline{p}_c|_{eq} = 0$.  This implies that $p_c|_{eq} =
\overline{p}_c|_{eq} + \pi_c|_{eq} = \pi_c|_{eq}$.  Therefore the capillary pressure at
equilibrium is interpreted as the difference in wetting potential and we arrive at an
expression that is similar to that found in \cite{Hassanizadeh2002}. To avoid possible
confusion we will continue to use the symbols $p_c|_{eq}$ in place of $\pi_c|_{eq}$ even
though they are understood to be the same.

We finally arrive at an expression relating the classical liquid-phase
pressure that appears in Darcy's law, $p^l|_{n.eq}$, and the capillary pressure,
$p_c|_{eq}$:
\begin{flalign}
    -p^l \Big|_{n.eq} = p_c\Big|_{eq} + \left( \pi_c\Big|_{n.eq} - \pi_c\Big|_{eq}
    \right) - p^g\Big|_{n.eq} - \tau \epsdot{l}.
    \label{eqn:pressure_wetting_epsdot}
\end{flalign}
If the deviation in the wetting potential from equilibrium is assumed to be small relative to
the pressure and the dynamic effects we can approximate the liquid pressure as
\begin{flalign}
    -p^l \Big|_{n.eq} \approx p_c\Big|_{eq} - p^g\Big|_{n.eq} - \tau \epsdot{l},
\end{flalign}
where it is possible that $p^g \approx 0$ as well (in fact, this is a common assumption).
To see why the deviation in wetting potential might be small, consider that in equation
\eqref{eqn:pc_neareq_epsdot} if $p_c|_{n.eq} \approx p_c|_{eq}$ then the saturation
dynamics is driven by the deviation in wetting potential.  The deviation in wetting
potential measures how much the shape of the curved liquid-gas interface is away from
equilibrium.  In slow flows it is unlikely that this deviation is significant.

As mentioned in Section \ref{sec:richards}, the (equilibrium) capillary pressure can be
related to the effective saturation through the van Genuchten $p_c - S$ relationship.
This relationship depends on several fitting parameters and is given as
\begin{flalign}
    p_c(S_e) = \left( \frac{1}{\al} \right) \left( S_e^{-1/m} - 1 \right)^{1-m},
    \label{eqn:van_Genuchten_capillary_pressure}
\end{flalign}
where $\al$ has units of reciprocal pressure and $m$ is the same fitting parameter as in
the relative permeabilities \eqref{eqn:vanGenuchten_krw} \cite{Bear1988,Pinder2006}. See
Figure \ref{fig:vanGenuchten_pcS_curves} for several examples of capillary pressure -
saturation curves for various sets of parameters.  Generally speaking, $m$ increases
(toward 1) as the soil becomes more densely packed.
\linespread{1.0}
\begin{figure}[ht*]
    \begin{center}
        \includegraphics[width=0.9\columnwidth]{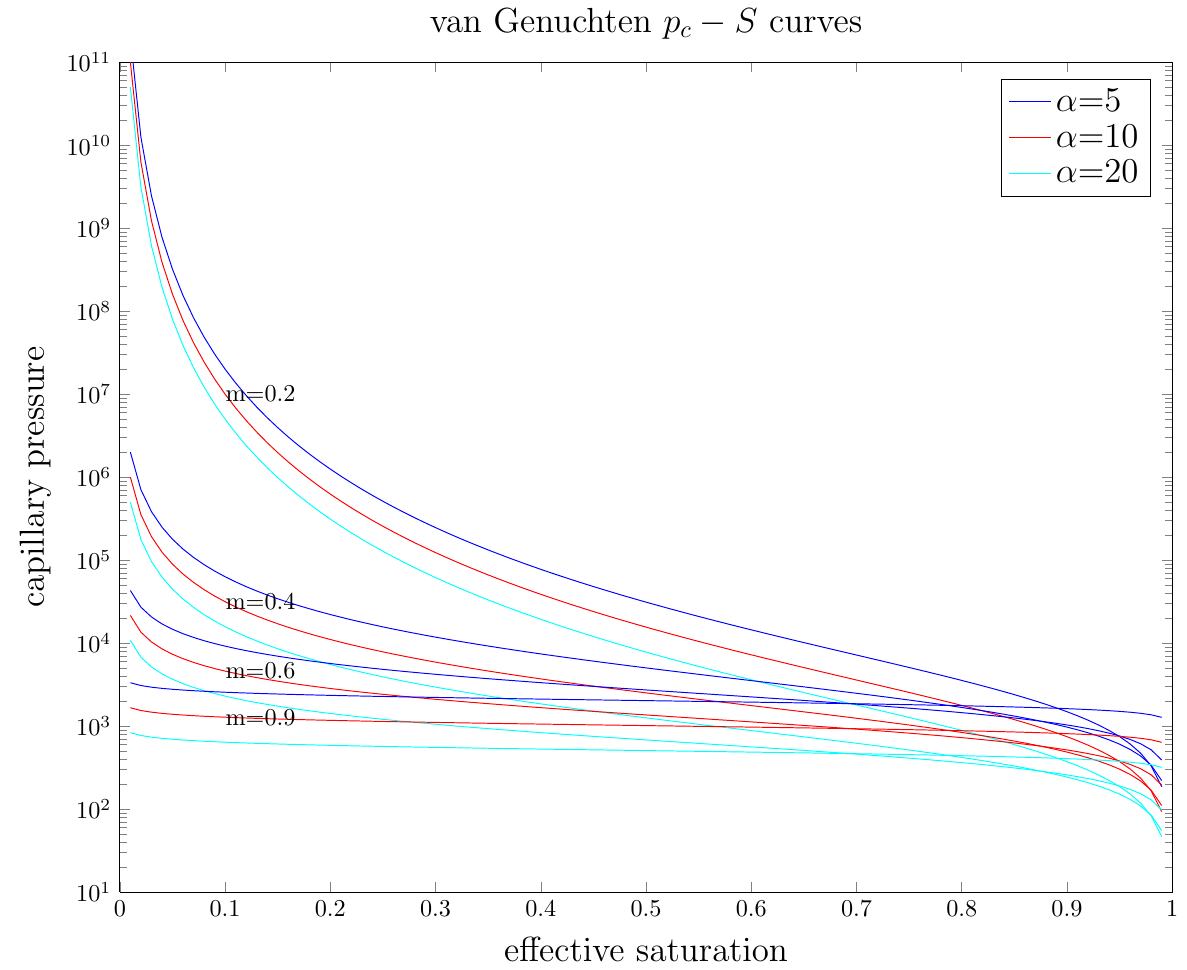}
    \end{center}
    \caption{Examples of van Genuchten capillary pressure - saturation curves for various
    parameters.}
    \label{fig:vanGenuchten_pcS_curves}
\end{figure}
\renewcommand{\baselinestretch}{\normalspace}

Substituting the capillary pressure and van Genuchten relationships into Darcy's law,
\eqref{eqn:Darcy_p_S_T_rh}, the fluid flux becomes
\begin{flalign}
    \notag \soten{R} \cd \left( \eps{l} \bv{l,s} \right) &= \left( \frac{d p_c}{dS} - C_S^l \right)
    \grad S - \tau \porosity \grad \dot{S} + \rho^l \foten{g} \\
    & \quad - \grad p^g - C_\rh^l \grad \rh - C_T^l \grad T + \grad 
    \underbrace{\left( \pi_c |_{n.eq} - \pi_c |_{eq} \right)}_{\approx 0}.
    \label{eqn:Darcy_S_Sdot_T_rh}
\end{flalign}
To understand the newly terms proposed here, we make the following three comments:
\begin{enumerate}
    \item First consider the gas pressure and relative humidity terms. In the absence of
        gravity, if the saturation, temperature, and the change in capillary wetting
        potential are held fixed then \eqref{eqn:Darcy_S_Sdot_T_rh} states that flow is
        driven by gradients in relative humidity and gas-phase pressure. The gas-phase
        pressure and the relative humidity are proportional to each other where the
        constant of proportionality is a function of temperature and the species
        densities. With this in mind, these two terms together can be rewritten as a gradient
        in gas pressure.  While a gradient in gas pressure can certainly cause flow, it is
        commonly assumed that $p^g$ is approximately constant (known as Richards'
        assumption \cite{Pinder2006}) and therefore these terms are typically neglected.
        If these terms are not neglected then they are best written as a single gradient
        of relative humidity for easy coupling with the gas-phase diffusion equation
        \[ C_\rh^l \grad \rh \gets \grad p^g + C_\rh^l \grad \rh. \]
    \item Next consider the gradient of temperature term. In the absence of gravity, if
        saturation and relative humidity are held fixed then \eqref{eqn:Darcy_S_Sdot_T_rh}
        states that flow is driven by a gradient in temperature. Saito et
        al.\,\cite{Saito2006} indicated that the thermally induced flow was negligible as
        compared to isothermal flow (also discussed in \cite{Smits2011,Webb1998}). This
        indicates that the $\grad T$ term in \eqref{eqn:Darcy_S_Sdot_T_rh} is likely quite
        small. 
    \item Finally we discuss the role of $C_S^l = 2 \rho^l \pd{\psi^l}{S}$. This function
        (or constant) relates the changes in energy with respect to saturation.  The term
        is already associated with the gradient in saturation as seen in equation
        \eqref{eqn:Darcy_S_Sdot_T_rh}. From the $\grad S$ term in this equation
        we can see $C_S^l$ as an {\it enhancement} of the capillary
        pressure - saturation relationship that directly models the affinity for the
        liquid phase to the other phases. It is entirely likely that this term is so
        closely linked with the capillary pressure that in experimental settings it is
        impossible to discern this effect from others.

\end{enumerate}

The saturation equation can finally be written as
\begin{flalign}
    \notag & \porosity \pd{S}{t} - \diver \left[ \soten{K}(S) \cd \left(  \left\{ -\frac{d
    p_c}{dS} + C_S^l \right\} \grad S_e + \tau \porosity \grad \dot{S}_e + C_T^l \grad T +
    C_\rh^l \grad \rh - \rho^l \foten{g} \right) \right] \\
    &= \frac{M}{\rho^l} \left( \rho^l - \rho^{g_v} \right)\left(
    \mu^l - \mu^{g_v} \right),
\end{flalign}
where we have assumed that $\pi_c|_{n.eq} - \pi_c|_{eq} \approx 0$ and, abusing notation
slightly, the $C_\rh^l$ term has be redefined to incorporate changes in the gas pressure.

To account for the residual (minimum) saturation, the saturation is scaled to the {\it
effective saturation} according to $S_e = (S - S_{min}) / (S_{max} - S_{min})$.  Defining
$\porosity_S$ as the product of porosity and the difference in maximal and minimal
saturation, $\porosity_S := \porosity (S_{max} - S_{min})$, and letting $S$ notationally
stand for $S_e$ allows us to write the saturation equation as
\begin{flalign}
    \notag & \pd{S}{t} - \diver \left[ \porosity_S^{-1} \soten{K}(S) \cd \left( \left\{ -
        \frac{dp_c}{dS} + C_S^l\right\} \grad S + \tau \porosity_S \grad \dot{S} +  C_T^l
    \grad T + C_\rh^l \grad \rh - \rho^l \foten{g} \right) \right] \\
    &\quad =
    \frac{M^l}{\porosity_S \rho^l} \left( \rho^l - \rho^{g_v} \right) \left( \mu^l -
    \mu^{g_v} \right)
    \label{eqn:saturation_with_T_rh}
\end{flalign}
At a quick
glance, the sign of the $\grad S$ term looks suspicious as it seems to indicate a
backward heat equation. Observe that $p_c'(S) < 0$ for all values of $S$.  Taking only
the first line with $C_S^l = 0$ returns Richards' equations exactly. The $\grad \dot{S}$
term (henceforth referred to as the {\it dynamic saturation term}) was originally proposed
by Hassanizadeh et al. in several publications (examples include
\cite{Hassanizadeh2002,Joekar-Niasar2010}) and is gaining more widespread acceptance in
the porous media community. Taking all of the terms on the first line (again with $C_S^l =
0$) along with the dynamic saturation term gives a closed pseudo-parabolic equation in
saturation.  The $C_S^l, C_T^l, C_\rh^l$ terms along with the form of the right-hand side
are all novel to this work. The temperature and relative humidity coupling terms can
certainly be taken to be zero in certain physical instances, but generally the relative
weight and functional forms of these terms is, as of yet, unknown.

We now turn out attention to the gas phase diffusion equation. Analysis
and numerical solutions to the saturation equation will be considered in Chapter
\ref{ch:Transport_Solution}.

\subsection{Gas Phase Diffusion Equation}\label{sec:vapor_diffusion}
In this subsection we make certain simplifications to the gas-phase diffusion equation so as
to tie the chemical potential formulation to the more classical enhanced diffusion model.
As a first step toward this simplification we consider the fact
that the gas phase chemical potentials are related to each other through equation
\eqref{eqn:diffusion_ficks_sum}; the expression for the relative motion of diffusing
species in a binary system:
\[ \sumj \left\{ \left( \frac{\rhoaj}{R^{g_j} T} \right) \soten{D}^\al \cd \left[ \grad
    \mu^{\aj} - \foten{g} \right] \right\} = \foten{0}. \]
With this, the gradient of chemical
potential of the inert air in
\eqref{eqn:gas_mass_balance_binary_ideal} can be rewritten as a function of the water
vapor chemical potential
\[ \rho^{g_a} \grad \mu^{g_a} = - \left( \frac{R^{g_a} \rho^{g_v}}{R^{g_v}} \right) \left(
    \grad \mu^{g_v} - \foten{g} \right) + \rho^{g_a} \foten{g}. \]
This means that the gas-phase mass balance equation can be rewritten as
\begin{flalign}
    \notag & \pd{ }{t} \left( \porosity  \rho^{g_v}_{sat} \rh (1-S) \right) \\
    \notag & \quad -\diver \left\{ \rho^{g_v} \left[ \soten{D}^{g_v} + \rho^{g_v} \left( 1
        - \frac{R^{g_a}}{R^{g_v}} \right) \soten{K}^g \right] \cd \left[ \grad \mu^{g_v} -
            \foten{g} \right] \right\} \\
    & \quad - \diver \left\{ \rho^g \rho^{g_v} \eta^g \soten{K}^g \cd \grad T\right\} =
    -M \left( \rho^l - \rho^{g_v} \right) \left( \mu^l - \mu^{g_v} \right).
        \label{eqn:gas_mass_balance_binary_ideal_simplified} 
\end{flalign}

Typically, one would choose a functional form of $\soten{D}^{g_v}$ to match the
enhancement model discussed in Section \ref{sec:phillip_devries} and the functional form
of $\soten{K}^g$ from the van Genuchten model discussed in Section
\ref{sec:liquid_simplifications}.  In the present case we argue to use different
functional forms of $\soten{D}^{g_v}$ and $\soten{K}^g$. This is done by considering the
conversions between the pore-scale density and chemical potential to the relative
humidity. For simplicity the tensorial notation is dropped and we assume that the
diffusion and conductivity tensors are all scalar multiples of the identity matrix.

We begin with some logical considerations for the gas-phase diffusion coefficient.  If the
gas-phase volume fraction were to drop to zero then there would be no gas in the pore
space (or their would be no pore space) and the diffusion coefficient should drop to zero.
Similarly, if the gas-phase volume fraction were to increase to 1 (100\% gas with no solid
or liquid), then the diffusion coefficient should return to the Fickian diffusion
coefficient $D^g$.  
With these two limiting cases in mind we first propose that $D^{g_v} =
C \eps{g} D^g$ where $C$ is a scaling parameter.  

As seen in Chapter \ref{ch:diffusion_comp}, the diffusion coefficient is modified for
Fick's law based on the dependent variable of interest.  In equations
\eqref{eqn:Ficks_mass_flux} and \eqref{eqn:Ficks_mass_flux_chempot} we see a scalar factor
of $1 / (R^{g_v} T)$ between the mass and chemical potential forms of Fick's law. Making
the same modification here along with the factor of $\eps{g}$ suggested above we get
\begin{flalign}
    \rho^{g_v} D^{g_v} \grad \mu^{g_v} \to \rho^{g_v} \left( \frac{\eps{g}}{R^{g_v} T}
    \right) D^g \grad \mu^{g_v} = \left( \frac{\porosity \rho_{sat}\rh(1-S)}{R^{g_v} T} \right)
    D^g \grad \mu^{g_v}
\end{flalign}
where $D^g$ is the same pore-scale diffusion coefficient as found in Chapter
\ref{ch:diffusion_comp}. One simple way to look at this conversion is that it scales out
the units and magnitude of the chemical potential when converting to relative humidity.
That is, $D^{g_v}\grad \mu^{g_v}$ and $D^g / (R^{g_v} T) \grad \rh$ have the same units
and magnitude. 
A further justification of this is found by recalling the pore-scale definition of the
chemical potential:
\begin{flalign}
    \notag \mu^{g_v} &= \mu^{g_v}_* + R^{g_v} T \ln \left( \frac{p^{g_v}}{p^g} \right) \\
    &= \mu^{g_v}_* + R^{g_v} T \ln \left( \lambda \rh \right),
    \label{eqn:chem_pot_pressure_rel_hum}
\end{flalign}
where $\lambda = p^{g_v}_{sat} / p^g$ and $p^{g_v}_{sat}$ is the partial pressure of the
water vapor under saturated conditions. Taking the gradient of
\eqref{eqn:chem_pot_pressure_rel_hum} and neglecting the temperature variation gives
\[ \grad \mu^{g_v} \approx \frac{R^{g_v} T}{\rh} \grad \rh. \]
Hence we see the exact conversion used in Fick's law.

Next we turn our attention to the hydraulic conductivity term that arose from Darcy's law:
$\rho^{g_v} K^g \grad \mu^{g_v}$. Similar to that of Fick's law, we need to scale the
conductivity to account for the fact that we're using the chemical potential as the
dependent variable. Unlike the Fickian diffusion coefficient, this term already has the
proper units since the units of $\rho^{g_v} \grad \mu^{g_v}$ are the same as the gradient
of pressure.  Therefore we seek a scaling that is unitless but scales the magnitude of the
chemical potential down to that of pressure. That is, we need a constant, $c$, such that 
$c \rho^{g_v} K^g \grad \mu^{g_v}$ and $K^g \grad p^{g}$ have approximately the same
magnitude.  

Taking the gradient of both sides of the first line of equation
\eqref{eqn:chem_pot_pressure_rel_hum} we arrive at
\begin{flalign*}
    \grad \mu^{g_v} &= \left( \frac{R^{g_v} T p^g}{p^{g_v}} \right) \grad \left(
    \frac{p^{g_v}}{p^g} \right) \\
    &= \left( \frac{R^{g_v} T p^g}{p^{g_v}} \right) \left( \left( \frac{1}{p^g} \right)
    \grad p^{g_v} - \left( \frac{p^{g_v}}{(p^g)^2} \right) \grad p^g \right) \\
    &= \left( \frac{R^{g_v} T}{p^{g_v}} \right) \grad p^{g_v} - \left(
    \frac{R^{g_v} T}{p^g} \right) \grad p^g.
\end{flalign*}
The coefficient of the gradient of gas-phase pressure can be rewritten as
\[ \frac{R^{g_v} T}{p^g} = \frac{\rho_{sat} R^{g_v} T}{\rho_{sat} p^g} =
    \frac{p_{sat}^{g_v}}{\rho_{sat} p^g} = 
    \frac{\lambda}{\rho_{sat}}. \]
Since the chemical potential form already has a factor of $\rho^{g_v} = \rho^{g_v}_{sat} \rh$ we
scale $K^g$ by $\lambda$ to account for the difference in magnitude
between the chemical potential and the pressure. Hence, the Darcy term in equation
\eqref{eqn:gas_mass_balance_binary_ideal_simplified} is rewritten as
\[ \rho^{g_v} \left( 1-\frac{R^{g_a}}{R^{g_v}} \right) K^g \grad \mu^{g_v} \to \rho^{g_v}
    \left( 1-\frac{R^{g_a}}{R^{g_v}} \right) \left( \lambda K^g \right) \grad \mu^{g_v}. \]
Keep in mind that this is a scaling of the hydraulic conductivity; just as the factor of
$1/(R^{g_v}T)$ is a scaling of the diffusion coefficient in Fick's law.  

One point of interest for this choice of scaling factor is that it is invisible when we
consider a {\it pure} gas phase.  That is, $\lambda = 1$ when no species are considered
since the saturated partial pressure will simply be the bulk pressure.  This indicates that
we have not actually changed Darcy's law. Instead we have simply made a conversion to
account for the use of a different dependent variable.

Next we focus on writing the gas-phase diffusion equation
\eqref{eqn:gas_mass_balance_binary_ideal_simplified} in terms of relative humidity,
saturation, and temperature. To do this we replace the chemical potential with relative
humidity and temperature via equation \eqref{eqn:chem_pot_pressure_rel_hum}. 
Taking the gradient of the chemical potential in equation
\eqref{eqn:chem_pot_pressure_rel_hum} we get
\[ \grad \mu^{g_v} = \frac{R^{g_v} T}{\rh} \grad \rh + \left( \frac{R^{g_v} T}{\lambda}
    \frac{d \lambda}{dT} + R^{g_v} \ln(\lambda \rh) \right) \grad T.\]
With the Fickian and Darcy terms written in terms of the relative humidity, along with the
fact that the saturated vapor density is a function of temperature, the vapor diffusion
equations can be written as
\begin{flalign}
    \notag & \pd{ }{t} \left( \porosity \rho_{sat} \rh (1-S) \right) \\
    \notag & \quad - \diver \left\{ \rho_{sat}\rh \left[ \frac{\porosity
    (1-S)}{R^{g_v} T} D^g + \rho_{sat}\rh \left( 1-\frac{R^{g_a}}{R^{g_v}} \right) \left(
    \lambda K^g \right) \right] \right. \\
    \notag & \quad \qquad \qquad \qquad \left. \cd \left[ \frac{R^{g_v} T}{\rh} \grad \rh
        + \left( \frac{R^{g_v} T}{\lambda} \frac{d \lambda}{dT} + R^{g_v} \ln(\lambda \rh)
    \right) \grad T \right] \right\} \\
    \notag & \quad - \diver \left\{ \rho^g \rho_{sat}\rh \eta^g \left( \lambda K^g \right)
\grad T \right\} = - M \left(\rho^l - \rho_{sat}\rh \right)\left( \mu^l - \mu^{g_v}
\right).
\end{flalign}
Combining like terms, dividing by the porosity, replacing the hydraulic conductivity by
the saturated and relative permeabilities,  and simplifying gives
\begin{flalign}
    \notag & \pd{ }{t} \left( \rho_{sat} \rh (1-S) \right) \\
    \notag & \quad -\diver \left\{ \rho_{sat}
        \mathcal{D}(\rh,S,T) \left[ \grad \rh - \frac{\foten{g} \rh}{R^{g_v} T} \right]
    \right\} - \diver \left\{ \rho_{sat} N^g(\rh,S,T) \grad T\right\} \\
    & \qquad =
    -\frac{ M^l \left( \rho^l - \rho^{g_v} \right)}{\porosity} \left( \mu^l - \mu^{g_v} \right),
    \label{eqn:full_diffusion_equation}
\end{flalign}
where the functions $\mathcal{D}$ and $N^g$ are 
\begin{subequations}
    \begin{flalign}
        \mathcal{D}(\rh,S,T) &:= (1-S) D^g + \rho_{sat} \rh R^{g_v} T \left( 1
        - \frac{R^{g_a}}{R^{g_v}} \right) \left( \frac{\lambda \kappa_s}{\porosity \mu_g} \right)
        \kappa_{rg}(S) \text{ and } 
        \label{eqn:modified_diffusion_coefficient} \\
        N^g(\rh,S,T) &:= \rh \left[ \mathcal{D}(\rh,S,T) \left(
            \frac{1}{\lambda} \frac{d\lambda}{dT} + \frac{R^{g_v} \ln(\lambda \rh)}{T}
            \right) + \rho^g \rho_{sat} \eta^g \left( \frac{\lambda \kappa_S}{\porosity
            \mu_g} \right) \kappa_{rg}(S) \right]
    \end{flalign}
\end{subequations}

The enhancement model suggested by de Vries, and subsequently used by several authors
\cite{Cass1984,Saito2006,Smits2011,Sakai2009,Webb1998}, is a multiplicative
combination of the pure Fickian diffusion coefficient, $D^g$, the tortuosity,
$\tau=\tau(\eps{g})$, and an enhancement factor, $\eta$:
\begin{flalign}
    D = \tau \eta D^g.
    \label{eqn:multiplicative_enhancement}
\end{flalign}
In these works, the functional
form of the enhancement
factor is taken to be of the form suggested by Cass et al. \cite{Cass1984} 
\begin{flalign}
    \eta_{(a)} = \left( a + 3 \frac{\eps{l}}{\porosity}
    \right) - (a-1) \text{exp} \left\{ -\left[ \left( 1 + \frac{2.6}{\sqrt{f_c}} \right)
        \frac{\eps{l}}{\porosity} \right]^3 \right\}.
    \label{eqn:enhancement_cass}
\end{flalign}
Here, $f_c$ is the mass fraction of clay in the soil. In the absence of clay the
enhancement factor is taken as 
\begin{flalign}
    \eta_{(a)} = a + 3 \frac{\eps{l}}{\eps{s}}
    \label{eqn:enhacement_cass_no_clay}
\end{flalign}
(for an example where $f_c \ne 0$ see Saito et al.\,\cite{Saito2006}).
The tortuosity is taken to be a function of the volumetric gas content, 
\begin{flalign}
    \tau = (2/3)\eps{g}. \label{eqn:tortuosity_cass}
\end{flalign}

Using equations \eqref{eqn:enhacement_cass_no_clay} and \eqref{eqn:tortuosity_cass} in
the multiplicative expansion of the diffusion coefficient,
\eqref{eqn:multiplicative_enhancement} gives a diffusion coefficient of
\begin{flalign}
    D = \left(a + 3 \frac{\eps{l}}{\porosity} \right)\left( \frac{2}{3} \eps{g} \right)
    D^g.
    \label{eqn:multiplicative_enhancement_combined}
\end{flalign}
The tortuosity and the porosity communicate to the diffusion coefficient the type of
geometry under consideration.  The present model (equation
\eqref{eqn:full_diffusion_equation}) communicates this information via the porosity, the
relative permeability, and the saturated permeability.  The diffusion model using equation
\eqref{eqn:multiplicative_enhancement_combined} relies on a fitting parameter, while the
present model avoids this trouble. In the author's opinion, this highlights the main
advantage to using the chemical potential as a modeling tool.

Comparing the enhancement model of Cass et al. (using the material parameters from the
experiment by Smits et al.\,\cite{Smits2011}) to the present model, we see, in Figure
\ref{fig:diffusion_coeff_comparison}, that the relative humidity level curves of the
present model underestimate the enhanced model for many values of the fitting parameter,
$a$.  That being said, these curves do suggest an enhancement over regular Fickan
diffusion and, depending on the parameters of interst, give {\it similar} levels of
enhancement as the model used in \cite{Smits2011}.  We simply state here that the present
model offers a modified view of the enhancement model. There are several parameters that
play roles in this model, but the advantage to the present approach is that all of the
parameters are readily measured for a given medium (at least in laboratory experiments).
There is no {\it fitting} parameter, so the type of material should dictate the
level of enhancement.
\linespread{1.0}
\begin{figure}[ht*]
    \begin{center}
        \includegraphics[width=0.75\columnwidth]{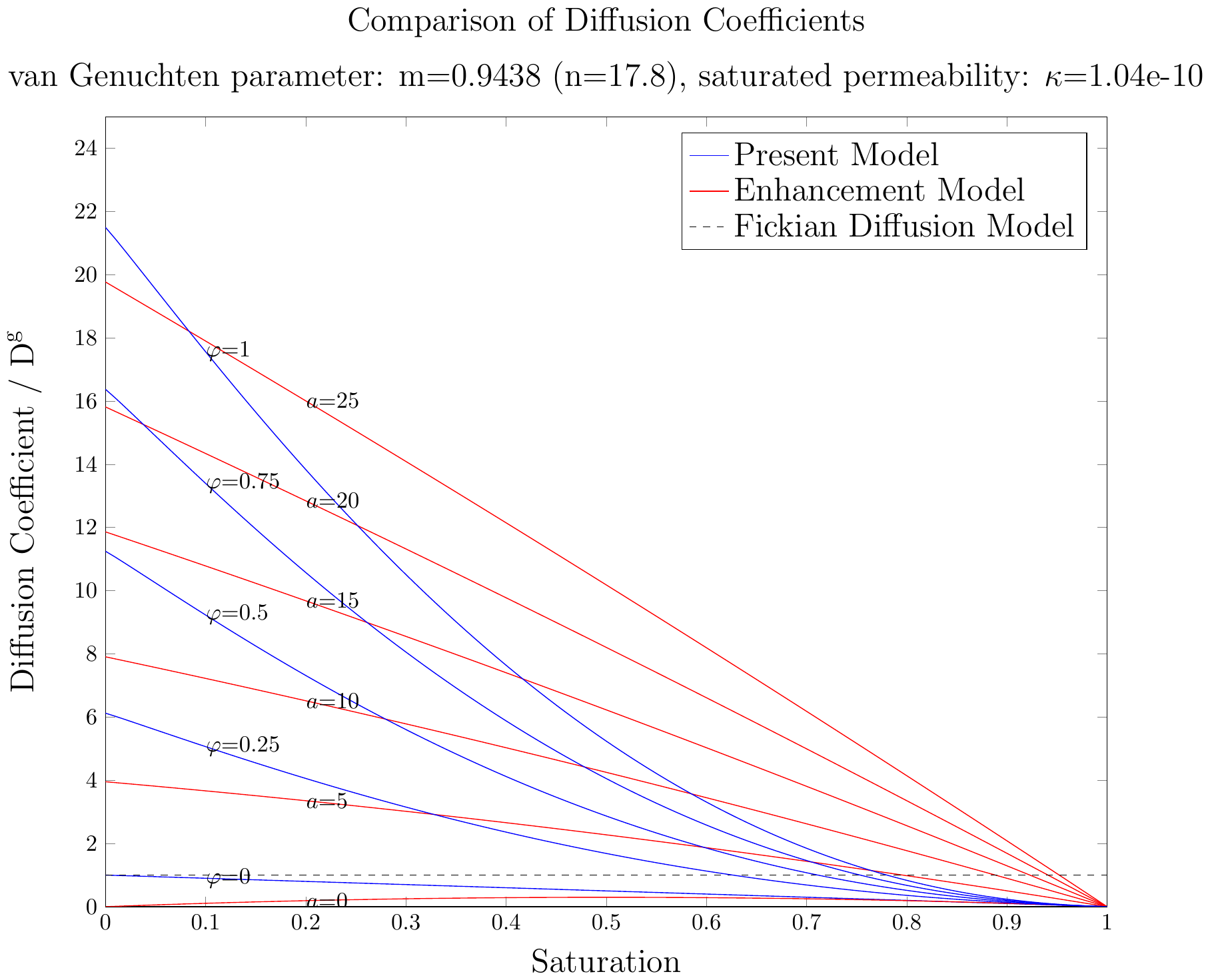} 
    \end{center}
    \caption{Comparison of different diffusion models at constant temperature ($T =
        295.15K$).  The value for the saturated permeability was chosen to match that of
        \cite{Smits2011} ($\kappa_S = 1.04 \times 10^{-10} m^2$), where they found a
    fitting parameter $a = 18.2$. The ``Present Model'' refers to equation
    \eqref{eqn:full_diffusion_equation} (with $\grad T = \foten{0}$ and no mass transfer) and the
    ``Enhancement Model'' refers to equation \eqref{eqn:deVries_original} along with
    \eqref{eqn:multiplicative_enhancement}, \eqref{eqn:enhacement_cass_no_clay}, and
    \eqref{eqn:tortuosity_cass} for the diffusion coefficient, enhancement factor, and
tortuosity respectively.}
    \label{fig:diffusion_coeff_comparison}
\end{figure}
\renewcommand{\baselinestretch}{\normalspace}

Another way to look at the present model is to consider that in most classical situations
the gas-phase pressure is considered constant. The effect of this assumption is that the
Darcy terms in the gas-phase mass balance equation are neglected.  This assumption is
valid in many cases, but in the present case the Darcy term is broken into component
parts (air and water vapor) via the chemical potentials.  The chemical potential
formulation draws influence from the Darcy-type movement, along with the Fickian
diffusion, of the individual constituents to define the general diffusion coefficient.

It is emphasized here that the traditional (de Vries-type) view of diffusion in porous
media is not taken here.  Shokri \cite{Shokri2009a} suggested that the mechanism of
enhanced diffusion is driven by the coupling of Darcy and Fickian diffusion.  The novelty
here is that the advection and diffusion are modeled in
terms of the same dependent variable; the chemical potential. This suggests that the
enhanced diffusion problem can be modeled by coupling Darcy-type flow along with Fickian
diffusion in the gas phase. The relationship between the enhancement model and the present
model will be discussed when we consider numerical solutions in Chapter
\ref{ch:Transport_Solution}.

\subsection{Total Energy Equation}\label{sec:total_energy_simplifications}
Continuing with the equation-by-equation derivation of the total heat and moisture
transport model, we now turn out attention to the total energy equation.  This picks up
from equation \eqref{eqn:energy_simplified} and we apply the simplifying assumptions
presented in the beginning of Section \ref{sec:simplifying_assumptions}.

If we assume that the vapor and air densities are functions of relative humidity and
temperature only, the total energy equation \eqref{eqn:energy_simplified} can be written as
\begin{flalign}
    \notag 0 &= \rho c_p \dot{T} - \diver \left( \soten{K} \cd \grad T \right) + \rho h + L
    \ehat{g}{l} \\
    \notag & \quad + \porosity \left[ p_c + 2\tau \porosity \Sdot + T \pd{ }{T} \left( \pi^g - \pi^l \right)
        \right] \Sdot \\
    \notag & \quad + \porosity (1-S) \left[ \sum_{j=v,a} \left[ \left( \mu^{g_j} - T
        \pd{\mu^{g_j}}{T} \right) \pd{\rho^{g_j}}{T} \right] - \left( \Gamma^g + T \eta^g
        \right) \pd{\rho^g}{T} \right] \dot{T} \\
    \notag & \quad + \porosity (1-S) \left[ \sum_{j=v,a} \left[ \left( \mu^{g_j} - T
        \pd{\mu^{g_j}}{T} \right) \pd{\rho^{g_j}}{\rh} \right] - \left( \Gamma^g + T \eta^g
        \right) \pd{\rho^g}{\rh} \right] \dot{\rh} \\
    %
        %
    \notag & \quad + \left[ \left( \rho^l c_p^l + e^l \frac{d \rho^l}{dT} \right) \grad T
        + \frac{T}{S} \pd{p^l}{T} \grad S \right] \cd \left( \eps{l} \bv{l,s} \right) \\
    \notag & \quad + \left[ \left( \rho^g c_p^g - \Gamma^g \pd{\rho^g}{T} +
        \sum_{j=v,a} \left[ \mu^{g_j} \pd{\rho^{g_j}}{T} \right] + T \pd{ }{T} \left(
        \sum_{j=v,a} \left[ \mu^{g_j} \pd{\rho^{g_j}}{T} \right] - \psi^g \pd{\rho^g}{T}
        \right) \right) \grad T \right.  \\
    \notag & \quad \qquad + \left( - \Gamma^g \pd{\rho^g}{\rh} +
    \sum_{j=v,a} \left[ \mu^{g_j} \pd{\rho^{g_j}}{\rh} \right] 
            + T \pd{ }{T} \left( \sum_{j=v,a} \left[ \mu^{g_j}
            \pd{\rho^{g_j}}{\rh}\right] - \psi^g \pd{\rho^g}{\rh} \right) \right) \grad \rh \\
    & \quad \qquad \left. - \frac{T}{(1-S)} \pd{p^g}{T} \grad S \right] \cd \left( \eps{g}
    \bv{g,s} \right).
    \label{eqn:energy_rh_S_T}
\end{flalign}
Recall that $\rho=\rho(\rh,S,T)$, $p_c = p_c(S_e)$, $\mu^{g_j} = \mu^{g_j}(\rh,T)$,
$\rho^{g_j} = \rho^{g_j}(\rh,T)$, $\Gamma^g = \Gamma^g(\rh,T)$, $\eta^g = \eta^g(T)$,
$\rho^l = \rho^l(T)$. Also recall that $\eps{\al} \bv{\al,s}$ represents the Darcy flux
for the $\al$ phase:  
\begin{subequations}
    \begin{flalign}
        \eps{l} \bv{l,s} &= -K^l \left[ \left\{ -p_c'(S_e) + C_S^l \right\} \grad S_e +
            \tau \porosity \grad \dot{S}_e + C_T^l \grad T + C_\rh^l \grad \rh - \rho^l
            \foten{g} \right] \\
        \eps{g} \bv{g,s} &= -K^g \left[ \lambda \rho^{g_v} \left(
            1-\frac{R^{g_a}}{R^{g_v}} \right) \left( \pd{\mu^{g_v}}{T} \grad T
            + \pd{\mu^{g_v}}{\rh} \grad \rh \right) + \rho^g \eta^g \grad T -
            \rho^g \foten{g} \right].
    \end{flalign}
    \label{eqn:Darcy_fluxes}
\end{subequations}

%
It is clear that there are several physical processes and couplings that occur for energy
balance to be achieved.  Equation \eqref{eqn:DeVries_heat} below shows the classical 1958
model of de Vries \cite{deVries1958} (which is similar to that of Bear \cite{Bear1988} and
is also presented in \cite{Bennethum1999} for the saturated case). 
\begin{flalign}
    \notag &\rho c_p \pd{T}{t} - \porosity \left( \rho^l W^l - \rho^g W^g \right) \pd{S}{t} \\
    & \quad = \diver
    \left( \soten{K} \grad T \right) - L \ehat{g}{l} - \left( \sum_{\al=l,g} \left(
    \frac{c_p^\al \rho^\al}{\eps{\al}}
    \right) \left( \eps{\al} \bv{\al,s} \right) \right) \grad T.
    \label{eqn:DeVries_heat}
\end{flalign}
In this form of the energy equation, $W^\al$ is a {\it differential heat of wetting}
\cite{Bennethum1999}, and the other variables are written in the present notation for
convenience. At first observation, the $\dot{T}$, $\dot{S}$, $\soten{K}$, $\ehat{g}{l}$,
and $\grad T$ terms in equation \eqref{eqn:energy_rh_S_T} are similar to terms found in
the de Vries model.  That is, we capture the standard effects of specific heat along with
differential heat of wetting, thermal conductivity, mass transfer, and convective heating.
Implicit in the $\dot{S}$ term in \eqref{eqn:energy_rh_S_T} is that we relate the partial derivative of the difference
in wetting potentials, $T \partial (\pi^g-\pi^l)/\partial T$, as a differential heat of
wetting.  The present model also captures the effects of changing relative humidity,
nonlinear effects such as $\grad S \cd \grad S$ and $\grad \rh \cd \grad \rh$, and cross
effects such as $\grad S \cd \grad \rh$.  It remains to determine which (if any) of these
effects are negligible as compared to the others. To make this determination we perform a
dimensional analysis in the next subsection. Let us first focus on the thermal
conductivity term, $\diver \left( \soten{K} \cd \grad T \right).$

The functional form of the thermal conductivity, $\soten{K}$, can be approximated in
several ways. A first approximation is to take the thermal conductivity as a weighted sum of
the conductivities of the individual phases
\begin{flalign}
    \soten{K} = \suma \epsa \soten{K}_T^\al.
    \label{eqn:thermal_weighted_sum}
\end{flalign}
Comparing to results in \cite{Smits2010a}, we note that this seems to overestimate the
measured thermal conductivity as well as fail to capture the experimentally measured
curvature of the thermal conductivity - saturation relationship. Since $\soten{K}$ is a
linearization constant that arose from the entropy inequality, it can depend on any
variable which is nonzero at equilibrium. In particular, $\soten{K}$ is a function of
saturation. Smits et al.\,\cite{Smits2010a} use a combination of the C\^ot\'e-Konrad and
Johansen models to estimate the thermal conductivity in the scalar case:
\begin{flalign}
    K(S) = K_e(S) \left( K_{sat} - K_{dry} \right) + K_{dry},
    \label{eqn:Johansen_thermal}
\end{flalign}
where $K_{sat}$ is the conductivity of the saturated medium, $K_{dry}$ is the conductivity
of the dry medium, and $K_e(S)$ is a ``normalized thermal conductivity known as the
Kersten number.'' C\^ot\'e and Konrad proposed a functional form of $K_e$ as
\begin{flalign}
    K_e(S) = \frac{\kappa S}{1 + (\kappa - 1)S}.
    \label{eqn:CoteKonrad}
\end{flalign}
The parameter, $\kappa$, is a fitting parameter that is presumed to be different for each
type of soil. In \cite{Smits2010a}, $\kappa$ was estimated for several types of sands and
several types of soil packs. Figure \ref{fig:CoteKonradThermal} shows a thermal conductivity
curve for \eqref{eqn:Johansen_thermal} with tightly packed 30/40 sand that has a porosity
of $0.334$. For comparison, equation \eqref{eqn:thermal_weighted_sum} is shown in red for
the same experiment.
\linespread{1.0}
\begin{figure}[ht!]
    \begin{center}
        \includegraphics[width=0.75\columnwidth]{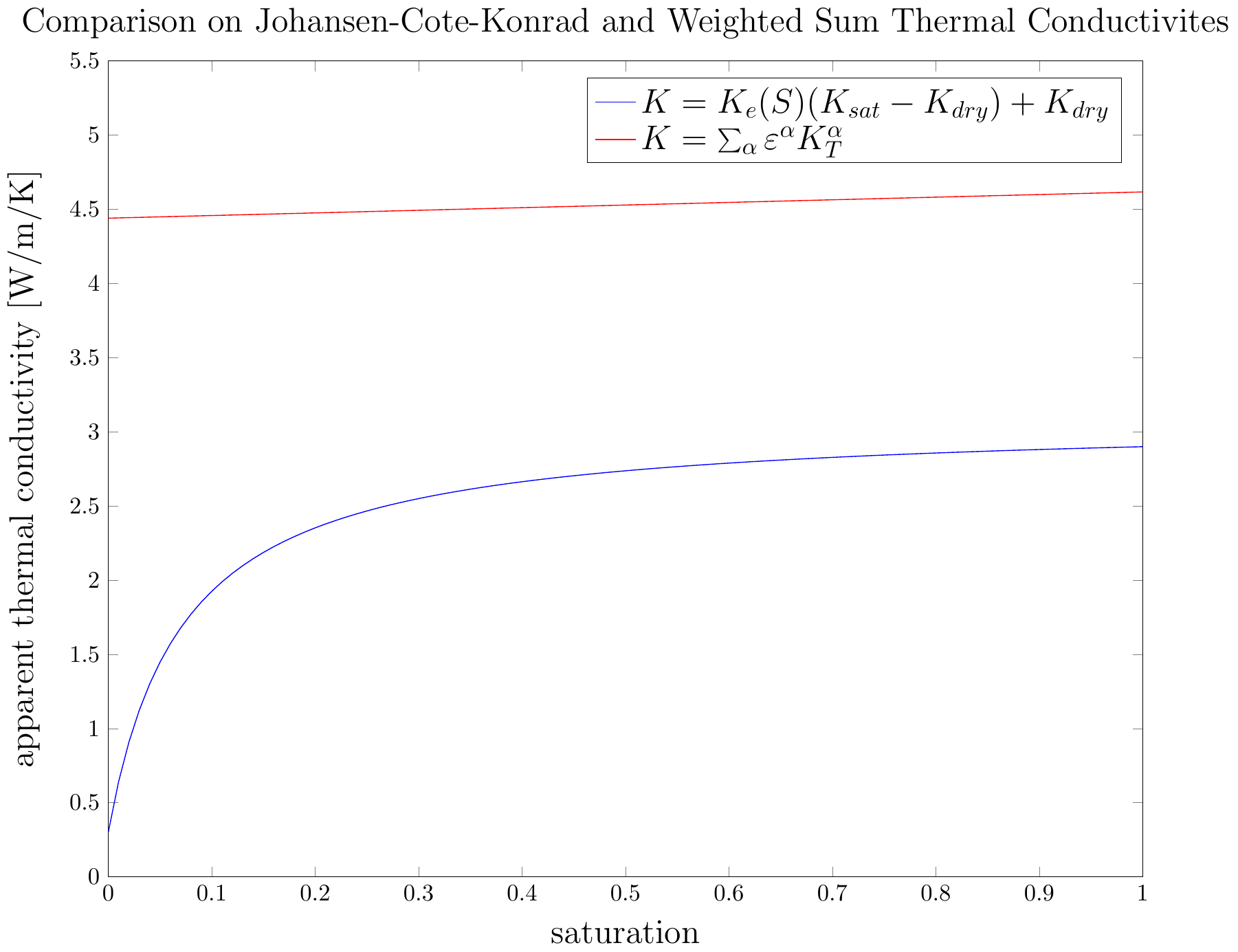}
    \end{center}
    \caption{Johansen thermal conductivity model with C\^ot\'e-Konrad $K_e - S$
    relationship (with $\kappa = 15$) plotted in blue, and the weighted sum
of the thermal conductivities of the individual phases plotted in red.}
    \label{fig:CoteKonradThermal}
\end{figure}
\renewcommand{\baselinestretch}{\normalspace}

We make some comments here giving some possible reasons for the discrepancy between the
weighted sum model (equation \eqref{eqn:thermal_weighted_sum}) and the model that more
closely matches what is experimentally observed (equation \eqref{eqn:Johansen_thermal}).
First, the thermal conductivity of air is neglected as compared to the thermal
conductivity of water or solid.  Also, the thermal conductivity of liquid is much smaller
than that of the solid, $K^l < K^s$.  Furthermore, the geometry of the packed solid plays
a crucial role.  Observe that if we idealize the soil grains as individual spheres then
there are relatively few contact points between the individual grains of the solid phase.
This idealization can be used as a partial explanation for the left-hand tail seen in the
C\^ot\'e-Konrad model depicted in Figure \ref{fig:CoteKonradThermal}. If there are few
contact points between the individual grains then it is much harder for heat to transfer
in the absence of a liquid phase connecting them.  
Thus, equation \eqref{eqn:Johansen_thermal} tells us that all pertinent information is
obtained by knowing what the thermal conductivity of the dry and saturated porous media is
as well as an interpolation function for effective saturation. This captures more of the
microscale geometry than just the volume fractions.  The effects of these two proposed
thermal conductivity functions on the behavior of the heat transport model will be
explored when we consider numerical solutions in Section \ref{sec:numerical_all_coupled}.

\subsubsection{Dimensional Analysis}
To determine which, if any, terms can be neglected from the energy transport equation we
perform a dimensional analysis.
Begin by noting that $\soten{K} / (\rho c_p)$ has units of area per time.  This
suggests a natural choice of time scale for the thermal problem of 
\[ t = \left( \frac{\rho c_p x_c^2}{K} \right) t', \]
where $t'$ is dimensionless time. Dividing by $\rho c_p$ (measured at a
reference state), introducing $x_c$ as a characteristic length (e.g. the height of a
column experiment), and multiplying by $t_c =
(\rho c_p x_c^2)/K$ gives the dimensionless form of the energy equation (the statement of
which is suppressed for the sake of brevity). 

Recall that the volumetric heat capacity, $\rho c_p$, is linearly related to the specific
heats of the individual phases
\[ \rho c_p = \suma \left( \epsa \rhoa c_p^\al \right) = \porosity S \rho^l c_p^l +
    \porosity (1-S) \rho^g c_p^g + (1-\porosity) \rho^s c_p^s. \]
Taking $S = 1$ as a reference state (or equivalently, $S=0$) gives a characteristic value
of $\rho c_p$.  Using values from Appendix \ref{app:dimensional_quantities} we see that
$\rho c_p \sim \mathcal{O}(10^6)$. Hence, several of the quantities in
\eqref{eqn:energy_rh_S_T} can be neglected:
\begin{subequations}
    \begin{flalign}
        &\left( \frac{\porosity}{\rho c_p} \right) \left[ \sum_{j=v,a} \left[ \left(
            \mu^{g_j} - T \pd{\mu^{g_j}}{T} \right) \pd{\rho^{g_j}}{T} \right] - \left(
            \Gamma^g + T \eta^g \right) \pd{\rho^g}{T} \right] \sim \mathcal{O}(10^{-4})
            \\
        &\left( \frac{\porosity}{\rho c_p} \right) \left[ \sum_{j=v,a} \left[ \left(
            \mu^{g_j} - T \pd{\mu^{g_j}}{T} \right) \pd{\rho^{g_j}}{\rh} \right] - \left(
            \Gamma^g + T \eta^g \right) \pd{\rho^g}{\rh} \right] \sim \mathcal{O}(10^{-4})
            \\
        \notag & \left( \frac{t_c}{x_c^2 \rho c_p} \right) \left[- \Gamma^g \pd{\rho^g}{T} +
            \sum_{j=v,a} \left[ \mu^{g_j} \pd{\rho^{g_j}}{T} \right] + T \pd{ }{T} \left(
            \sum_{j=v,a} \left[ \mu^{g_j} \pd{\rho^{g_j}}{T} \right] - \psi^g
            \pd{\rho^g}{T} \right) \right] \\
        & \qquad \sim \mathcal{O}(10^{-2}) \\
        \notag & \left( \frac{t_c}{x_c^2 \rho c_p} \right) \left[- \Gamma^g \pd{\rho^g}{\rh} +
            \sum_{j=v,a} \left[ \mu^{g_j} \pd{\rho^{g_j}}{\rh} \right] + T \pd{ }{T} \left(
            \sum_{j=v,a} \left[ \mu^{g_j} \pd{\rho^{g_j}}{\rh} \right] - \psi^g
            \pd{\rho^g}{\rh} \right) \right] \\
        & \qquad \sim \mathcal{O}(10^{-2}) \\
        & \left( \frac{t_c}{x_c^2 \rho c_p} \right) e^l \pd{\rho^l}{T} \sim
            \mathcal{O}(10^{-1})
    \end{flalign}
    \label{eqn:energy_neglected}
\end{subequations}
In order to make these approximations it is assumed that Gibbs potentials are given by the
Gibbs-Duhem relationship, \eqref{eqn:Gibbs_Duhem}, and that the Helmholtz potential and
internal energy are approximately the same order of magnitude as the Gibbs potential. 

With these considerations we can rewrite the present version of the energy equation as
\begin{flalign}
    \notag 0 &= \rho c_p \pd{T}{t} - \diver \left( \soten{K} \cd \grad T \right) + \rho h + L
    \ehat{g}{l} + \porosity \left[ p_c + 2\tau \porosity \Sdot + T \pd{ }{T} \left(
        \pi^g - \pi^l \right) \right] \pd{S}{t} \\
    %
    %
        %
    \notag & \quad + \left[ \rho^l c_p^l \grad T
        + \frac{T}{S} \pd{p^l}{T} \grad S \right] \cd \left( \eps{l} \bv{l,s} \right) \\
        & \quad + \left[ \rho^g c_p^g \grad T - \frac{T}{(1-S)} \pd{p^g}{T} \grad S
            \right] \cd \left( \eps{g} \bv{g,s} \right).
    \label{eqn:energy_rh_S_T_simplified}
\end{flalign}
Unfortunately this analysis leads us to the conclusion that this new version of the heat
transport equation is only {\it slightly} different than those proposed in past works
\cite{Bennethum1999,deVries1958}.
The major differences are the $\grad S$ terms associated with the Darcy fluxes, the
capillary pressure adjustment to the differential heat of wetting term, and the
Darcy fluxes themselves.  Recalling the forms of the Darcy fluxes from equations
\eqref{eqn:Darcy_fluxes}, the energy equation can be
rewritten in a more compact notation as
\begin{flalign}
    \notag 0 =& \rho c_p \pd{T}{t} - \diver \left( \soten{K} \cd \grad T \right) + \rho h
    + L \ehat{g}{l} + \mathcal{W} \pd{S}{t} \\
    \notag & + \left( \chi_1 \grad S + \chi_2 \grad T + \chi_3 \grad \rh \right) \cd \grad
    T \\
    & + \left( \chi_4 \grad S + \chi_5 \grad \rh \right) \cd \grad \rh + \chi_6 \grad S
    \cd \grad S
    \label{eqn:energy_final}
\end{flalign}
where $\mathcal{W}$ and each $\chi_j$ are implicitly defined via equations
\eqref{eqn:energy_rh_S_T_simplified} and \eqref{eqn:Darcy_fluxes}.  It remains to
determine the functional form(s) of the several constitutive variables in
\eqref{eqn:energy_final}.

\subsection{Constitutive Equations}\label{sec:constitutive_equations}
Hidden within the coefficients of \eqref{eqn:energy_final},
\eqref{eqn:saturation_with_T_rh}, and \eqref{eqn:gas_mass_balance_binary_ideal_simplified}
are a few final relationships necessary for closure.  In particular, we need constitutive
equations for
\begin{subequations}
    \begin{flalign}
        \tau &= \pd{p_c}{\epsdot{l}} \label{eqn:constitutive_tau} \\
        \ehat{l}{g_v} &= M \left( \rho^l - \rho^{g_v} \right) \left( \mu^{l} - \mu^{g_v}
        \right) \label{eqn:constitutive_ehat}\\
        \mathcal{W} &= p_c(S) + 2 \tau \porosity \dot{S} + T \pd{ }{T} \left( \pi^g - \pi^l
        \right) = p_c(S) + 2 \tau \porosity \dot{S} + W \label{eqn:constitutive_W}\\
        C_S^l &= 2 \rho^l \pd{\psi^l}{S} = \porosity \left( \ppi{l}{l} - \ppi{l}{g}
        \right) \label{eqn:constitutive_CSl}\\
        C_T^l &= \rho^l \sumj \left( \pd{\psi^l}{\rho^{g_j}} \pd{\rho^{g_j}}{T} \right) -
        \sum_{\al=g,s} \left( \frac{\eps{\al} \rho^\al}{\eps{l}} \pd{\psia}{\rho^l}
        \pd{\rho^l}{T} \right)\label{eqn:constitutive_CTl}\\
        C_\rh^l &= \rho^l \sum_j \left( \pd{\psi^l}{\rho^{g_j}} \pd{\rho^{g_j}}{\rh}
        \right).\label{eqn:constitutive_Crhl}
    \end{flalign}
\end{subequations}
The simplest possible assumption would be that $\tau, C_S^l, C_T^l, C_\rh^l,$
and $W$ are constants.  This would allow for the easiest sensitivity analysis but is
likely contrary to physical reality. The following paragraphs discuss each of these terms
and propose functional forms in terms of saturation, relative humidity, and temperature.
The sensitivity of the numerical solution to several of these parameters is discusses in
Chapter \ref{ch:Transport_Solution}.

It is generally assumed that $\tau$ in equation \eqref{eqn:constitutive_tau} is constant
\cite{Hassanizadeh2002,Peszynska2008}, but according to the linearization process in HMT,
$\tau$ can be a function of any variable that is not zero at equilibrium.  In particular,
it is possible that $\tau$ is a function of $S$; but {\it which} function? In
\cite{Berentsen2006}, the authors suggest several functional forms (constant, linear,
quadratic, Gaussian, and error) and compare to experimental findings.  Their findings
suggest that ``\ldots an error function or Gaussian relationship for the damping
coefficient $\tau$ provides reasonable agreement between data and simulations.'' Thus we
consider the following forms:
\begin{subequations}
    \begin{flalign}
        \tau &= \tau_{max} \\
        \tau &= \frac{\tau_{max}}{2} \left( 1 - \text{erf}\left( \frac{S - \mu}{\sigma}
        \right) \right) \\
        \tau &= \tau_{max} \text{exp}\left(- \frac{(S - \mu)^2}{2 \sigma^2} \right).
    \end{flalign}
    \label{eqn:proposed_tau_functions}
\end{subequations}
Plots of equations \eqref{eqn:proposed_tau_functions} are shown in Figure
\ref{fig:proposed_tau_functions} with typical mean and standard deviation parameters.  To
the author's knowledge, no other experiments have been conducted to make a better
determination as to the functional form of $\tau$. This being said, since $\tau$ is a
measure of the rate at which the pore-scale saturation profile rearranges in a dynamic
situation, it is reasonable to assume that as $S \to 1$ the effect of this term should be
minimized and as $S \to 0$ the effect should be maximized. Hence, in the author's opinion
an error function is more sensible. It remains, of course, to determine the values of the
maximum, mean and standard deviation parameters which are likely themselves functions of
material properties.
\linespread{1.0}
\begin{figure}[ht*]
    \begin{center}
        \includegraphics[width=0.7\columnwidth]{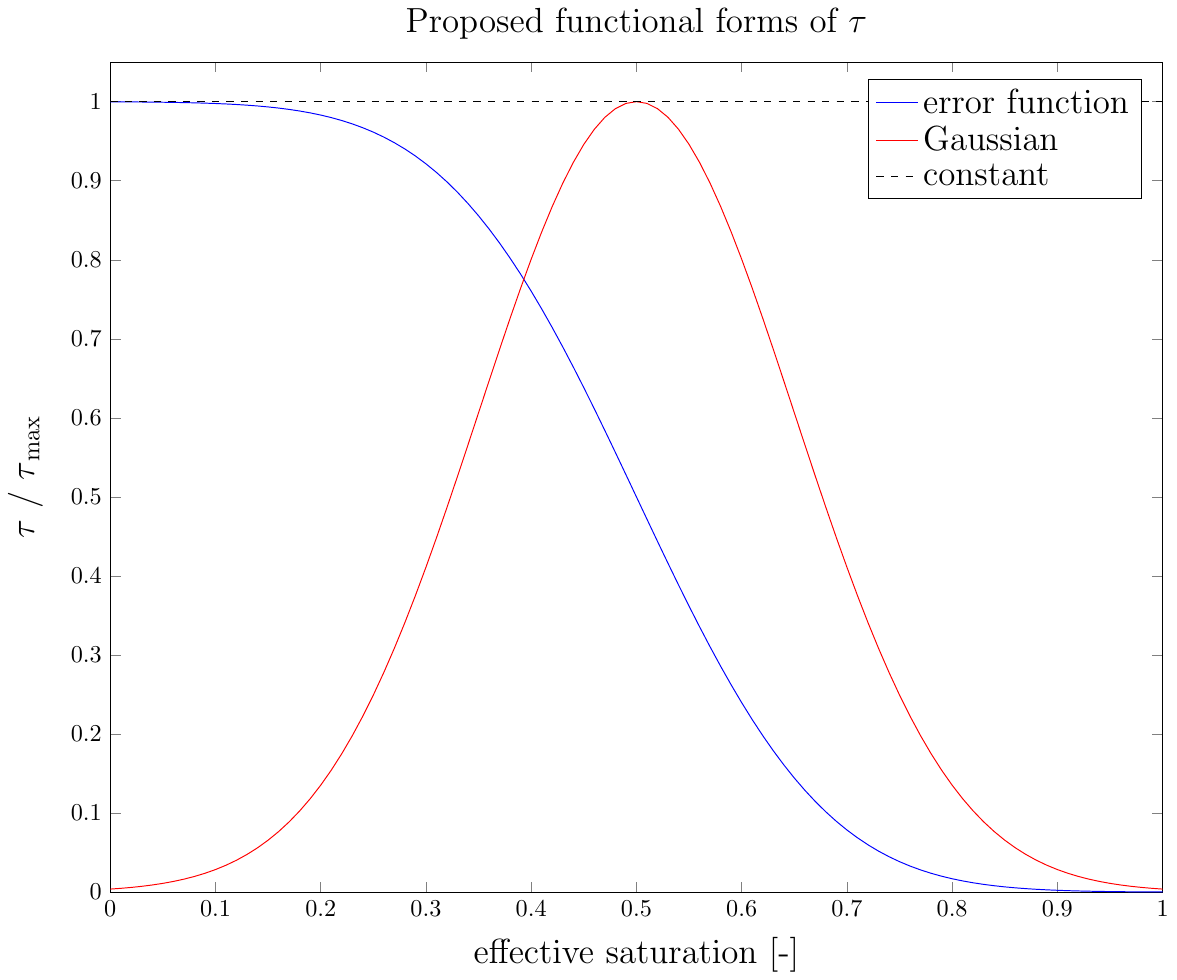}
    \end{center}
    \caption{Three proposed functional forms of $\tau = \tau(S)$}
    \label{fig:proposed_tau_functions}
\end{figure}
\renewcommand{\baselinestretch}{\normalspace}

The evaporation rate term, $\ehat{l}{g_v}$, given in equation
\eqref{eqn:constitutive_ehat} is written as a function of the difference
between the liquid and vapor chemical potentials.  The chemical potential in the water
vapor is a function of temperature and relative humidity \cite{Callen1985},
\[ \mu^{g_v} = \mu^{g_v}_* + R^{g_v} T \ln(\lambda \rh). \]
The liquid-phase chemical potential, on the other hand, does not have such a natural
description.  At equilibrium, $\mu^l = \mu^{g_v}$.  Away from equilibrium we only know
that 
\[ \mu^l = \Gamma^l = \psi^l + \frac{p^l}{\rho^l}, \]
and therefore is a function of every variable that $\psi^l$ depends.  In the most
simplistic form we can assume that the liquid chemical potential is $\mu^l = \mu^l_* +
(p^l - p^l_*)/\rho^l$. This assumption is taken from classical thermodynamics (see
\cite{Callen1985} for example).  Furthermore, $\mu^l_* \approx \mu^{g_v}_*$ if we take the
reference state to be equilibrium.  Therefore, 
\begin{flalign}
    \notag \ehat{l}{g_v} &\approx \frac{M\rh}{\rho^l} \left( \rho^l - \rho^{g_v} \right) \left(
    \frac{p^l - p^l_0}{\rho^l} - R^{g_v} T \ln \left( \lambda \rh \right)
    \right) \\
    &\approx \frac{M\rh}{\rho^l} \left( \rho^l - \rho^{g_v} \right) \left( \frac{-p_c + \tau
        \porosity \dot{S} - p^l_0}{\rho^l} - R^{g_v} T \ln \left( \lambda \rh \right)
        \right),
    \label{eqn:evaporation_rule}
\end{flalign}
where $M$ is a fitting parameter. The factor of relative humidity is included to achieve a
better match with existing empirical models (discussed in the next
paragraph). 

There are several empirical rules for evaporation in porous media.  One such rule, given
by Bixler \cite{Bixler1985} and repeated in Smits et al.\,\cite{Smits2011}, is 
\begin{flalign}
    \ehat{l}{g_v} = b (\eps{l} - \eps{l}_{r}) R^{g_v} T \left( \rho_{sat} - \rho^{g_v}
    \right),
    \label{eqn:evaporation_Bixler}
\end{flalign}
where $b$ is a fitting parameter and $\eps{l}_r$ is the residual volumetric water content.
Equations \eqref{eqn:evaporation_rule} and \eqref{eqn:evaporation_Bixler} are quite
different, but under proper scaling they are {\it close} as seen in Figures
\ref{fig:evap_vs_relhum_vs_sat}.  From
these plots it is also clear that there is a large dicrepancy between these model at very
low saturations. These plots are generated at standard temperature with $\dot{S} = 0$.
The dynamic saturation term will change the shape of these curves, but as the Bixler
model, \eqref{eqn:evaporation_Bixler},
is not dynamic we compare only with the steady state form of \eqref{eqn:evaporation_rule}.
Furthermore, the present model depends on the van Genuchten parameters for capillary
pressure.  In Figures \ref{fig:evap_vs_relhum_vs_sat} the parameters $m = 0.944$ and $\al =
5.7$ are used along with $b \approx 2.1 \times10^{-5}$ to match the values used in
\cite{Smits2011}.
\linespread{1.0}
\begin{figure}[ht*]
    \centering
    \subfigure[Comparison of mass transfer rate vs. effective saturation shown with level
    curves in relative humidity.]{
        \includegraphics[width=0.45\columnwidth]{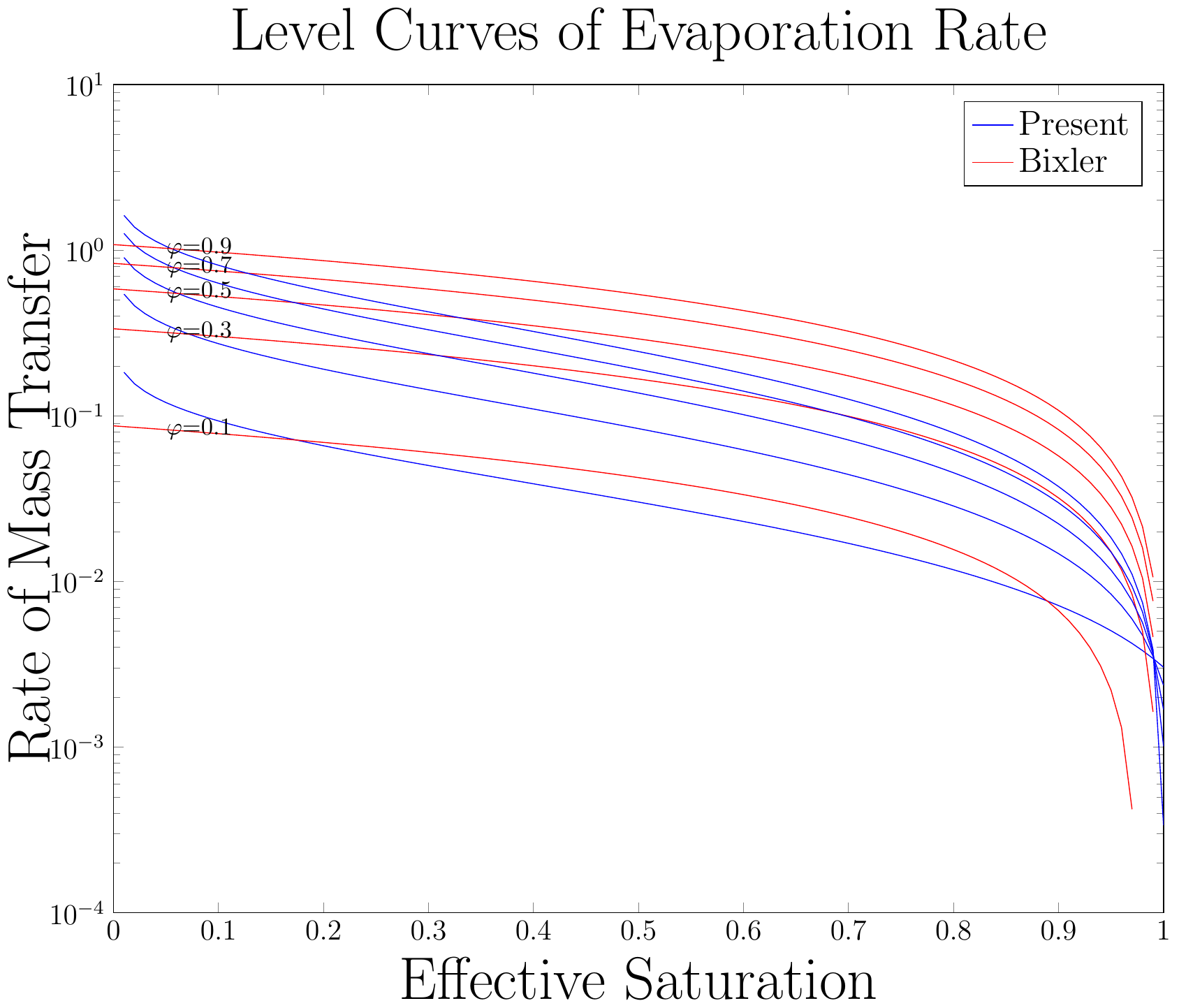}
        \label{fig:evaporation_vs_saturation}
    }
    \subfigure[Comparison of mass transfer rate vs. relative humidity shown with level
    curves in saturation.]{
        \includegraphics[width=0.45\columnwidth]{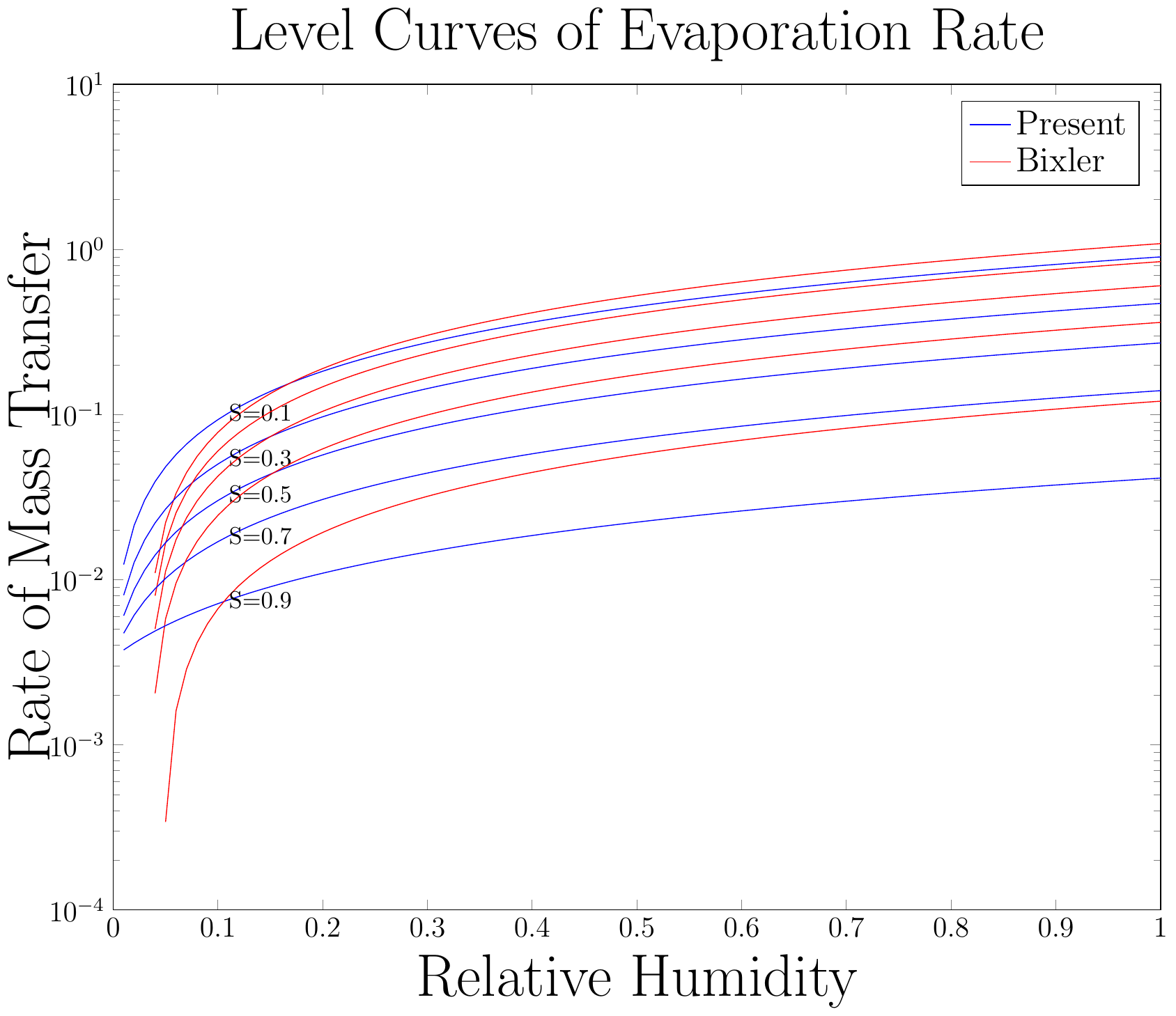}
        \label{fig:evaporation_vs_relativehumidity}
    }
    \caption{Level curves of mass transfer rate functions.}
    \label{fig:evap_vs_relhum_vs_sat}
\end{figure}
\renewcommand{\baselinestretch}{\normalspace}

The differential heat of wetting, $W$, in equation \eqref{eqn:constitutive_W} represents
the heat gained or lost due to changes in saturation and adsorption.  The present
generalization suggests that the differential heat of wetting be supplemented by the
capillary pressure and time rate of change of saturation. According to \cite{Prunty2002a},
the typical value of the differential heat of wetting is on the order to $10^3 J / kg$
depending on the type of soil. This value will be taken as constant throughout, but in
reality value should be a function of saturation.

Finally, the values of $C_S^l, C_T^l,$ and $C_\rh^l$ in equations
\eqref{eqn:constitutive_CSl} - \eqref{eqn:constitutive_Crhl} are new and hence there is no
existing literature for which to make estimates or comparisons.  For this reason we make
the initial assumption that these terms are constant. This allows for relatively simple
sensitivity analysis without introducing any unnecessary mathematical difficulties. As
discussed in Section \ref{sec:cap_pressure_and_dyn_cap_pressure}, the value of $C_T^l$ is
likely quite small since some research has been done to determine the affect of thermal
gradients on Darcy flow \cite{Saito2006}.

\section{Conclusion and Summary}\label{sec:simplifications_summary}
In this chapter we have derived several new equations and terms for heat and moisture
transport in unsaturated porous media.  
For the sake of readability, we summarize the results, assumptions, and equations derived
here within Chapter \ref{ch:Transport}.

The main assumptions are:
\begin{description}
    \item[Assumption \#1] The solid phase is rigid, incompressible, and inert.
    \item[Assumption \#2] The liquid and gas phases are composed of $N$ constituents.
        (this was later relaxed to let $N=2$ in the gas phase and $N=1$ in the liquid phase).
    \item[Assumption \#3] No chemical reactions take place in any of these phase.
    \item[Assumption \#4] Diffusion with the liquid phase is negligible compared to the
        advection of the liquid phase.
    \item[Assumption \#5] The liquid phase is incompressible.
    \item[Assumption \#6] The gas phase is an ideal binary gas mixture of water vapor and
        inert {\it air}.
    \item[Assumption \#7] The gas-phase chemical potentials and densities are functions of
        relative humidity and temperature.
\end{description}
The secondary assumptions used up to this point are (in order of appearance): 
\begin{itemize}
    \item the medium of interest is granular so angular momentum conservation yields a symmetric stress tensor,
    \item the material is {\it simple} in the sense of Coleman and Noll \cite{Coleman1963},
    \item the phase interfaces are assumed to contain no mass, momentum, or energy,
    \item second-order effects in velocity are negligible (e.g. $\bv{\aj,\al} \otimes
        \bv{\aj,\al} \ll \bv{\al}$),
    \item the species in the solid phase do not diffuse, 
    \item inertial terms in the momentum balance equation are negligible,
    \item the capillary pressure - saturation
        relationship is given by the van Genuchten function,
    \item the deviation in wetting potential is approximately zero ($
        (\pi_c|_{n.eq} - \pi_c|_{eq}) \approx 0$), 
    \item the coefficient of the dynamic saturation term, $\tau$, is constant, 
\end{itemize}

Considering assumptions \#1 - \#7 along with all of the secondary assumptions, the final
system of equations proposed to model heat and moisture transport in unsaturated porous
media is:
\begin{subequations}
    \begin{flalign}
        \notag & \pd{S}{t} - \diver \left[ \porosity_S^{-1} K^l \left( \left\{ p_c'
        + C_S^l \right\} \grad S + \tau \porosity_S \grad \dot{S} + C_T^l \grad T +
        C_\rh^l \grad \rh - \rho^l \foten{g} \right) \right] \\
        & \quad = \frac{\ehat{l}{g_v}}{\rho^l} \label{eqn:system_final_saturation} \\
        & \pd{ }{t} \left( \rho_{sat} \rh (1-S) \right) - \diver \left[ \rho_{sat} \left(
            \mathcal{D} \grad \rh + N^g \grad T \right) \right]  = -\ehat{l}{g_v}
            \label{eqn:system_final_humidity}\\
        \notag & 0 = \rho c_p \pd{T}{t} + \mathcal{W} \pd{S}{t} - \diver \left( \soten{K}
        \cd \grad T \right) + \rho h + L \ehat{g}{l} \\
        \notag &\quad + \left( \chi_1 \grad S + \chi_2 \grad T + \chi_3 \grad \rh \right)
        \cd \grad T \\
        & \quad + \left( \chi_4 \grad S + \chi_5 \grad \rh \right) \cd \grad \rh + \chi_6
        \grad S \cd \grad S,
        \label{eqn:system_final_energy} 
    \end{flalign}
    \label{eqn:system_final}
\end{subequations}
where the relevant empirical, constitutive, and derived relations are
\begin{subequations}
    \begin{flalign}
        K^l(S) &= \frac{\kappa_s}{\mu_l} \kappa_{rl} = \frac{\kappa_s}{\mu_l} \sqrt{S}
        \left( 1 - \left[ 1 - S^{1/m} \right]^m \right)^2
        \label{eqn:system_final_conductivity} \\
        K^g(S) &= \frac{\kappa_s}{\mu_g} \kappa_{rg} = \frac{\kappa_s}{\mu_g} \left( 1-S
        \right)^{1/3} \left( 1-S^{1/m} \right)^{2m}
        \label{eqn:system_final_conductivity_gas} \\
        p_c(S) &= \frac{1}{\al} \left( S^{-1/m} - 1 \right)^{1-m}
        \label{eqn:system_final_capillary_pressure} \\
        \tau &= \pd{p_c}{\epsdot{l}} \quad \text{(see equations
            \eqref{eqn:proposed_tau_functions})} \label{eqn:system_final_tau} \\
        \ehat{l}{g_v} &= M\rh \left( \rho^l - \rho^{g_v} \right) \left( \frac{-p_c + \tau
            \porosity \dot{S} - p^l_0}{\rho^l} - R^{g_v} T \ln \left( \lambda \rh \right)
            \right) \label{eqn:system_final_mass_transfer} \\
        \mathcal{D}(\rh,S,T) &:= (1-S) D^g + \rho_{sat} \rh R^{g_v} T \left( 1
        - \frac{R^{g_a}}{R^{g_v}} \right) \left( \frac{\lambda \kappa_s}{\porosity \mu_g} \right)
        \kappa_{rg}(S) \label{eqn:system_final_diffusion} \\
        %
        N^g(\rh,S,T) &= \rh \left[ \mathcal{D}(\rh,S,T) \left(
            \frac{1}{\lambda} \frac{d\lambda}{dT} + \frac{R^{g_v} \ln(\lambda \rh)}{T}
            \right) + \rho^g \rho_{sat} \eta^g \left( \frac{\lambda \kappa_S}{\porosity
            \mu_g} \right) \kappa_{rg}(S) \right] 
            \label{eqn:system_final_gas_entropy} \\
            D^g(T) &= 2.12 \times 10^{-5} \left(\frac{T}{273.15}\right)^2
        \label{eqn:system_final_diffusion_v_temp} \\
        \rho c_p &= \suma \epsa \rhoa c_p^\al \label{eqn:system_final_bulk_specific_heat}
        \\
        \rho h &= \suma \rhoa h^\al \label{eqn:system_final_heat_source} \\
        \mathcal{W} &= p_c + 2 \tau \porosity \dot{S} + W 
        \label{eqn:system_final_diff_heat_wetting} \\
        \soten{K} &= \suma \epsa \soten{K}^\al_T, \quad \text{or} \quad K = \left(
        \frac{\kappa S\left( K_{sat} - K_{dry} \right)}{1+(\kappa-1)S} \right) + K_{dry}
        \label{eqn:system_final_thermal_conductivity} \\
        \rho^l(T) &= 3.79 \times 10^{-5} (T-277.15)^3 - 7.37
        \times 10^{-3} (T-277.15)^2 + 10^3 \label{eqn:system_final_rhol}\\
        \rho_{sat}(T) &= \frac{1}{T} \exp \left(31.37 - 7.92 \times
        10^{-3}T - \frac{6014.79}{T}\right) 10^{-3} \label{eqn:system_final_rho_sat} \\
        \notag \mu_l(T) &= \left( -2.56109 \times10^{-6} (T-273.15)^3 + 0.00057672
    (T-273.15)^2 \right. \\
        & \qquad \left. -0.0469527 (T-273.15)+1.75202 \right) 10^{-3} \label{eqn:system_final_liquid_vis} \\
        \mu_g(T) &=\left( \frac{1.02312 T^3}{10^9}-\frac{3.62788 T^2}{10^6}+0.00665915
    T+0.11767 \right) 10^{-5} \label{eqn:system_final_gas_vis} \\
        L(T) &= 2.501 \times 10^6-2369.2 (T-273.15) \label{eqn:system_final_latent_heat} \\
        \notag \eta^g(T) &= 6.1771 \times 10^{-4} (T-273.15)^4 - 7.3971
        \times 10^{-2} (T-273.15)^3 \\ 
        & \quad +3.1324 (T-273.15)^2-34.4817 (T-273.15)+191.208 \label{eqn:system_final_gas_ent}.
    \end{flalign}
    \label{eqn:system_final_constitutive}
\end{subequations}

Equations \eqref{eqn:system_final} coupled with equations
\eqref{eqn:system_final_constitutive} give several adjustments to the classical models for
saturation (Richards'), vapor diffusion (Phillip and de Vries), and heat transport (de
Vries) presented in Section \ref{sec:classical_models}.  In order for
the present models to be accepted in the hydrology community we must show that the
proposed terms are non-negligible and in some way put some of the empirical relations on a
firmer theoretical footing. The proposed vapor diffusion equation
\eqref{eqn:system_final_humidity} is a prime example of this as there are no empirical
fitting parameters within the diffusion coefficient (hence removing the need for an
empirical {\it enhancement factor}).

In Chapter
\ref{ch:Existence_Uniqueness} we discuss the mathematical questions of existence and
uniqueness of solutions to the individual equations. In Chapter
\ref{ch:Transport_Solution} we discuss numerical simulations of the models. 

\newpage
\chapter{Existence and Uniqueness Results}\label{ch:Existence_Uniqueness}
In this chapter we discuss the necessary regularity and assumptions for existence and
uniqueness of solutions for the three equations.   As the main thrust of this work is not
to prove existence and uniqueness for general classes of systems of partial differential
equations, we approach these problems by stating relevant existing theorems from the
literature and satisfying the hypotheses of these theorems. The saturation and gas
diffusion equations are both of parabolic type and can be treated similarly.  The heat
transport equation is an advection-reaction-diffusion equation that, in principle, should
be parabolic in nature.  The advection terms force a different approach to this equation.  
In Section \ref{sec:existence_sat_hass}, an existence and uniqueness result for the
saturation equation with the third-order term (due to
Mikeli\'c \cite{Mikelic2010}) is
outlined.  The theorems of Alt and Luckhaus \cite{Alt2012,Alt1983} are outlined in Section
\ref{sec:alt_luckhaus} and then used in Sections \ref{sec:existence_sat_rich} and
\ref{sec:existence_diffusion} to prove existence and uniqueness results for Richards'
equation and the vapor diffusion equation respectively. Finally, an existence and
uniqueness result for a special case of the heat transport equation is presented in
Section \ref{sec:existence_uniqueness_heat}.

\section{Saturation Equation with $\tau \ne 0$}\label{sec:existence_sat_hass}
The saturation equation has been well studied since Richards' first introduced it in the
1930's.  Recent modeling efforts, including those of Hassanizadeh et al., have
introduced a new term into the classical Richards' equation and this has caused a
resurgence in the analytical study of the saturation equation.  The 2010 paper by Andro
Mikeli\'c \cite{Mikelic2010} gives the necessary conditions for existence and uniqueness
of a weak solution to the following equation:
\begin{subequations}
    \begin{flalign}
        \pd{S}{t} = \diver \left\{ K(S) \left( -\frac{d p_c}{dS} \grad
        S + \tau \grad \left( \pd{S}{t} \right) + \foten{e}_3 \right)
    \right\} & \text{ in } \Omega_T = \Omega \times (0,T) \label{eqn:hass_eqn_a} \\
    S = S_D &\text{ on } \Gamma_D = \partial_D \Omega \times (0,T) \\
    K(S) \left( -\frac{d p_c}{dS} \grad S + \tau \grad \left(
    \pd{S}{t} \right) + \foten{e}_3 \right) \cd \nu = R &\text{ on } \Gamma_N =
    \partial_N \Omega \times (0,T) \\
        S(x,t=0) = S_i(x) &\text{ on } \Omega
    \end{flalign}
    \label{eqn:hass_eqn}
\end{subequations}
Here,  $\foten{e}_3$ is a unit vector pointing the $z$ direction to account for
gravitational effects, $\nu$ is an outward pointing normal, and the subscripts $D$ and
$N$ represent Dirichlet and Neumann conditions repsectively. Notice that
\eqref{eqn:hass_eqn} is a simplification of the present saturation equation as it contains
no evarporation (source) term and no coupling with relative humidity or temperature.

Mikeli\'c's theorem is stated here for completeness.
\begin{theorem}[Mikeli\'c 2010 \cite{Mikelic2010}, Theorems 3 \& 4]\label{thm:Mikelic}
Consider the following hypotheses:
\begin{description}
    \item[H1:] there are constants $\beta>0, S_K>0$ and a nonnegative function $f \in
        C_0^\infty(\mathbb{R})$ such that $K$ is given by 
        \[ K(z) = \frac{S_k z^\beta}{1+S_K z^\beta f(z)}, \, z \in [0,1] \]
    \item[H2:] there exists $\lambda>0, S_p>0, M_p>0$ and an arbitrary function $g \in
        C_0^\infty(\mathbb{R})$ such that $-p_c'$ is written as
        \[ -p_c'(z) = \frac{S_p z^{-\lambda}}{1+M_p z^\lambda g(z)}, \, z \in [0,1] \]
    \item[H3:] the product of the functions $K$ and $p_c'$ is bounded on $[0,1]$. 
    \item[H4:] the initial Dirichlet data is smooth: $S_D \in C^1([0,T];
        H^1(\Omega))$, and is bounded away from zero $0 < S_{D,min} \le
        S_D(x,t) \le 1$ (or impose that $S_D \ne 0 a.e.$)
    \item[H5:] $R = R_0 \zeta;$ where $R_0 \in C^1(\overline{\Gamma}_N \times [0,T]), \,
        R_0 \ge 0$ and $\zeta \in C_0^\infty(\mathbb{R}), \, \zeta(z) \ge 0 \text{ for } z >
        0,$ and $z\zeta(z) \ge 0 \text{ for } z<0$.
    \item[H6:] Initial moisture content satisfies a ``finite entropy'' condition: $\int_{\Omega}
        \left( S_0(x)\right)^{2-\beta} dx < + \infty$ where $\beta \ge \lambda > 2$
\end{description}
Under these hypotheses there is a weak solution for \eqref{eqn:hass_eqn} where $S
\in H^1(\Omega_T)$ such that $0 \le S(x,t) \, a.e.$ on $\Omega_T$, $\grad \partial_t S \in
L^2(\Omega_T)$ and $S-S_D \in L^2(0,T;V)$ for $V = H^1(0,1)$.
\end{theorem}

The proof of this theorem is beyond the scope of this work, but it indicates that under
constant relative humidity and temperature conditions, where no mass
transfer is expected, there exists a weak solutution to the saturation equation. The sixth
hypothesis restricts the shape of the initial condition. Simply put, the initial condition
cannot drop to zero in such a way as to make $\int_\Omega S_0^{2-\beta} dx$ go to
infinity.  This avoids the natural degenerate nature of the problem. The regularity
expected for the solution ($H^1$) is a nice result given that this is actually a
third-order differential equation.

\section{Alt and Luckhaus Existence and Uniqueness Theorems}\label{sec:alt_luckhaus}
We now turn our attention to demonstrating the necessary conditions for existence and
uniqueness of Richards' equation (saturation with $\tau = 0$) and the vapor diffusion
equation in the special cases where the other dependent variables are held fixed (possibly
even constant).  The two equations are treated together in this section since they both fall under the
class of quasi-linear parabolic equations.  As such, they can be analyzed using similar
theory.  For the purposes of demonstrating existence and uniqueness we apply general
theorems by Alt and Luckhaus \cite{Alt2012,Alt1983} to these equations.

The following paragraphs are paraphrased from Alt and Luckhaus \cite{Alt1983} and are
presented here to introduce the reader to the notation used therein and for future reference.

Consider the general initial boundary value problem (IBVP) for a system of quasilinear
elliptic-parabolic differential equations
\begin{subequations}
    \begin{flalign}
        \partial_t b^j(\foten{u}) - \diver a^j(b(\foten{u}),\grad \foten{u}) &= f^j(b(\foten{u})) \quad \text{ in } (0,T)
        \times \Omega, \, j=1:m \label{eqn:alt_eqn} \\
        b(\foten{u}) &= b^0 \quad \text{ on } \{0\} \times \Omega \label{eqn:alt_ic} \\
        \foten{u} &= \foten{u}^D \quad \text{ on } (0,T) \times \Gamma \label{eqn:alt_dirichlet} \\
        a^j(b(\foten{u}),\grad \foten{u}) \cd \nu &= 0 \quad \text{ on } (0,T) \times (\partial \Omega
        \setminus \Gamma), \, j=1:m \label{eqn:alt_neumann}
    \end{flalign}
\label{eqn:alt}
\end{subequations}
In equations \eqref{eqn:alt}, $\foten{u}\in \mathbb{R}^m$, $b:\mathbb{R}^m \to
\mathbb{R}^m$, and $a: \mathbb{R}^m \times \mathbb{R}^{m\times N} \to \mathbb{R}^m$ where
$N$ is the spatial dimension of the problem and $m$ is the number of equations. 

We call $\foten{u}$ in the affine space $\foten{u}^D + L^r(0,T;V)$ a weak solution of
\eqref{eqn:alt} if the following two properties are fulfilled:
\begin{enumerate}
    \item $b(\foten{u}) \in L^{\infty}(0,T;L^1(\Omega))$ and $\partial_t b(\foten{u}) \in
        L^{r^*}(0,T;V^*)$ with initial values $b^0$, that is
        \[ \int_0^T \left< \partial_t b(\foten{u}), \zeta \right> + \int_0^T \int_\Omega \left(
            b(\foten{u}) - b^0 \right) \partial_t \zeta = 0 \]
        for every test function $\zeta \in L^r(0,T;V) \cap W^{1,1}(0,T;L^\infty(\Omega))$
        with $\zeta(T) = 0$.
    \item $a(b(\foten{u}),\grad \foten{u}), \, f(b(\foten{u})) \in L^{r^*}( (0,T) \times
        \Omega)$ and $\foten{u}$ satisfies the differential equation, that is,
        \[ \int_0^T \left< \partial_t b(\foten{u}), \zeta \right> + \int_0^T \int_\Omega
            a(b(\foten{u}),\grad \foten{u}) \cd \grad \zeta = \int_0^T \int_\Omega
            f(b(\foten{u})) \zeta \]
        for every $\zeta \in L^r(0,T;V)$.
\end{enumerate}
Recall from Functional Analysis that $V^*$ is the dual space of the vector space $V$, and
$W^{k,p} = \{ u \in L^p(\Omega) \, : \, D^\al u \in L^p(\Omega) \, \forall |\al| \le k \}$
with the weak derivative $D^\al u$.  The reader should also recall the common simplified
notation $W^{k,2}(\Omega) = H^k(\Omega)$.

Consider the following hypotheses:
\begin{description}
    \item[H1:] $\Omega \subset \mathbb{R}^n$ is open, bounded, and connected with Lipschitz
        boundary, $\Gamma \subset \partial \Omega$ is measurable with $H^{n-1}(\Gamma) >
        0$ and $0 < T < \infty$.
    \item[H2:] $b$ is a monotone vector field and a continuous gradient, that is, there is
        a convex $C^1$ function $\Phi: \mathbb{R}^m \to \mathbb{R}$ with $b = \grad \Phi$.
        We can assume that $b(0) =0$.  
        The convexity of $\Phi$ then implies that we can define
        \[ B(\foten{z}) :=  b(\foten{z}) \cd \foten{z} -
            \Phi(\foten{z}) + \Phi(0). \]
        \item[H3:] $a(b(\foten{z}),\foten{p})$ is continuous in $\foten{z}$ and
            $\foten{p}$ and elliptic in the sense that 
            \[ \left( a(b(\foten{z}),\foten{p}^{(1)}) - a(b(\foten{z}),\foten{p}^{(2)})
                \right) \cd \left( \foten{p}^{(1)} - \foten{p}^{(2)} \right) \ge C \left|
                \foten{p}^{(1)} - \foten{p}^{(2)} \right|^r \]
                with $1 < r < \infty$ and $f(b(\foten{z}))$ continuous in $\foten{z}$.
    \item[H4:] The following growth condition is satisfied:
        \[ \left| a(b(\foten{z}),\foten{p}) \right| + \left| f(b(\foten{z})) \right| \le c
            \left( 1 + B(\foten{z})^{(r-1)/r} + |\foten{p}|^{r-1} \right). \]
    \item[H5:] We assume that $u^D$ is in $L^r(0,T;W^{1,r}(\Omega))$ and in
        $L^{\infty}\left( (0,T) \times \Omega \right)$ and we define
        \[ V := \{ v \in W^{1,r}(\Omega) : v = 0 \text{ on } \Gamma \} \]
    \item[H6:] Assume either that $b^0$ maps into the range of $b$ and therefore there is a
        measurable function $\foten{u}^0$ with $b^0 = b(\foten{u}^0)$ or that 
        \[ \partial_t u_D \in L^1(0,T;L^{\infty}(\Omega)). \]
\end{description}

The existence and uniqueness theorems of Alt and Luckhaus \cite{Alt1983} are stated here
for convenience and reference.

\begin{theorem}[Alt and Luckhaus \cite{Alt1983}, Theorem 1.7]\label{thm:alt}
Suppose the data satisfy H1 - H6, and assume that $\partial_t u^D \in
L^1(0,T;L^\infty(\Omega))$. Then there is a weak solution to \eqref{eqn:alt}.
\end{theorem}

\begin{theorem}[Alt and Luckhaus \cite{Alt1983}, Theorem 2.4]\label{thm:alt_unique}
Suppose that the data satisfy H1 - H6 with $r=2$ and 
\[ a(t,x,b(\foten{z}),\foten{p}) = A(t,x)\foten{p} + e(b(\foten{z})) \]
where $A(t,x)$ is a symmetric matrix and measurable in $t$ and $x$ such that for $\al > 0$
\[ A - \al I \quad \text{and} \quad A + \al \partial_t A \]
are positive definite.  Moreover assume that
\[ \left| e(b(\foten{z}_2)) - e(b(\foten{z}_1)) \right|^2 + \left| f(b(\foten{z}_2)) -
    f(b(\foten{z}_1)) \right|^2 \le C \left( b(\foten{z}_2) - b(\foten{z}_1) \right)
    \left( \foten{z}_2 - \foten{z}_1 \right). \]
Then there is at most one weak solution.
\end{theorem}

\subsection{Existence and Uniqueness for Richards' Equation}\label{sec:existence_sat_rich}
The existence and uniqueness of weak solutions to Richards equation is known, and a
general tool for handling this problem is the Alt-Luckhaus theorem stated above.  In this
subsection we set up and state the theorems. 
We will show that the hypotheses of Theorems \ref{thm:alt} and \ref{thm:alt_unique} are
satisfied under restrictions on the Dirichlet boundary conditions and appropriate
boundedness assumptions.  The equation 
\begin{flalign}
    \notag \pd{S}{t} &= \pd{ }{x} \left( \frac{\kappa_{rw}(S)}{\rho^l g} \pd{p^l}{x} -
    \kappa_{rw}(S) \right) \\
    &= \pd{ }{x} \left( \kappa_{rw}(S) \pd{h}{x} - \kappa_{rw}(S) \right)
    \label{eqn:Richards_head}
\end{flalign}
is Richards' equation in dimensionless time and one spatial dimension. In this formulation
we take $h$ as the hydraulic head: $h = p^l / (\rho^l g)$.  Recall from previous
discussions that the pressure (or head) is a function of saturation.  This relationship is
invertible so here we note that saturation can be written as a function of pressure (or
head).  As suggested in \cite{Alt1983,Pop2004},
``saturation may be less regular than pressure, therefore we expect to achieve better
[regularity] results by applying a Kirchhoff transform''. A Kirchhoff transformation gives
a smoothed relationship between head and a new unknown; a generalized pressure head
\[ \mathscr{K}: \mathbb{R} \to \mathbb{R} \quad \text{as} \quad \mathscr{K}(h) =
    \int_{-\infty}^h \kappa_{rw}(S(q)) dq := u. \]
The pressure head is taken to be negative by convention (opposite sign of capillary
pressure).
The variable $u$ now becomes the primary unknown of \eqref{eqn:Richards_head} since the
spatial derivative can be written as
\[ \frac{du}{dx} = \frac{d \mathscr{K}}{dh} \frac{dh}{dx} = \kappa_{rw}
    \frac{dh}{dx}, \]
and if $b(u)$ is defined as $b(u) = S(\mathscr{K}^{-1}(u))-1$ then
\begin{flalign}
    \pd{ }{t} \left( b(u) \right) = \pd{ }{x} \left( \pd{u}{x} -  \kappa_{rw}(b(u)+1)
    \right).
\end{flalign}
Notice that the definition of $b$ depends on the
invertibility of $\mathscr{K}$. Also, since $\mathscr{K}^{-1}(u) = h$, $b$ can be seen as
a function of $h$: $b(h) = S(h)-1$ (see Figure
\ref{fig:KirchhoffRichards_bu_m08}) .  Given the van Genuchten capillary pressure - saturation
relationship,
$ S(h) = \left( \left( \al h \right)^{1/(1-m)} + 1 \right)^{-m},$
and the van Genuchten relative permeability function,
$\kappa_{rw}(S) = \sqrt{S} \left( 1 - \left( 1 - S^{1/m} \right)^{m} \right)^2,$
Figure \ref{fig:KirchhoffRichards_m08} shows several plots of $\mathscr{K}$ for various
parameter values. There is a horizontal asymptote as $h\to -\infty$ and it is evident from
the plot that $\mathscr{K}$ is one-to-one and onto for all values of $h \in (-\infty,0)$,
but as $h$ gets {\it large} in absolute value the inverse becomes
unstable. 
\linespread{1.0}
\begin{figure}[H]
    \begin{center}
        \includegraphics[width=0.75\columnwidth]{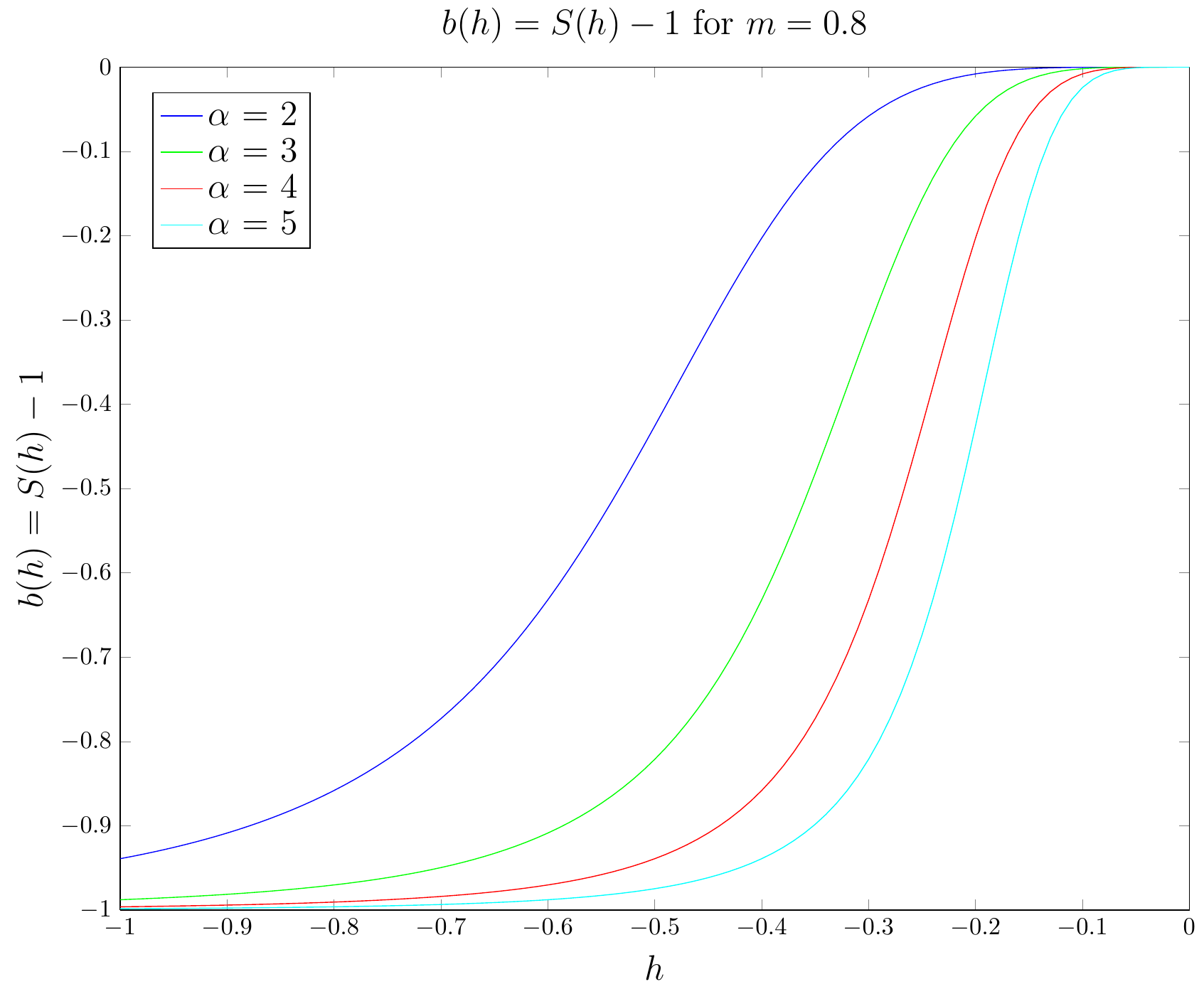}
    \end{center}
    \caption{The function $b(h) = S(h)-1$ for $m = 0.8$ and various values of
$\alpha$.}
    \label{fig:KirchhoffRichards_bu_m08}
\end{figure}
\renewcommand{\baselinestretch}{\normalspace}
\linespread{1.0}
\begin{figure}[H]
    \begin{center}
        \includegraphics[width=0.75\columnwidth]{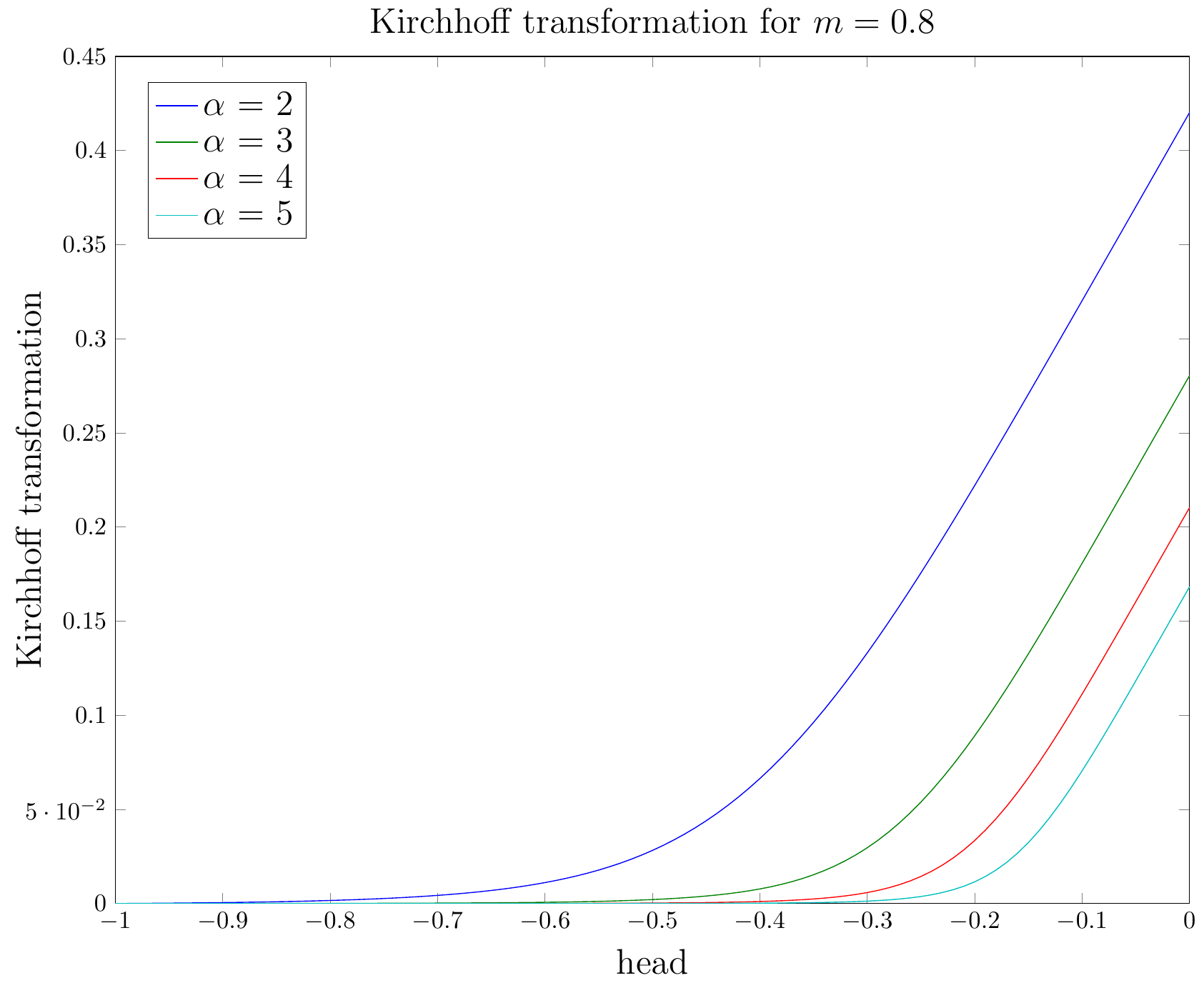}
    \end{center}
    \caption{Kirchhoff transformation $\mathscr{K}$ for $m = 0.8$ and various values of
$\alpha$.}
    \label{fig:KirchhoffRichards_m08}
\end{figure}
\renewcommand{\baselinestretch}{\normalspace}

Matching to equation \eqref{eqn:alt_eqn} we note that $j=1$, define $a(b(u),\grad u)$ as 
\[ a\left(b(u),\pd{u}{x} \right) = \pd{u}{x} + \kappa_{rw}(b(u)+1), \]
and notice that $f = 0$.  Given constant head Dirichlet boundary conditions, we finally
rewrite Richards' equation as
\begin{subequations}
    \begin{flalign}
        \pd{ }{t} \left( b(u) \right) &= \pd{ }{x} \left( \pd{u}{x} - \kappa_{rw}(b(u))
        \right) \quad \text{in} \quad (0,T) \times \Omega \\
        u &= u_D \quad \text{on}\quad  (0,T) \times \{0,1\} \\
        u &= u_0 \quad \text{on}\quad  0 \times \Omega
    \end{flalign}
    \label{eqn:Richards_Kirchhoff_PDE_IBC}
\end{subequations}
With this form of Richards' equation we propose the following existence and uniqueness
result.

\begin{theorem}
    \label{thm:Richards_Kirchhoff_exist_unique}
    Suppose that the following conditions hold for the generalized head, $u$, in
    equation \eqref{eqn:Richards_Kirchhoff_PDE_IBC}.  
    \begin{enumerate}
        \item $\Omega = (0,1)$ and $t \in (0,T) \subset (0,\infty)$
        \item $u_D \in L^{\infty}((0,T)\times\Omega)$
        \item $\partial_t u_D \in L^1(0,T;L^{\infty}(\Omega))$
    \end{enumerate}
    Then there exists a unique weak ($u \in u_D + L^2(0,T;H_0^1(\Omega))$) solution to
    \eqref{eqn:Richards_Kirchhoff_PDE_IBC}.
\end{theorem}
The proof of Theorem \ref{thm:Richards_Kirchhoff_exist_unique} has been discussed in
several articles. In particular, the transformation of Richards equation to the form seen
in equations \eqref{eqn:Richards_Kirchhoff_PDE_IBC} are discussed as a model problem for
the Alt and Luckhaus theorems \cite{Alt1983}.  Furthermore, this proof is presented in
\cite{Pop2004} as part of their numerical formulation of Richards' equation.  The
fundamental reason for presenting this result here is that the value of $\tau$ in the new
saturation equations is not yet 
well known in experimental studies.  Presenting this case simply covers all of the
possible bases.

In the cases where $\grad \rh$, $\grad T$, or mass transfer terms are non-zero, these
terms become source terms that depend on $x$. This means that $f = f(x,b(u)) \ne 0$.  According
to section 1.10 in \cite{Alt1983} ``it makes no difference if $a$ and $f$ depend on $x$.''
This is made more clear in their subsequent work \cite{Alt2012} where the theorem is
explicitly stated to allow for $x$ and $t$ dependence.

For comparison sake we observe the difference between the regularity required for Richards
equation (Theorem \ref{thm:Richards_Kirchhoff_exist_unique}) and for the extended
saturation equation with the third-order term (Theorem \ref{thm:Mikelic}). For the
equation with the third-order term ($\diver \grad \dot{S}$), the weak solution is in $H^1$
while in the second-order equation the weak solution is in $L^2$.  This extra required
regularity is expected.

\subsection{Vapor Diffusion Equation}\label{sec:existence_diffusion}
To prove existence for the gas diffusion problem we proceed using the theorem of Alt and
Luckhaus as with the saturation equation.  Recall that under constant temperature and
fixed saturation conditions
\begin{flalign}
    (1-S(x)) \pd{\rh}{t} - \pd{ }{x} \left( \mathcal{D}(\rh,S(x)) \pd{\rh}{x} \right) = 0
    \label{eqn:diffusion_exist}
\end{flalign}
where 
\begin{flalign}
    \mathcal{D}(\rh,S(x)) = (1-S(x)) D^g + \rho_{sat} \rh R^{g_v} T \left( 1 -
    \frac{R^{g_a}}{R^{g_v}} \right) \lambda K(S(x)).
    \label{eqn:diff_coeff_exist}
\end{flalign}
Allowing $S$ to be a function of $x$ constitutes a departure from the exact form of the
parabolic-elliptic system found in Alt and Luckhaus (see equations \eqref{eqn:alt}) as
this is now a non-autonomous differential equation.  In \cite{Alt2012} this proof was
generalized to allow for $b = b(x,u)$ and for $a = a(x,u,\grad u)$ (see section 11 of
\cite{Alt2012}). The only additional assumptions for the existence theorems are that $b :
\Omega \times \mathbb{R} \to \mathbb{R}$ and $a: \Omega \times \mathbb{R} \times
\mathbb{R}^N \to \mathbb{R}^N$ are measurable in the first argument and continuous in the
others. With this addition to Theorem \ref{thm:alt} we proceed with the existence theorem
for the gas-phase equation.  

Assume that the initial-boundary conditions are
\begin{subequations}
    \begin{flalign}
        \rh(0,t) &= \rh_{D,0} \in (0,1), \quad \forall t \in (0,T) \\
        \rh(1,t) &= \rh_{D,1}(t) \in (0,1), \quad \forall t \in (0,T) \\
        \rh(x,0) &= \rh_0(x), \quad x \in \Omega.
    \end{flalign}
        \label{eqn:gas_ibc}
\end{subequations}
Note here that the Dirichlet boundary condition on the {\it right-hand} side of $\Omega$
is time dependent and the one on the left is independent of time. The problem could also
be restated where the right-hand boundary is of Neuman type. The conditions are chosen to
better match the experimental data that will be considered in Section
\ref{sec:FullyCoupledSolutions}.

\begin{theorem}[Existence of Weak Solution to Diffusion
    Equation]\label{thm:existence_diffusion}
    Suppose that the following conditions hold for equations \eqref{eqn:diffusion_exist} -
    \eqref{eqn:gas_ibc}.
    \begin{enumerate}
        \item $\Omega = (0,1)$ and $t \in (0,T) \subset (0,\infty)$
        \item $\rh \in (0,1-\epsilon] $ for all $x \in \Omega$ and for all $t \in
            (0,T)$ where $0 < \epsilon \ll 1$
        \item $S \in [\epsilon,1-\epsilon]$ and $S(x) \in C^1(\Omega)$ (independent of
            time) where $0 < \epsilon \ll 1$
        \item $\rh_{D,1} \in L^2(0,T;H^1(\Omega))$ and $L^{\infty}( (0,T) \times
            \Omega)$
    \end{enumerate}
    Then there exists a weak solution to \eqref{eqn:diffusion_exist} -
    \eqref{eqn:gas_ibc} in the sense that $\rh \in \rh_D + L^2(0,T;H_0^1(\Omega))$.
\end{theorem}

Matching equation \eqref{eqn:diffusion_exist} to the form of Alt and Luckhaus (equation
\eqref{eqn:alt_eqn}) we have 
\begin{align}
    b(x,z) &= (1-S(x))z, & a(x,z,p) &= \mathcal{D}(x,z) p, &
    f &= 0, & m &= 1.
\end{align}
In the conditions for Theorem \ref{thm:existence_diffusion} we use the
parameter $\epsilon$ to define two different sets.  This is a small abuse of notation
since $\rh$ and $S$ need not belong to exactly the same set.  We are simply stating that
both of these functions must be bounded away from $0$ and $1$.

\begin{proof}
    We proceed by verifying hypotheses H1 - H6 of Theorem \ref{thm:alt} noting the
    extension proposed in \cite{Alt2012} to non-autonomous functions. 
    \begin{itemize}
        \item[H1:] In 1 spatial dimension it is clear that $\Omega$ is an open, bounded,
            and connected domain with Lipschitz boundary.  $\Gamma = \{0,1\}$, and
            $H^0(\Gamma) > 0$ and $0 < T < \infty$.

        \item[H2:] In this case we note that $b(x,z) =(1-S(x)) z$. Clearly $b(x,0) = 0$.
            Define $\Phi(x,z) = (1-S(x)) z^2/2$ and observe that $\partial \Phi / \partial
            z = (1-S(x)) z = b(x,z)$ and $\partial^2 \Phi / \partial z^2 = 1 - S(x) > 0$
            since $S(x) \in (0,1)$. Since $S \in C^1(\Omega)$ (assumption \#3 in the
            statement of the theorem) it is clear that $b$ is
            measurable in the first component.  Furthermore, $b$ is a continuous gradient
            of a convex function in the second component. Define $B(x,z) = b(x,z) z -
            \Phi(x,z) + \Phi(x,0) = (1-S(x)) z^2/2$.

        \item[H3:] Since $a$ is a linear function of $p$ it is easy to see that 
            \begin{flalign}
                \notag \left( a(x,z,p^{(1)}) - a(x,z,p^{(2)}) \right)\left( p^{(1)} -
                p^{(2)} \right) &= \mathcal{D}(x,z) \left( p^{(1)} - p^{(2)} \right)^2 \\
                &\ge C_\epsilon \left( p^{(1)} - p^{(2)} \right)^2
            \end{flalign}
            where $\mathcal{D}(x,z) \ge C_\epsilon$ for all $x,z$ (this
            $\epsilon$-dependence reflects the choice of the saturation function, $S(x)$).
            Given the functional form of $\mathcal{D}$ it is obvious that $a$ is
            continuous and bounded on $z \in (0,1)$, $p \in \mathbb{R}$, and is
            measurable in $x$. Hence $a$ satisfies the ellipticity condition. Simply
            stated, the ellipticity of the diffusion coefficient means that the operator
            in question is bounded away from zero and is therefore invertible.
            
        \item[H4:] Let $z \in [\epsilon,1-\epsilon]$ and $p \in \mathbb{R}$. From the
            definition of $a$ and $f$,
            \begin{flalign}
                \notag |a(x,z,p)| + |f(z)| = |a(x,z,p)| &= |\mathcal{D}(x,z)p| \\
                \notag &\le c_{\mathcal{D},\epsilon} |p| \\
                &\le c_\epsilon \left( 1 + \sqrt{B(z)} + |p| \right)
            \end{flalign}
            for all $z$, where $c_{\mathcal{D},\epsilon}$ is the upper bound on
            $\mathcal{D}(x,z)$ over $z$. Therefore
            $a$ (and $f$) satisfy the growth condition.

        \item[H5:] The left Dirichlet boundary condition is fixed in time, $\rh(0,t) =
            \rh_{D,0}$.  It is assumed that $\rh_{D,0} \in L^2(0,T;H^1(\Omega))$ and
            $L^{\infty}( (0,T) \times \Omega)$. The right Dirichlet boundary condition is
            allowed to vary in time.  Assumption \#4 in the statement of this theorem
            guarantees that hypothesis H5 is satisfied for this boundary condition.

        \item[H6:] Since $b(x,z) = (1-S(x))z$ it is clear that $b$ is surjective so long
            as $S(x) \ne 1$ and that $b^0 = \rh_0 / (1-S(x))$. That is, there exists a
            function $\rh_0/(1-S(x))$ such that $b^0 = b(x,\rh^0)$.
    \end{itemize}
    Given the final assumption in the statement of this theorem we have, in particular,
    $\rh_D \in L^1(0,T;L^{\infty}(\Omega))$ since $L^1(\Omega) \subset L^2(\Omega) \subset
    L^\infty(\Omega)$ for sets $\Omega$ of finite measure \cite{Folland1999}.  Therefore, from
    Theorem \ref{thm:alt} there exists a weak solution, $\rh$, in the affine space $\rh_D
    + L^2(0,T;V)$ where $V = \{ v \in H^1(\Omega): v = 0 \text{ on } \{0,1\} \}$.
\end{proof}

The uniqueness of the weak solution to \eqref{eqn:diffusion_exist} - \eqref{eqn:gas_ibc},
unfortunately, doesn't fit Theorem \ref{thm:alt_unique} because the diffusion operator
cannot be decomposed in the manner required. This does not mean that the weak solution is
not unique, it simply means that this is not the tool to prove uniqueness.  This small
problem is left for future research.

\subsection{Limits of the Alt and Luckhaus Theorem}
The theorem of Alt and Luckhaus does not apply to the heat transport equation since there
are advection-type terms present in that equation that can not satisfy the assumed form of
Theorem \ref{thm:alt}. The next logical direction is to see if this tool can be used to
prove existence of the coupled saturation-humidity system at constant temperature.  The
forcing term on the right-hand side of each equation is now non-zero.  The equations are
\begin{subequations}
    \begin{flalign}
        \pd{S}{t} - \pd{ }{x} \left( (-D(S) + C_S^l) \pd{S}{x} + C_\rh^l \pd{\rh}{x} -
        K(S) \rho^l g \right) &= \ehat{l}{g_v}(\rh,S) \\
        (1-S) \pd{\rh}{t} - \rh \pd{S}{t} - \pd{ }{x} \left( \mathcal{D}(\rh,S)
        \pd{\rh}{x} \right) &= -\ehat{l}{g_v}(\rh,S).
    \end{flalign}
    \label{eqn:sat_diffusion_system}
\end{subequations}
If we were to define $b(z) : \mathbb{R}^2 \to \mathbb{R}^2$ as
\[ b(z) = \begin{pmatrix} 1 & 0 \\ -z_2 & (1-z_1) \end{pmatrix} \begin{pmatrix} z_1 \\ z_2
    \end{pmatrix} \]
one can show that there does not exist a function $\Phi: \mathbb{R}^2 \to \mathbb{R}$ such
that $b = \grad \Phi$.  For this reason we restate the equations with a consolidated form
of the time derivatives in the second equation
\begin{subequations}
    \begin{flalign}
        \pd{S}{t} - \pd{ }{x} \left( (-D(S)+C_S^l) \pd{S}{x} + C_\rh^l \pd{\rh}{x} - K(S)
        \rho^l g \right) &= \ehat{l}{g_v}(\rh,S) \\
        \pd{u}{t} - \pd{ }{x} \left( \mathcal{D}(\rh,S) \pd{\rh}{x} \right) &=
        -\ehat{l}{g_v}(\rh,S) \\
        u &= (1-S) \rh
    \end{flalign}
\end{subequations}
Solving for the relative humidity in equation (c) and substituting into equations (a) and
(b) gives 
\begin{subequations}
    \begin{flalign}
        \pd{S}{t} - \pd{ }{x} \left( (-D(S)+C_S^l) \pd{S}{x} + C_\rh \pd{ }{x}\left(
        \frac{u}{1-S}
        \right) - K(S) \rho^l g \right) &=
        \ehat{l}{g_v}\left( \frac{u}{1-S},S \right) \\
        \pd{u}{t} - \pd{ }{x} \left( \mathcal{D}\left( \frac{u}{1-S} ,S\right) \pd{}{x}
        \left( \frac{u}{1-S} \right) \right) &=
        -\ehat{l}{g_v}\left( \frac{u}{1-S} , S \right). 
    \end{flalign}
    \label{eqn:humidity_saturation_strong_coupling}
\end{subequations}

It can be seen from this form that the coupling in the time derivatives has been moved to
a stronger coupling with the diffusion terms.  It can be shown that the associated
$a(\cd,\cd)$ function is not elliptic in the sense required in Theorem \ref{thm:alt}.
Therefore we have determined that the Alt and Luckhaus theorems don't apply to the coupled
system in this form. 

Equations
\eqref{eqn:humidity_saturation_strong_coupling} poses the system in a form of strong
coupling known as a {\it triangular system}.  A triangular parabolic system has two
equations; one parabolic
equation with a contribution to diffusion from both dependent variables and the other with
a contribution to diffusion from only one variable \cite{Le2005}. Future research into the
existence and uniqueness results will likely start here as the theory of triangular
systems is fairly well developed and may provide a springboard to results for this
problem.

\section{Heat Transport Equation}\label{sec:existence_uniqueness_heat}
In this section we consider the question of existence and uniqueness for the heat transport
equation.  This is done under the assumptions that the relative humidity and saturation
profiles are fixed in space and time.

If the saturation and the relative humidity are considered fixed and constant then the thermal
transport equation \eqref{eqn:system_final_energy} collapses to
\begin{flalign}
    \rho c_p \pd{T}{t} - \diver \left( K \grad T \right) + \rho h + \chi_2 \grad T \cd
    \grad T = 0,
    \label{eqn:thermal_fixed_S_rh}
\end{flalign}
where $\chi_2$ is given as
\[ \chi_2 = \rho^l c_p^l C_T^l K(S) + \rho^g c_p^g N(S,T). \]
In the absence of heat sources and if $\chi_2$ is neglected we arrive at the standard heat
equation; the existence and uniqueness results of which are well known (see any standard
text on PDEs). It is likely that $C_T^l \approx 0$ since, in Saito \cite{Saito2006}, the
authors indicated that the thermal liquid flux is negligible as compared to isothermal
liquid flux. The entropy term appearing in $N$, on the other hand, is likely
non-negligible and therefore must be considered.  In the case where $S$ and $\rh$ are
fixed but non-constant, the terms in equation \eqref{eqn:system_final_energy} associated
with $\grad S$ and $\grad \rh$  are combined as a source term which depends on $x, t,$ and
$T$. Therefore, we only need to consider thermal equations in the form of
\eqref{eqn:thermal_fixed_S_rh}. If $h=0$ and Dirichlet boundary conditions are considered then this is the exact form of the equation
considered by Rincon et al.\,in \cite{Rincon2005}:
\begin{subequations}
    \begin{flalign}
        &\pd{u}{t} - \diver \left( a(u) \grad u \right) + b(u) |\grad u|^2 = 0 \quad \text{in}
        \quad \Omega \times (0,T) \\
        &u = 0 \quad \text{on} \quad \partial \Omega \times (0,T) \\
        &u(x,0) = u_0(x) \quad \text{in} \quad \Omega.
    \end{flalign}
    \label{eqn:rincon_equation}
\end{subequations}
where we have defined $u$ such that $u \gets T-T_{ref}$ with reference temperature
$T_{ref}$.  Taking $h=0$ means that we must assume that both $S$ and $\rh$
are constant in space and fixed in time. This is not entirely physical, but it is a step
toward a general existence uniqueness theory for the present equations. In this problem,
$a(u)$ is the diffusion coefficient, $a(u) = K(u)$, defined either by the weighted sum of
the thermal conductivities (equation \eqref{eqn:thermal_weighted_sum}) or by the
Johansen thermal conductivity function (equation \eqref{eqn:Johansen_thermal}). The $b$
function is defined as $\chi_2$ as above.

For the Rincon existence and uniqueness theorem we consider the following hypotheses:
\begin{description}
    \item[H1:] $a(u)$ and $b(u)$ belong to $C^1(\mathbb{R})$ and there are positive
        constants $a_0, a_1$ such that $a_0 \le a(u) \le a_1$ and $b(u) u \ge 0$.
    \item[H2:] There is a positive constant $M>0$ such that 
        \[ \text{max}_{s} \left\{ \left| \frac{d a}{du}(s) \right|, \left|
            \frac{db}{du} (s) \right|\right\} \le M. \]
    \item[H3:] $u_0 \in H_0^1(\Omega) \cap H^2(\Omega)$ such that $\| \Delta u_0
        \|_{L^2(\Omega)} < \epsilon$ for some constant $\epsilon > 0$.
\end{description}

\begin{theorem}[Rincon et al.\,\cite{Rincon2005}, Theorem
    2.1]\label{thm:rincon_nonlinear_heat}
    Under hypotheses H1 - H3 there exists a positive constant $\epsilon_0$ such that if $0
    < \epsilon < \epsilon_0$ then the problem \eqref{eqn:rincon_equation} admits a unique
    solution satisfying
    \begin{enumerate}
        \item[i.] $u \in L^2(0,T;H_0^2(\Omega) \cap H^2(\Omega))$ and $\partial_t u \in
            L^2(0,T;H_0^1(\Omega))$
        \item[ii.] $\pd{u}{t} - \diver (a(u) \grad u) + b(u) |\grad u|^2 = 0$ in
            $L^2(\Omega \times (0,T))$
        \item[iii.] $u(0) = u_0$
    \end{enumerate}
\end{theorem}

\begin{theorem}
    There exists a unique solution to equation \eqref{eqn:thermal_fixed_S_rh} under the
    following conditions:
    \begin{enumerate}
        \item $u(0,t) = u(1,t)$
        \item $h=0$
        \item $u(x,0) = u_0 \in H^1_0(\Omega) \cap H^2(\Omega)$ and there exists $\epsilon
            > 0$ such that $\| \Delta u_0 \|_{L^2(\Omega)} < \epsilon$
    \end{enumerate}
    Here we are using $u = T - T_{ref}$ for a scaled temperature (so that capital $T$ will
    represent a finite time as in Theorem \ref{thm:rincon_nonlinear_heat}).
\end{theorem}

\begin{proof}
    We will proceed by verifying the hypotheses of Theorem \ref{thm:rincon_nonlinear_heat}
    \begin{itemize}
        \item[H1:] From the derivation of the heat transport equation, $a(u)$ is a
            weighted sum of thermal conductivities from the individual phases.  The
            particular form of the weighted sum comes from either equation
            \eqref{eqn:thermal_weighted_sum} or \eqref{eqn:Johansen_thermal}, but in this
            scenario, the saturation is presumed to be constant. Therefore, in this case
            $a(u)$ is constant and is trivially $C^1(\Omega)$.  The functional form of
            $b$ depends on the functional form of the entropy and the saturation.  So long
            as the saturation is fixed away from zero then $b$ is in $C^1(\Omega)$.
            Furthermore, $b(u)$ is positive so $b(u)u$ is also positive for all $u$.
        \item[H2:] Since $a$ is a constant, $da / du = 0$ for all $u$.  The functional
            form of $b$, on the other hand, is not constant so this hypothesis simply
            states that $\chi_2$ needs to have a bounded first derivative.  Taking the entropy term from the Darcy flux as
            \[ \eta^g = c_p^g \ln\left( \frac{T}{T_{ref}} \right) + \eta_{ref}, \]
            (from the definition of the specific heat) and defining $\chi_2$ accordingly
            we see that $b$ will have a bounded first derivative so long as $u+T_{ref} =
            T$ remains bounded away from $0$.  This is, of course, always true since $T$
            is the absolute temperature.
        \item[H3:] The third assumption of the theorem satisfies this hypothesis.
    \end{itemize}
    Therefore, there exists a unique solution to \eqref{eqn:thermal_fixed_S_rh} with no
    sources and equal Dirichlet boundary conditions.
\end{proof}

\section{Conclusion}
At this point we turn our attention to the analysis and comparisons of numerical solutions
of the equations (both individually and coupled). The existence and uniqueness theory
presented here is by no means complete.  In particular, we are missing a uniqueness result
for the vapor diffusion equation, the theorem used for the heat transport equation is very
limiting with respect to boundary conditions and sources, and we have not mentioned
results for any of the coupled systems.  Many numerical solvers will iterate coupled
systems across the equations, so an existence and uniqueness theory for each equation is
essential to give hope that the
numerical method converges to {\it the} solution.  These results are left for future work
as the ultimate crux of this thesis is to justify the modeling technique against physical
experimentation and classical models.

\newpage
\chapter{Numerical Analysis and Sensitivity Studies}\label{ch:Transport_Solution}
In this chapter we build and analyze the solution(s) to the heat and moisture transport
model summarized in equations \eqref{eqn:system_final_saturation} -
\eqref{eqn:system_final_energy} with constitutive equations
summarized in equations \eqref{eqn:system_final_conductivity} -
\eqref{eqn:system_final_gas_ent}.  To simplify matters we
henceforth assume a 1-dimensional geometry modeling a column experiment common to soil
science.  Figure \ref{fig:column} gives a cartoon drawing of a typical column experiment
with a definition of the geometric variable $x$.  The grains represent a packed porous
medium.  Flow, diffusion, and heat transport are assumed to travel solely in the $x$
direction (up or down).
\linespread{1.0}
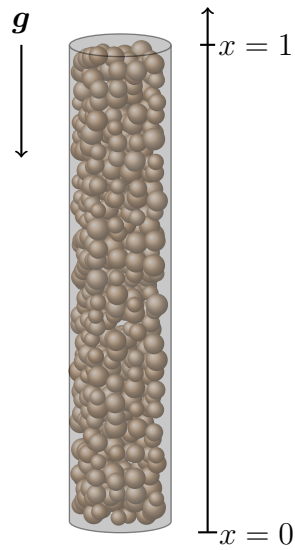
\begin{figure}[H]
    \begin{center}
        \begin{tikzpicture}
            \draw[thick,|->] (2,0) node[anchor=west]{$x=0$} -- (2,7);
            \draw (2,6.5) node[anchor=west]{$x=1$};
            \draw[thick] (1.9,6.5) -- (2.1,6.5);
            \includegraphics[width=0.1\columnwidth]{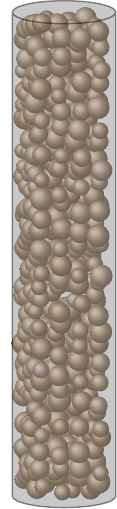}
            \draw[thick,->] (-2,6.5) node[anchor=south]{$\foten{g}$} -- (-2,5);
        \end{tikzpicture}
    \end{center}
    \caption{Cartoon of a 1-dimensional packed column experimental apparatus.}
    \label{fig:column}
\end{figure}
\renewcommand{\baselinestretch}{\normalspace}

In this chapter we are interested in the behavior of equations \eqref{eqn:system_final} in
several situations related to the apparatus depicted in Figure \ref{fig:column}; some physical and some merely hypothetical.
\begin{enumerate}
    \item In a {\bf drainage experiment} the column is saturated with the wetting phase and then
        allowed to drain under the influence of gravity. \\
        Possible simplifying assumptions include: constant relative humidity and temperature.
    \item In an {\bf imbibition experiment} the column starts partially saturated (or dry) and
        the wetting phase is introduced either at $x=1$ or $x=0$.  If the wetting phase is
        introduced at $x=1$ then the primary force driving the liquid flow will be
        gravity, and if it is introduced at $x=0$ then the pressure head from the
        reservoir drives the flow. \\
        Possible simplifying assumptions include: constant relative humidity and temperature.
    \item In {\bf evaporation studies}, a gradient in relative humidity is introduced
        between $x=0$ and $x=1$ and relative humidity is tracked throughout the column. \\
        Possible simplifying assumptions include: constant temperature and/or fixed saturation
        profile.
    \item In {\bf Coupled saturation and evaporation experiments} the saturation and
        relative humidity are tracked throughout the column under boundary conditions that
        drive both.  \\
        Possible simplifying assumptions include: constant (or at least fixed) temperature.
    \item In {\bf fully coupled} systems we consider a heat source (typically located at
        $x=1$) and boundary conditions that drive all three equations.
\end{enumerate}

In Chapter \ref{ch:Existence_Uniqueness} we discussed the questions of existence and
uniqueness of solutions to equations \eqref{eqn:system_final}. We now turn to numerical
analysis.  In Sections \ref{sec:numerical_saturation},
\ref{sec:numerical_vapor_diffusion}, and \ref{sec:numerical_coupled_sat_vap} we discuss
various numerical solutions associated with the situations outlined above. For example, in
Section \ref{sec:numerical_saturation}, we examine numerical solutions associated with
drainage and imbibition experiments (types 1 and 2 above). In Section
\ref{sec:FullyCoupledSolutions} we compare with a 1-dimensional column experiment outlined
in Smits et al. \cite{Smits2011}. No two- or three-dimensional experiments are performed
in this work.

\section{Saturation Equation}\label{sec:numerical_saturation}
In this subsection we consider the saturation equation \eqref{eqn:system_final_saturation}
with fixed and constant relative humidity and temperature and no mass transfer.  That is,
we consider
\begin{flalign}
    \pd{S}{t} = \pd{ }{x} \left( \porosity_S^{-1} K(S) \left( \left[ -p_c'(S) + C_S^l
        \right] \pd{S}{x} + \tau \porosity_S \pddm{S}{x}{t} - \rho^l g \right) \right).
\end{flalign}
These assumptions are natural in an oil-water system or simply unsaturated systems where
the relative humidity is considered fixed experimentally.
We would like to determine qualitative behavior of solutions to this equation under
certain boundary conditions, experimental setups, van Genuchten parameters, and values (or
functional forms) of $\tau$ and $C_S^l$.  As a first step toward this analysis let us
consider dimensionless spatial and temporal scalings. Notice that the spatial dimension
can already be viewed as dimensionless as seen in Figure \ref{fig:column}.  A
characteristic time for this equation is $t_c = x_c / k_c = 1/k_c$ where $k_c = (\rho^l g
\kappa_s)/\mu_l$ is the hydraulic conductivity. Multiplying by $t_c$ and henceforth
understanding $t$ and $x$ as dimensionless we get
\begin{flalign}
    \notag \pd{S}{t} &= \pd{ }{x} \left( t_c \porosity_S^{-1} K(S) \left[ -p_c'(S) + C_S^l
        \right] \pd{S}{x} \right) \\
    & \quad + \pd{ }{x} \left( \tau K(S) \pddm{S}{x}{t} \right) - \pd{ }{x} \left( t_c
    K(S) \porosity_S^{-1} \rho^l g \right).
\end{flalign}

In the case where $\tau = 0$, the qualitative behavior can be analyzed via the P\'eclet
number; the ratio of the advective to diffusive coefficients
\begin{flalign}
    Pe = \frac{\rho^l g}{-p_c'(S) + C_S^l} = \frac{\rho^l g}{\left( \frac{\rho^l g
    (1-m)}{\alpha m} \right) S^{-(1+1/m)} \left( S^{-1/m} -1 \right)^{-m} + C_S^l }.
\end{flalign}
Since the diffusive coefficient depends on the dependent variable it is immediately clear
that the P\'eclet number will change in time and space (in the study of linear PDEs the
P\'eclet number is a fixed ratio that does not depend on the dependent variable). If $Pe <1$ then the problem is
{\it diffusion dominated} whereas if $Pe >1$ then the problem is {\it advection dominated}. In a
diffusion dominated problem we expect a smooth solution that spreads spatially in time, and in an
advection dominated problem we expect more advection (transport) than smoothing.  In
quasilinear advection diffusion equations (see a standard PDE text discussing the method of
characteristics (e.g. \cite{Zach,Evans2010}), if the diffusion term is not significantly
weighted then the advective term may yield shock-type solutions.  For example, if the
material parameters for a particular experiment are located in the top right of Figure
\ref{fig:Peclet_compare} then the diffusive term is weighted very small as compared to the
advection and a shock is more likely to develop.  That being said, a shock-type solution is
non-physical so it is not expected in these experiments. This gives an indication that if
a shock does occur then the parameters must be non-physical or the numerical method is not
accurately capturing the diffusion.

From the definition of the P\'eclet number it is clear that the action of $C_S^l$ is to
increase the damping of the diffusion term.  Given the form of the P\'eclet number, it
stands to reason that damping similar to that of $C_S^l$ can be achieved by choosing
different van Genuchten parameters.  For this reason we presume for the remainder of this
work that the effects of $C_S^l$ are inseparably tied up with the effects of the $p_c-S$
relationship. Hence we can assume that $C_S^l \approx 0$.  Recall that $C_S^l$ is defined
(see equation \eqref{eqn:CSl_defn}) as 
\[ C_S^l = \ppi{l}{l} - \ppi{l}{g} = 2 \pd{\psi^l}{S} \] 
and is interpreted as a wetting potential.

With the assumption that $C_S^l \approx 0$ (or is at least inseparable experimentally from
$p_c'(S)$), the P\'eclet number becomes
\[ Pe = \left( \frac{\al m}{1-m} \right) S^{(1+1/m)} (S^{-1/m} - 1)^m. \]
The van Genuchten parameters, $\al$ and $m$, are independent in this form of the P\'eclet
number.  Furthermore, the van Genuchten capillary pressure - saturation function is only
one of several choices for this relationship. Other common forms are the Brooks-Corey and
Fayer-Simmons models; each of which will have their own associated P\'eclet number. Figure
\ref{fig:Peclet_compare} shows the nature of the P\'eclet number as a function of these
parameters as well as the saturation.
\linespread{1.0}
\begin{figure}[H]
        \centering
        \subfigure[Log of P\'eclet number for $S = 0.1$]{
                \includegraphics[trim = 1cm 0cm 0cm 0cm, clip=true,
                    width=0.4\textwidth]{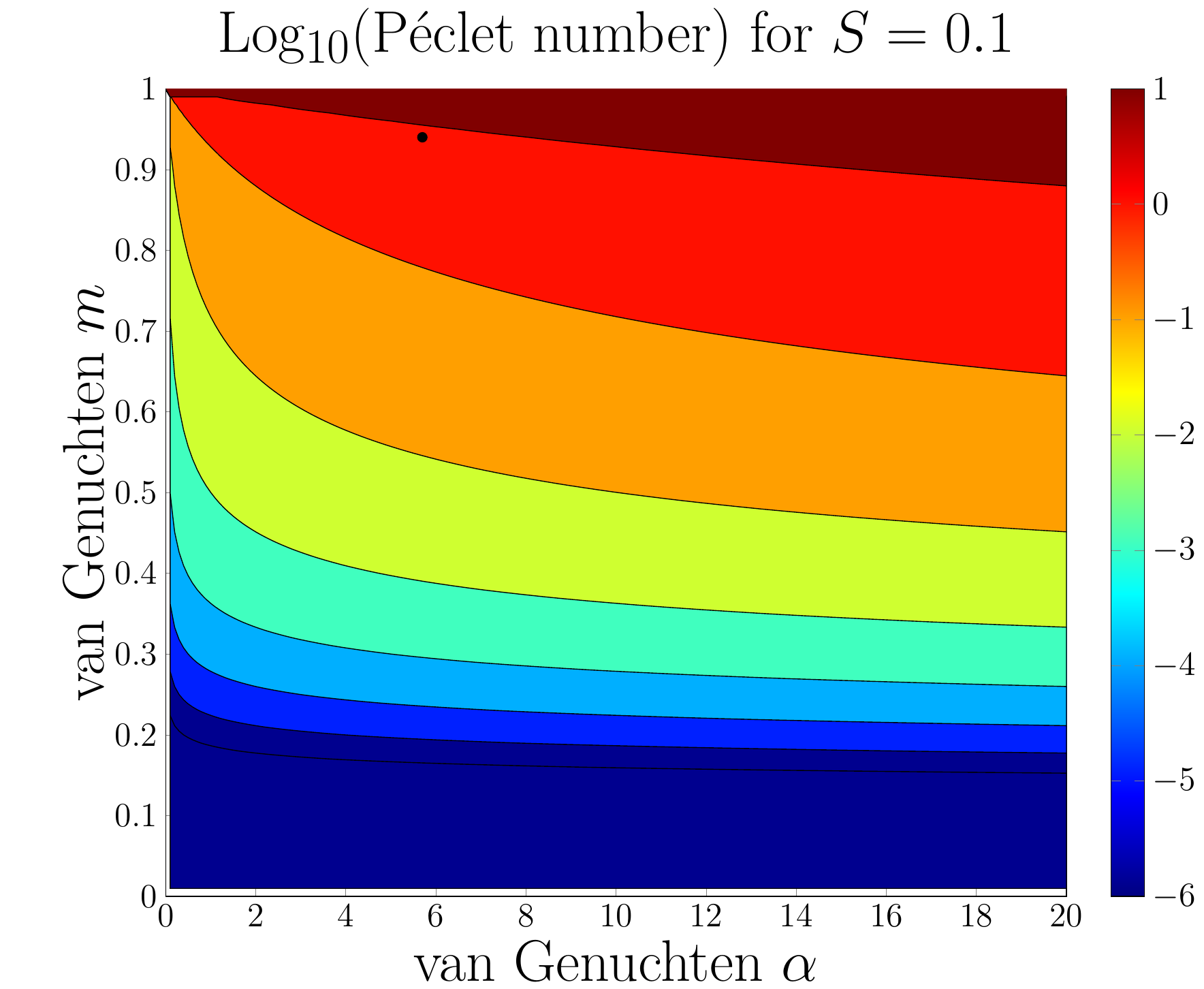}
                \label{fig:Peclet_S01}
            }
        \subfigure[Log of P\'eclet number for $S = 0.4$]{
                \includegraphics[trim = 1cm 0cm 0cm 0cm, clip=true,
                    width=0.4\textwidth]{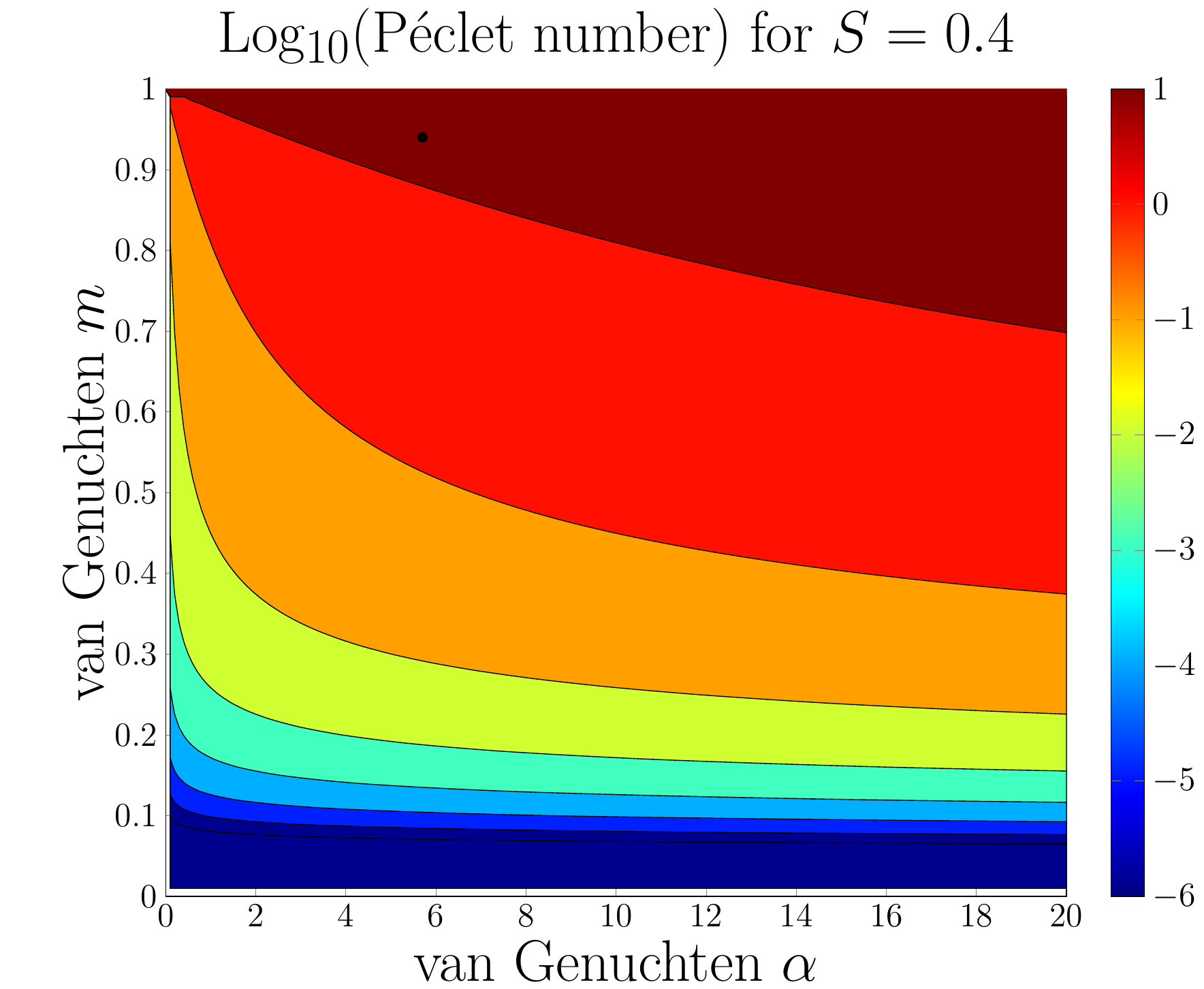}
            \label{fig:Peclet_S04}
        }
        \subfigure[Log of P\'eclet number for $S = 0.7$]{
                \includegraphics[trim = 1cm 0cm 0cm 0cm, clip=true,
                    width=0.4\textwidth]{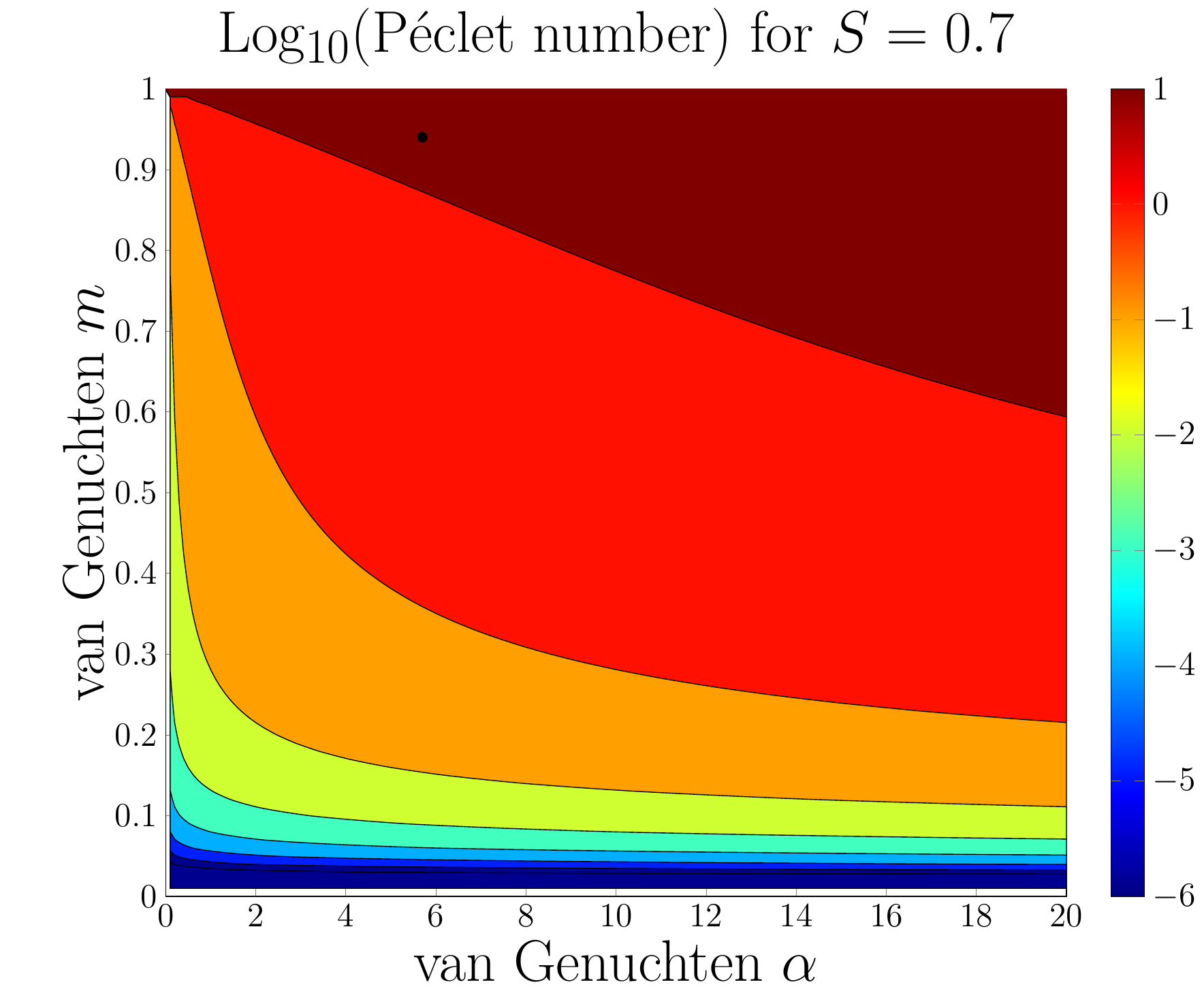}
            \label{fig:Peclet_S07}
        }
        \subfigure[Log of P\'eclet number for $S = 0.99$]{
                \includegraphics[trim = 0.5cm 0cm 0cm 0cm, clip=true,
                    width=0.4\textwidth]{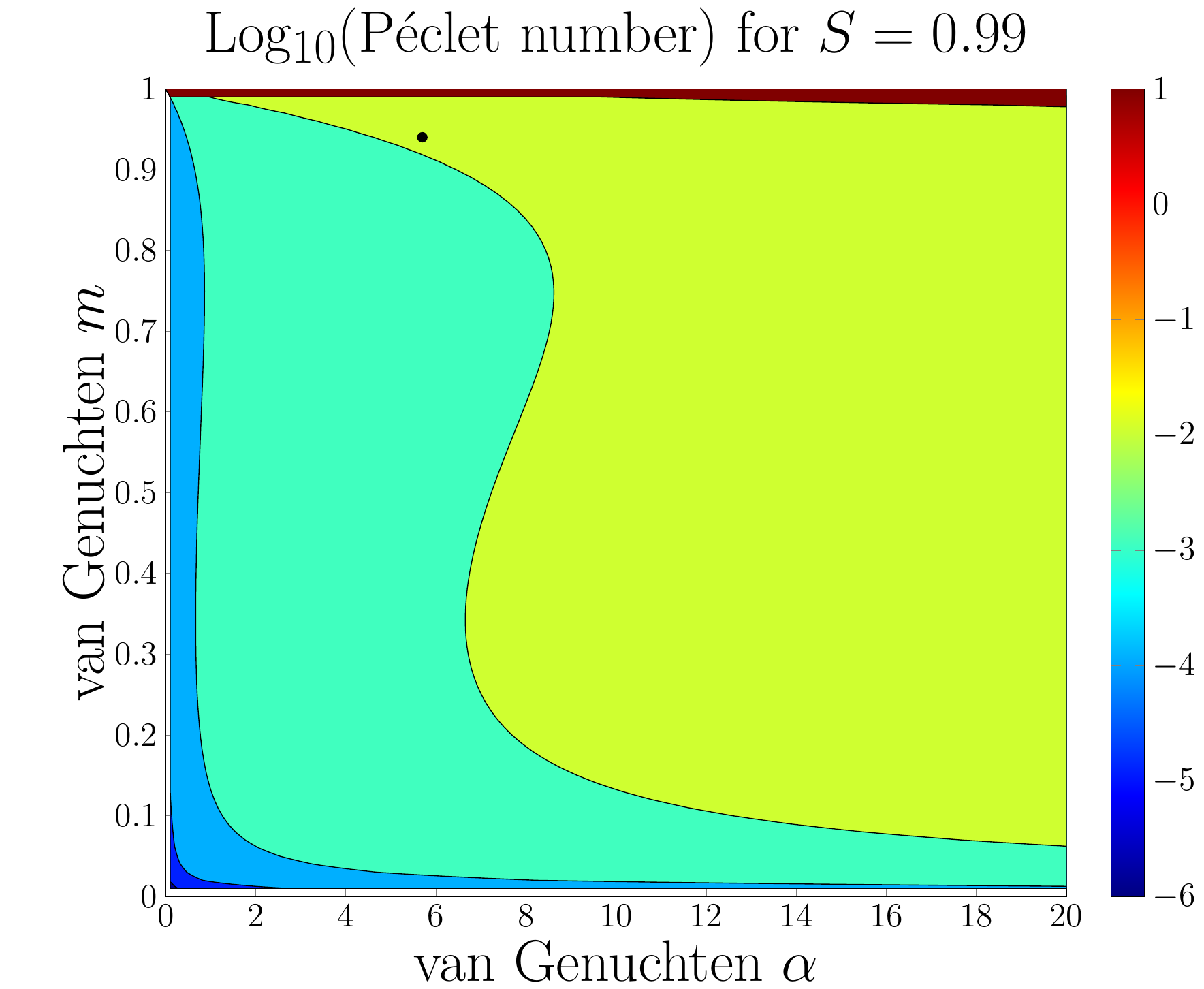}
                \label{fig:Peclet_S099}
        }
        \caption{Log of P\'eclet numbers for various values of
        saturation. The point at $(\al = 5.7, m=0.94)$ indicates the values used in Smits et
        al.\,\cite{Smits2011}. Warmer colors are associated with higher P\'eclet number
    and therefore associated with an advective solution.}\label{fig:Peclet_compare}
\end{figure}
\renewcommand{\baselinestretch}{\normalspace}

In Figure \ref{fig:Peclet_compare} it appears that the solutions to the saturation
equation (with $\tau = 0$) become more diffusion dominated for smaller values of van Genuchten
parameters.  As $S \to 1$ the diffusion term gains more traction and hence dampens the
advection.  Of course, one cannot simply choose a set of van Genuchten parameters.
Instead, the parameters are dictated by the material properties of the soil.  In the study
by Smits et al.\,\cite{Smits2011}, $\al = 5.7$ and $m\approx0.94$ (indicated by the point
in Figure \ref{fig:Peclet_compare}). In this instance, we expect an advection dominated
solution with very little diffusive damping.  This poses a danger numerically as it is
close to the regime where shock-type solutions could arise.

The third-order term can be analyzed in a similar manner. To the author's knowledge there
is no {\it name} for the ratio of the coefficients of the third-order term to the
diffusive term  
\begin{flalign}
    \notag H &:= \frac{\tau \porosity_S}{-t_c p_c'(S)} = \frac{\tau \porosity_S \rho^l g
    \kappa_s}{-p_c'(S) \mu_l} \\
    %
    %
    \notag &= \left( \frac{\tau \porosity_S \kappa_S}{\mu_l} \right) \left( \frac{\al m
    }{1-m} \right) S^{(1+1/m)} (S^{-1/m} - 1)^m \\
    &= \left( \frac{\tau \porosity_S \kappa_S}{\mu_l} \right) Pe := H_0 Pe
    \label{eqn:hassnizadeh_number}
\end{flalign}
Thus the plots in Figure \ref{fig:Peclet_compare} are simply scaled versions of $H$.  The
question that remains is what effect the third-order term has on the solution.  To answer
this questions we examine a few solution plots.  These solutions are found using
\texttt{Mathematica}'s \texttt{NDSolve} function.  This build-in command is a general
differential equation solver handling ordinary and partial differential equations, systems
of equations, vector equations, and stiff systems.  For partial differential equations it
uses a finite difference approach to discretize the spatial variable and a version of
Gear's method for implicit stiff time stepping following a method-of-lines approach
\cite{WolframNDSolve2012}. 

Figure \ref{fig:S_Drain_compare} shows a drainage experiment for various values of $H_0 =
(\tau \porosity_S \kappa_S)/\mu_l$. The initial condition is given in black.  A Dirichlet
boundary condition ($S = S_0$) is given at $x=1$ and a homogeneous Neumann condition
($\partial S / \partial x |_{x=0} = 0$) is imposed
at $x=0$. Gravity points in the negative $x$ direction, so that the liquid present
in the column is expected to drain over time. Figure \ref{fig:S_Drain_t025} shows that at earlier
times a larger value of $H_0$ gives a steeper front with plausibly physical saturation
profiles. Non-physical, non-monotonic, results are observed for
$H_0 = 10^{-2}$ as seen near $x = 0.8$ in Figures \ref{fig:S_Drain_t050} -
\ref{fig:S_Drain_t100}. For values of $H_0$ smaller than $10^{-2}$ we continue to observe
a sharper front as compared to solutions for $\tau = 0$ (shown in blue).
\linespread{1.0}
\begin{figure}[H]
        \centering
        \subfigure[Saturation profiles at $t=t_1$]{
                \includegraphics[width=0.45\textwidth]{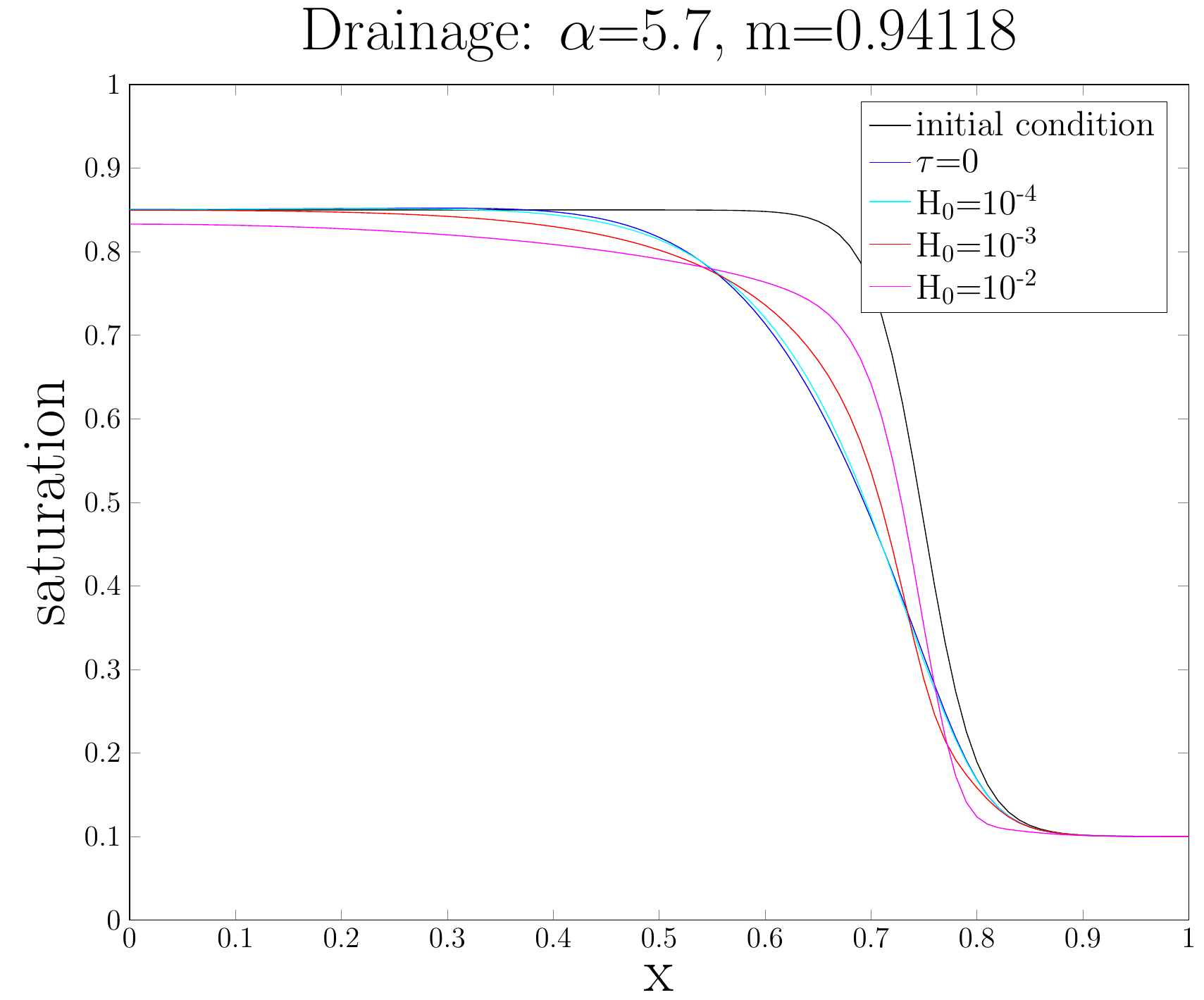}
                    \label{fig:S_Drain_t025}
            }
        \subfigure[Saturation profiles at $t=t_2$]{
                \includegraphics[width=0.45\textwidth]{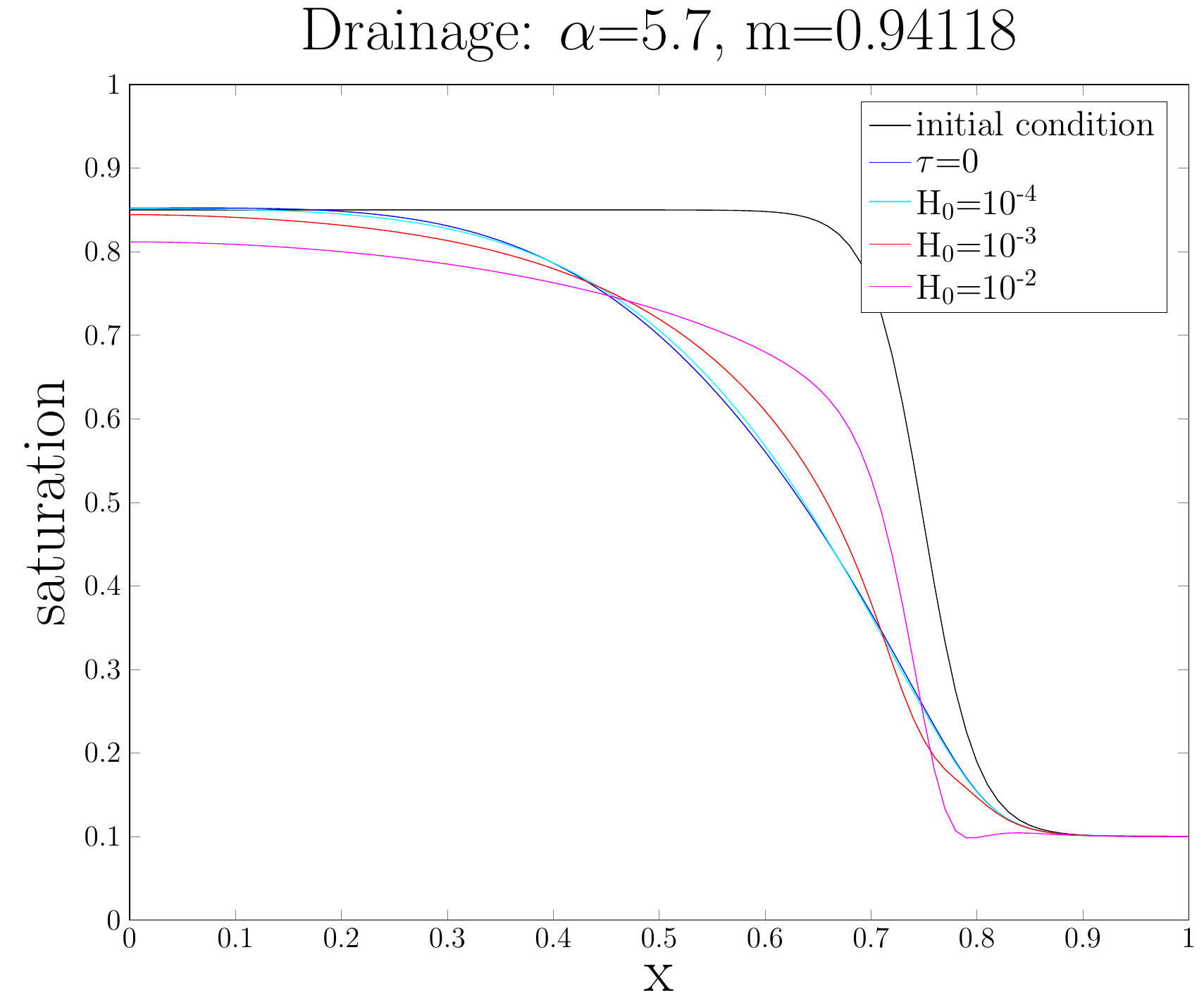}
                    \label{fig:S_Drain_t050}
        }
        \subfigure[Saturation profiles at $t=t_3$]{
                \includegraphics[width=0.45\textwidth]{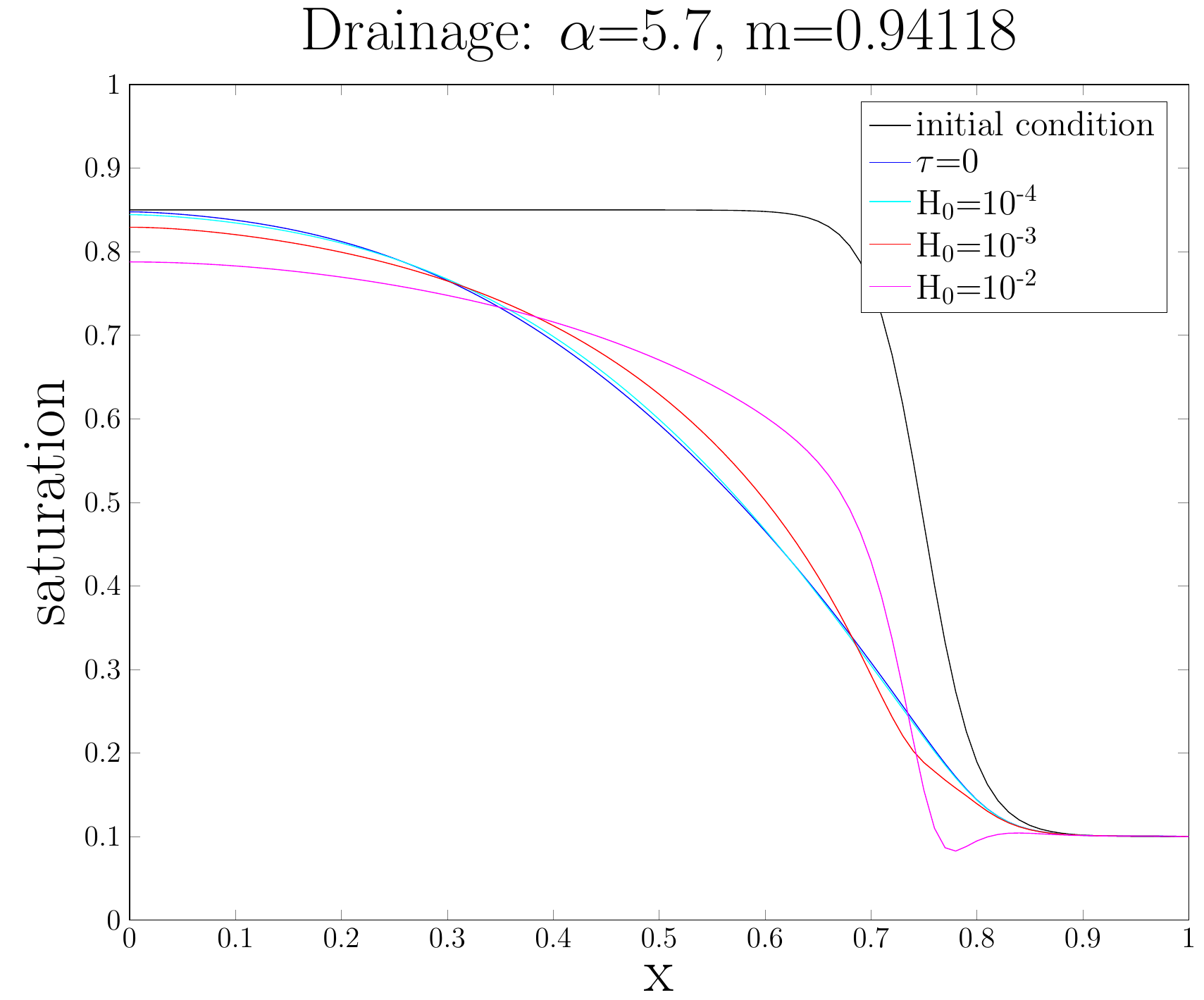}
                    \label{fig:S_Drain_t075}
        }
        \subfigure[Saturation profiles at $t=t_4$]{
                \includegraphics[width=0.45\textwidth]{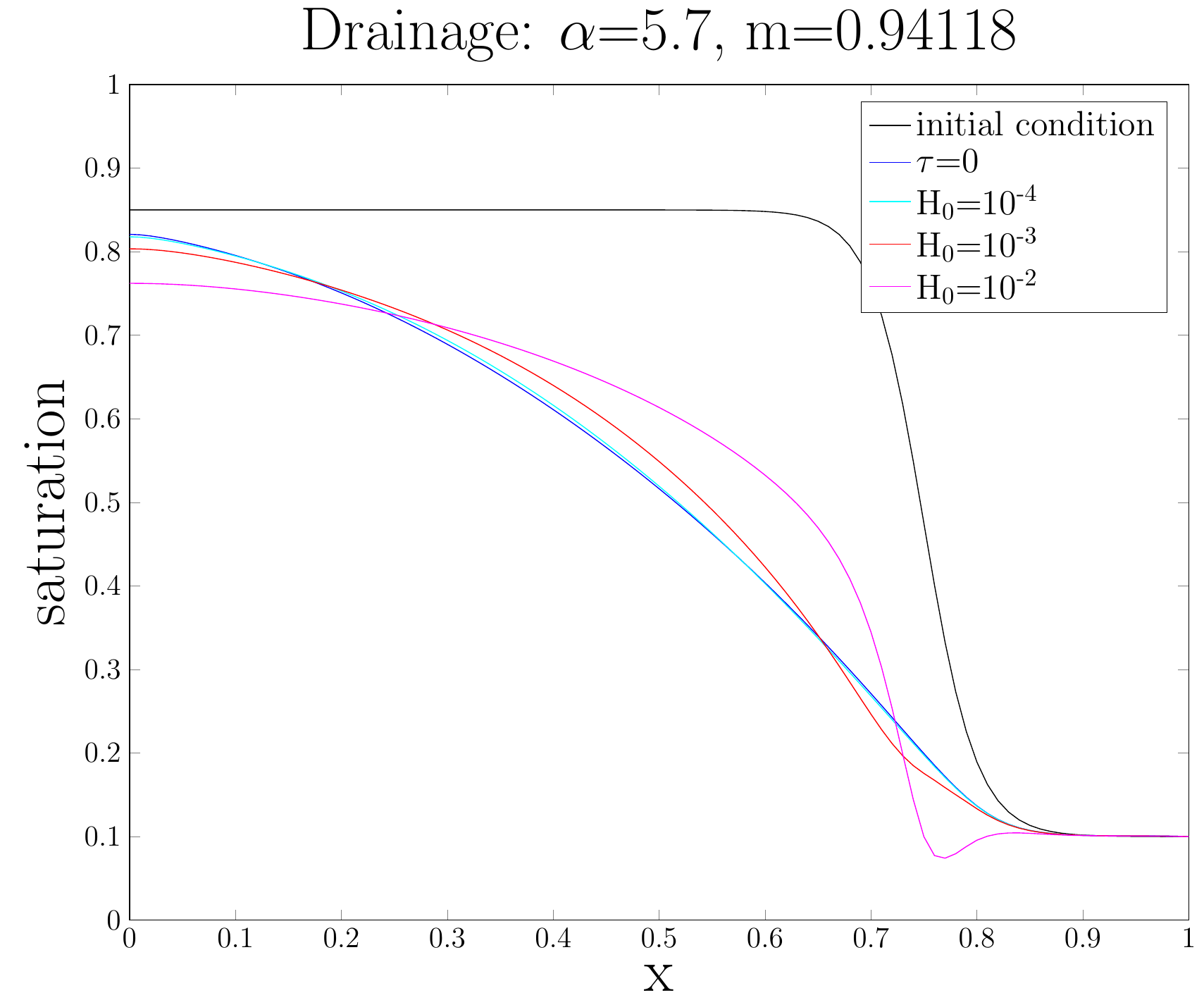}
                    \label{fig:S_Drain_t100}
        }
        \caption{Saturation profiles at various times in a drainage experiment with $\al = 5.7, n =
        17$.}\label{fig:S_Drain_compare}
\end{figure}
\renewcommand{\baselinestretch}{\normalspace}

To be sure that the non-physical results observed for $H_0 = 10^{-2}$ are not due to
numerical noise we complete a numerical convergence test on this particular set of initial
boundary conditions. A typical convergence test of a numerical method would compare
against a known analytic solution, but in this case there is no known analytic solution. For this
reason, we allow \texttt{Mathematica} to solve the problem using the default spatial and
temporal tolerances and then compare solutions with fixed grids consisting of fewer mesh points
to this solution.  \texttt{Mathematica}'s differential equation solver uses a finite
difference approach for spatial discretization. The defaults for this scheme are
fourth-order central differences where spatial points are on a static grid and the number
of grid points is chosen automatically based on the initial condition. For the tests shown
in Figure \ref{fig:S_Drain_compare} there were 103 grid points selected automatically.   To
check this solution, we examine the relative $L^2(0,1)$ error as a function of time,
\[ E_{(N)}(t) = \frac{\| S_{(k)} - S_{(103)} \|_{L^2}}{\|S_{(103)} \|_{L^2}}, \]
where $N$ is the number of spatial points. In Figure \ref{fig:Drain_Conv_Test} we measure
$E_{(N)}(t)$ for $N$ ranging from 20 spatial points to 100 spatial points. Notice that for
any fixed {\it small} time the relative error decreases with increasing grid size; hence
indicating numerical convergence at that fixed time.  For dimensionless time greater than
approximately $0.05$, on the other hand, the error decreases at a slower rate and there is
evidence that the numerical method may not be converging. In all cases the relative error
grows in time until approximately $0.25$. While the {\it bump} that appears in Figure
\ref{fig:S_Drain_compare} is certainly non-physical, Figure \ref{fig:Drain_Conv_Test}
seems to indicate that the numerical method is failing in this case and the results may
not be trust-worthy for this set of parameters and initial boundary conditions.
\linespread{1.0}
\begin{figure}[H]
    \begin{center}
        \includegraphics[width=0.70\columnwidth]{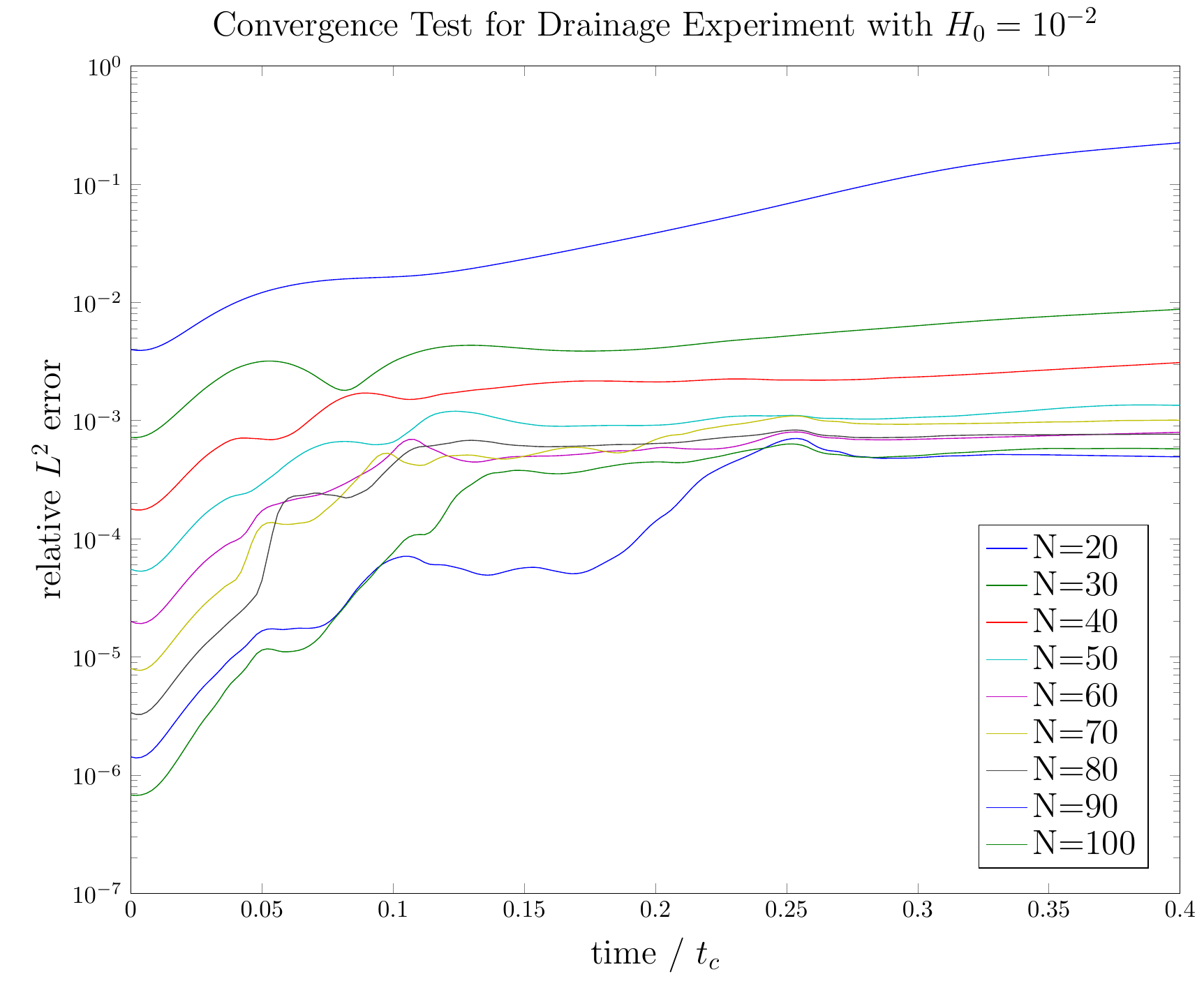}
    \end{center}
    \caption{Convergence test for drainage experiment depicted in Figure
        \ref{fig:S_Drain_compare}. $N$ is the number of spatial grid points. In Figure
        \ref{fig:S_Drain_compare}, $t_1 = 0.025t_c$, $t_2 = 0.050t_c$, $t_3 = 0.075t_c$,
    and $t_4 = 0.010t_c$}
    \label{fig:Drain_Conv_Test}
\end{figure}
\renewcommand{\baselinestretch}{\normalspace}

Figure \ref{fig:S_Imbibe_compare} shows an imbibition experiment for various value of
$H_0$.  As before, the initial condition is shown in black.  For this experiment,
Dirichlet boundary conditions are enforced at both $x=0$ and $x=1$.  Gravity points in the
negative $x$ direction, and the boundary condition at $x=1$ indicates that wetting fluid
is being added over time.  For $H_0 = 10^{-4}$ and $H_0 = 10^{-3}$ we see plausibly
physical results and we see sharper wetting fronts as in the drainage experiment.  For $H_0
\ge 10^{-2}$ we almost immediately see a non-physical non-monotonicity appear at the top
edge of the wetting front.  Similar behavior was observed by Peszynsk\'a and Yi
\cite{Peszynska2008} for their numerical scheme, and they stated 
\begin{quote}
    ``\ldots we cannot speculate whether the apparent nonmonotonicity of profiles \ldots relates
    to a numerical instability, or to a physical phenomenon.''
\end{quote}
It is reasonable to ask whether this is
associated with numerical noise, and Figure \ref{fig:Imbibe_Conv_Test} shows a convergence
test similar to that shown with the drainage experiment.  From Figure
\ref{fig:Imbibe_Conv_Test} it appears as if the numerical method is converging under mesh
refinement for dimensionless time approximately less than $0.1$.  The non-monotonicity
appears in the region where the method should be stable so we
tentatively conclude that this effect is not a numerical artifact. Finally, we observe that for $H_0 = 10^{-1}$ the advection
term has been overwhelmed by the diffusion and the third-order term and the numerical
results are completely non-physical.
\linespread{1.0}
\begin{figure}[H]
        \centering
        \subfigure[Saturation profiles at $t=t_1$]{
                \includegraphics[width=0.45\textwidth]{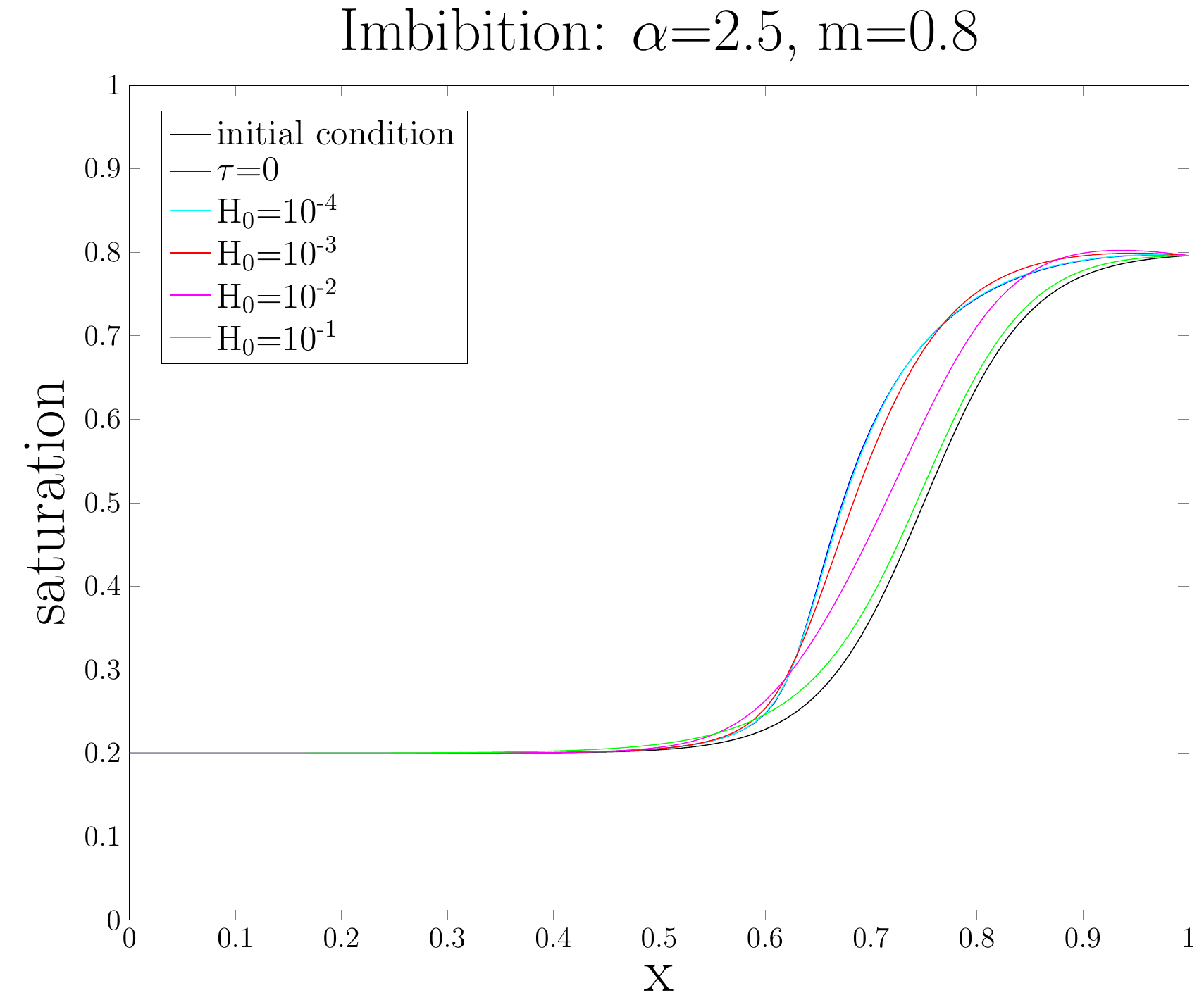}
                    \label{fig:S_Imbibe_t025}
            }
        \subfigure[Saturation profiles at $t=t_2$]{
                \includegraphics[width=0.45\textwidth]{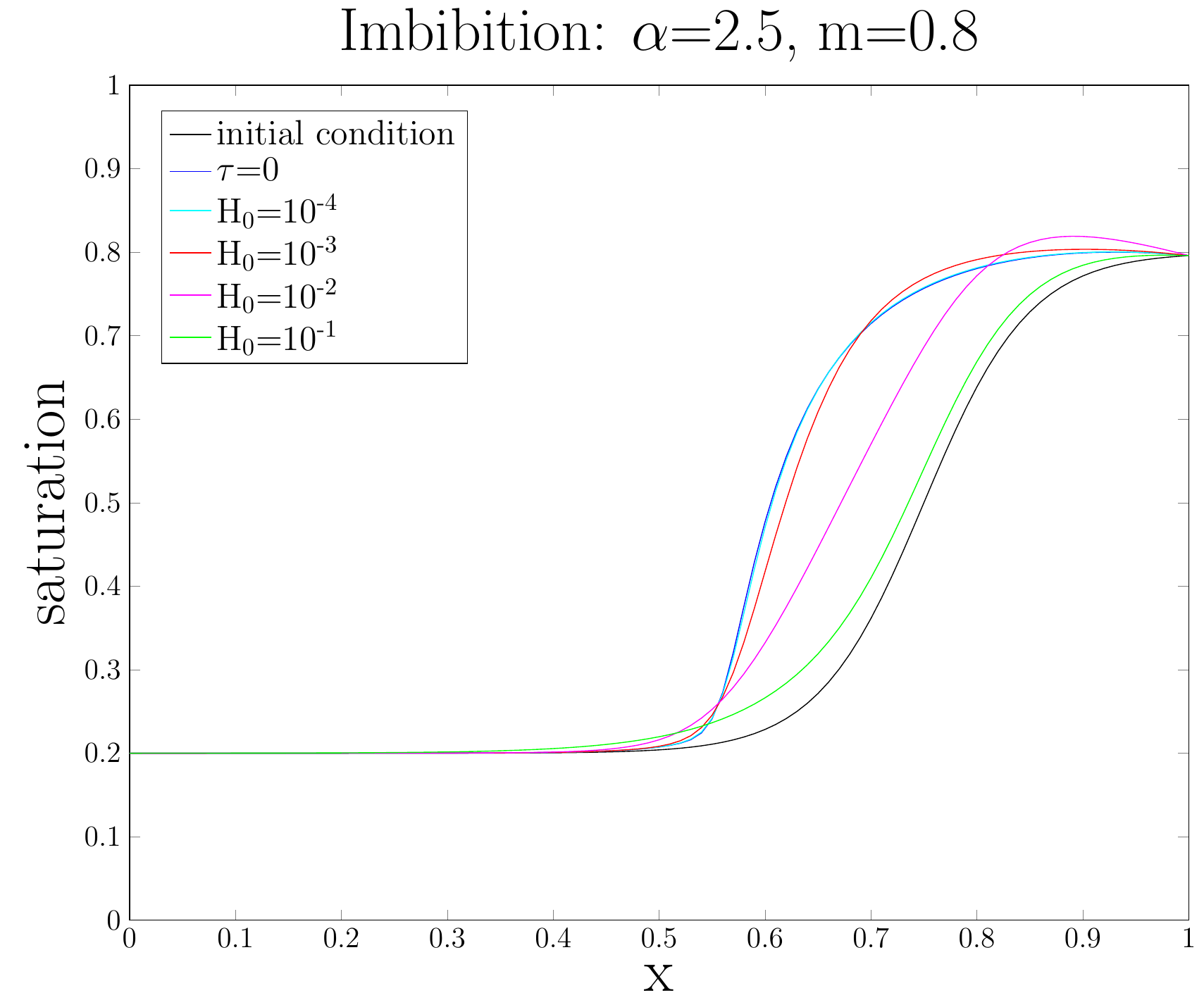}
                    \label{fig:S_Imbibe_t050}
        }
        \subfigure[Saturation profiles at $t=t_3$]{
                \includegraphics[width=0.45\textwidth]{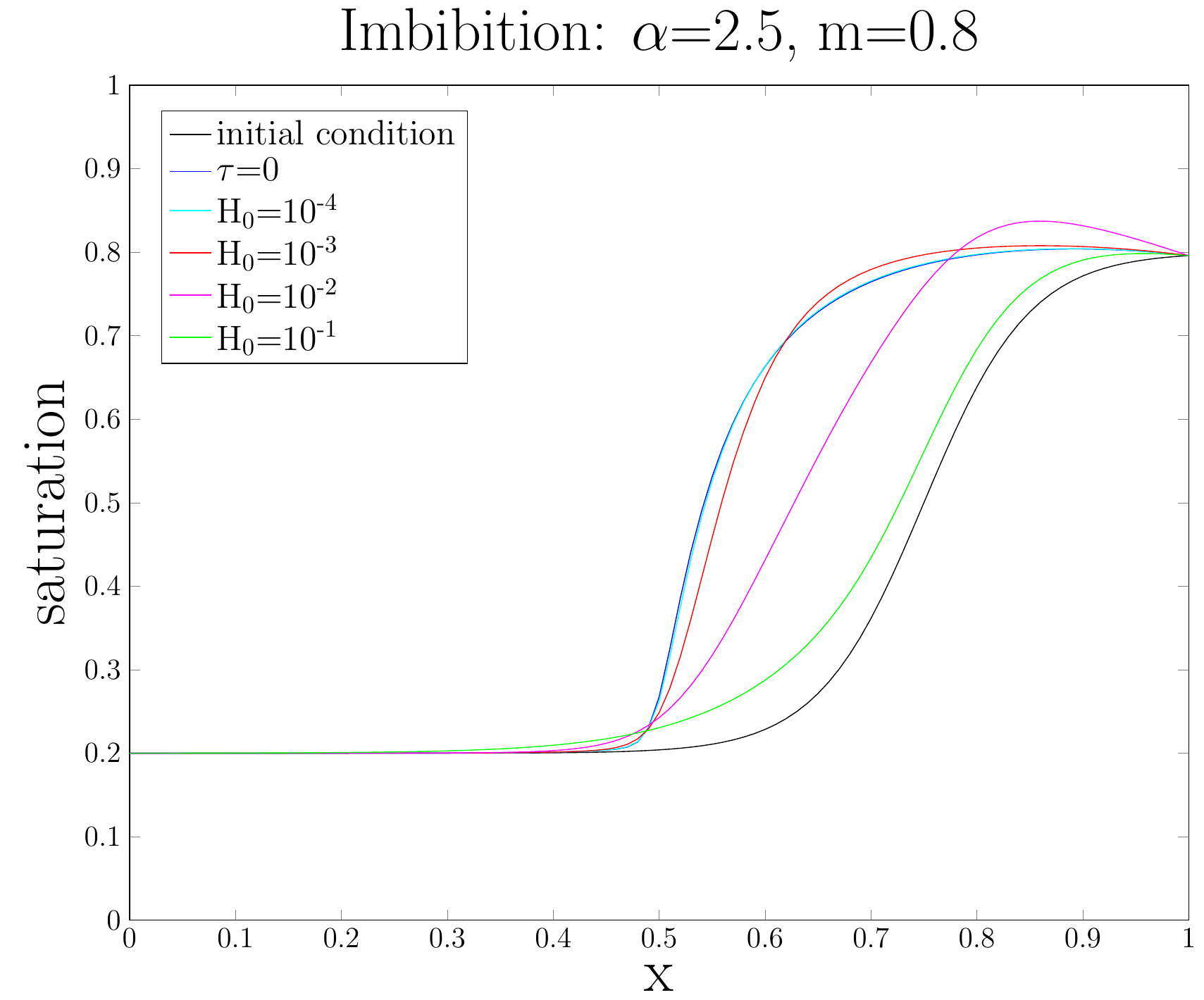}
                    \label{fig:S_Imbibe_t075}
        }
        \subfigure[Saturation profiles at $t=t_4$]{
                \includegraphics[width=0.45\textwidth]{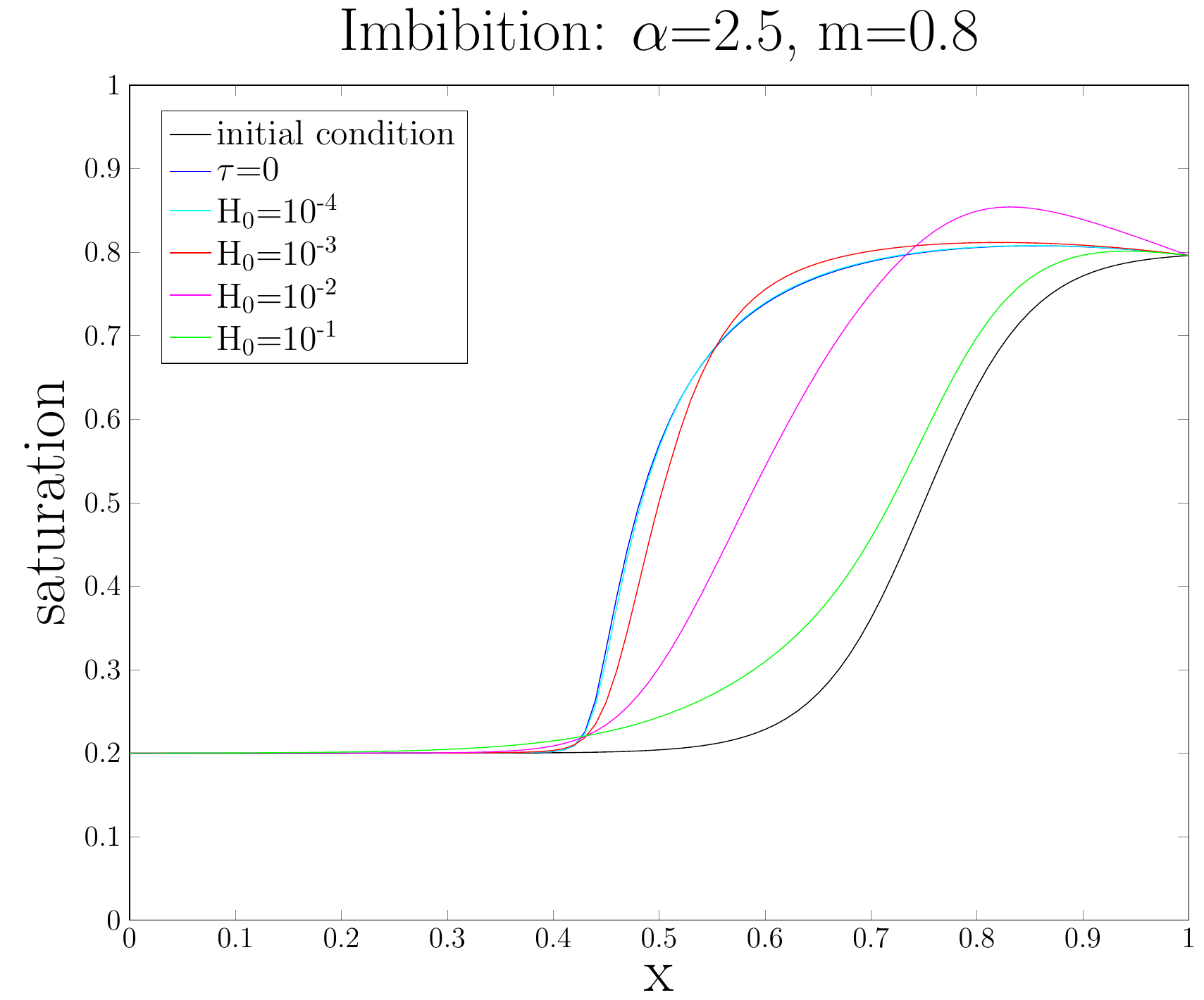}
                    \label{fig:S_Imbibe_t100}
        }
        \caption{Saturation profiles in a imbibition experiment with $\al = 2.5, n =
        5$.}\label{fig:S_Imbibe_compare}
\end{figure}
\renewcommand{\baselinestretch}{\normalspace}

\linespread{1.0}
\begin{figure}[H]
    \begin{center}
        \includegraphics[width=0.75\columnwidth]{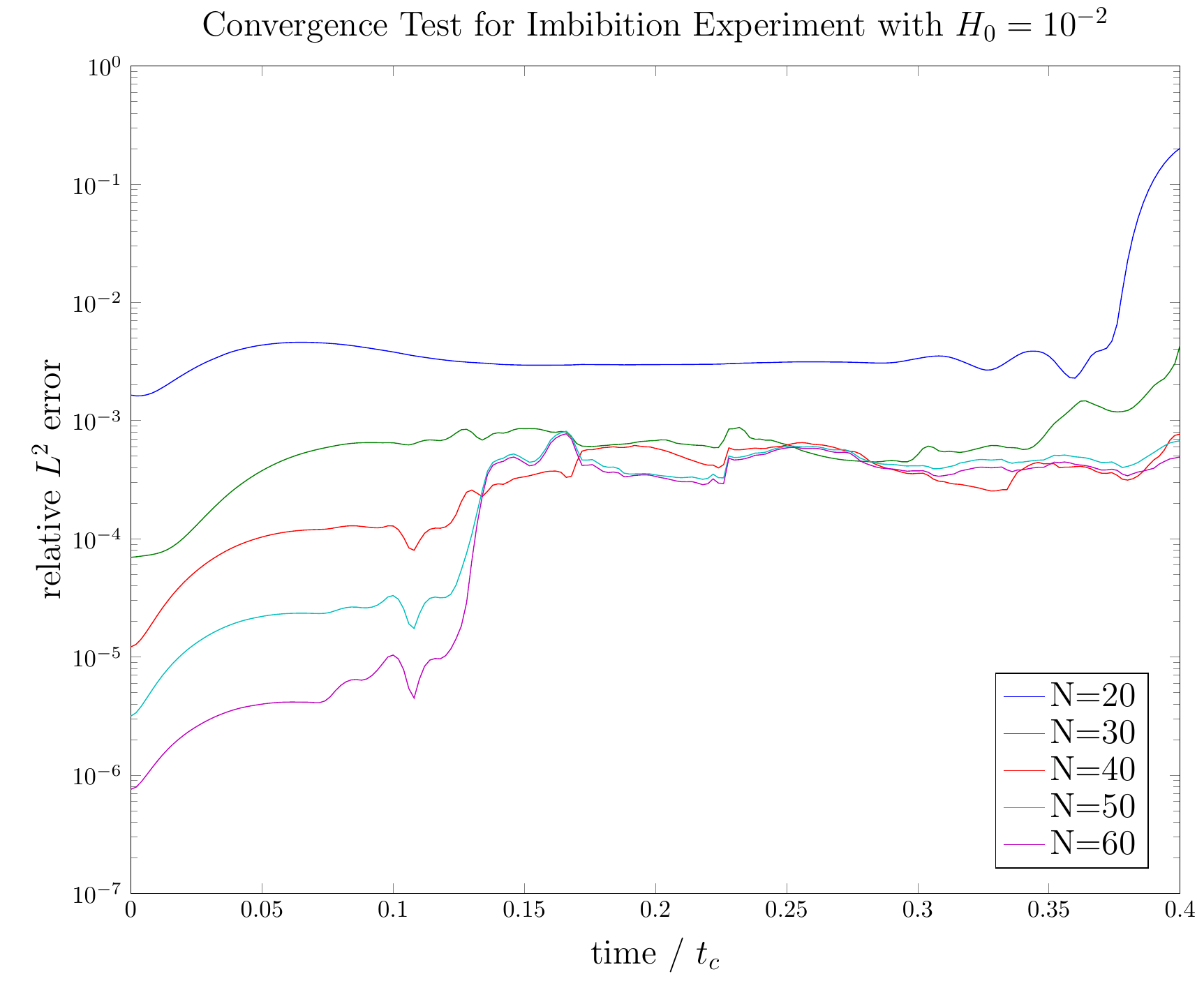}
    \end{center}
    \caption{Convergence test for imbibition experiment depicted in Figure
        \ref{fig:S_Imbibe_compare}.  In Figure
        \ref{fig:S_Imbibe_compare}, $t_1 = 0.025t_c$, $t_2 = 0.050t_c$, $t_3 = 0.075t_c$,
    and $t_4 = 0.010t_c$}
    \label{fig:Imbibe_Conv_Test}
\end{figure}
\renewcommand{\baselinestretch}{\normalspace}

Clearly there are infinitely many choices of initial boundary conditions, and the results
presented herein inherently depend on the conditions chosen.  Similar types of
non-physical behavior can be observed for other families of van Genuchten parameters, but
the associated plots are excluded here for brevity.  An empirical conclusion is that for $H_0 =
(\tau \porosity_S \kappa_S)/\mu_l$ greater than $10^{-3}$ possibly leads to non-physical
behavior in the numerical solution. 

To the author's knowledge, an analysis of parameters of this type has not been completed
in the literature.  We have shown in this subsection that for reasonably small values of
$\tau$ we predict sharper fronts than with the traditional Richards' equation. 

\section{Vapor Diffusion Equation}\label{sec:numerical_vapor_diffusion}
Next let us consider the vapor diffusion equation under assumptions of fixed constant
temperature and a fixed saturation profile.  This particular study is a bit peculiar since
it is unlikely that a saturation profile will remain fixed during an evaporation (or
condensation) study. Of course, we could consider $S \equiv 0$ everywhere and study only
evaporation in dry porous media, but this is also not realistic as enhancement models
depend partly on the presence of a liquid phase. In this section we compare the present model
proposed in Section \ref{sec:vapor_diffusion} to the classical enhancement model and to
Fickian diffusion.
\begin{flalign}
    (1-S) \pd{\rh}{t} &= \pd{ }{x} \left( \mathcal{D}(\rh,S) \pd{\rh}{x} \right) \quad
    \text{(present model)} \\
    (1-S) \pd{\rh}{t} &= \pd{ }{x} \left( \eta_{(a)}(S) \tau(S) D^g \pd{\rh}{x} \right) \quad
    \text{(enhancement model)} \\
    (1-S) \pd{\rh}{t} &= D^g \pdd{\rh}{x} \quad \label{eqn:fickian_diffusion_model}
    \text{(Fickian diffusion model)}.
\end{flalign}
Recall that $\eta_{(a)}(S)$ is the empirical {\it enhancement factor} traditionally used,
$\tau(S)$ is the tortuosity, and $D^g$ is the constant Fickian vapor diffusion
coefficient (see equations \eqref{eqn:enhancement_cass}, \eqref{eqn:tortuosity_cass}, and
obviously \eqref{eqn:fickian_diffusion_model}) .  The reader should note that we are slightly abusing notation given that
$\eta$ previously stood for intensive entropy and $\tau$ is the label for the relaxation term in
the saturation equation.  This abuse of notation is contained to this section and should
not cause confusion.

Qualitatively, the {\it shape} of the diffusion curve in the $x-\rh$ plane for the present
model is rather different than those of the enhancement and Fickian models. Figure
\ref{fig:DiffComp_compare} gives several snapshots of a sample diffusion experiment with
enhancement parameter $a=25$, van Genuchten parameter $m=0.9$, and saturated permeability
$\kappa_S = 1.04 \times 10^{-10} m^2$. Observe further that the steady state solutions are
different for the two models. This is no surprise since the nonlinearities in the
diffusion coefficient have different functional forms.
\linespread{1.0}
\begin{figure}[H]
        \centering
        \subfigure[Relative humidity profiles at $t=t_1$]{
                \includegraphics[width=0.45\textwidth]{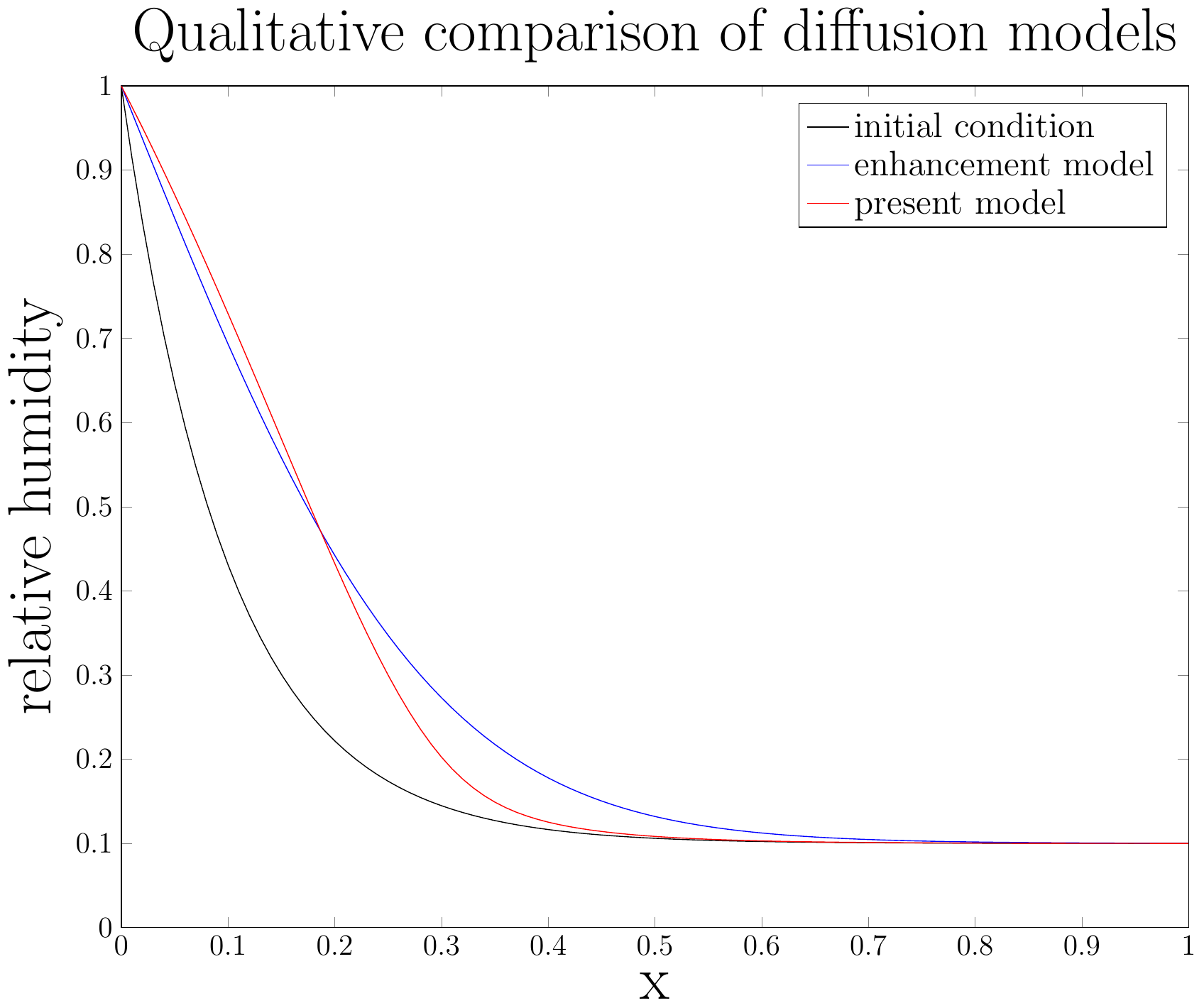}
                    \label{fig:DiffComp_t002}
            }
        \subfigure[Relative humidity profiles at $t=t_2$]{
                \includegraphics[width=0.45\textwidth]{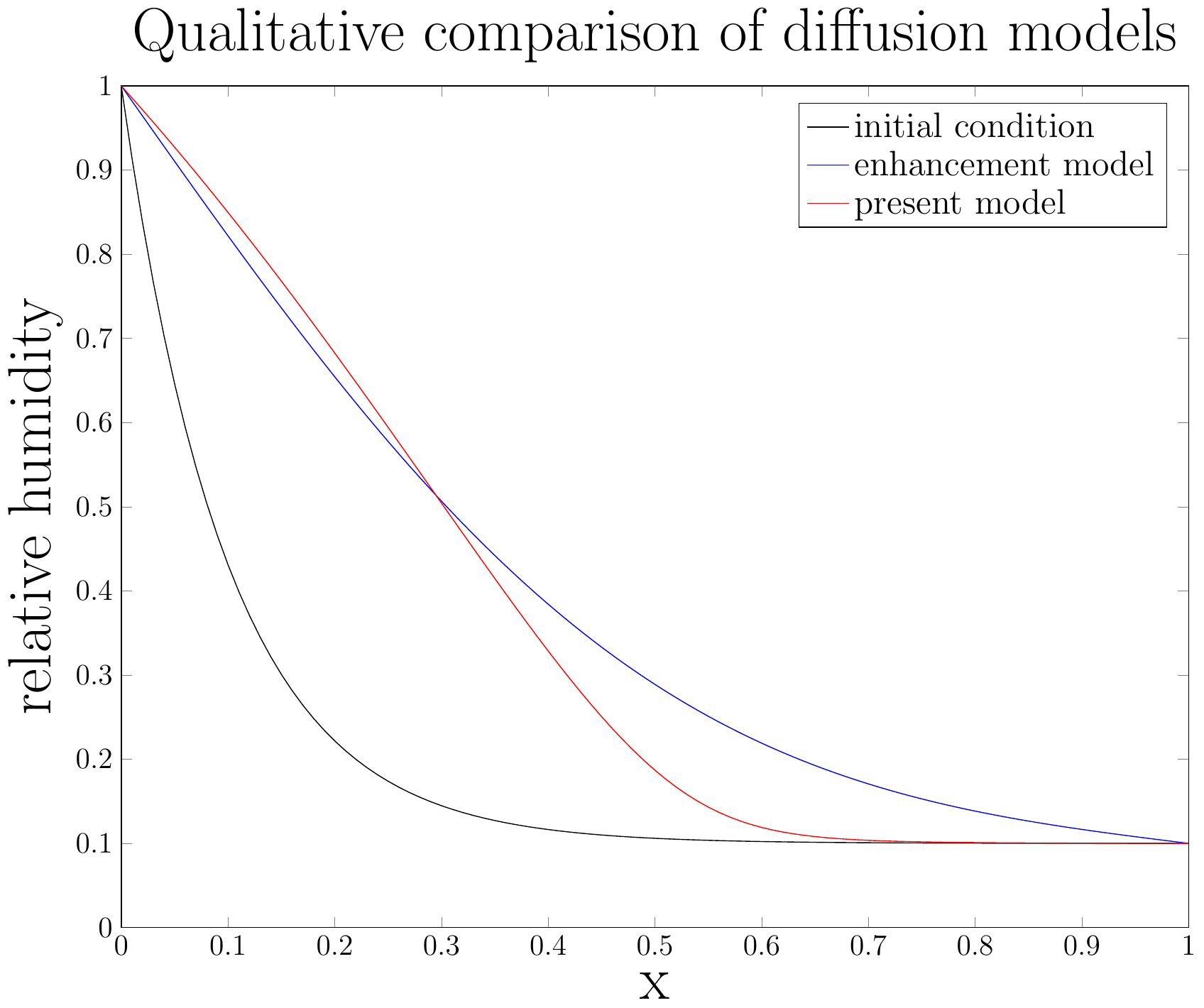}
                    \label{fig:DiffComp_t005}
        }
        \subfigure[Relative humidity profiles at $t=t_3$]{
                \includegraphics[width=0.45\textwidth]{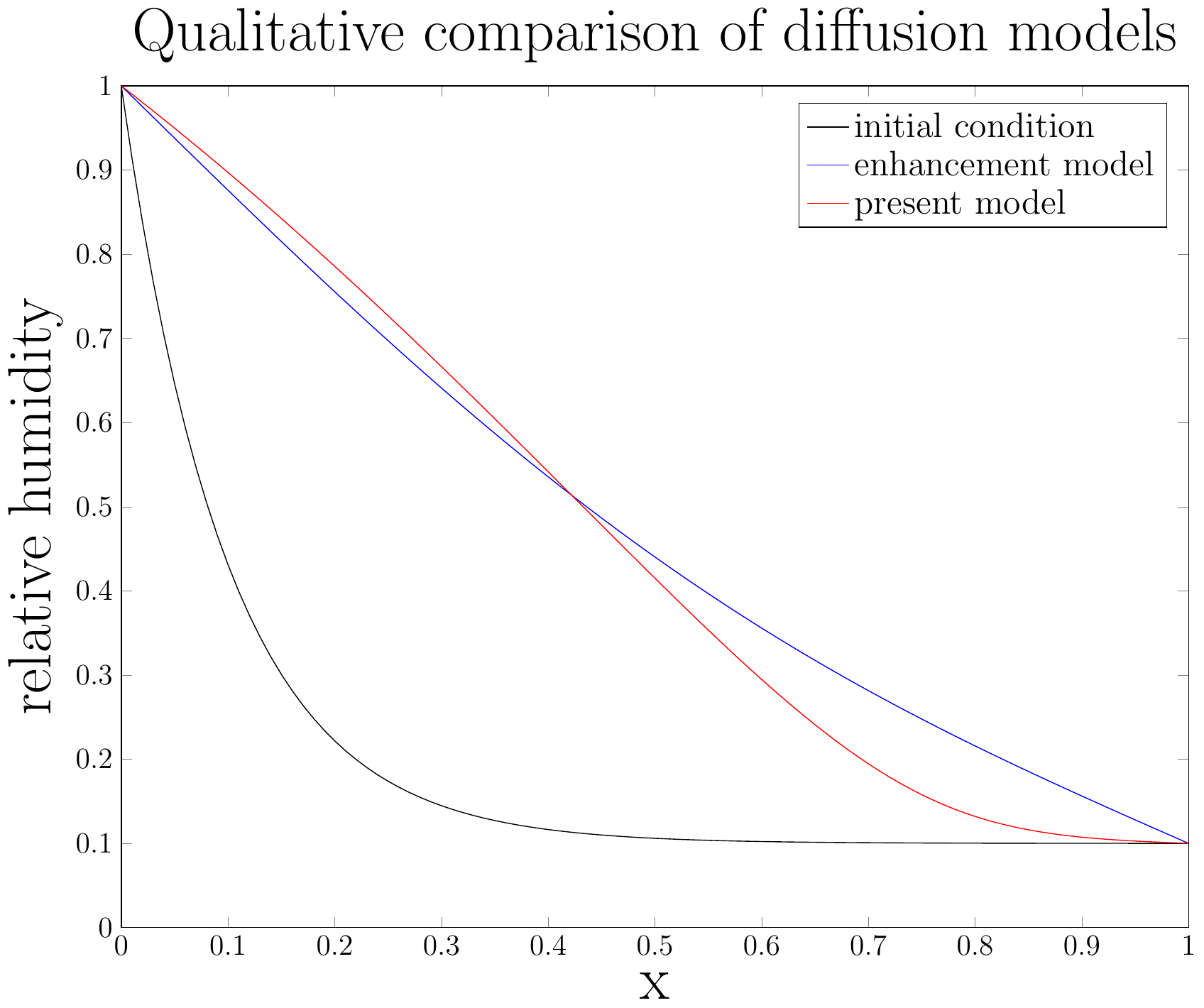}
                    \label{fig:DiffComp_t010}
        }
        \subfigure[Relative humidity profiles at $t=t_4$]{
                \includegraphics[width=0.45\textwidth]{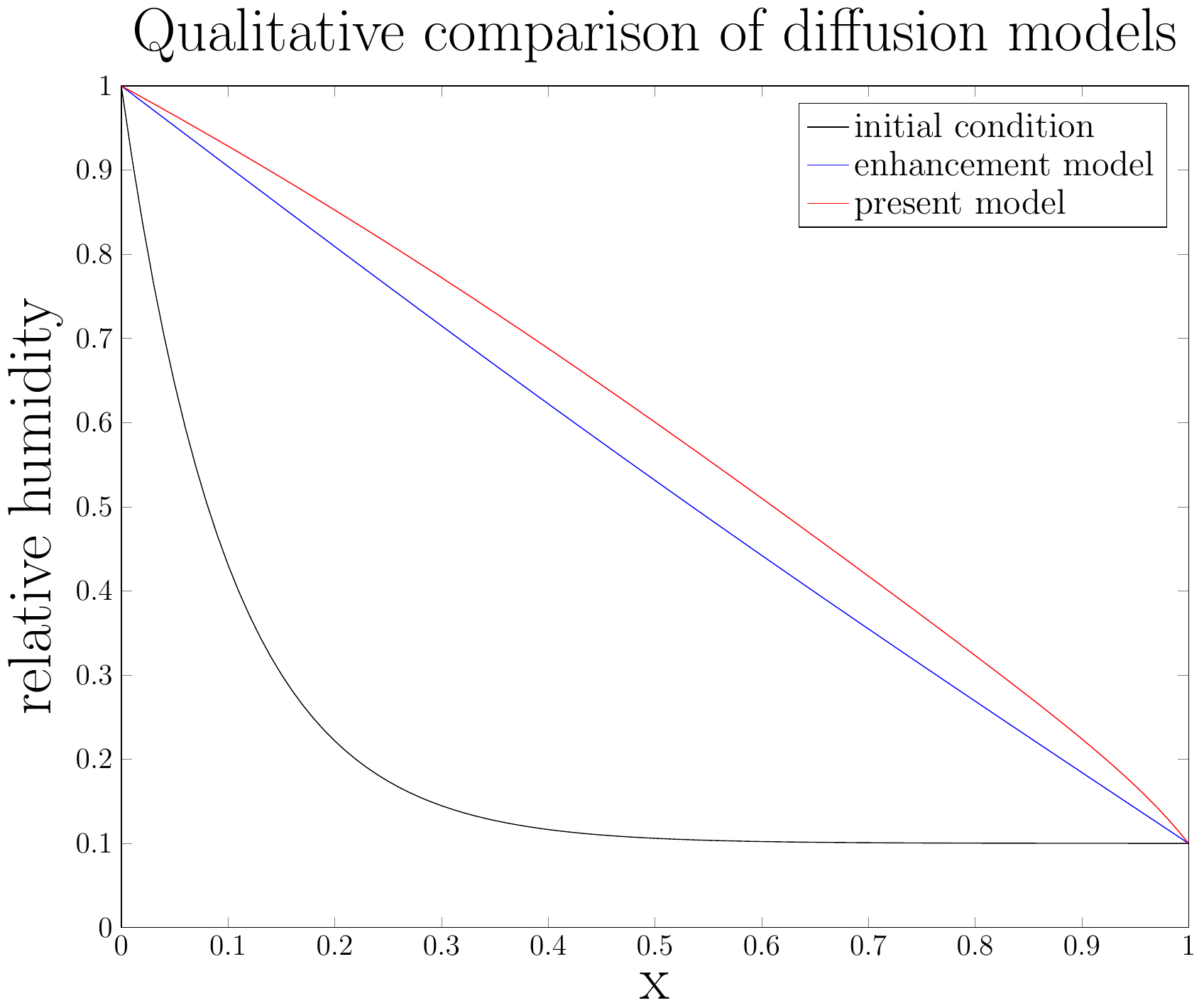}
                    \label{fig:DiffComp_t020}
        }
        \caption{Sample diffusion experiment comparing the enhancement model to the
            present model. Here, $a = 25$, $m=0.9$ ($n=10$), and $\kappa = 10^{-10}$ with
            Dirichlet boundary conditions and an exponential initial profile.
        }\label{fig:DiffComp_compare}
\end{figure}
\renewcommand{\baselinestretch}{\normalspace}

In Figure \ref{fig:diffusion_coeff_comparison} we saw that there is potentially a marked
difference between the diffusion coefficient in the present model and the enhancement
model.  Figures \ref{fig:EnhVSull_compare} and \ref{fig:EnhVSullKappa2_compare} show a
comparison of the diffusion coefficients for several values of the van Genuchten $m$
parameter and two different saturated permeabilities.  The functional dependence of the
diffusion coefficient in the present model on the van Genuchten parameter can be readily
seen between Figures \ref{fig:EnhVSull1_67} and \ref{fig:EnhVSull20} (similarly,
\ref{fig:EnhVSullKappa21_67} and \ref{fig:EnhVSullKappa220}), and the functional
dependence on the saturated permeability can be seen between the two sets of figures. 

\linespread{1.0}
\begin{figure}[H]
        \centering
        \subfigure[Comparison for $m=0.4$ ($n=1.67$)]{
            \includegraphics[width=0.45\textwidth]{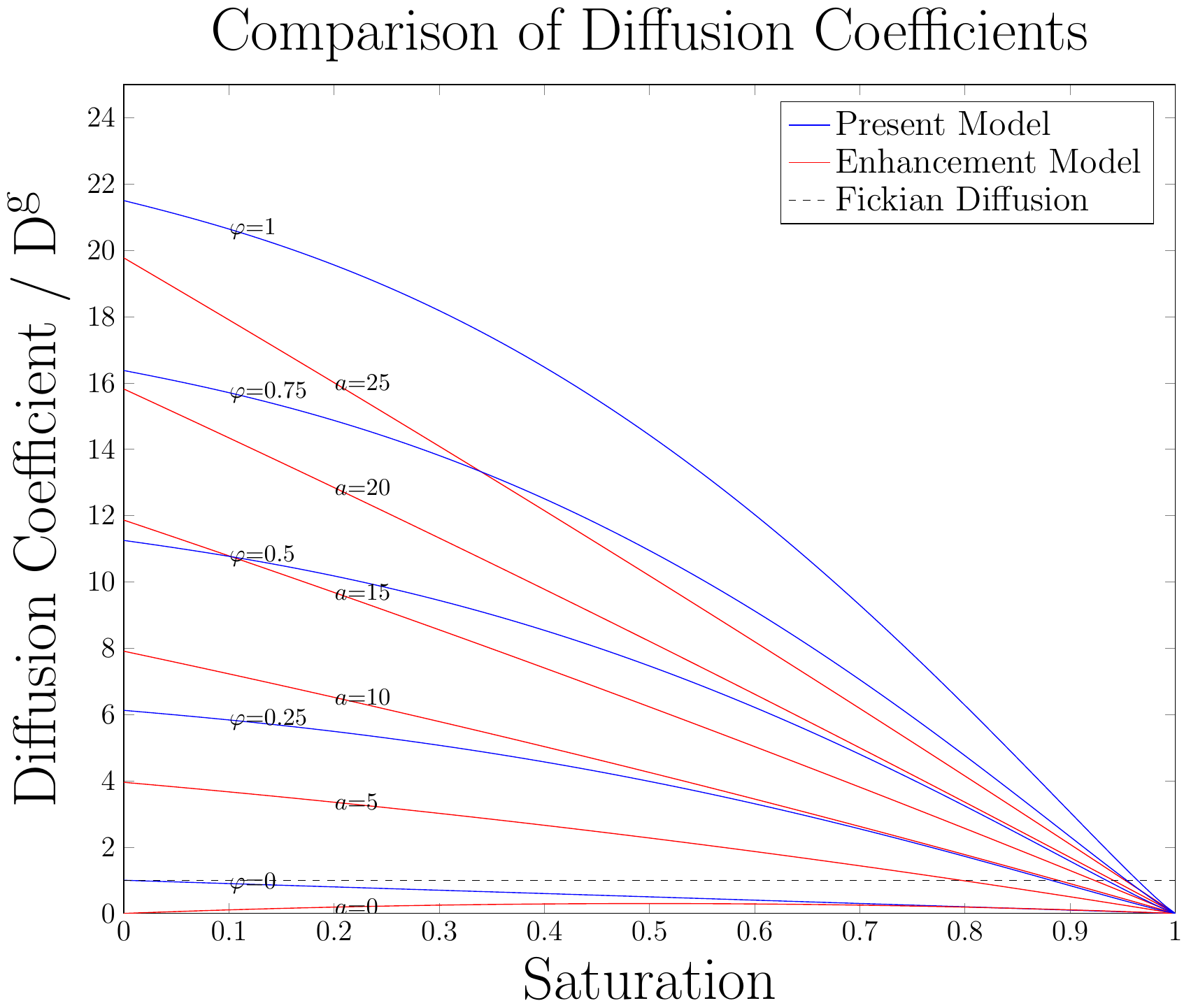}
                    \label{fig:EnhVSull1_67}
            }
        \subfigure[Comparison for $m=0.6$ ($n=2.5$)]{
            \includegraphics[width=0.45\textwidth]{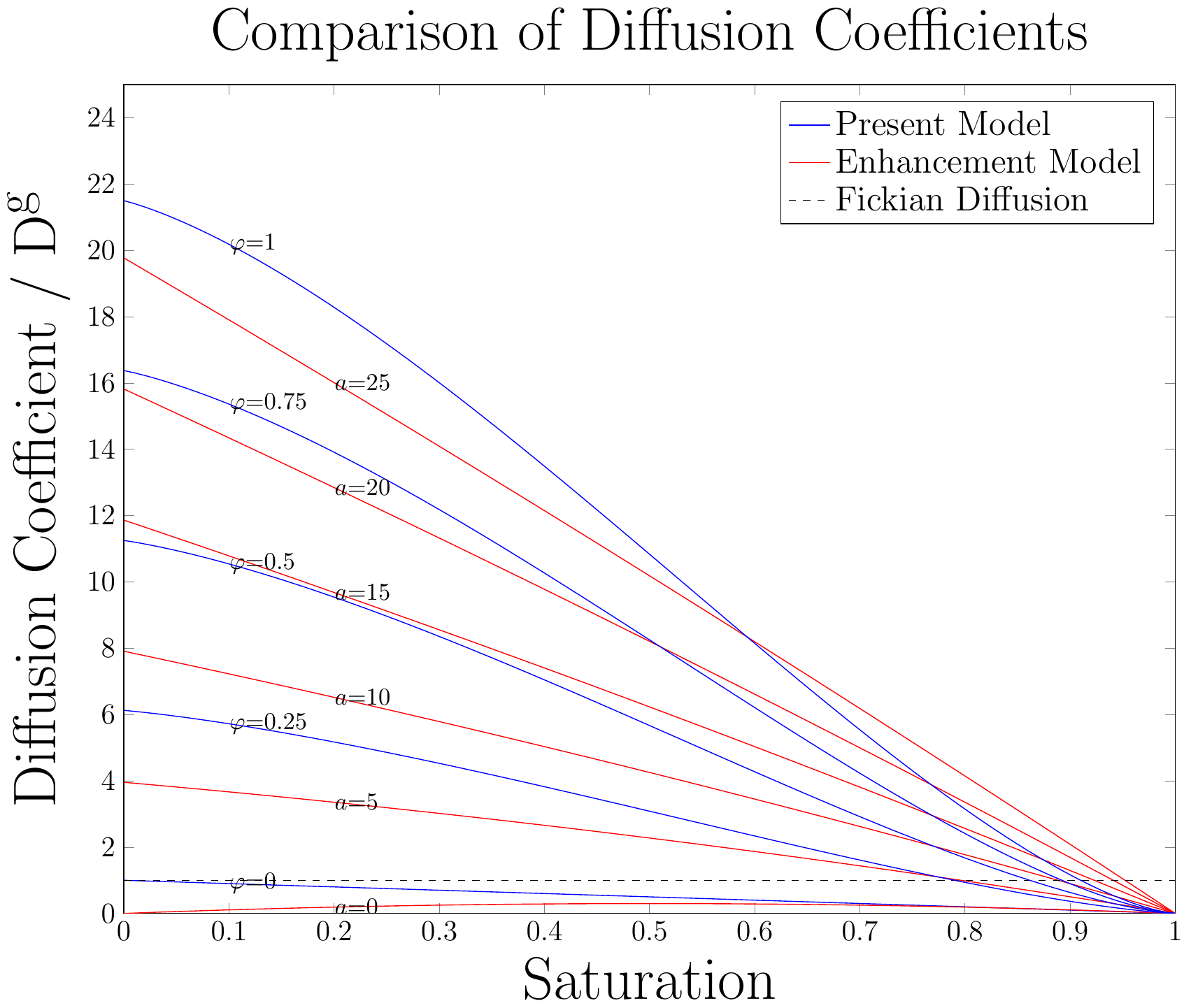}
                    \label{fig:EnhVSull2_5}
        }
        \subfigure[Comparison for $m=0.8$ ($n=5$)]{
            \includegraphics[width=0.45\textwidth]{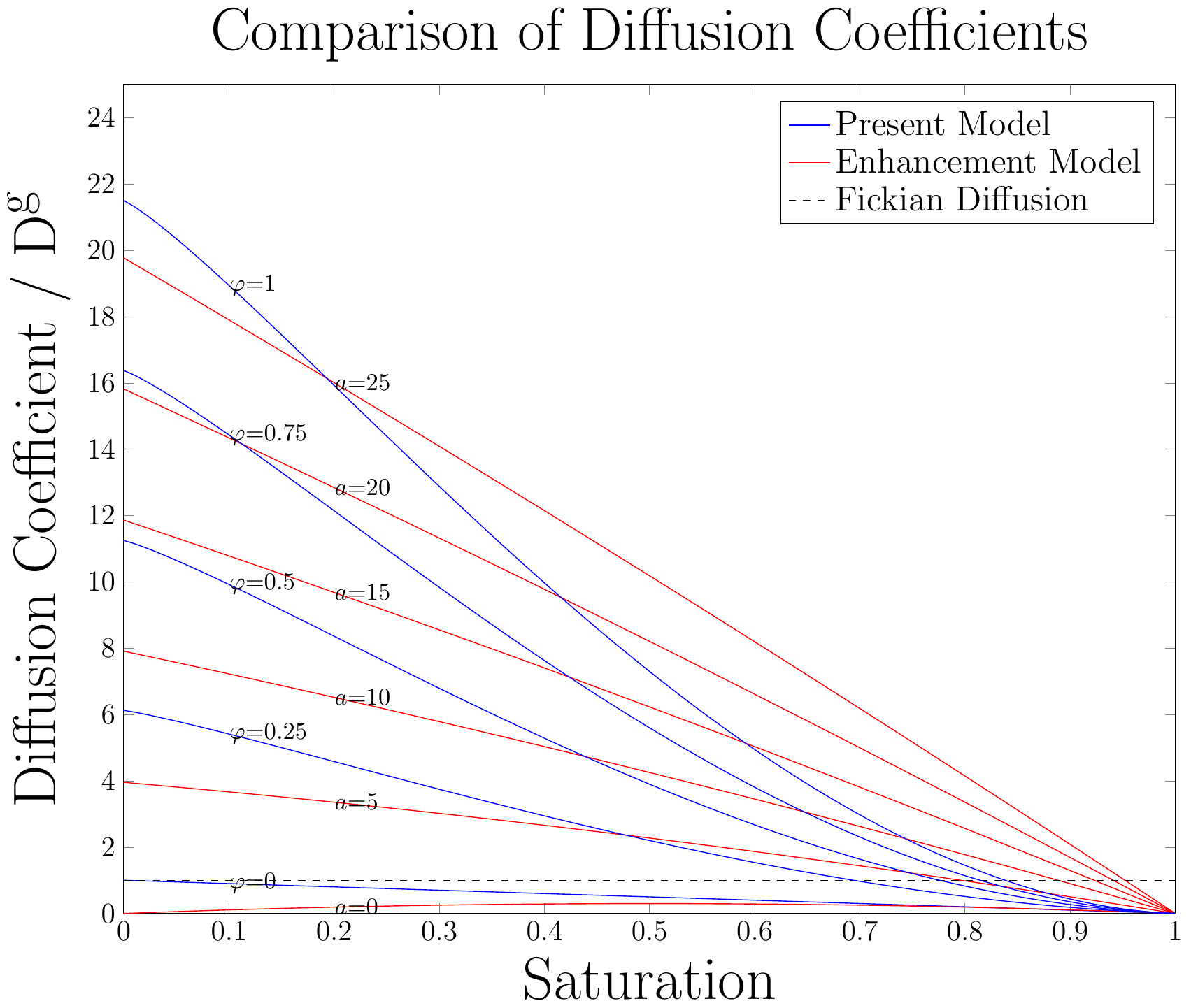}
                    \label{fig:EnhVSull5}
        }
        \subfigure[Comparison for $m=0.95$ ($n=20$)]{
            \includegraphics[width=0.45\textwidth]{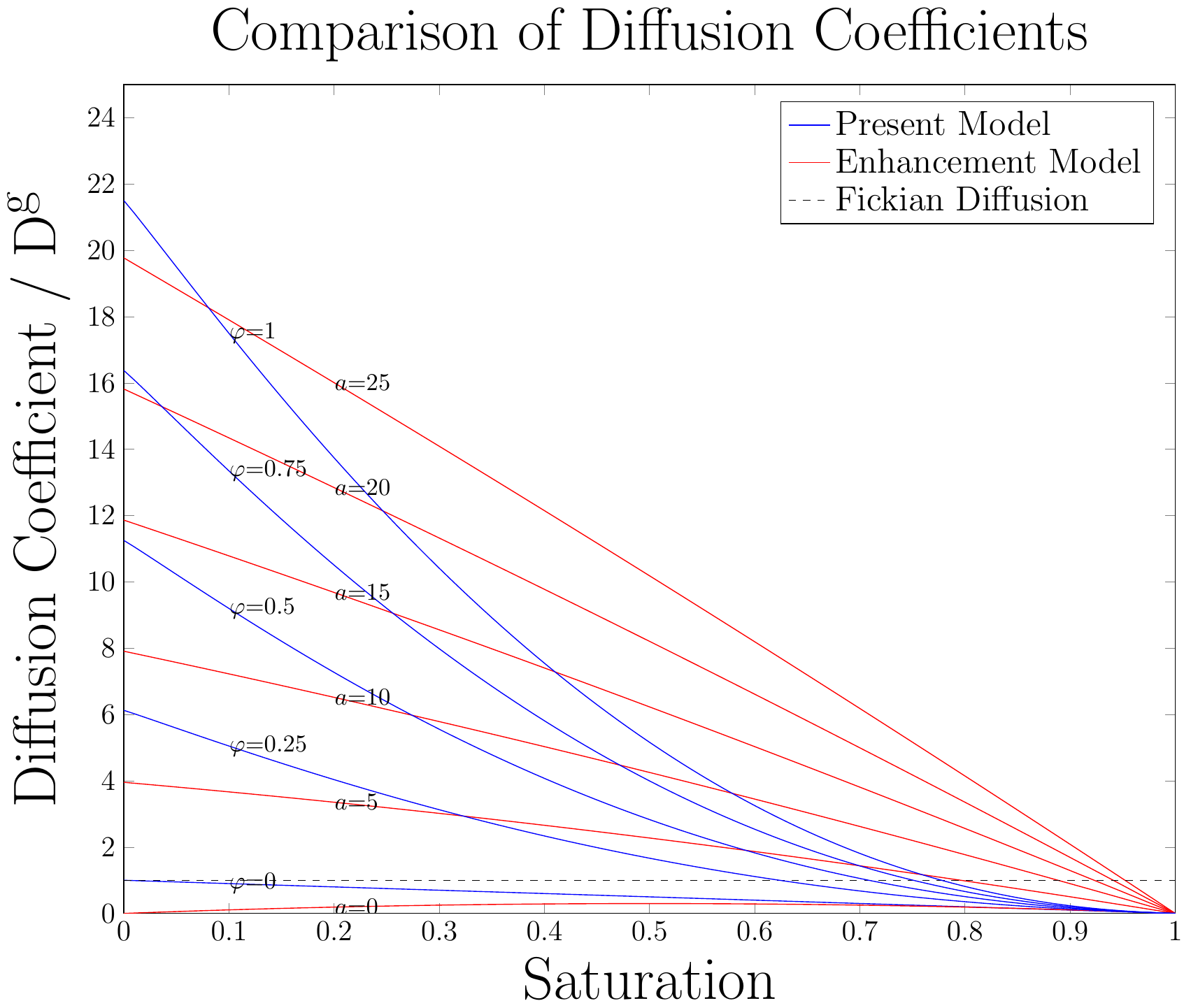}
                    \label{fig:EnhVSull20}
        }
        \caption{Comparison of diffusion coefficients for various van Genuchten
            parameters all taken with $\kappa_s = 1.04 \times 10^{-10}$ and $\porosity =
            0.334$ to match the experiment in \cite{Smits2011}.}\label{fig:EnhVSull_compare}
\end{figure}
\renewcommand{\baselinestretch}{\normalspace}
\linespread{1.0}
\begin{figure}[H]
        \centering
        \subfigure[Comparison for $m=0.4$ ($n=1.67$)]{
            \includegraphics[width=0.45\textwidth]{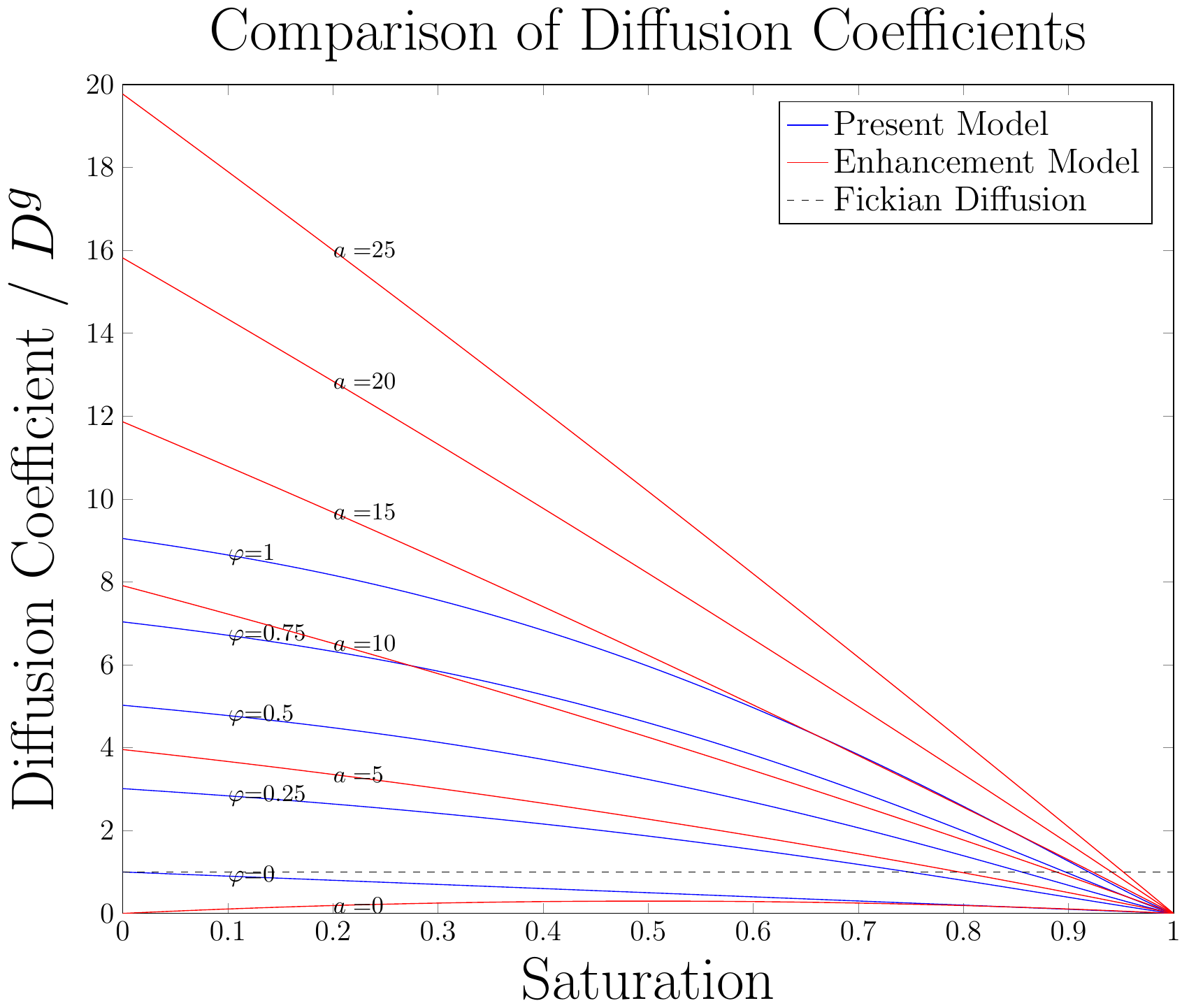}
                    \label{fig:EnhVSullKappa21_67}
            }
        \subfigure[Comparison for $m=0.6$ ($n=2.5$)]{
            \includegraphics[width=0.45\textwidth]{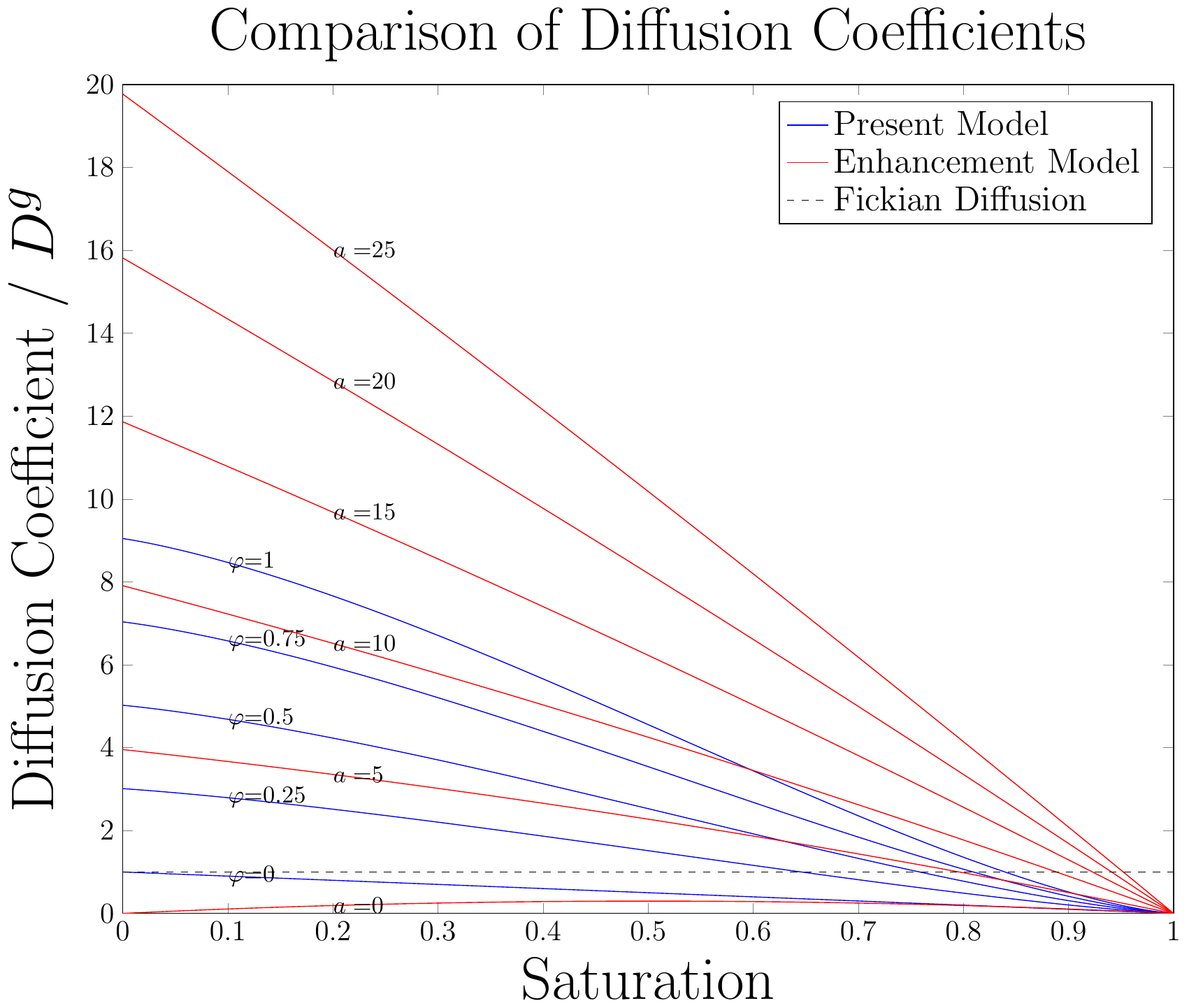}
                    \label{fig:EnhVSullKappa22_5}
        }
        \subfigure[Comparison for $m=0.8$ ($n=5$)]{
            \includegraphics[width=0.45\textwidth]{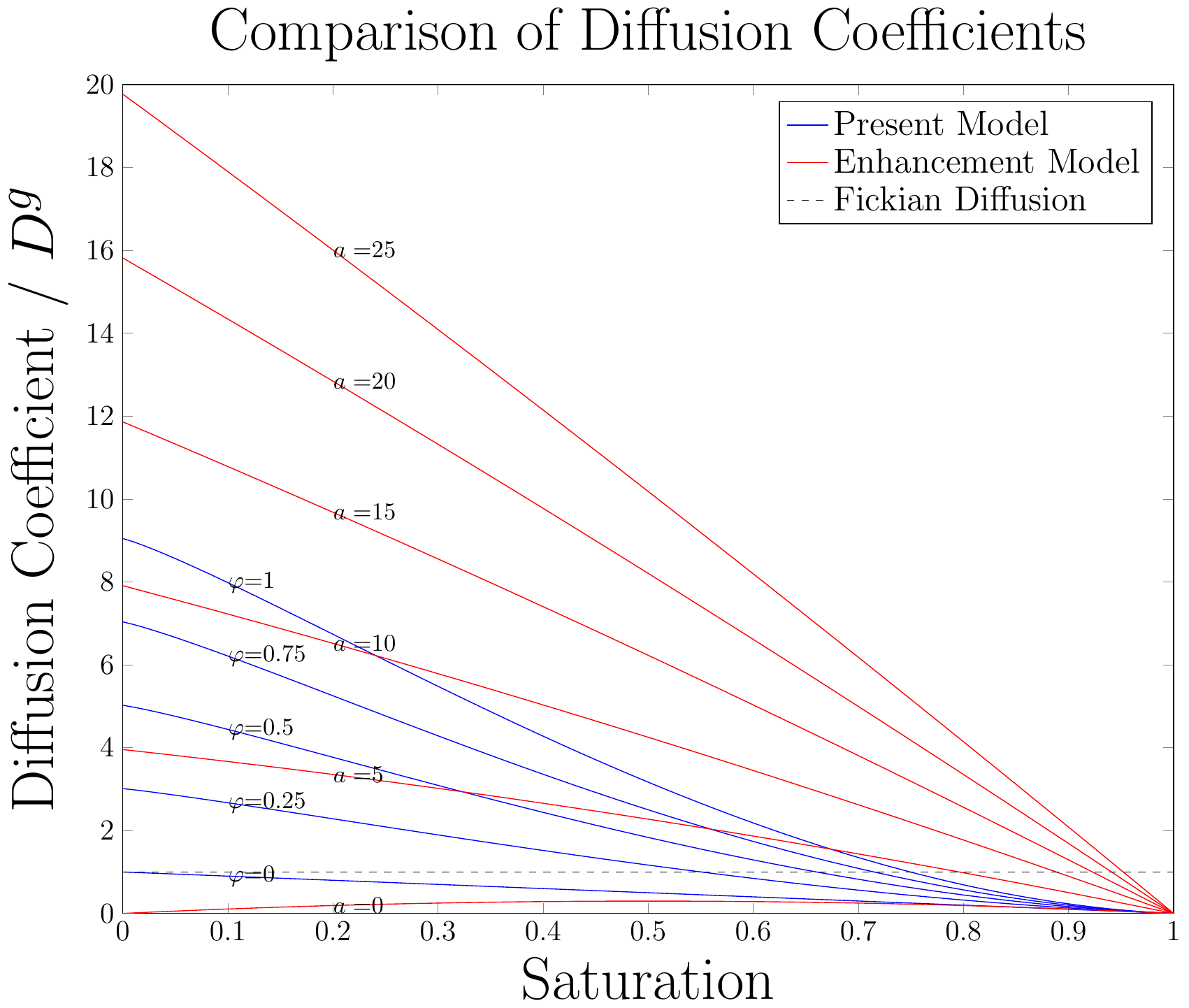}
                    \label{fig:EnhVSullKappa25}
        }
        \subfigure[Comparison for $m=0.95$ ($n=20$)]{
            \includegraphics[width=0.45\textwidth]{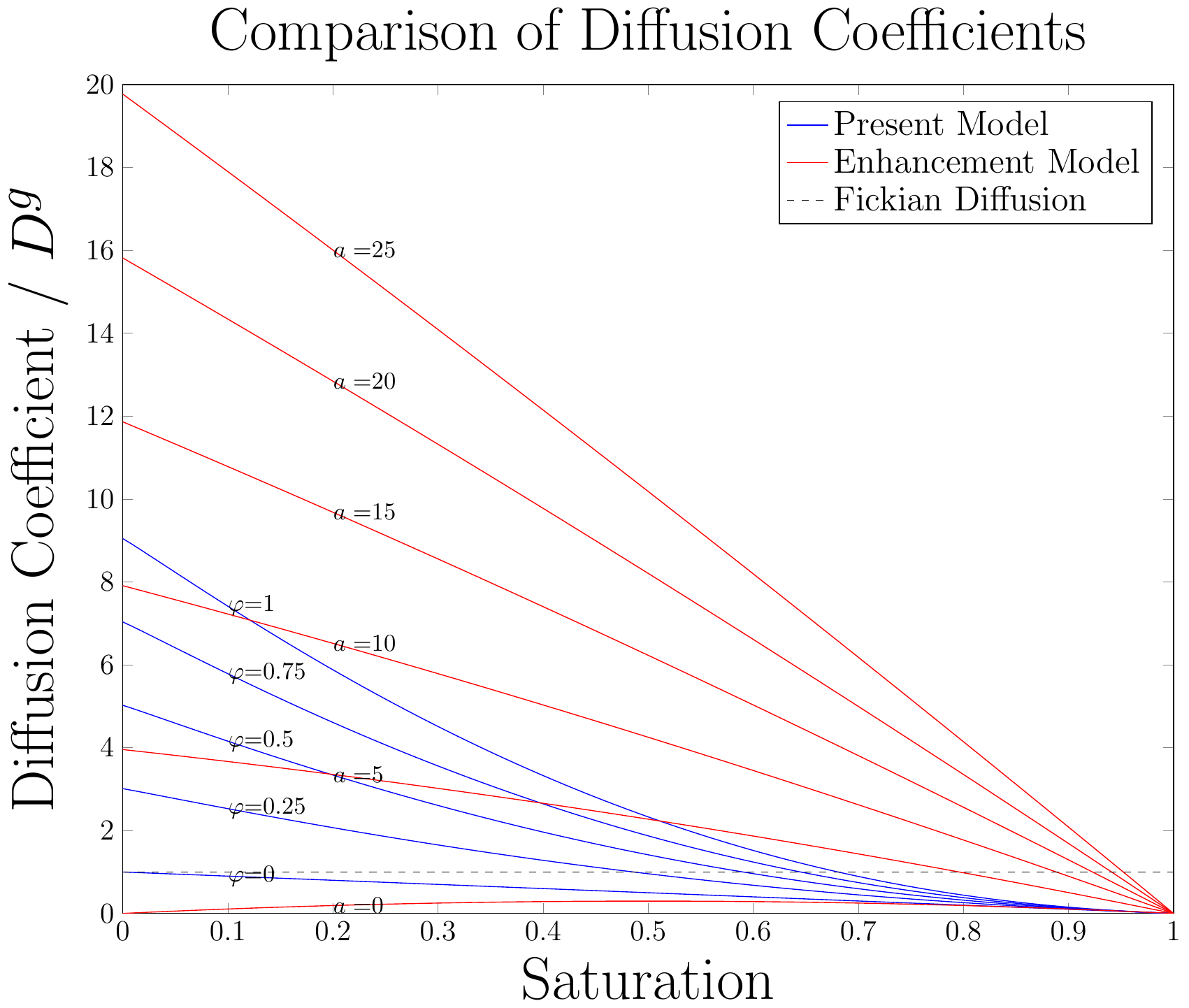}
                    \label{fig:EnhVSullKappa220}
        }
        \caption{Comparison of diffusion coefficients for various van Genuchten
            parameters all taken with $\kappa_s = 4.0822 \times 10^{-11}$ and $\porosity =
            0.385$ to match the experiment in
            \cite{Sakai2009}.}\label{fig:EnhVSullKappa2_compare}
\end{figure}
\renewcommand{\baselinestretch}{\normalspace}

In this thesis we propose that there is a relationship between the fitted $a$ value and
the material properties. 
\begin{prop}\label{prop:m_kappa_a}
    Given the van Genuchten $m$ (or equivalently, $n$) parameter and the saturated
    permeability, $\kappa_S$, of the soil there is an a-priori estimate of the fitting
    parameter $a$. 
\end{prop}
The immediate consequence of Propostion \ref{prop:m_kappa_a} is that if the fitting
parameter can be predicted with the use of experimentation then it is, indeed,
unnecessary.

To test Proposition \ref{prop:m_kappa_a} we use a simple heuristic approach to match
the material coefficients, $m$ and $\kappa_S$, to the calculated fitting parameter, $a$.
This is done on the experiments by Smits et al.\,\cite{Smits2011} and
Sakai et al.\,\cite{Sakai2009}. 
In Smits et al., $\kappa_S = 1.04\times 10^{-10}$ [m$^2$] and $m=0.944$ with a statistically
tuned $a$-value of $18.2$. In Sakai et al., $\kappa_S = 4.082 \times 10^{-11}$ [m$^2$] and
$m\approx0.799$ with $a$-values of $5,8$, and $15$ considered. Sakai et al.\,indicated
the best agreement with $a=8$ while simultaneously considering a modified van Genuchten
(Fayer-Simmons) model for the soil water retention curve. We do not consider the
Fayer-Simmons model here, but as the Fayer-Simmons model is designed to give better
agreement of the capillary pressure - saturation relationship with very
low saturations we don't believe this negates our approach.

The heuristic tests of Proposition \ref{prop:m_kappa_a} is as
follows. The steady-state mass fluxes predicted by the propsed new model and the
tranditional enhanced diffusion model are 
\[ \rho_{sat} \mathcal{D}_{(m,\kappa_S)}(\rh,S) \grad \rh \quad \text{and} \quad \rho_{sat}
    \eta_{(a)}(S) \tau(S) D^g \grad \rh.\]
Assuming that the mass fluxes are equal gives the equation
\[ \mathcal{D}_{(m,\kappa_S)}(\rh,S) \grad \rh = \eta_{(a)}(S)\tau(S) D^g \grad \rh. \]
Making the further assumption that the gradients in relative humidity are the same at
steady state then the diffusion coefficients must be equal.  If this mass flux is taken at
a liquid-gas interface we can assume that $\rh=1$.  Hence, the left-hand side of this
equation is a function of $S$, $m$, and $\kappa_s$ while the right-hand side is a function
of $S$ and $a$  
\[ \mathcal{D}_{(m,\kappa_S)}(1,S) = \eta_{(a)}(S)\tau(S) D^g. \]
At this point we could proceed by simply choosing a value for $S$ and making comparisons
or we could consider the integral over all of $S$ to remove the dependence on the
saturation.  We choose the latter as it gives a cumulative effect of the diffusion
coefficient over the entire range of saturations.

Therefore, for each $m$ and $\kappa_S$ and for 
fixed $\rh=1$, there is a value of $a$ such that
\begin{flalign}
    \int_0^1 \mathcal{D}_{(m,\kappa_S)}(1,S) dS = \int_0^1 \eta_{(a)}(S)
    \tau(S) D^g dS.
    \label{eqn:heuristic_1}
\end{flalign}
The left-hand side is a function of material parameters and the right-hand side is a
function of $a$. The right-hand side of equation \eqref{eqn:heuristic_1} integrates easily
to a linear function of $a$
\[ \int_0^1 \eta_{(a)}(S) \tau(S) D^g dS = \frac{D^g}{3} \left( a + \frac{1}{1-\porosity}
    \right). \]
The left-hand side of equation \eqref{eqn:heuristic_1}, on the other hand, isn't readily
integrable due to the nonlinear nature of the van Genuchten relative permeability
function. For this reason we seek an approximate solution to equation
\eqref{eqn:heuristic_1}.

Figure \ref{fig:heuristic_1} shows the left- and right-hand
sides of equation \eqref{eqn:heuristic_1}.  The intersections indicate the triple
$(m,\kappa_S,a)$ where the equation is true, and hence indicates where the two
models have the same cumulative diffusive effect over $S \in [0,1]$. For example,
in Figure \ref{fig:heuristic_1}, if $m\approx 0.5$ and $\kappa_s \approx 10^{-10}$
then we predict a fitting parameter of $a \approx 30$. 

In Figure \ref{fig:heuristic_1}, the blue and green curves are the right-hand sides of
equation \eqref{eqn:heuristic_1} for different saturated permeabilities. The blue curve is included to show the agreement with
Smits et al. The green curve is included to show the agreement with Sakai et al.  Observe
that the experimental values are {\it close} to the values that make equation
\eqref{eqn:heuristic_1} true (the intersections indicated in the figure).  This is to say
that given $m$ and $\kappa_S$, equation \eqref{eqn:heuristic_1} could have been used as an
a-priori estimate of the value of $a$ in these two experiments.  Table
\ref{tab:heuristic_1} gives a more concise summary of the results found in Figure
\ref{fig:heuristic_1}.
\linespread{1.0}
        \begin{figure}[H]
            \begin{center}
                \includegraphics[width=0.75\columnwidth]{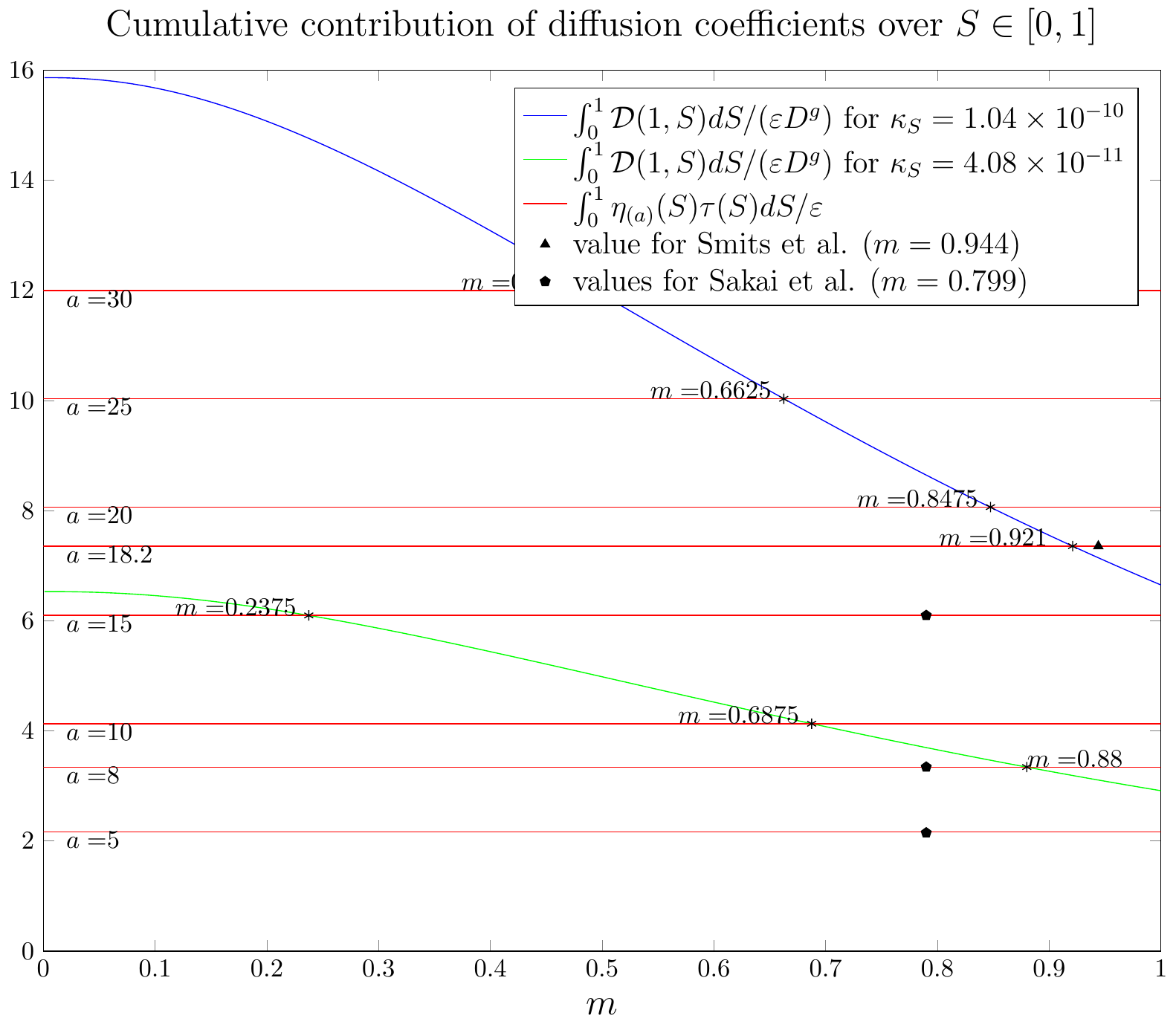}
            \end{center}
            \caption{The blue and green curves show the left-hand side of equation
                \eqref{eqn:heuristic_1} for different saturated permeabilities, and the
            red lines show level curves for right-hand side for various values of $a$. The
        blue and green curves can be used to predict the value of $a$ before
    experimentation.}
            \label{fig:heuristic_1}
        \end{figure}
\renewcommand{\baselinestretch}{\normalspace}
\linespread{1.0}
\begin{table}[ht!]
    \centering
    \begin{tabular}{|c|cc|}
        \hline
                    & $a$ measured  & $a$ predicted from \eqref{eqn:heuristic_1} \\ \hline \hline
        Smits et al. & 18.2      & $\approx 18$ \\ 
        Sakai et al. & 8         & $\approx 9$ \\ \hline
    \end{tabular}
    \caption{Measured and predicted value of the fitting parameter $a$ based on equation
        \eqref{eqn:heuristic_1}.}
    \label{tab:heuristic_1}
\end{table}
\renewcommand{\baselinestretch}{\normalspace}

Comparisons with the experiments of Smits et al.\,and Sakai et al.\,indicate that, while
not perfect, the present model gives a diffusion equation that matches experimental
findings reasonably well without the necessity of an a-posteriori fitting parameter.

\section{Coupled Saturation and Vapor Diffusion}\label{sec:numerical_coupled_sat_vap}
In this section we couple the saturation and vapor diffusion equations under reasonable
boundary conditions.  This is done while holding the temperature fixed.  The purpose of
this short study is to determine the roles of the mass transfer, $\ehat{l}{g_v}$, the $C_\rh^l \partial \rh /
\partial x$ term appearing in the saturation equation, and the time
rate of change of saturation that appears in the vapor diffusion equation.  The equations
are restated here for reference.
\begin{subequations}
    \begin{flalign}
        \pd{S}{t} - \pd{ }{x} \left( \porosity_S^{-1} K(S) \left( -p_c'(S) \pd{S}{x} + \tau
        \porosity_S \pddm{S}{x}{t} + C_\rh^l \pd{\rh}{x} - \rho^l g \right) \right) &=
        \frac{\ehat{l}{g_v}}{\porosity_S \rho^l} \\
        (1-S) \pd{\rh}{t} - \rh \pd{S}{t} - \pd{ }{x} \left( \porosity_S^{-1} \mathcal{D}(\rh,S)
        \pd{\rh}{x} \right) &= \frac{-\ehat{l}{g_v}}{\porosity_S \rho_{sat}} 
    \end{flalign}
\end{subequations}
where
\begin{flalign*}
    \ehat{l}{g_v} = M \rh \left( \rho^l - \rho_{sat} \rh \right) \left(
    \frac{-p_c + \tau \porosity \dot{S} - p_0^l}{\rho^l} - R^{g_v} T \ln(\lambda \rh)
    \right).
\end{flalign*}

This is a system of advection-diffusion-reaction equations with a pseudo-parabolic damping
term in saturation (for $\tau \ne 0$) and no advection term in the relative humidity
equation.  
To judge the relative affect of $C_\rh^l$ compare the rate of movement of the water
through the liquid phase with the rate of movement of water in the gas phase. As such, we consider the ratio of the coefficient of this
term to the saturation diffusion
\begin{flalign} 
    \frac{C_\rh^l}{-p_c'} = \left( \frac{C_\rh^l}{\rho^l g} \right) Pe = \left(
    \frac{C_\rh^l \al m}{\rho^l g (1-m)}
    \right) S^{1+1/m}(S^{-1/m}-1)^m. \label{eqn:Crhl_Pe}
\end{flalign}
In the (unlikely) case that ratio \eqref{eqn:Crhl_Pe} is approximately 1 then the
diffusion in relative humidity has equal effect as the diffusion in saturation in
controlling the transient nature of the saturation.  This does not fit with our physical
experience so we conjecture that the ratio is much smaller.  Figures
\ref{fig:Coupled_SatVap_H03_al4_n3} show time
snapshots of an imbibition experiment with simultaneous vapor diffusion and (temperature
independent) evaporation.  Both saturation and relative humidity are controlled with fixed
Dirichlet boundary conditions and initial profiles consistenti with imbibition into a low
saturation column.
\linespread{1.0}
\begin{figure}[H]
        \centering
        \subfigure[Comparison at $t=t_1$]{
            \includegraphics[width=0.45\textwidth]{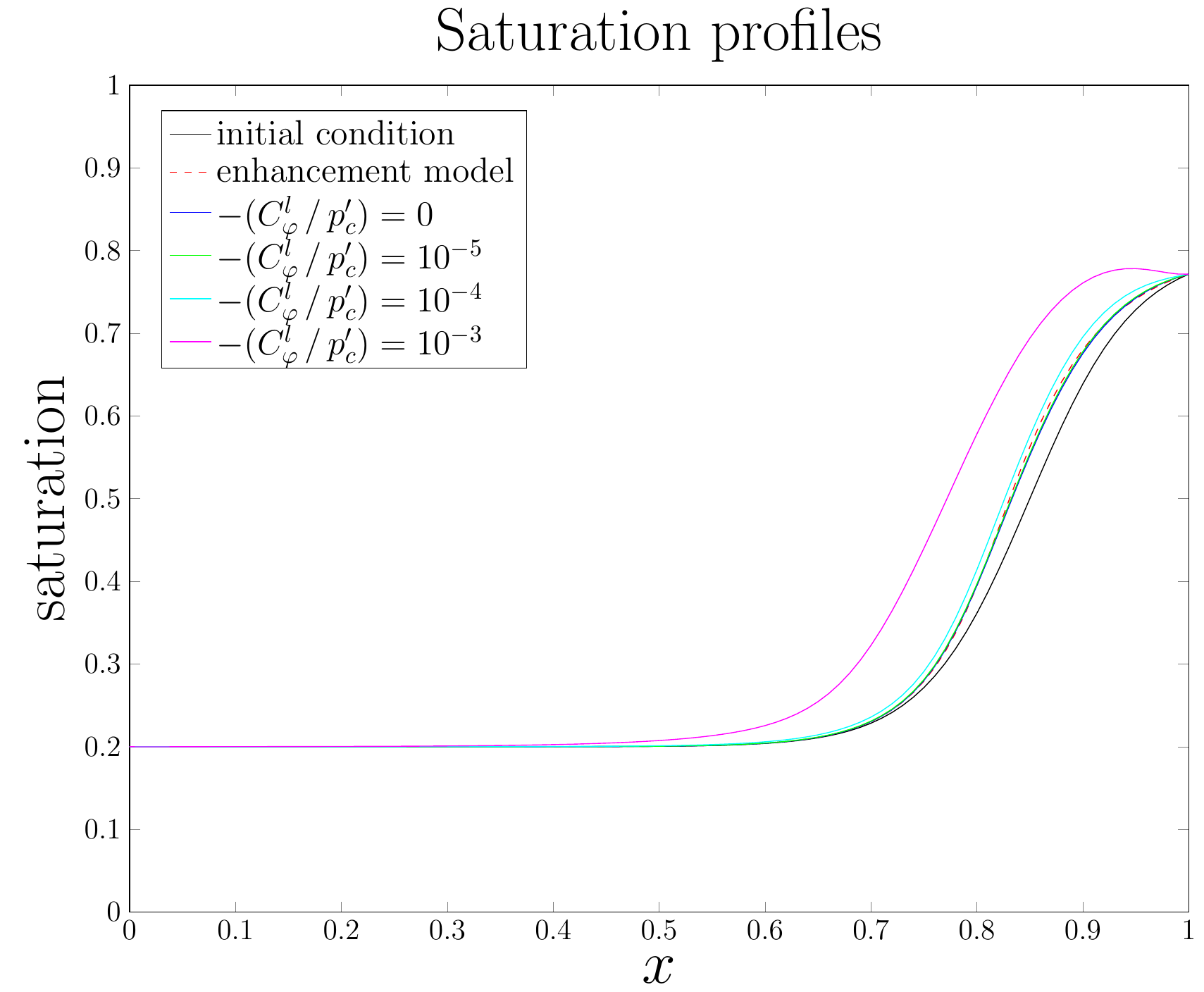}
            \label{fig:Coupled_SatVap_S_t001_H03_al4_n3}
            }
        \subfigure[Comparison at $t=t_2$]{
            \includegraphics[width=0.45\textwidth]{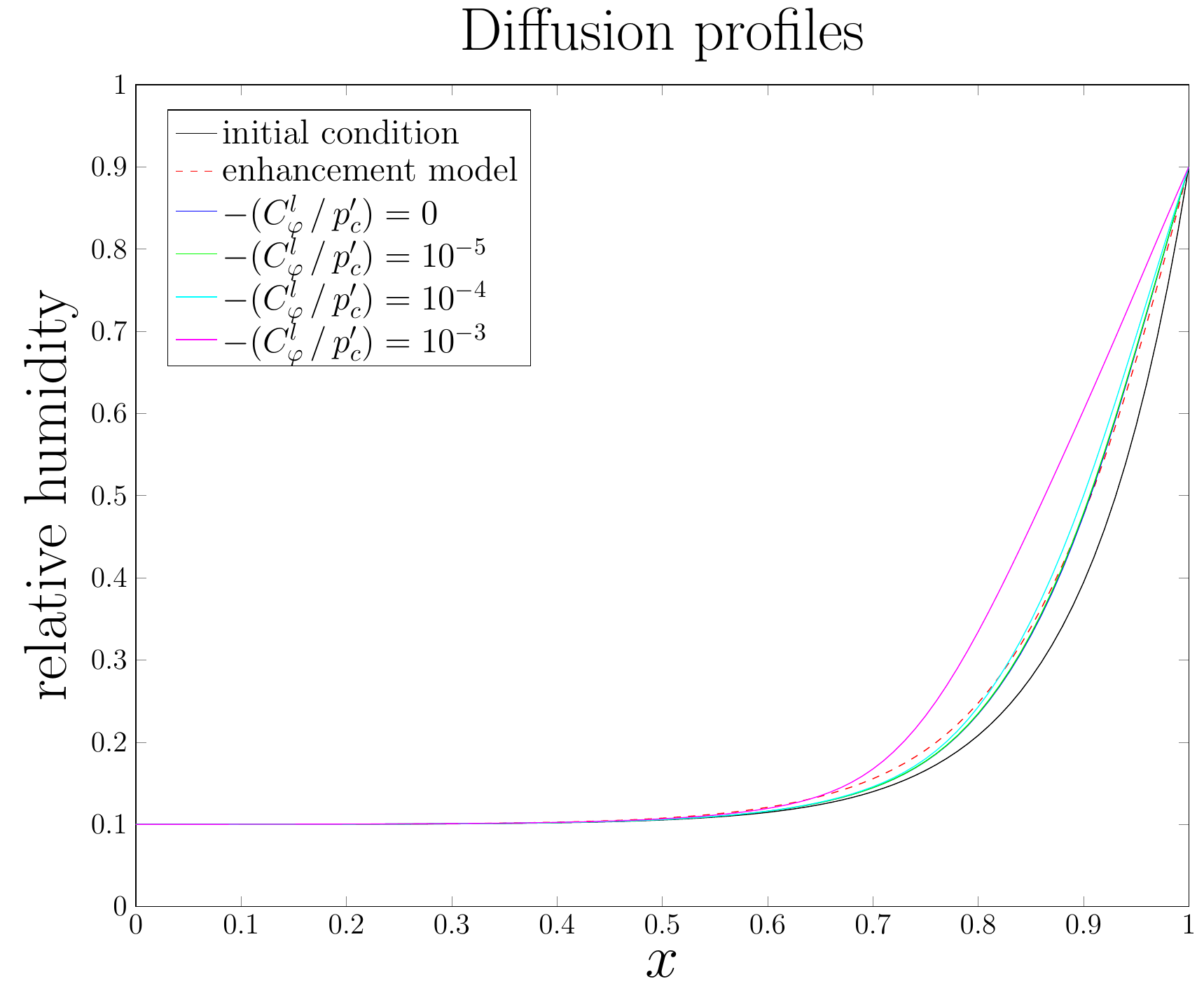}
            \label{fig:Coupled_SatVap_RH_t001_H03_al4_n3}
        }
        \subfigure[Comparison at $t=t_3$]{
            \includegraphics[width=0.45\textwidth]{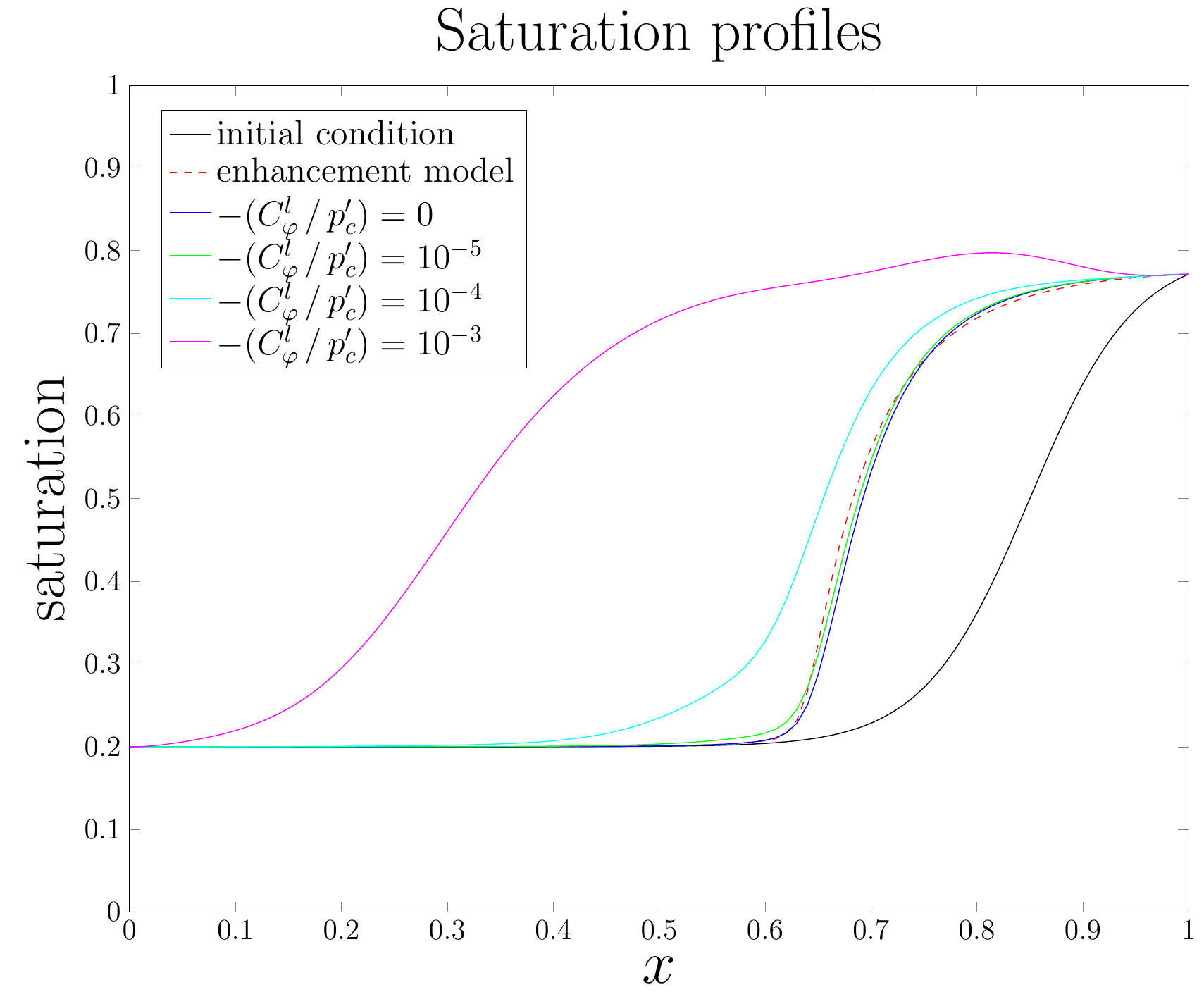}
            \label{fig:Coupled_SatVap_S_t010_H03_al4_n3}
            }
        \subfigure[Comparison at $t=t_4$]{
            \includegraphics[width=0.45\textwidth]{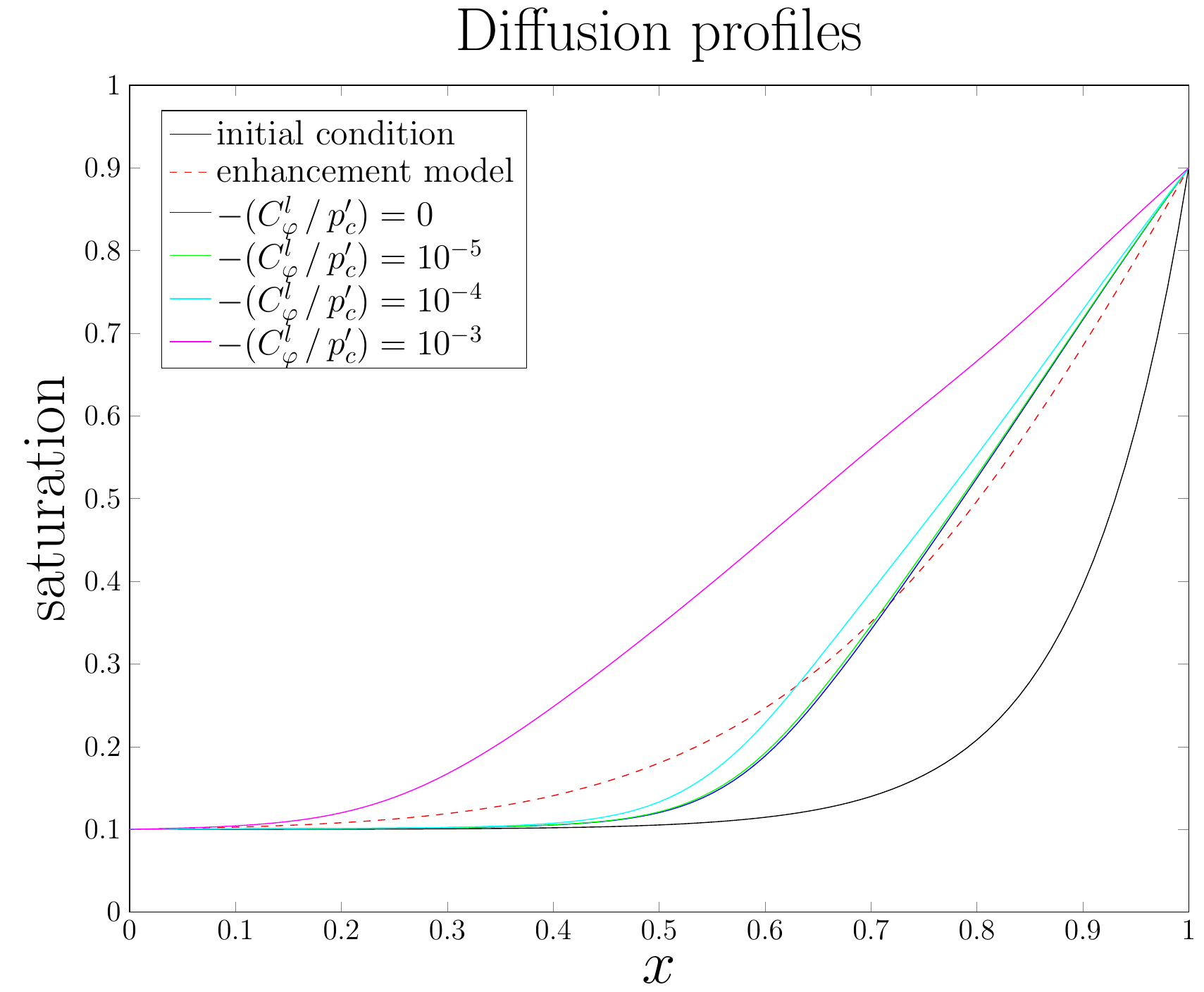}
            \label{fig:Coupled_SatVap_RH_t010_H03_al4_n3}
        }
        \caption{Comparison of coupled saturation-diffusion models for various weights of
            $C_\rh^l$ with parameters: $\kappa_s = 1.04 \times 10^{-10}, \porosity =
            0.334$, $H_0 = 10^{-3}, \alpha = 4$, and $m = 0.667$.}
        \label{fig:Coupled_SatVap_H03_al4_n3}
\end{figure}
\renewcommand{\baselinestretch}{\normalspace}

From Figures \ref{fig:Coupled_SatVap_H03_al4_n3}, if ratio \eqref{eqn:Crhl_Pe} is greater
than or equal to $10^{-3}$ then non-physical results are observed for this particular set
of initial boundary conditions. As there are infintiely many sets of initial boundary
conditions we only present this one particular case as a proof of concept. In general we
observe that this ratio must be kept below $10^{-3}$. With a ratio this small we are
simply saying that the enhancement in saturation seen due to increased levels of relative
humidity have a small affect as compared to gradients in capillary pressure in the case of
fixed temperature.

\section{Coupled Heat and Moisture Transport System}\label{sec:FullyCoupledSolutions}
In this final section we examine the fully coupled system of saturation, vapor diffusion,
and heat transfer. The culminating goal of this section is to compare the numerical
solution of the present model with the experimental results associated with
\cite{Smits2011}.  Dr. Smits was generous enough to share the experimental results for
this comparison. In Section \ref{sec:column_experiment}, the physical apparatus is
discussed as well as material parameters and initial boundary conditions.  In Section
\ref{sec:numerical_all_coupled} the full system is solved numerically and compared to the
experimental data.

\subsection{Experimental Setup, Material Parameters, and IBCs}\label{sec:column_experiment}
The experiment of interest is to track temperature, relative humidity, and saturation in a
column of packed sand. Soil moisture, relative humidity, and temperature sensors were
placed throughout a 111cm column of packed sand. A heat source was turned on and off above
the surface of the soil (to simulate natural temperature cycles). The goal of Smits et
al.\,was to determine whether the equilibrium assumption between phases was valid in
porous media evaporation studies.  For the our purposes we use this data simply as a
validation of the present modeling effort. 

A schematic of the experimental apparatus used in Smits et al.\,\cite{Smits2011} is shown
in Figure \ref{fig:KateSchematic} (recreated from Dr. Smits' notes).  Saturation and
temperature sensors \#1 - \#10, are placed every 10cm from the bottom.  Saturation
and temperature sensor \#11 is 1cm under the surface of the sand. Sensor \#12 is 10cm
above the surface. Sensor \#13 is on the surface (``in good contact''). Temperature
sensors \#14 and \#15 are placed within the insulation surrounding the apparatus (to
measure the lateral heat loss (see the top view in Figure \ref{fig:KateSchematic})).
Relative humidity sensor \#1 is 1cm under the surface, and sensor \#2 is on the surface.
The gray shaded area in Figure \ref{fig:KateSchematic} represents the location of the soil
pack. The initial water level is the surface of the soil pack.  The spatial variable to be
used numerically is $x \in [0,1]$ where $x=0$ represents the cool end of the apparatus and
where $x=1$ represents the surface of the soil 111cm above the cool end.  The material
properties used in this experiment are shown in Table
\ref{tab:SmitsMaterial}.  

\linespread{1.0}
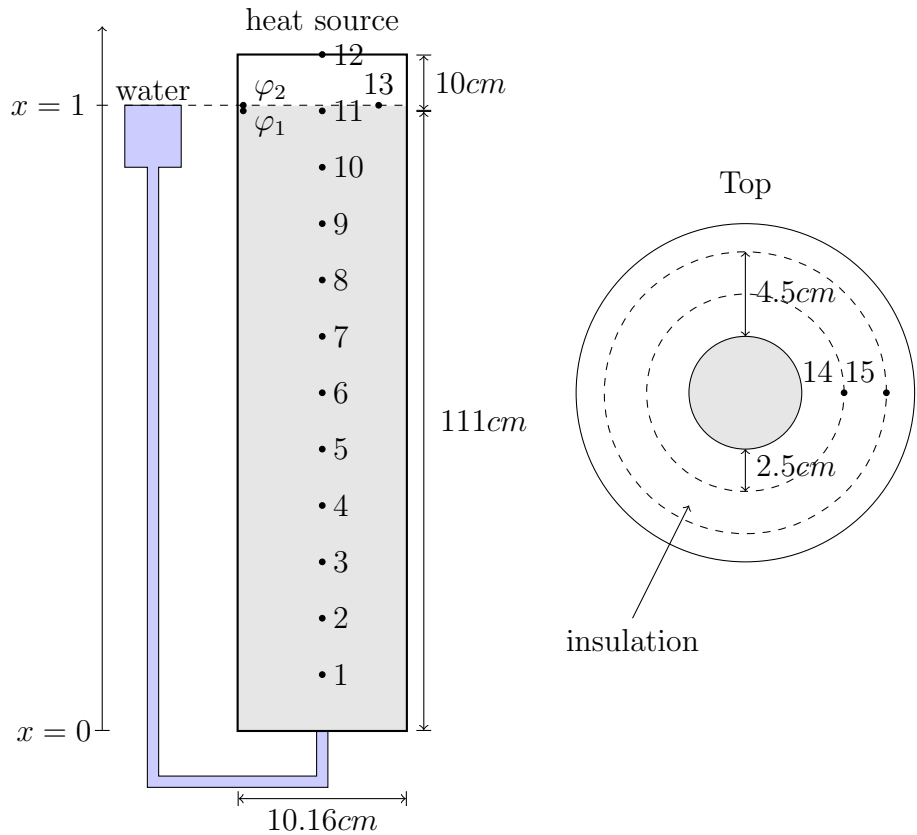
\begin{figure}[ht*]
    \begin{center}
        \begin{tikzpicture}[scale=0.75]
            \draw[color=white,fill=gray!20] (0,0) -- (0,11.1) -- (3,11.1) -- (3,0) -- cycle;
            \draw[fill=blue!20] (1.4,0) -- (1.6,0) -- (1.6,-1) -- (-1.6,-1) -- (-1.6,10)
            -- (-2,10) -- (-2,11.1) -- (-1,11.1) -- (-1,10) -- (-1.4,10) -- (-1.4,-0.8) --
            (1.4,-0.8) -- cycle;
            \draw (-1.5,11) node[anchor=south]{water};
            \draw[thick] (0,0) -- (0,12) -- (3,12) -- (3,0) -- cycle;
            \draw[dashed] (-2,11.1) -- (3,11.1);
            \draw[|<->|] (3.3,11) -- (3.3,12); \draw (3.3,11.5) node[anchor=west]{$10cm$};
            \draw[|<->|] (3.3,0) -- (3.3,11); \draw (3.4,5.5) node[anchor=west]{$111cm$};
            \draw[|<->|] (0,-1.2) -- (3,-1.2); 
            \draw (1.5,-1.2) node[anchor=north]{$10.16cm$};
            \foreach \i in {1,...,12}
            {
                \draw[fill=black] (1.5,\i) circle(0.05cm) node[anchor=west]{\i};
            }
            \draw[fill=black] (0.1,11.0) circle(0.05cm) node[anchor=south west]{$\rh_2$};
            \draw[fill=black] (0.1,11.1) circle(0.05cm) node[anchor=north west]{$\rh_1$};
            \draw[fill=black] (2.5,11.1) circle(0.05cm) node[anchor=south]{13};
            \draw (9,6) circle(3cm);
            \draw[fill=gray!20] (9,6) circle(1.0cm);
            \draw[fill=black] (10.75,6) circle(0.05cm) node[anchor=south east]{$14$};
            \draw[dashed] (9,6) circle(1.75cm);
            \draw[fill=black] (11.5,6) circle(0.05cm) node[anchor=south east]{$15$};
            \draw[dashed] (9,6) circle(2.5cm);
            \draw[<->] (9,5) -- (9,4.25); \draw (9,4.7) node[anchor=west]{$2.5cm$};
            \draw[<->] (9,7) -- (9,8.5); \draw (9,7.8) node[anchor=west]{$4.5cm$};
            \draw[->] (7,2) node[anchor=north]{insulation} -- (8,4);
            \draw (9,9.25) node[anchor=south]{Top};
            \draw (1.5,12.25) node[anchor=south]{heat source};
            \draw[<-|] (-2.4,12.5) -- (-2.4,0) node[anchor=east]{$x=0$};
            \draw (-2.3,11.1) -- (-2.5,11.1) node[anchor=east]{$x=1$};
        \end{tikzpicture}
    \end{center}
    \caption{Schematic of the Smits et al.\,experimental apparatus. Saturation and
        temperature sensors numbered 1 - 11, temperature sensors 12 - 15, and relative
        humidity sensors 1 and 2 \cite{Smits2011}. The geometric $x$ coordinate is shown
        on the left. (Image recreated with permission from \cite{Smits2011})}
    \label{fig:KateSchematic}
\end{figure}
\renewcommand{\baselinestretch}{\normalspace}

\linespread{1.0}
\begin{table}[ht*]
    \centering
    \begin{tabular}{| c | c | c |}
        \hline
        Parameter & Value & Units \\ \hline \hline
        Sand Number & 30/40 & [$-$] \\ 
        Dry Bulk Density & 1.77 & [g  cm$^{-3}$] \\
        Porosity & 0.318 & [$-$] \\
        Residual Water Content & 0.028 & [$-$] \\
        Saturated Hydraulic Conductivity & 0.104 & [cm \, s$^{-1}$] \\
        van Genuchten $\alpha$ & 5.7 & [m$^{-1}$] \\
        van Genuchten $n$ ($m=1-1/n$) & 17.8 (0.9438) & [$-$] \\
        \hline
    \end{tabular}
    \caption{Material parameters for experimental setup \cite{Smits2011}.}
    \label{tab:SmitsMaterial}
\end{table}
\renewcommand{\baselinestretch}{\normalspace}
The experiment was run for 32 days, at which point there was a power outage and the
experiment was stopped. In the midst of the experiment there were two sensors that failed:
saturation sensor \#3 (after the 1847$^{th}$ measurement ($t>12.8$days)), and relative
humidity sensor \#1 (after the $2155^{th}$ measurement ($t>14.9$days)) (see Figure
\ref{fig:BrokenSensors}). The saturation sensors are accurate to within $\pm 2\%$ soil
moisture content after soil calibration (performed by Smits et al.). The relative humidity
sensor accuracy ranges between $\pm 2 \%$ (for mid-range temperatures and humidities) and
$\pm 12 \%$ (for extreme temperatures and humidities). The temperature sensors are accurate to
within $0.5^\circ$C for the temperature ranges of interest (\verb|www.decagon.com|). 
\linespread{1.0}
\begin{figure}[ht]
    \begin{center}
        \includegraphics[width=0.65\columnwidth]{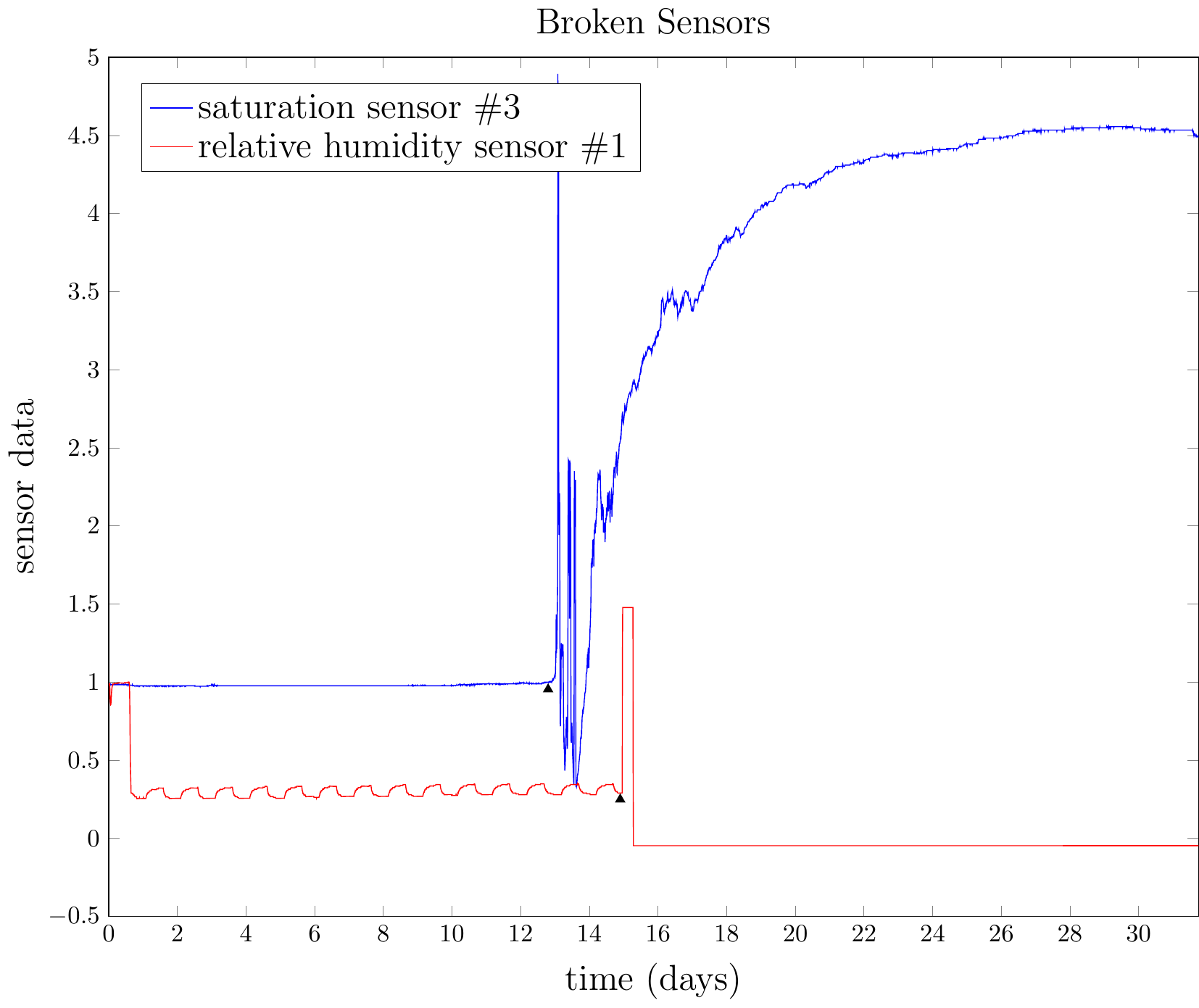}
    \end{center}
    \caption{Broken sensor data. Saturation sensor \#3 shown in blue and relative humidity
    sensor \#1 shown in red. It is evident from these plots that these sensors are not
working properly as they give non-physical readings. (Image recreated with permission from \cite{Smits2011})}
    \label{fig:BrokenSensors}
\end{figure}
\renewcommand{\baselinestretch}{\normalspace}

The initial and boundary conditions for the forthcoming numerical experiments can be taken
from any point within the data set. The logical initial point for the numerical
experiment is the beginning of the physical experiment. This particular point is of
interest to the experimentalist as some of the interesting transient behavior occurs
during this period.  That being said, there is a
significant amount of sensor noise in the initial phases of the experiment (see Figure
\ref{fig:InitialRelHumData}), and if a simple {\it proof of concept} is all that is needed
for the purposes of this work, then a later time is preferred so as to avoid complications
related to this noise.
\linespread{1.0}
\begin{figure}[H]
    \centering
    \subfigure[Initial sensor data for relative humidity]{
        \includegraphics[width=0.47\columnwidth]{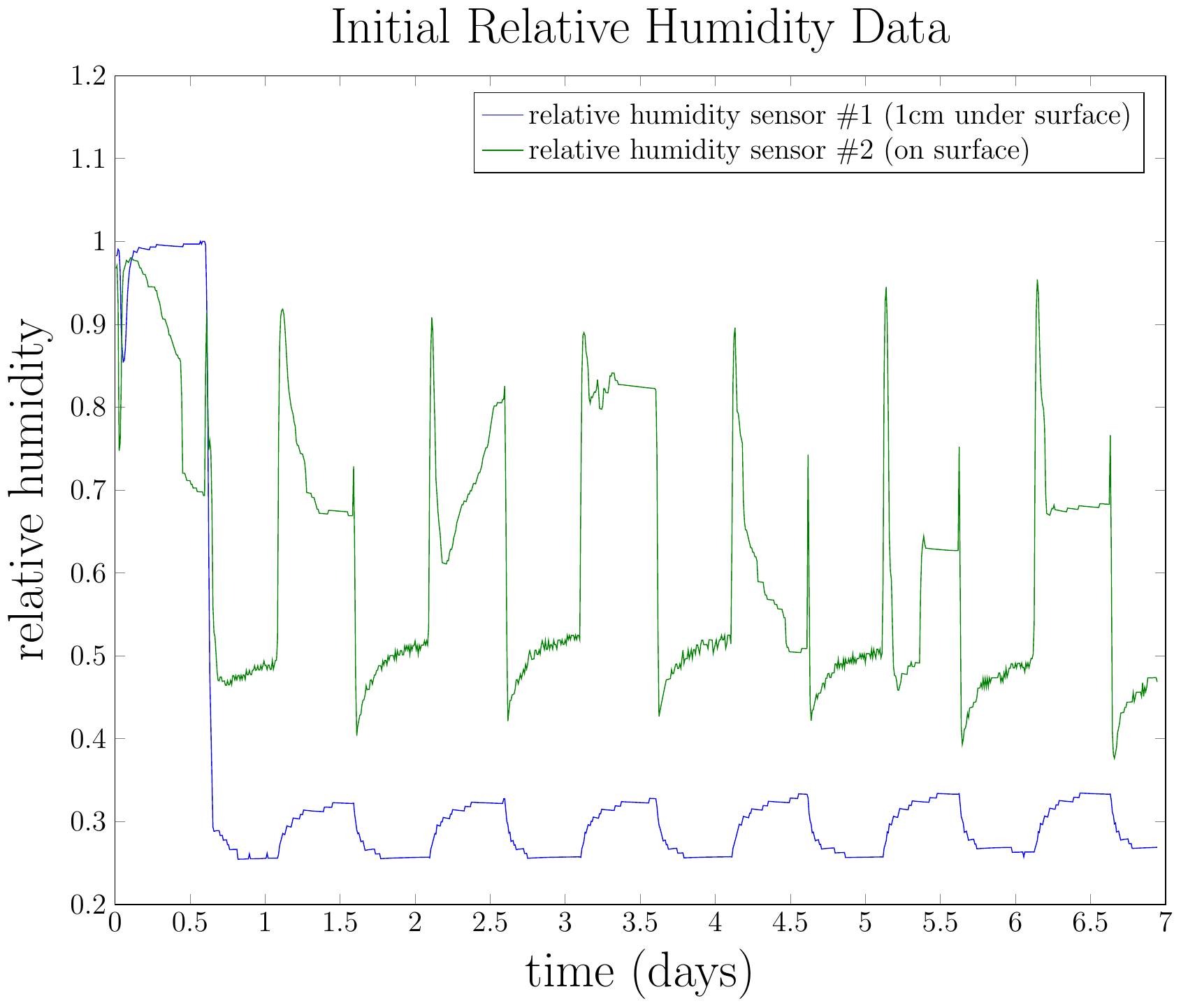}
        \label{fig:InitialRelHumData}
    }
    \subfigure[Initial sensor data for temperature]{
        \includegraphics[width=0.47\columnwidth]{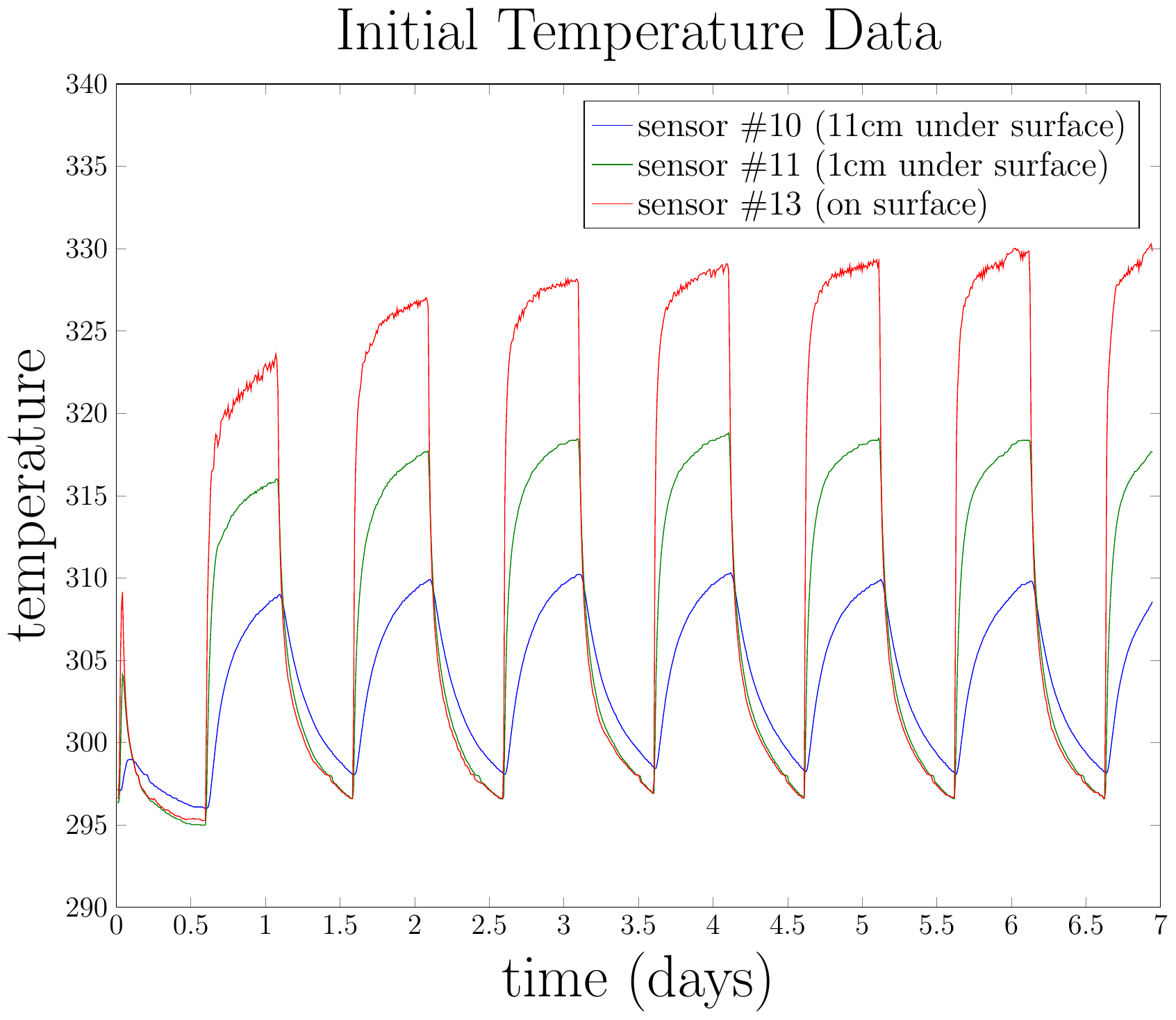}
        \label{fig:InitialTempData}
    }
    \caption{Relative humidity and temperature data showing measurement variations in the
    first few days of the experiment. (Image recreated with permission from \cite{Smits2011})}\label{fig:initial_rh_temp_data}
\end{figure}
\renewcommand{\baselinestretch}{\normalspace}
The section of data where we will initially focus is between time measurements 1800 (12.5
days) and 2150 (14.9 days). This section of data is chosen since, qualitatively, it shows
the least amount of sensor noise in both relative humidity and temperature.  Saturation
sensor \#3 is faulty in this time region, but the adjacent sensors indicate that there is
little to no deviation from full saturation for these times. The relative humidity and
temperature data for this time region are shown in Figure \ref{fig:SensorDataWindow}.
\linespread{1.0}
\begin{figure}[H]
    \centering
    \subfigure[Sensor data for relative humidity]{
        \includegraphics[width=0.47\columnwidth]{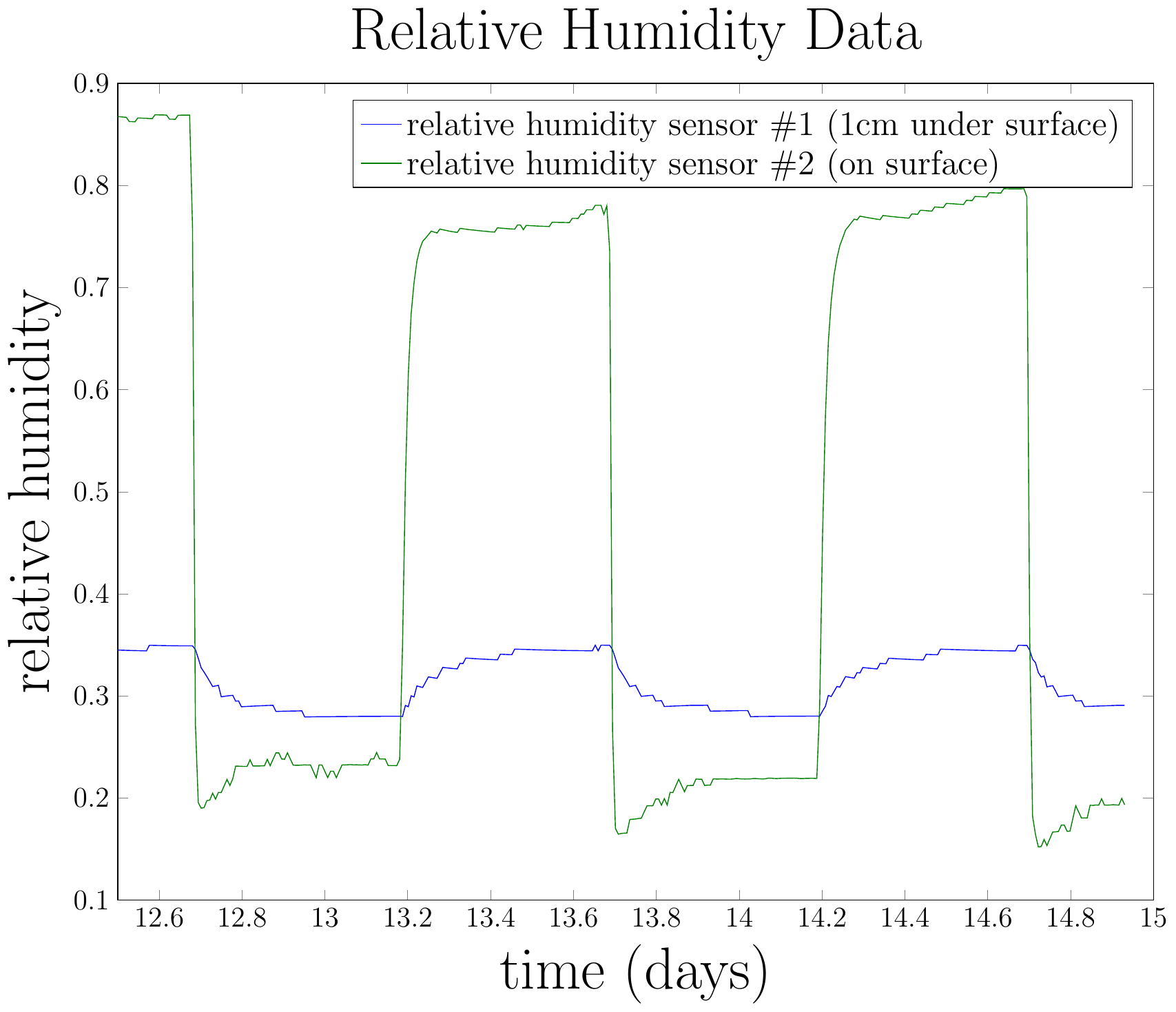}
        \label{fig:SensorDataWindowRH}
    }
    \subfigure[Sensor data for temperature]{
        \includegraphics[width=0.47\columnwidth]{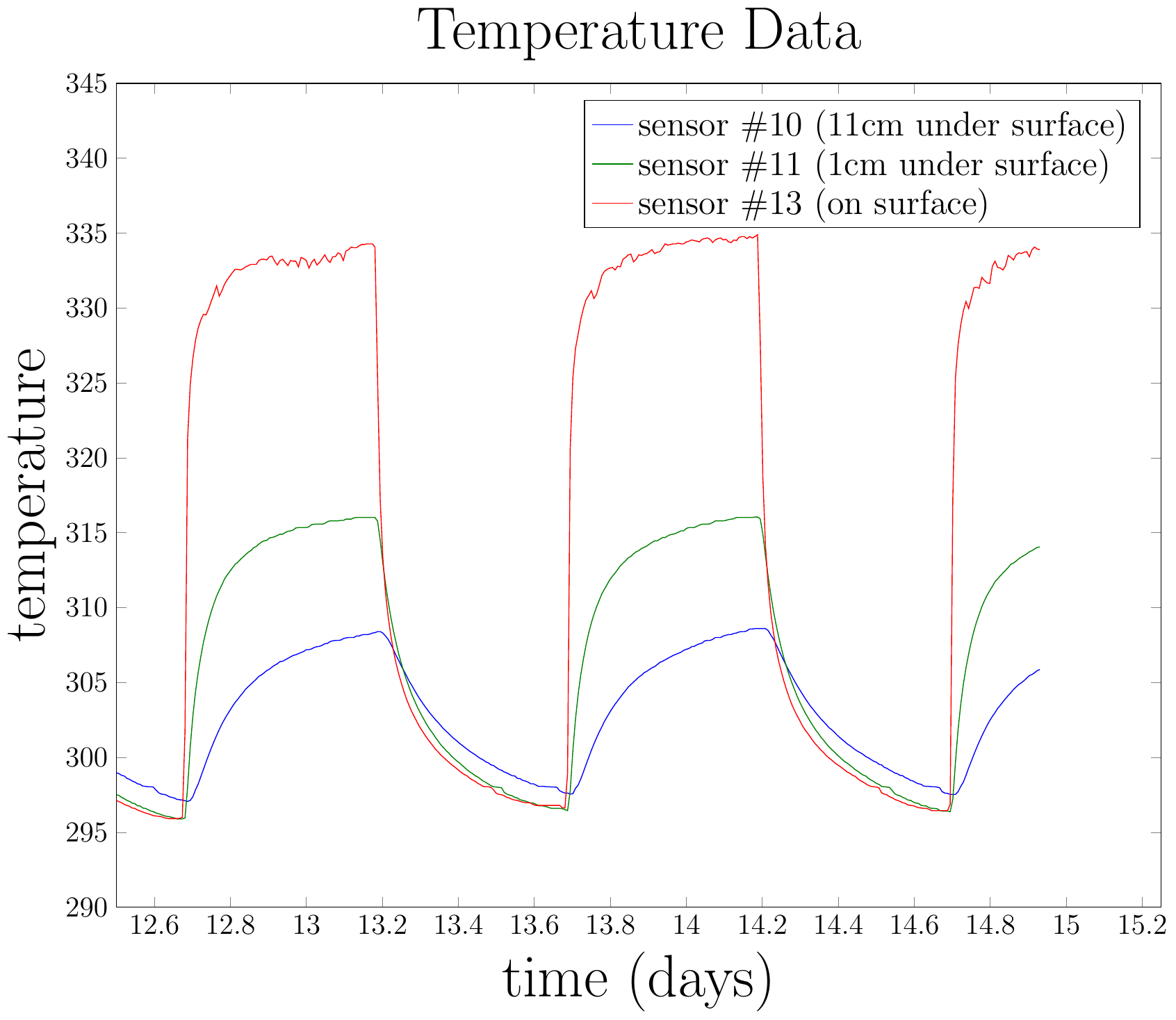}
        \label{fig:SensorDataWindowTemp}
    }
    \caption{Relative humidity and temperature data at a window beginning roughly 12.5
    days into the experiment.  This window is chosen since the sensor noise is
qualitatively minimal in this region. (Image recreated with permission from \cite{Smits2011})}\label{fig:SensorDataWindow}
\end{figure}
\renewcommand{\baselinestretch}{\normalspace}

The peaks and valleys of the temperature and relative humidity data (associated with the
on-off cycle of the heat lamp) have small variations that are likely due to sensor noise.
To avoid modeling this noise directly we can approximate the data with either a simple
sinusoidal function or a square wave approximation (found by applying the \texttt{sign} function
to the sinusoidal approximation).  The data suggests a square wave approximation, but
the jumps in data may cause numerical difficulties as the derivatives at the points of
discontinuity are technically delta functionals. A graphic of these approximations is
shown in Figure \ref{fig:BC_Sinusoidal_Square}.
\linespread{1.0}
\begin{figure}[H]
    \centering
    \subfigure[relative humidity]{
        \includegraphics[width=0.47\columnwidth]{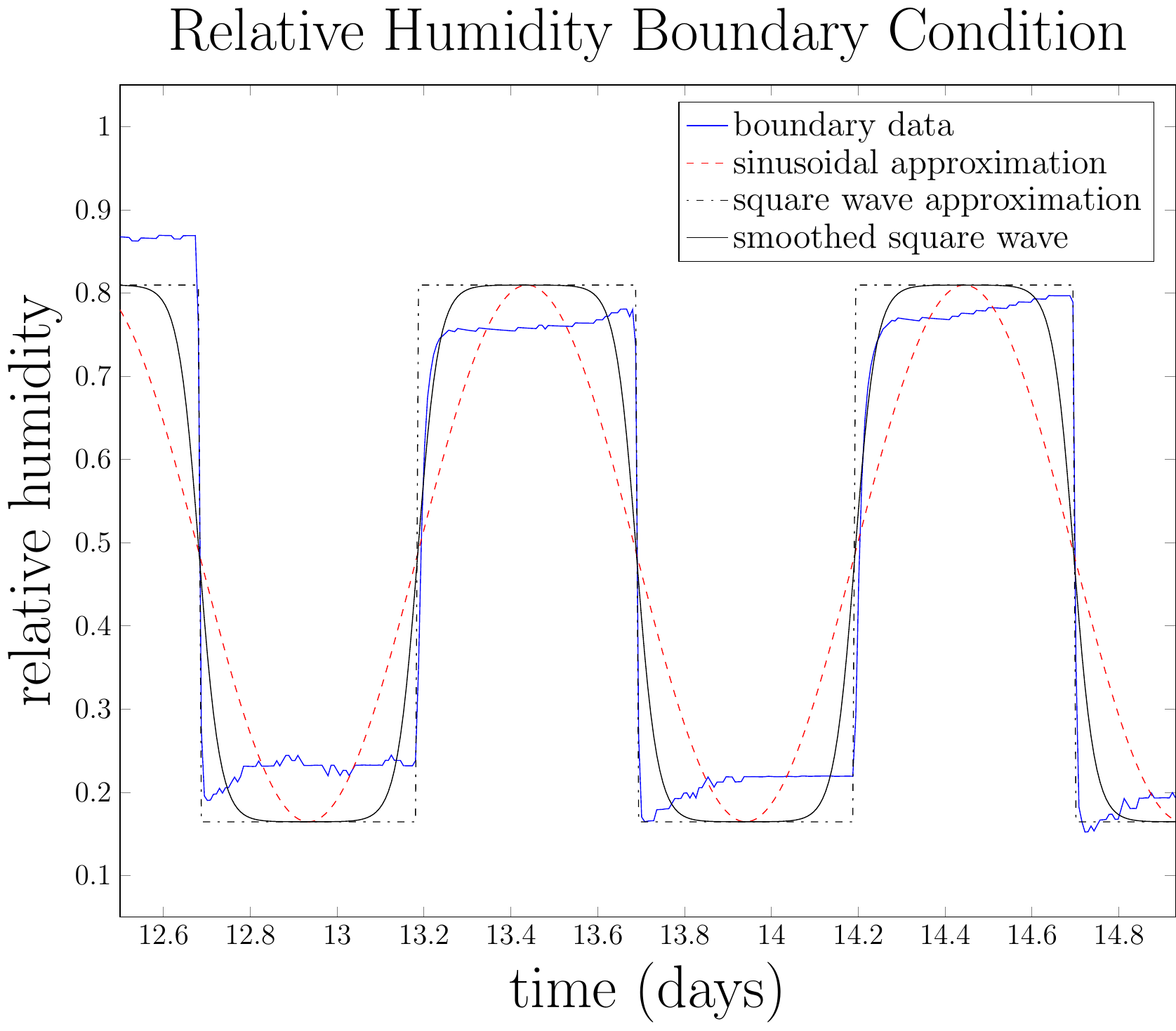}
        \label{fig:RH_BC_Sinusoidal}
    }
    \subfigure[temperature]{
        \includegraphics[width=0.47\columnwidth]{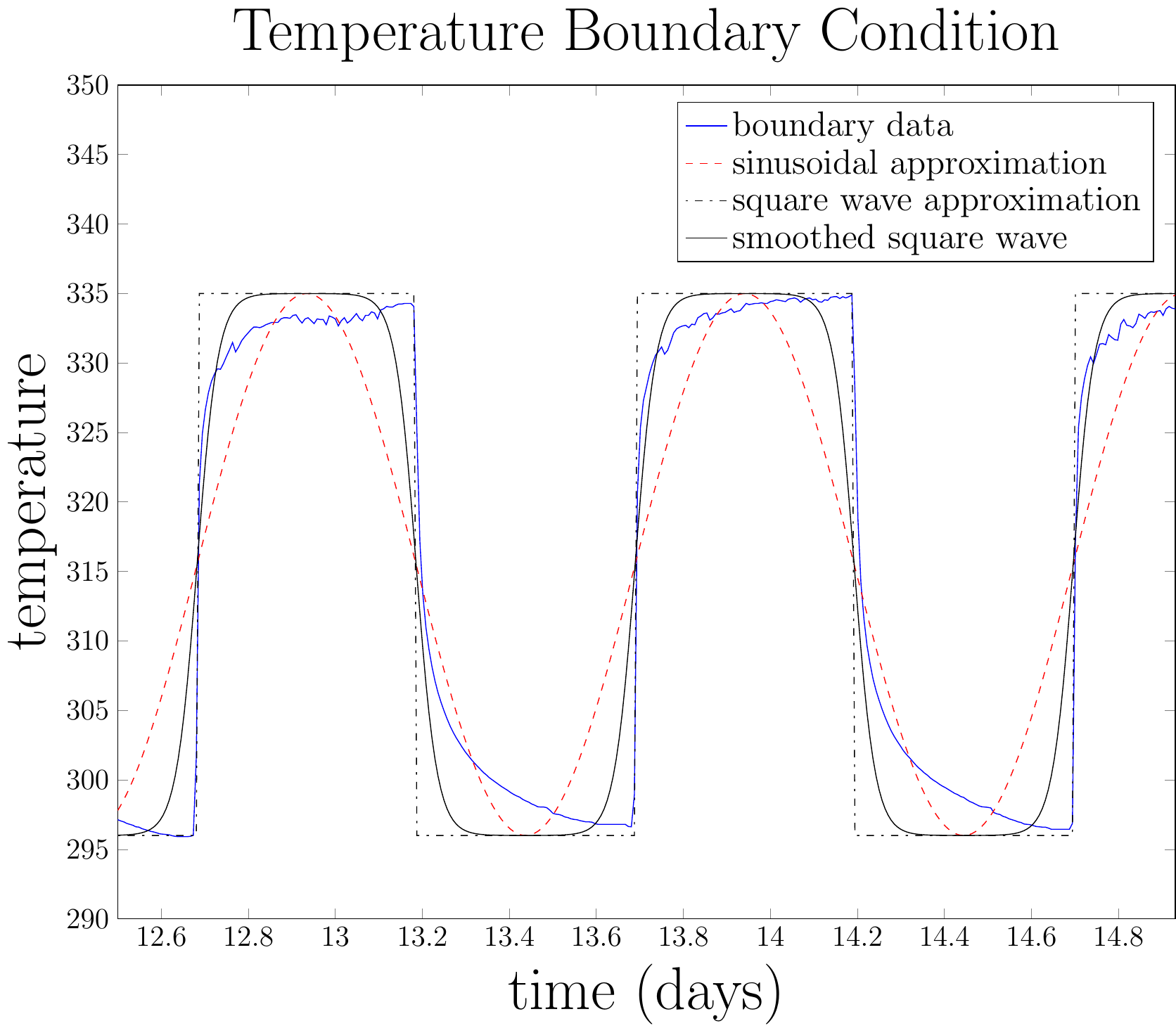}
        \label{fig:T_BC_Sinusoidal}
    }
    \caption{Approximations to relative humidity and temperature boundary
    conditions at the surface of the soil.}\label{fig:BC_Sinusoidal_Square}
\end{figure}
\renewcommand{\baselinestretch}{\normalspace}

Any starting point can be taken within this window of time. We choose the 2000$^{th}$ time step as the initial condition (somewhat
arbitrarily) and fit functions to the coarse spatial data for saturation, relative
humidity, and temperature.  For the relative humidity and saturation profiles we choose
hyperbolic tangent functions since they exhibit the primary features observed in the data
(see Figures \ref{fig:InitialCond_Approx2000_RH} and  \ref{fig:InitialCond_Approx2000_Sat}
respectively). For the temperature initial condition we choose an exponential function
(see Figure \ref{fig:InitialCond_Approx2000_Temp}).
\linespread{1.0}
\begin{figure}[ht!]
    \centering
    \subfigure[relative humidity]{
        \includegraphics[width=0.47\columnwidth]{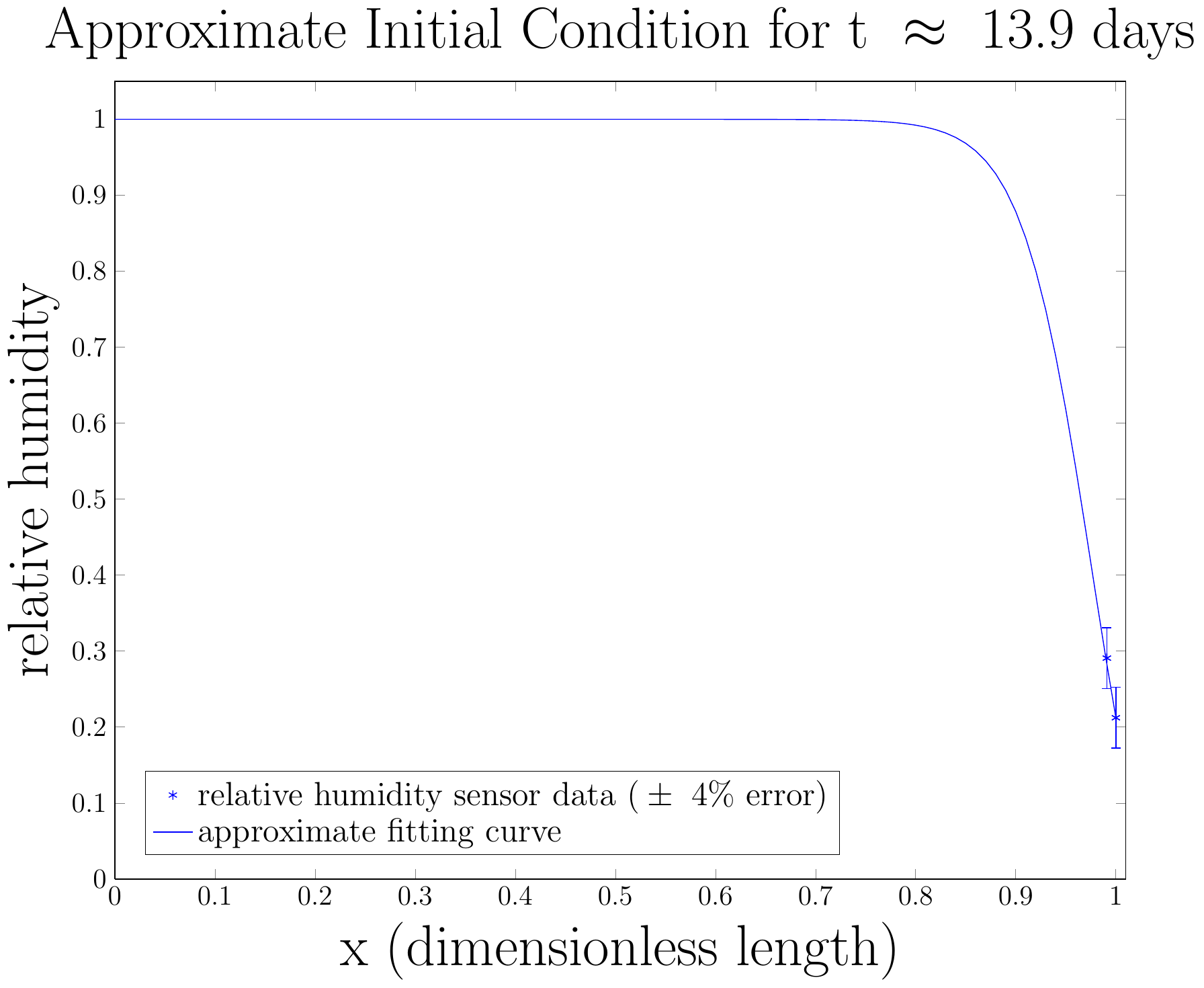}
        \label{fig:InitialCond_Approx2000_RH}
    }
    \subfigure[saturation]{
        \includegraphics[width=0.47\columnwidth]{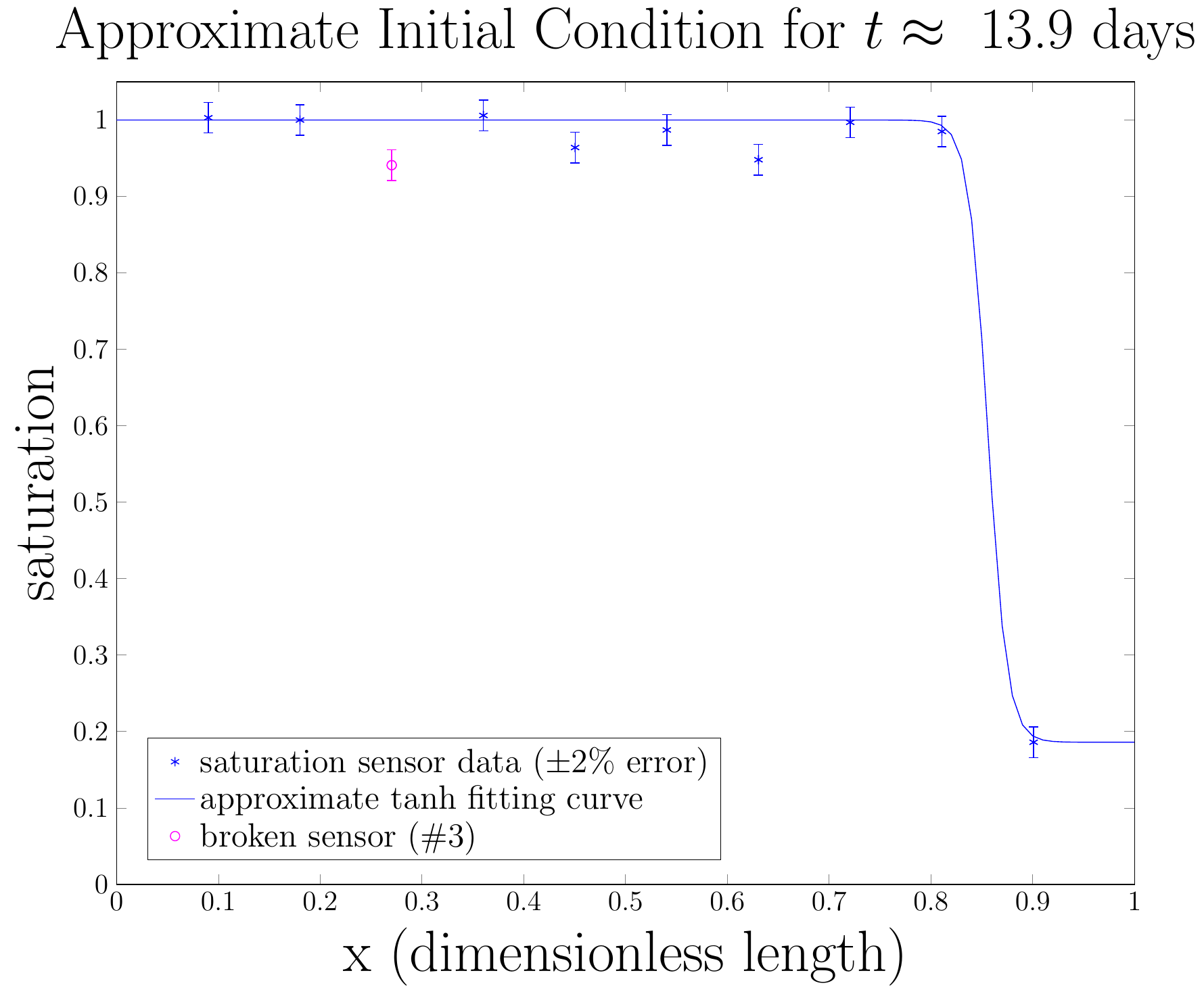}
        \label{fig:InitialCond_Approx2000_Sat}
    }
    \subfigure[temperature]{
        \includegraphics[width=0.47\columnwidth]{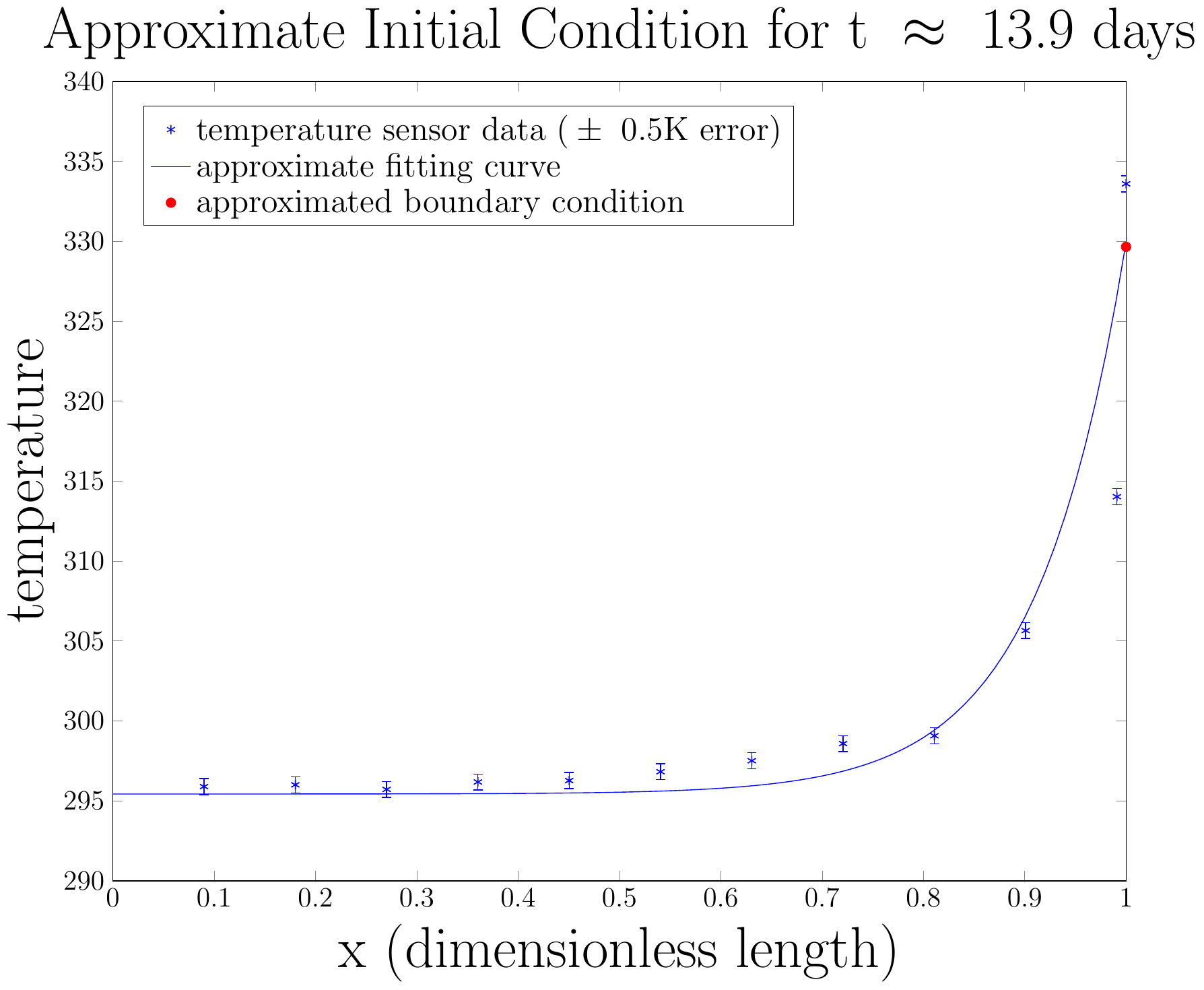}
        \label{fig:InitialCond_Approx2000_Temp}
    }
    \caption{Approximate initial conditions at the 2000$^{th}$ data point ($t \approx
    13.9$ days). Error bars indicate approximate sensor accuracy.}\label{fig:InitialCond_Approx2000}
\end{figure}
\renewcommand{\baselinestretch}{\normalspace}

To summarize, thus far we have boundary conditions for relative humidity and temperature
at $x=1$ and we have initial conditions for all of the variables.  The boundary conditions
at $x=0$ are much simpler.  For saturation and relative humidity we can take $S(x=0,t)=1$
and $\rh(x=0,t)=1$ based on the fact that the saturation is fixed mechanically at 100\% at
the bottom end of the apparatus.  For the temperature we can either take $T(x=0,t) = T_0$
or $\partial T / \partial x (x=0,t) = 0$.  The Dirichlet condition simply states that the
temperature is fixed, and the Neumann condition states that the bottom of the apparatus is
insulated so that no heat is lost. Throughout the course of this experiment, the
thermal effects are not appreciably translated to the bottom of the apparatus so either
boundary condition would be sufficient.    Finally, the only boundary condition remaining is the
saturation condition at the surface of the soil.  A simple condition is to state that the
flux of liquid is zero across this boundary.  Mathematically, this translates to the
Neumann condition $\partial S / \partial x (x=1,t) = 0$.

A fine point needs to be stated regarding the relative humidity equation.  The saturation
initial condition states that much of the experimental apparatus is completely saturated
with liquid water (below sensor \#9 approximately).  The issue is that there is no gas
phase present when $S=1$. Mathematically this translates to a Stefan-type problem where
the lower boundary for the gas phase is actually moving spatially as the liquid water
evaporates.  For the sake of illustration let us simply assume for a moment that the
saturation equation is a hyperbolic linear advection equation where the front simply {\it
advects} in time.  Figure \ref{fig:MovingInterface} illustrates how the gas-phase domain
might evolve in time in this simplified example. Of course, the saturation equation is not
such a simple equation but the essence of the moving domain is the same regardless.
\linespread{1.0}
\begin{figure}[ht*]
    \begin{center}
        \includegraphics[width=0.7\columnwidth]{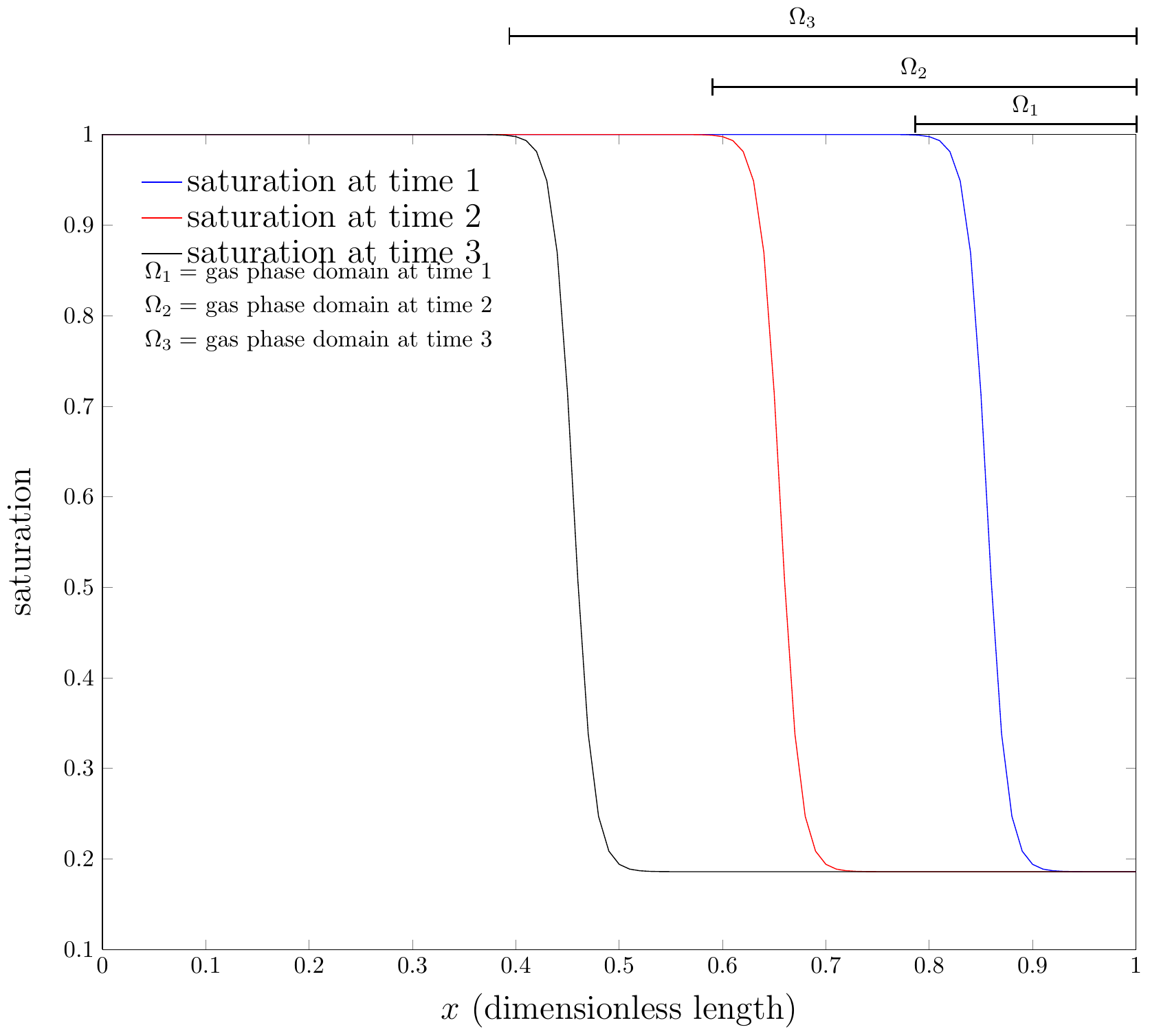}
    \end{center}
    \caption{Illustration of how the gas-phase domain might evolve in time.}
    \label{fig:MovingInterface}
\end{figure}
\renewcommand{\baselinestretch}{\normalspace}

One way to model this Stefan problem is to assume that when $S = 1$ then $\rh = 1$. Under
this assumption we can define the gas-phase equation as a piecewise-defined differential
equation:
\begin{flalign}
    \pd{\rh}{t} = \left\{ \begin{array}{cl} 0, & S = 1 \\ \displaystyle \frac{-\rh \pd{
        }{t}\left( (1-S)\rho_{sat} \right) + \diver \left( \rho_{sat} \mathcal{D} \grad
        \rh \right) + \ehat{l}{g_v}}{(1-S)\rho_{sat}}, & S \in (0,1) \end{array} \right..
    \label{eqn:piecewise_rh}
\end{flalign}
This is a somewhat artificial setup, but it allows us to assume that a relative humidity
exists for the entire domain even when a gas phase doesn't (strictly speaking) exist.
Simply stated, equation \eqref{eqn:piecewise_rh} indicates that if the saturation is 100\%
then there is no change in relative humidity (even though technically there is no gas
phase).  Then if we take the initial condition as $\rh(x,0) = 1$ when $S=1$ and the
left-hand boundary condition as $\rh(0,t) = 1$ we (artificially) create a relationship for
the relative humidity that holds over the entire spatial domain.  There are several
concerns with this approach, not the least of which is that the existence and uniqueness
theory discussed previously does not cover this sort of case.  Moreover, numerically
requiring that the transition point is exactly 1 is not reasonable and some artificial
cutoff, $S_*<1$, should be used to loosen this condition in numerical simulations.

A simpler way to model the relative humidity in this experiment is to prescribe an initial
saturation that allows for {\it some} gas phase to exist throughout the experiment for all
time.  This is achieved by setting the initial saturation less than 1.  From a numerical
standpoint this makes the equations easier to solve as the boundaries are all stationary
in time.  On the other hand, this choice of initial condition does not match the
experimental setup and is therefore less desirable for the purposes of model validation.
The forthcoming numerical experiments are performed using a combination of these two
approaches; if $S < S_* < 1$, then the relative humidity equation is {\it turned on} and
diffusion is allowed to occur.

In the experiment, the relative humidity was only measured at the surface, at 1cm below
the surface, and in the surrounding (ambient) conditions. Hence, we can only truly compare
with this data in regions very close to the surface. The saturation and temperature
sensors, on the other hand, are placed coarsely throughout and the model can be validated
over the entire spatial domain.

\subsection{Numerical Simulations}\label{sec:numerical_all_coupled}
In this subsection we perform numerical simulations of the coupled system based on the
initial and boundary conditions in Section \ref{sec:column_experiment}. These numerical
simulations are a first step toward validation of the newly proposed mathematical model.
The main thrust of this work was not to create robust numerical solvers for coupled
systems of PDEs. As such, we rely on the \texttt{NDSolve} package built into
\texttt{Mathematica} as the primary numerical solver.  The plotting and post processing
are performed on a mix of \texttt{Mathematica} and \texttt{MATLAB}. The purpose of these
numerical experiments is to provide a validation for the proposed equations. As such, we
chose to directly model the Stefan problem with the relative humidity equation as shown in
\eqref{eqn:piecewise_rh}.

The \texttt{NDSolve} package is based on a method of lines approach to numerical time
integration with a finite difference spatial discretization.  In time we use a
fourth-order Gear's scheme and in space we use a central differencing scheme. One
disadvantage to using this type of spatial scheme in this problem is that the saturation
and heat equations have advective components.  It is well known
\cite{LeVeque1992,LeVeque2007} that upwind schemes are typically better at capturing
the physics of advective flow problems and the central differencing schemes will usually
introduce artificial diffusion into the solution. In an evaporation-type experiment such
as this one, the advective flow is expected to be less dominant than in drainage or
imbibition experiments. Hence, the artificial diffusion introduced with a central
difference scheme is expected to have little impact on the solution quality. To date, the
use of non-central schemes and adaptive mesh refinement are not supported by
\texttt{Mathematica}'s differential equation solving package.

The parameters that can be varied in this experiment are the coefficient of the dynamic
saturation term, $\tau$, the weight of the evaporation coefficient, $M$, and the weights
of the $\grad \rh$ and $\grad T$ terms in the saturation equation, $C_\rh^l$ and $C_T^l$.
There are no parameters that can be varied in the relative humidity equation due to the
newly proposed model for diffusion.  This is in contrast to the standard enhancement model
where the empirical fitting parameter for diffusion is used to match the experimental
data. The fact that we have three different fitting parameters in the saturation equation
simply allows us to fine tune the {\it shape} of the saturation solution beyond what is
predicted by the traditional Richards' equation. Recall from Section
\ref{sec:numerical_saturation} that larger values of the dynamic saturation term affects
the sharpness of the moving front in the saturation equation.

In Section \ref{sec:constitutive_equations} it was demonstrated that
for certain choices of $M$ in the evaporation rule, the present evaporation model
approximated that of Bixler \cite{Bixler1985}
\[ \ehat{l}{g_v} = b \left( \eps{l} - \eps{l}_r \right) R^{g_v} T \left( \rho_{sat} -
    \rho^{g_v} \right). \]
In \cite{Smits2011}, the fitting parameter
for Bixler's model was $b = 2.1 \times 10^{-5}$. The corresponding fitted parameter in
the present model is $M \approx 10^{-15}$ where $\ehat{l}{g_v}$ is given as
\[ \ehat{l}{g_v} = M \rh \left( \rho^l - \rho^{g_v} \right)\left( \mu^l - \mu^{g_v} \right).
    \]
This is only an order of magnitude approximation and fine tuning can be made to better fit
the data.

Several experimental estimates of $\tau$ were presented in Table 2 of
\cite{Hassanizadeh2002}.  From this data, $\tau$ could possibly span several orders of
magnitude: $10^4 < \tau < 10^8$. Unfortunately, the soil types were only listed as
``sand'' (or ``dune sand'') and the relevant permeabilities and van Genuchten parameters
were absent from this summary. These values at least give a {\it ballpark} estimate for
experimentation with $\tau$. The coefficients of $\grad \rh$ and $\grad T$ in the 
saturation equation, on the other hand, are new to this study and appropriate values have
yet to be determined. As such, we study different orders of magnitude for these values to
estimate the effect of the terms to the overall numerical solution. The material parameters are all chosen to match those in Table
\ref{tab:SmitsMaterial}.

The initial simulations will be run with sinusoidal boundary conditions as shown in
Figures \ref{fig:BC_Sinusoidal_Square}. This is to give a qualitative estimation of the
behavior of the solutions without the trouble of the jump discontinuities associated with
the square wave approximation or data interpolation. A smoothed square wave approximation
(also shown in Figures \ref{fig:BC_Sinusoidal_Square}) is then used to give a {\it closer}
match to the experimental data. The smoothing is achieved by taking piecewise defined
hyperbolic tangent functions to approximate the steps. To measure the error between the
data and the numerical solution we use a sum of the squares of the residual
values measured at each sensor location:
\begin{flalign}
    e_u(t) = \frac{\sum_{\xi \in data} \left( u_{data}(\xi,t) - u_{num}(\xi,t)
\right)^2}{\sum_{\xi \in data} \left( u_{data}(\xi,t) \right)^2},
    \label{eqn:l2_data_error}
\end{flalign}
where $u$ represents any of the three dependent variables of interest ($S$, $\rh$, or $T$)
and the subscript indicates where the value is taken from. Stating that ``$\xi \in data$''
simply means that $\xi$ spans the sensor locations relevant for the given $u$ (i.e. $\xi
\in \{110/111, \, 1 \}$ for $u=\rh$). Obviously $e_u(t)$ is a function
of time so to get a single measure that describes the error we take the maximum of
$e_u(t)$ over the length of an experimental day
\begin{flalign}
    E_u = \max_{t} \left( e_u(t) \right).
    \label{eqn:l2_data_error_max}
\end{flalign}
A single experimental day was chosen due to numerical difficulties and due to loss of
relative humidity sensor information.
Equation \eqref{eqn:l2_data_error_max} gives a single numerical value measuring the fit of
the numerical solution to the data.  In the relative humidity this is a very simplistic
exercise as there is only 1 data point to compare against; the sensor located 1cm
below the surface of the soil.  For the saturation and temperature, on the other hand,
this gives a better quantitative measure. 

Table \ref{tab:rel_err_classical} gives errors measured with equation
\eqref{eqn:l2_data_error_max} against the {\it classical} system of equations (Richards',
Enhanced Diffusion, and de Vries).  Table 
\ref{tab:rel_err_sq_wave} shows the error as measured with equation
\eqref{eqn:l2_data_error_max} for various values of $\tau$, for different functional forms
of the thermal conductivity (see Section \ref{sec:total_energy_simplifications}), and for
various values of $C_\rh^l$ and $C_T^l$.  In order to make comparisons with $\tau,
C_\rh^l,$ and $C_T^l$ we use the ratio of this coefficient as compared to the diffusive
term in the saturation equation. This was done in Sections \ref{sec:numerical_saturation}
and \ref{sec:numerical_coupled_sat_vap}, and the ratios of interest are repeated here for
clarity:
\begin{subequations}
    \begin{flalign}
        \left( \frac{\tau \porosity_S}{-t_c p_c'(S)} \right) &= \left( \frac{\tau \porosity_S
        \kappa_S}{\mu_l} \right) Pe(S) = H_0 Pe(S) \\
        -\left( \frac{C_\rh^l}{p_c'(S)} \right) &= \left( \frac{C_\rh^l}{\rho^l g} \right)
        Pe(S) = R_0 Pe(S) \\
        -\left( \frac{C_T^l}{p_c'(S)} \right) &= \left( \frac{C_T^l}{\rho^l g} \right) Pe(S) =
        \theta_0 Pe(S).
    \end{flalign}
    \label{eqn:dimensionless_ratios}
\end{subequations}
Since each ratio is relative to the P\'eclet number (which is a function of $S$) we focus
only on the ratios on the right-hand sides of equations \eqref{eqn:dimensionless_ratios}.
Due to the fact that this is a large parameter space, only some of the notable
relative errors are presented. Mesh refinement was used in the comparisons in several
instances to minimize numerical artifacts. Spatially, the meshes ranged between 100 and
1024 points. Only a uniform mesh was considered.
\linespread{1.0}
\begin{table}
    \centering
    \begin{tabular}{|c|c||c|c|c|}
        \hline
        \multicolumn{2}{|c||}{} & \multicolumn{3}{c|}{Relative Errors} \\ \hline
        Conductivity & Boundary Cond. & Saturation & Rel. Humidity & Temperature \\
        \hline \hline
        Weighted Sum   & Smoothed Square  & 0.00356 & 1.54048 & 0.000515 \\
        C\^ot\'e-Konrad & Smoothed Square  & 0.00357 & 1.27818 & 0.000631 \\ \hline
    \end{tabular}
        \caption{Relative errors measured using equation \eqref{eqn:l2_data_error_max} for
        the classical mathematical model consisting of Richards' equation for saturation,
        the enhanced diffusion model for vapor diffusion, and the de Vries model for heat
        transport. These are compared for the two thermal conductivity functions of
        interest (weighted sum \eqref{eqn:thermal_weighted_sum} and C\^ot\'e-Konrad
        \eqref{eqn:CoteKonrad}).}
    \label{tab:rel_err_classical}
\end{table}
\renewcommand{\baselinestretch}{\normalspace}

\linespread{1.0}
    \begin{table}[ht!]
        \centering
        \begin{tabular}{|c|c|c|c||c|c|c|}
            \hline
            \multicolumn{4}{|c||}{{\bf Parameters \& Functions}} & \multicolumn{3}{c|}{{\bf Relative
            Errors}} \\ \hline
            {\bf Conductivity} & {\bf $R_0$} & {\bf $\theta_0$} & {\bf $H_0$} & {\bf
            Saturation} & {\bf Rel. Humidity} & {\bf Temperature} \\ \hline \hline
            Weighted Sum & \multicolumn{3}{c||}{Classical Model} & 0.00356 & 1.54048 & 0.000515  \\ \hline \hline
            Weighted Sum    & 0             & 0             & $10^{-2.5}$   & 0.011966 & 1.206005  & 0.000463 \\
                            &               &               & $10^{-3.0}$   & 0.009020 & 1.201793  & 0.000463 \\
                            &               &               & $10^{-3.5}$   & 0.006508 & 1.198274  & 0.000463 \\
                            &               &               & $10^{-4.0}$   & 0.004904 & 1.199174  & 0.000463 \\
                            &               &               & $10^{-4.5}$   & 0.004076 & 1.195180  & 0.000463 \\
                            &               &               & $10^{-5.0}$   & 0.003712 & 1.196750  & 0.000463 \\ 
                            &               &               & $0$           & 0.003536 & 1.199584  & 0.000463 \\ \hline \hline
            Weighted Sum    & $10^{-5}$     & $10^{-5}$     & $10^{-5}$     & 0.003712 & 1.196751   & 0.000463 \\
                            & $10^{-4}$     & $10^{-4}$     & $10^{-5}$     & 0.003712 & 1.196756   & 0.000463 \\
                            & $10^{-3}$     & $10^{-3}$     & $10^{-5}$     & 0.003710 & 1.196807   & 0.000463 \\
                            & $10^{-2}$     & $10^{-2}$     & $10^{-5}$     & 0.003692 & 1.197041   & 0.000463 \\ 
                            & $10^{-1}$     & $10^{-1}$     & $10^{-5}$     & 0.003509 & 1.202644   & 0.000462 \\ 
                $\star$     & $1$           & $1$           & $10^{-5}$     & $\star$  & $\star$   & $\star$ \\ \hline \hline
            C\^ot\'e-Konrad & \multicolumn{3}{c||}{Classical Model} & 0.00357 & 1.27818 & 0.000631  \\ \hline \hline
            C\^ot\'e-Konrad & 0             & 0             & $10^{-2.5}$   & 0.011964 & 0.946441   & 0.000516 \\
                            &               &               & $10^{-3.0}$   & 0.009022 & 0.951573   & 0.000515 \\
                            &               &               & $10^{-3.5}$   & 0.006513 & 0.950024   & 0.000515 \\
                            &               &               & $10^{-4.0}$   & 0.004910 & 0.948023   & 0.000515 \\
                            &               &               & $10^{-4.5}$   & 0.004084 & 0.944724   & 0.000516 \\
                            &               &               & $10^{-5.0}$   & 0.003719 & 0.946599   & 0.000516 \\
                            &               &               & $0$           & 0.003545 & 0.942403   & 0.000516  \\ \hline \hline
            C\^ot\'e-Konrad & $10^{-5}$     & $10^{-5}$     & $10^{-5}$     & 0.003719 & 0.948032   & 0.000516 \\ 
                            & $10^{-4}$     & $10^{-4}$     & $10^{-5}$     & 0.003720 & 0.941801   & 0.000516 \\ 
                            & $10^{-3}$     & $10^{-3}$     & $10^{-5}$     & 0.003717 & 0.945905   & 0.000515 \\ 
                            & $10^{-2}$     & $10^{-2}$     & $10^{-5}$     & 0.003698 & 0.947884   & 0.000515 \\ 
                            & $10^{-1}$     & $10^{-1}$     & $10^{-5}$     & 0.003510 & 0.966405   & 0.000513 \\ 
                 $\star$    & $1$           & $1$           & $10^{-5}$     & $\star$ &  $\star$  & $\star$ \\ \hline \hline
        \end{tabular}
        \caption{Relative errors measured using equation \eqref{eqn:l2_data_error_max} for
    instances within the parameter space consisting of the thermal conductivity function (weighted sum \eqref{eqn:thermal_weighted_sum} and C\^ot\'e-Konrad
        \eqref{eqn:CoteKonrad}),
$C_\rh^l, C_T^l,$ and $\tau$. These are taken for a (smoothed) square wave approximation to the
boundary conditions. (The starred rows indicate failure of the numerical method, and the
errors from the classical model are repeated for clarity)}
\label{tab:rel_err_sq_wave}
    \end{table}
\renewcommand{\baselinestretch}{\normalspace}

It is apparent in Table \ref{tab:rel_err_sq_wave} that the
best error approximations for saturation, relative humidity, and temperature are found
with smaller values of $H_0$ (or equivalently, $\tau$). This observation is particular to
a drainage-type experiment.  If the experiment were an imbibition-type then it is
conjectured (based on the results in Section \ref{sec:numerical_saturation}) that the
value of $\tau$ would play a larger role.  Also apparent in Table
\ref{tab:rel_err_sq_wave}, we see that the values of $C_\rh^l$ and $C_T^l$ play little
role in the overall dynamics of the problem.

Keep in mind that the relative humidity errors are really just the difference between 1
single sensor and the numerical solution at that physical location.  In the author's
opinion it is unreasonable to judge the effectiveness of the numerical solution based
solely on one point.  One complication that arose within this solution is that the
relative humidity exhibits small periods of non-physical behavior for certain
parameters and boundary conditions.  Mesh refinement removes some of this effect, but even
with further mesh refinement not all of the non-physical regions were removed.  Possible
sources of this problem are: (1) the fact that \texttt{Mathematica} uses cubic
interpolation polynomials to deliver the solutions to numerical differential equations
(cubic interpolation can overshoot sharp transitions in data), and (2) the Stefan nature
of the problem causes numerical {\it stiffness} at the point of transition. Further
studies are needed to determine the exact cause of this non-physical behavior. 

Figures \ref{fig:SatRelHumCompare_Cosine} show individual time steps of several solutions
for various parameters with a sinusoidal approximation to the relative humidity and
temperature boundary conditions.  Figures \ref{fig:SatRelHumCompare_Square} show the same
plots with smoothed square wave boundary conditions. The square wave boundary conditions
obviously give closer approximation to the boundary data, and at the same time the use of
the square wave boundary conditions removes the non-physical behavior in this case. The
plots associated with the sinusoidal approximation to the boundary conditions are
presented here for comparison between very smooth and slightly sharper transitions in
boundary data.  A
closeup of the regions of non-physical behavior is shown in Figures
\ref{fig:SatRelHumCompare_Cosine_Blowup_t100} and
\ref{fig:SatRelHumCompare_Square_Blowup_t100}. Observe in these figures that the diffusion equation
solved with smaller values of $H_0$ and larger diffusion (from the weighted sum equation)
give the most plausible solutions. Figures \ref{fig:SatRelHumCompare_Cosine_Blowup2_t100}
and \ref{fig:SatRelHumCompare_Square_Blowup2_t100} give an indication of the difference
in the relative humidity equations given different thermal conductivity functions.

\linespread{1.0}
\begin{figure}[ht!]
        \centering
        \subfigure[Comparison at $t=1 \times 600$ sec]{
            \includegraphics[width=0.45\textwidth]{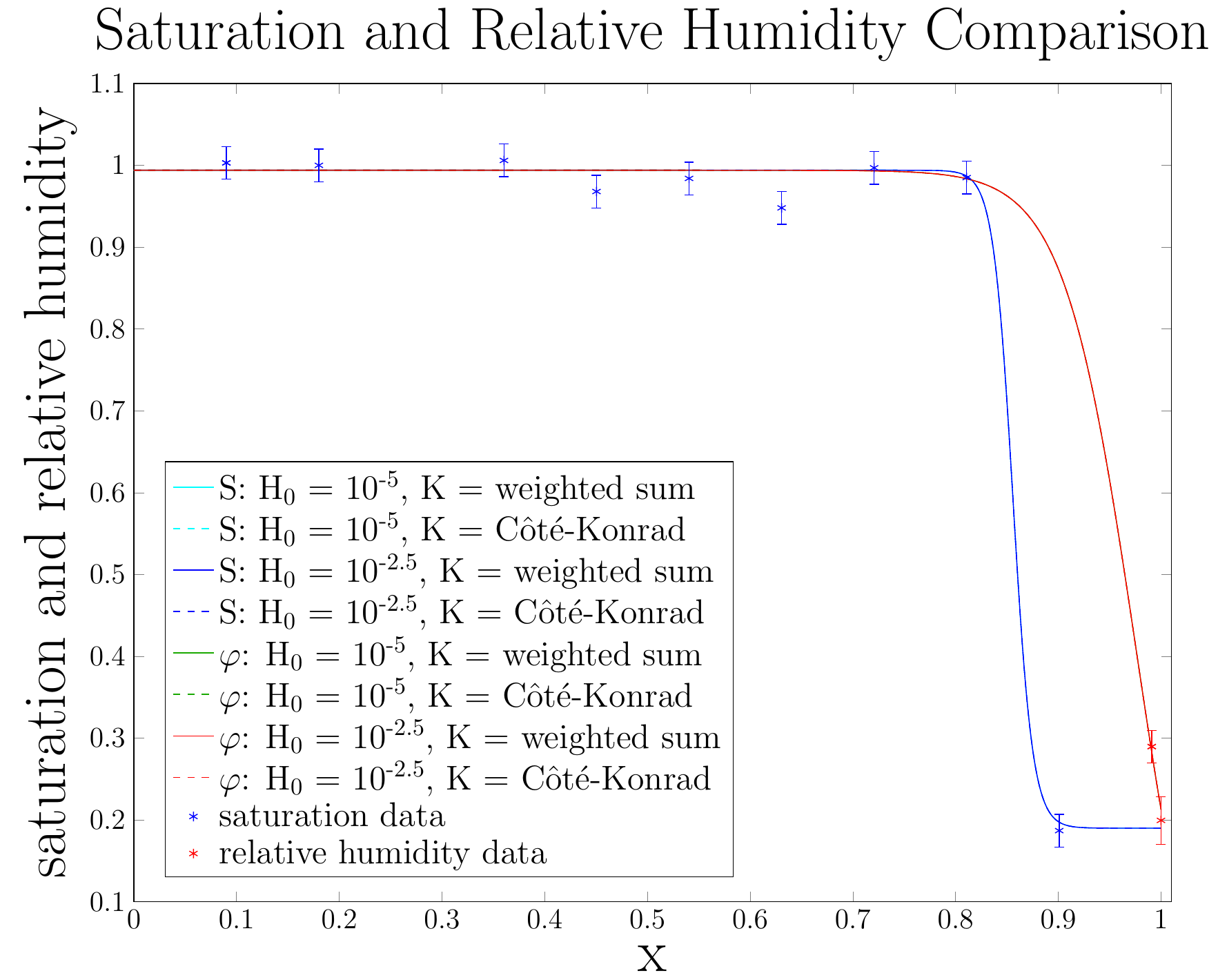}
            \label{fig:SatRelHumCompare_Cosine_t001}
            }
        \subfigure[Comparison at $t=50 \times 600$ sec]{
            \includegraphics[width=0.45\textwidth]{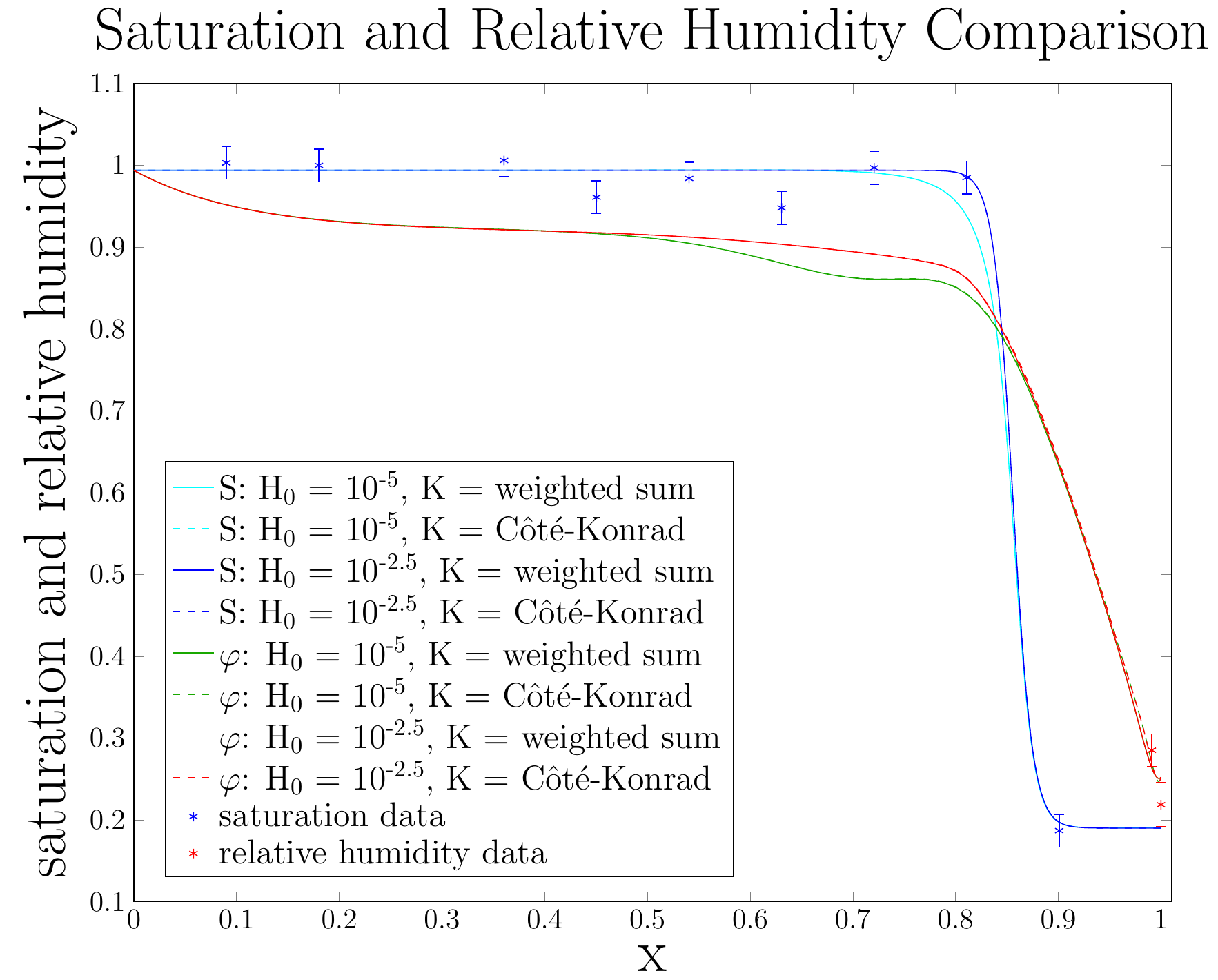}
            \label{fig:SatRelHumCompare_Cosine_t050}
        }
        \subfigure[Comparison at $t=100 \times 600$ sec]{
            \includegraphics[width=0.45\textwidth]{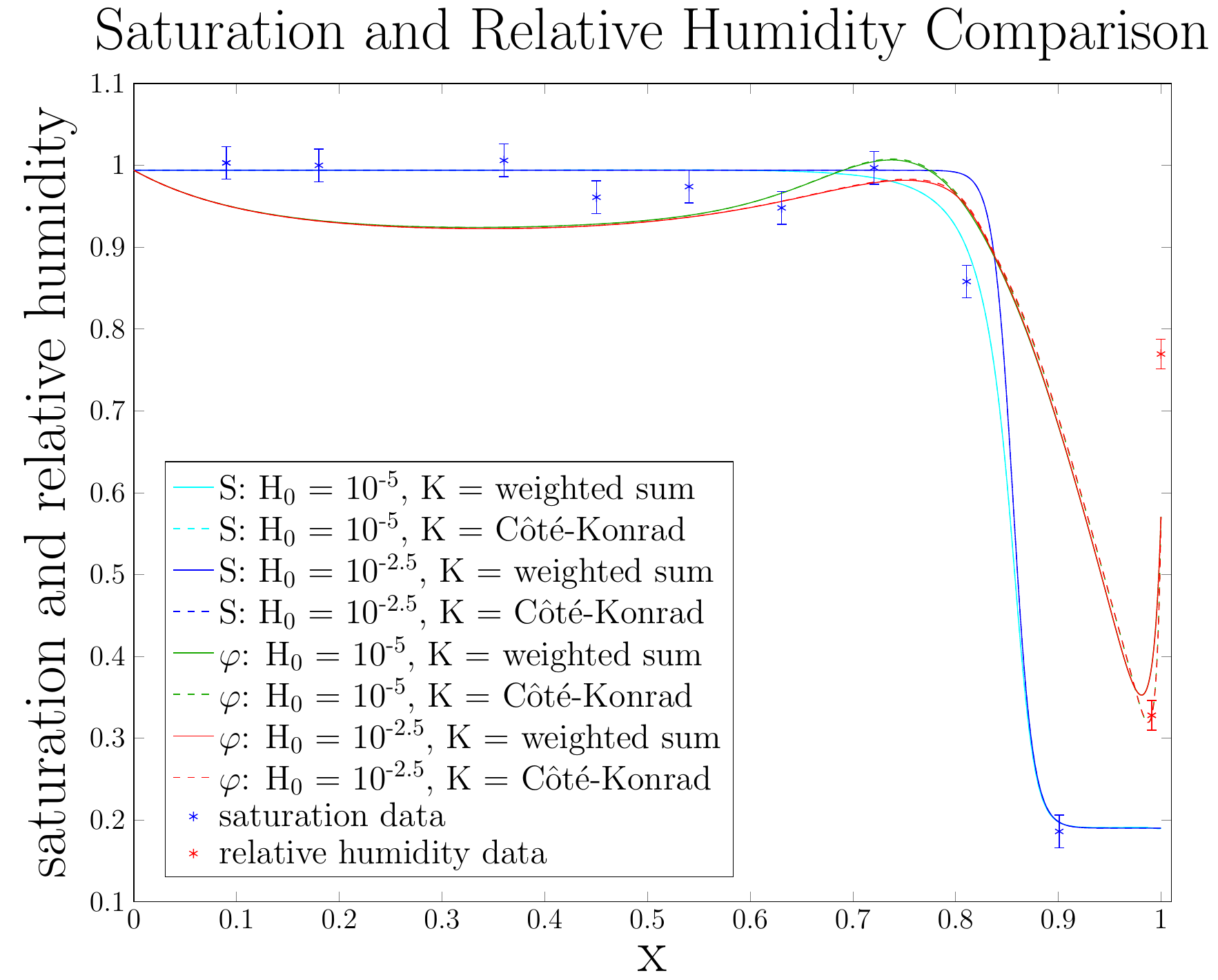}
            \label{fig:SatRelHumCompare_Cosine_t100}
            }
        \subfigure[Comparison at $t=150 \times 600$ sec]{
            \includegraphics[width=0.45\textwidth]{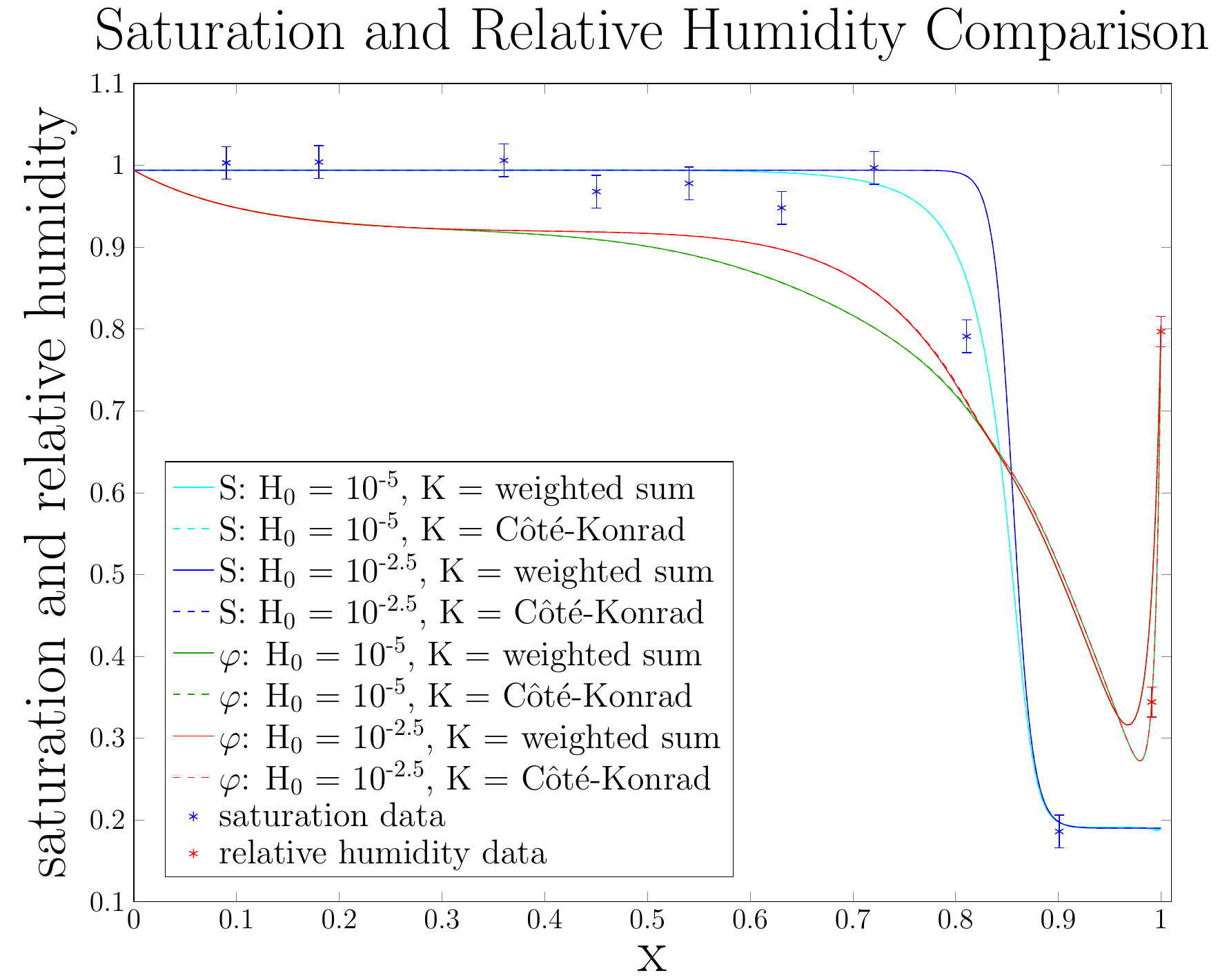}
            \label{fig:SatRelHumCompare_Cosine_t150}
        }
        \caption{Comparison of relative humidity and saturation for the fully coupled
            saturation-diffusion-temperature model as compared to data from
            \cite{Smits2011}. Boundary conditions are taken from a sinusoidal
        approximation of boundary data. Thermal conductivities are taken as either
        weighted sum \eqref{eqn:thermal_weighted_sum} or C\^ot\'e-Konrad
        \eqref{eqn:CoteKonrad}.}
        \label{fig:SatRelHumCompare_Cosine}
\end{figure}
\renewcommand{\baselinestretch}{\normalspace}

\linespread{1.0}
\begin{figure}[ht!]
        \centering
        \subfigure[Comparison at $t=1 \times 600$ sec]{
            \includegraphics[width=0.45\textwidth]{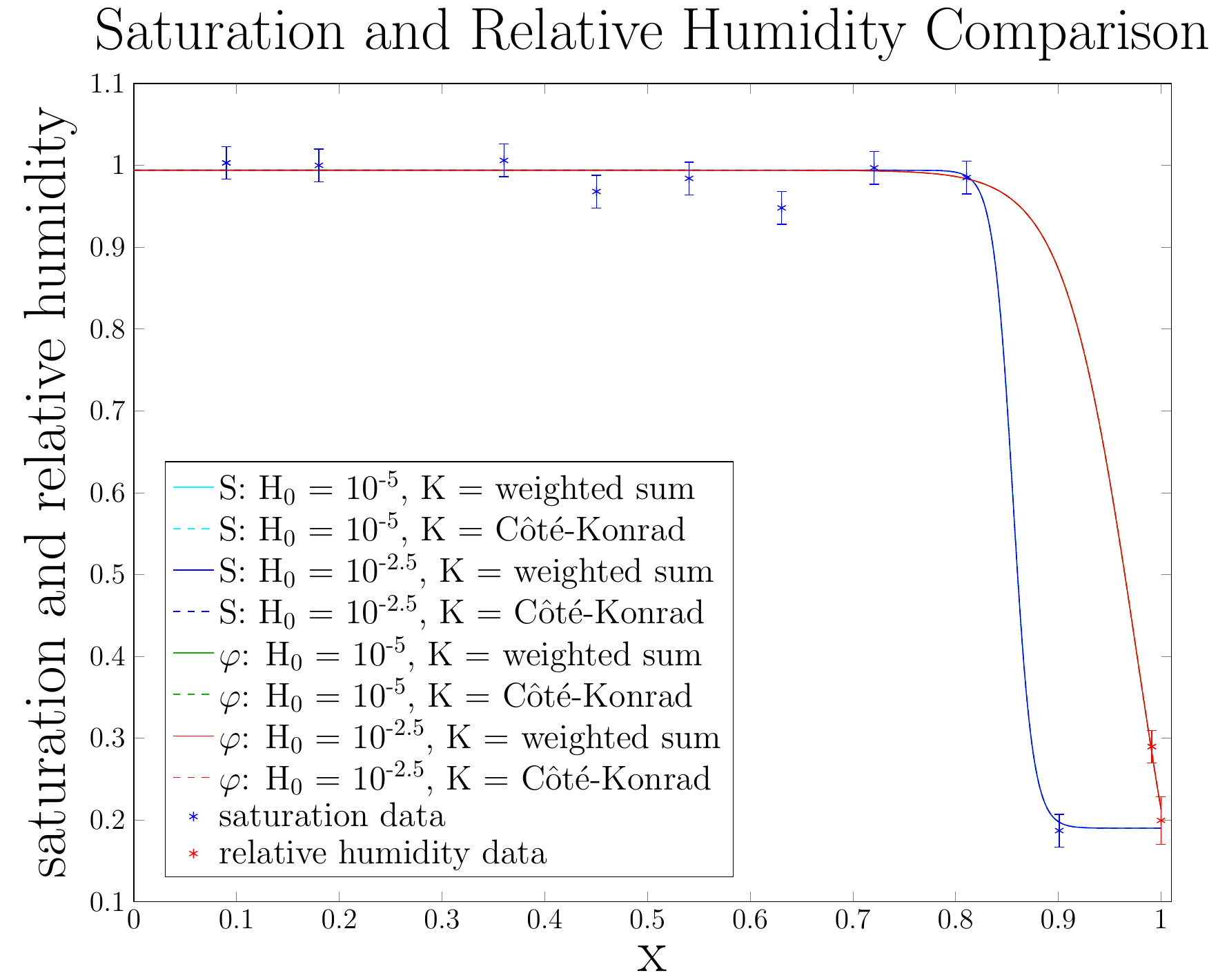}
            \label{fig:SatRelHumCompare_Square_t001}
            }
        \subfigure[Comparison at $t=50 \times 600$ sec]{
            \includegraphics[width=0.45\textwidth]{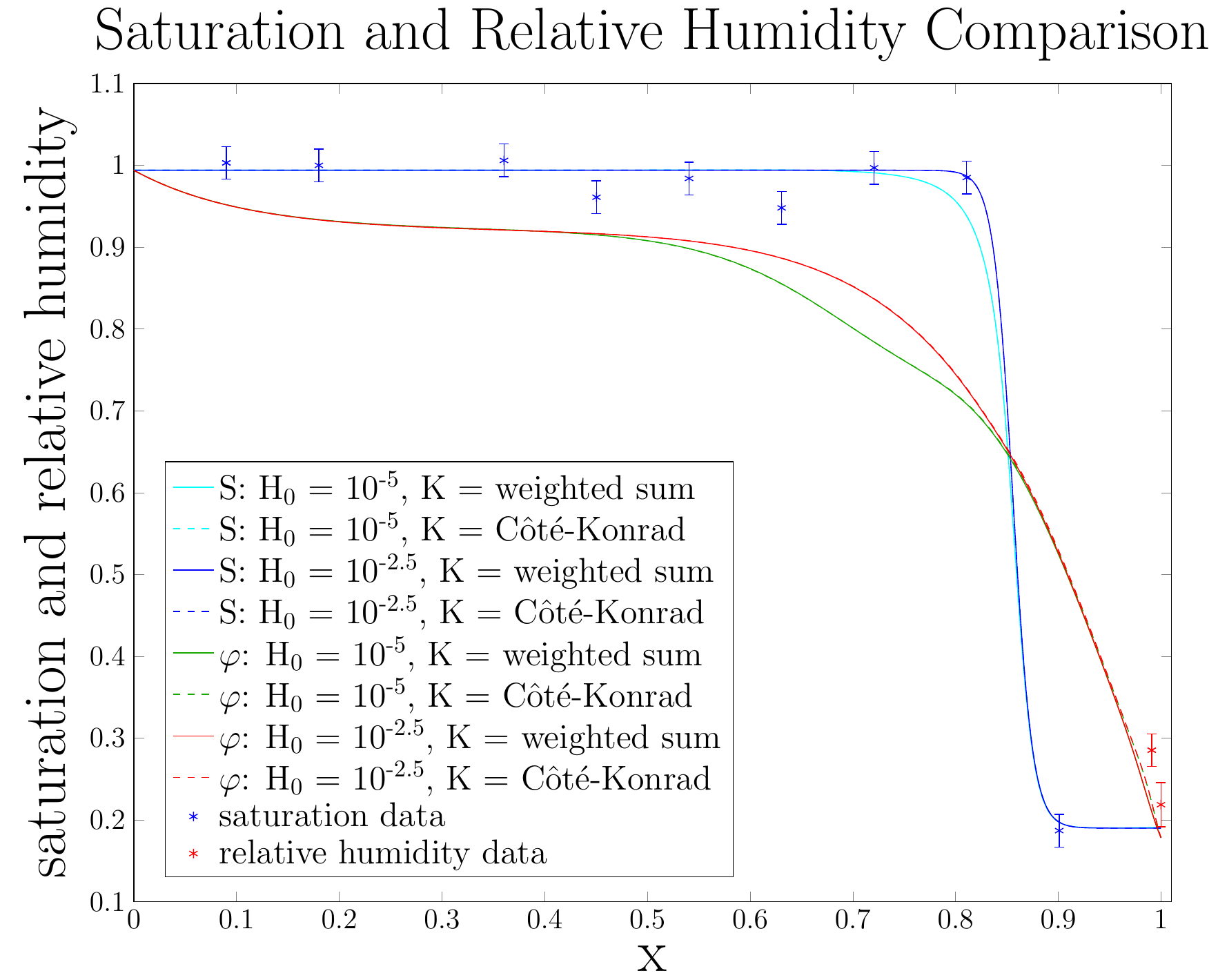}
            \label{fig:SatRelHumCompare_Square_t050}
        }
        \subfigure[Comparison at $t=100 \times 600$ sec]{
            \includegraphics[width=0.45\textwidth]{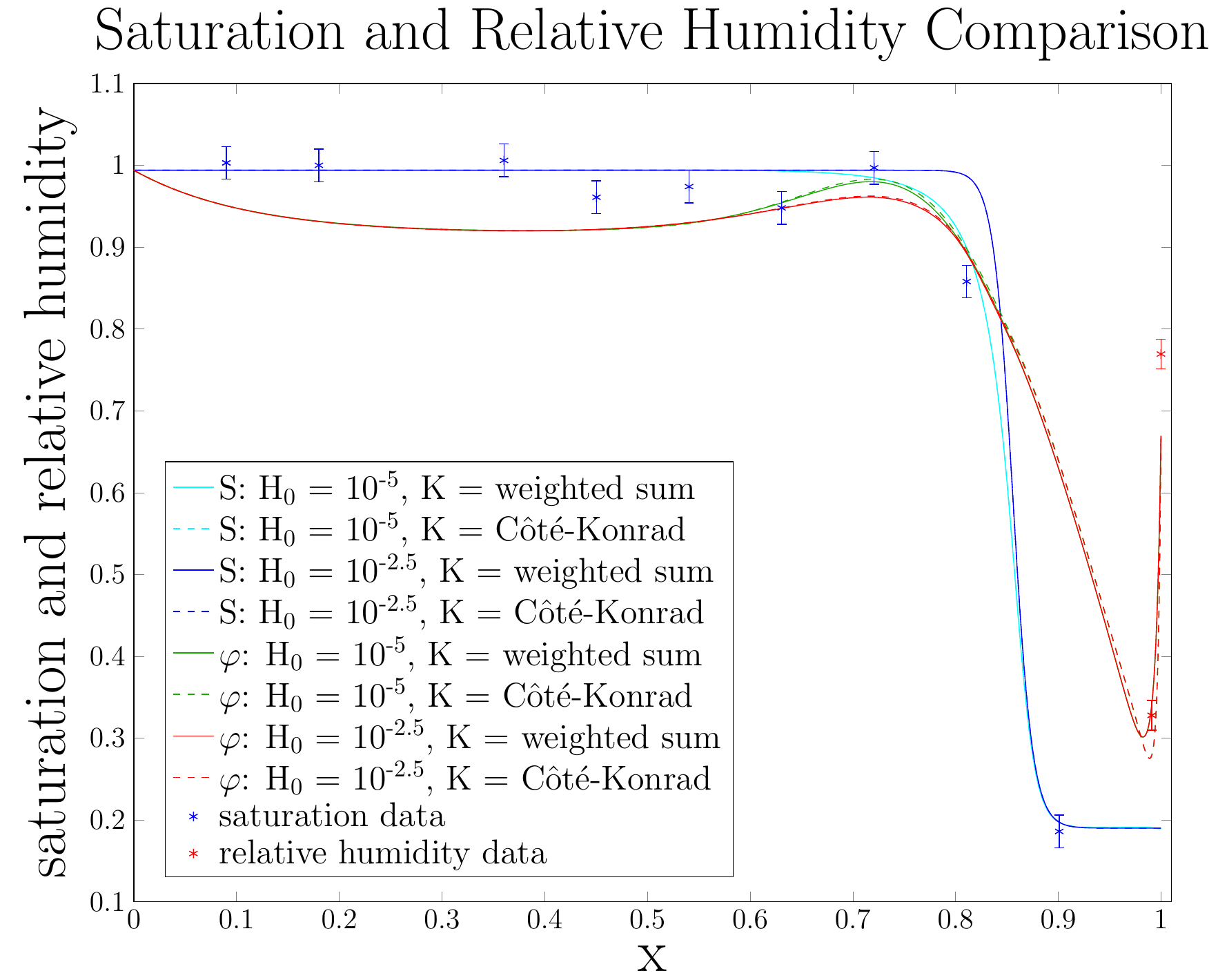}
            \label{fig:SatRelHumCompare_Square_t100}
            }
        \subfigure[Comparison at $t=150 \times 600$ sec]{
            \includegraphics[width=0.45\textwidth]{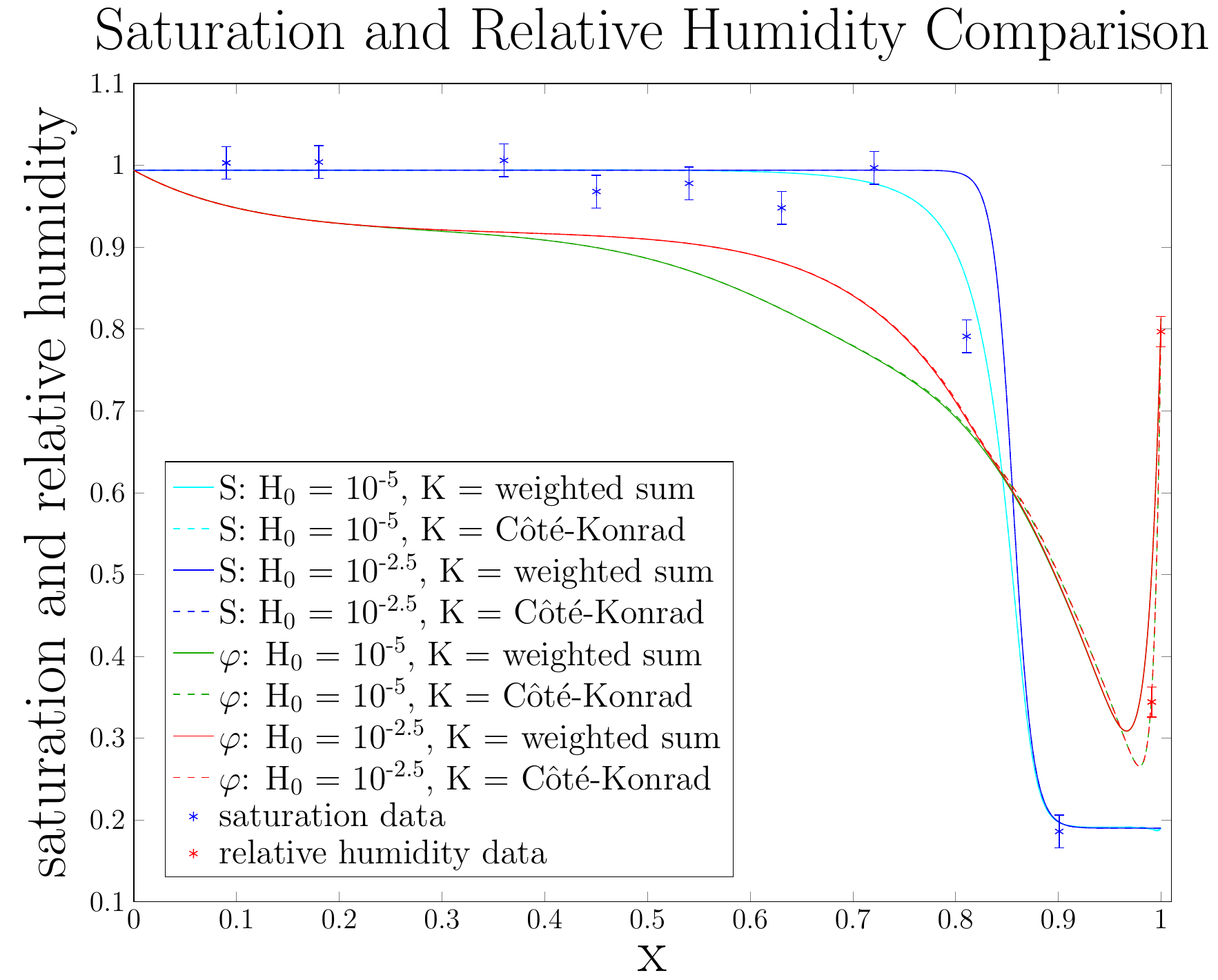}
            \label{fig:SatRelHumCompare_Square_t150}
        }
        \caption{Comparison of relative humidity and saturation for the fully coupled
            saturation-diffusion-temperature model as compared to data from
            \cite{Smits2011}. Boundary conditions are taken from a smoothed square wave
        approximation of boundary data. Thermal conductivities are taken as either
        weighted sum \eqref{eqn:thermal_weighted_sum} or C\^ot\'e-Konrad
        \eqref{eqn:CoteKonrad}.}
        \label{fig:SatRelHumCompare_Square}
\end{figure}
\renewcommand{\baselinestretch}{\normalspace}

\linespread{1.0}
\begin{figure}[ht!]
        \centering
        \subfigure[Blowup comparison at $t=100 \times 600$ sec. (sinusoidal approximated
        boundary)]{
            \includegraphics[width=0.45\textwidth]{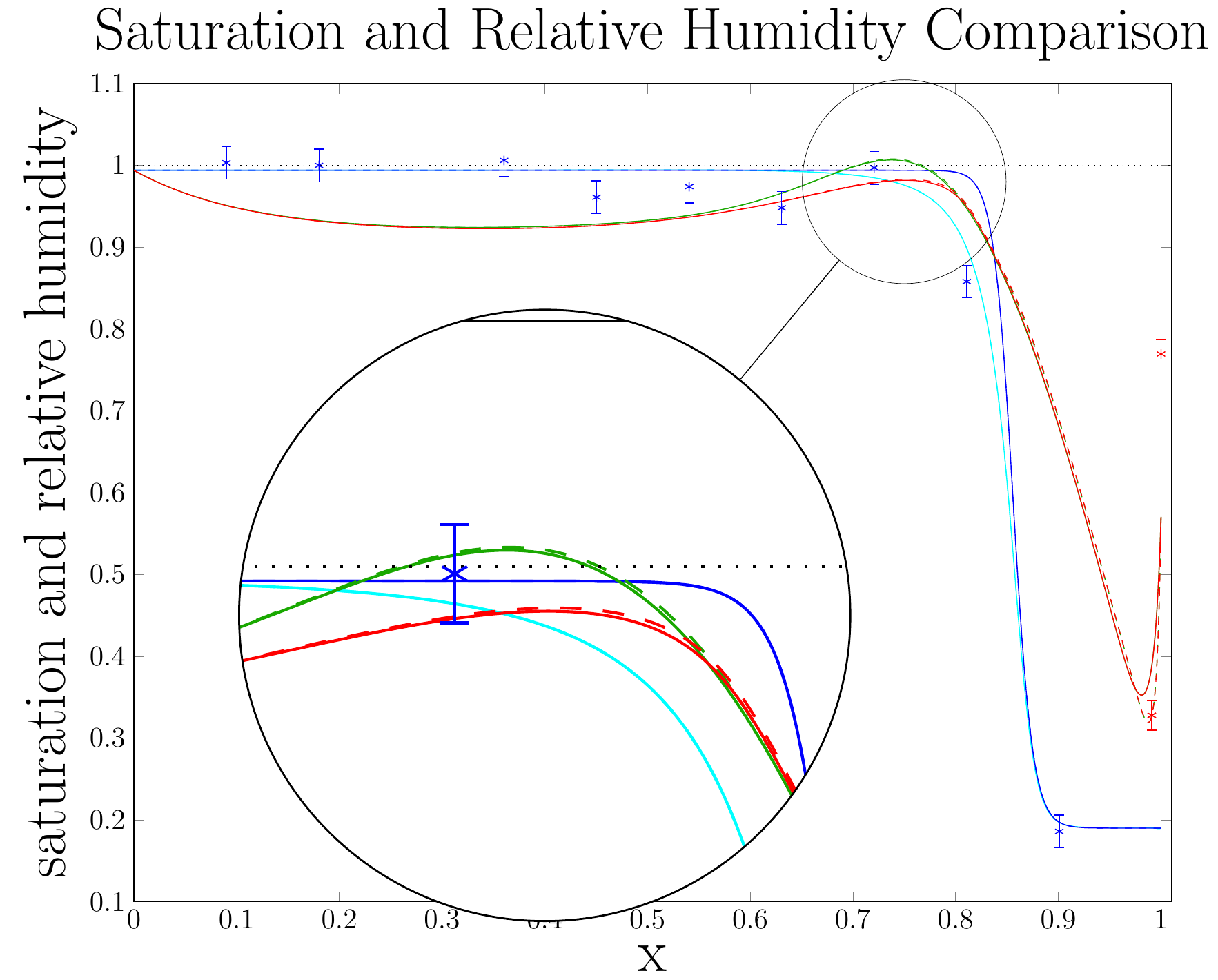}
            \label{fig:SatRelHumCompare_Cosine_Blowup_t100}
        }
        \subfigure[Blowup comparison at $t=100 \times 600$ sec. (square wave approximated
        boundary)]{
            \includegraphics[width=0.45\textwidth]{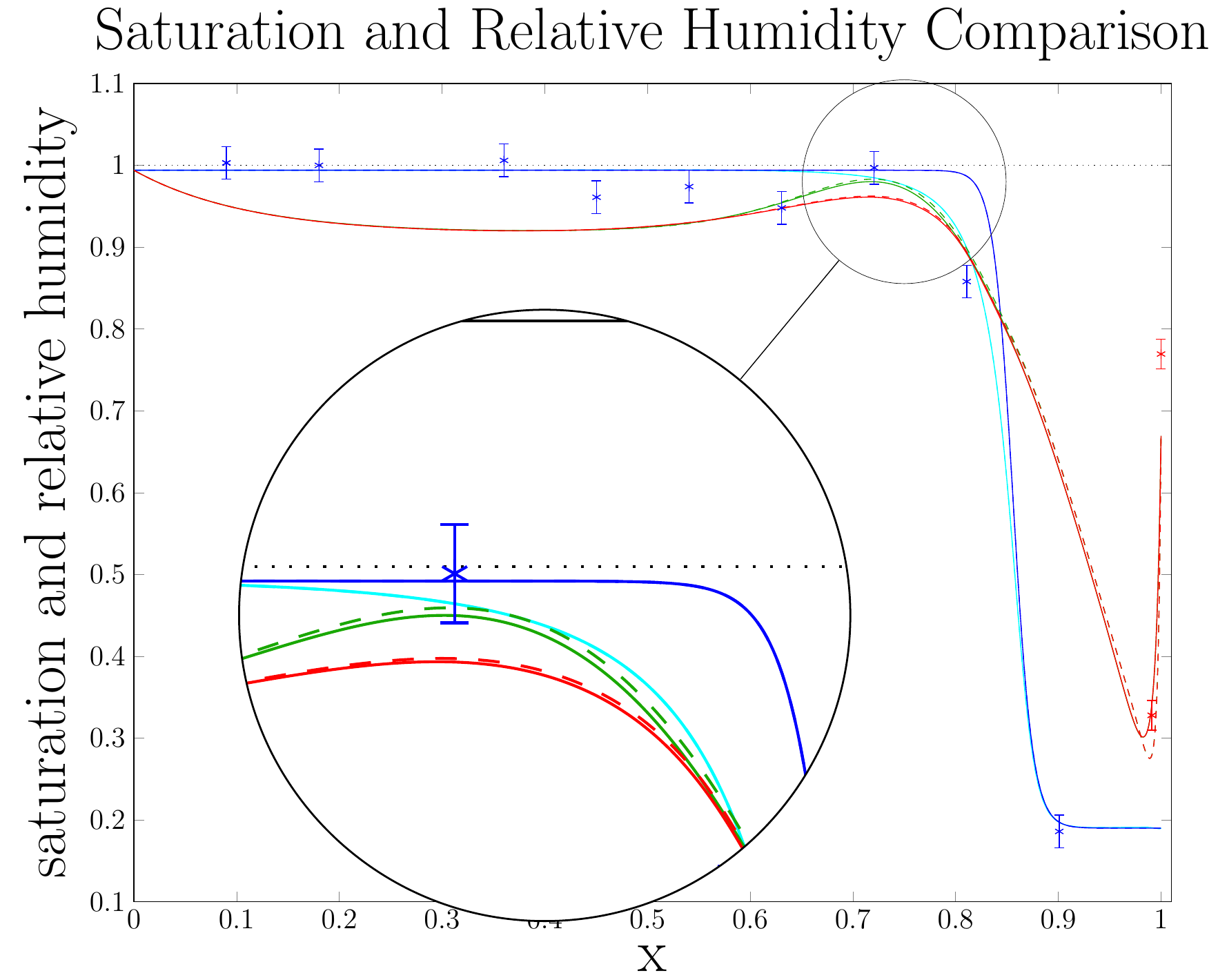}
            \label{fig:SatRelHumCompare_Square_Blowup_t100}
            }
        \subfigure[Blowup comparison at $t=100 \times 600$ sec. (sinusoidal approximated
        boundary)]{
            \includegraphics[width=0.45\textwidth]{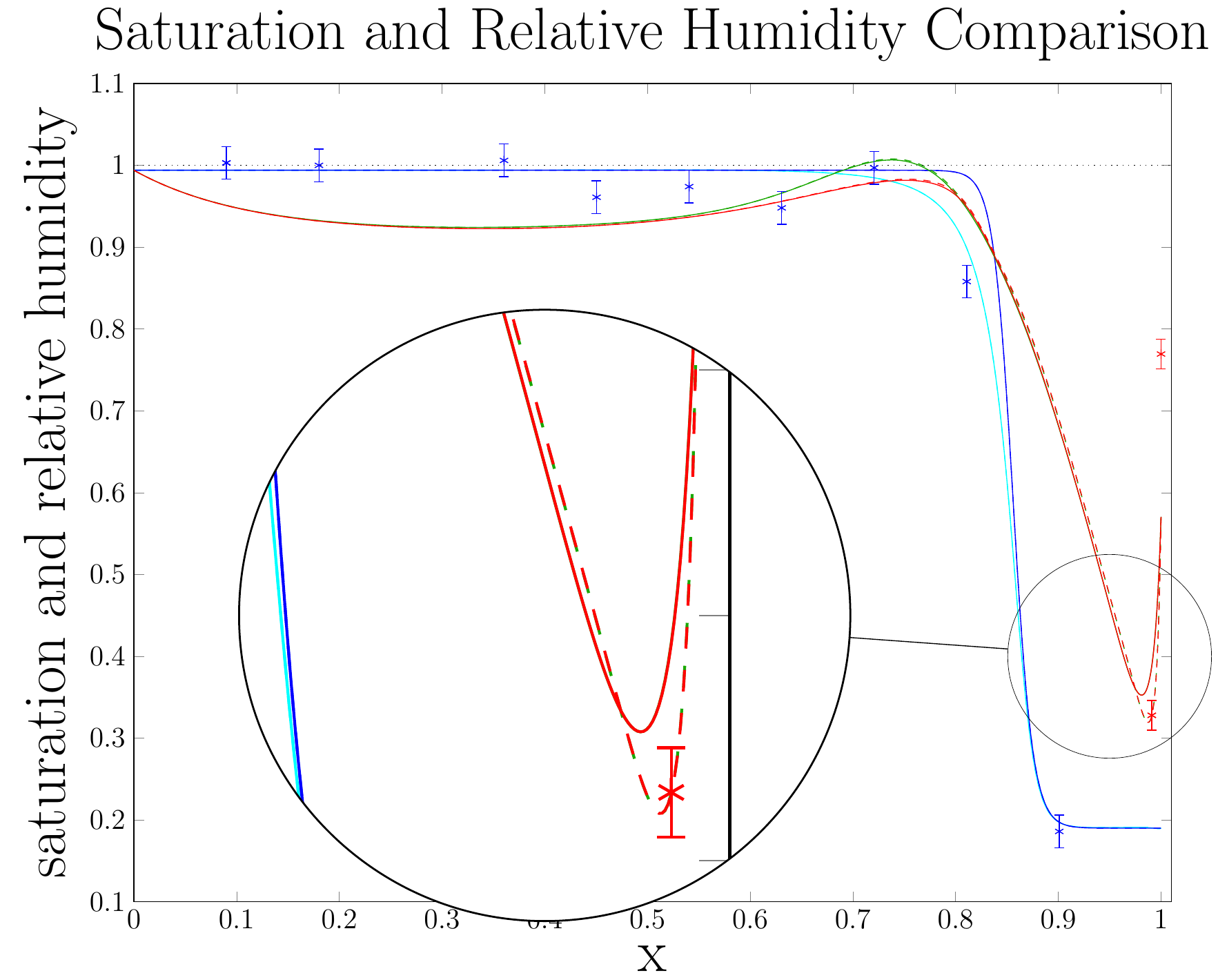}
            \label{fig:SatRelHumCompare_Cosine_Blowup2_t100}
        }
        \subfigure[Blowup comparison at $t=100 \times 600$ sec. (square wave approximated
        boundary)]{
            \includegraphics[width=0.45\textwidth]{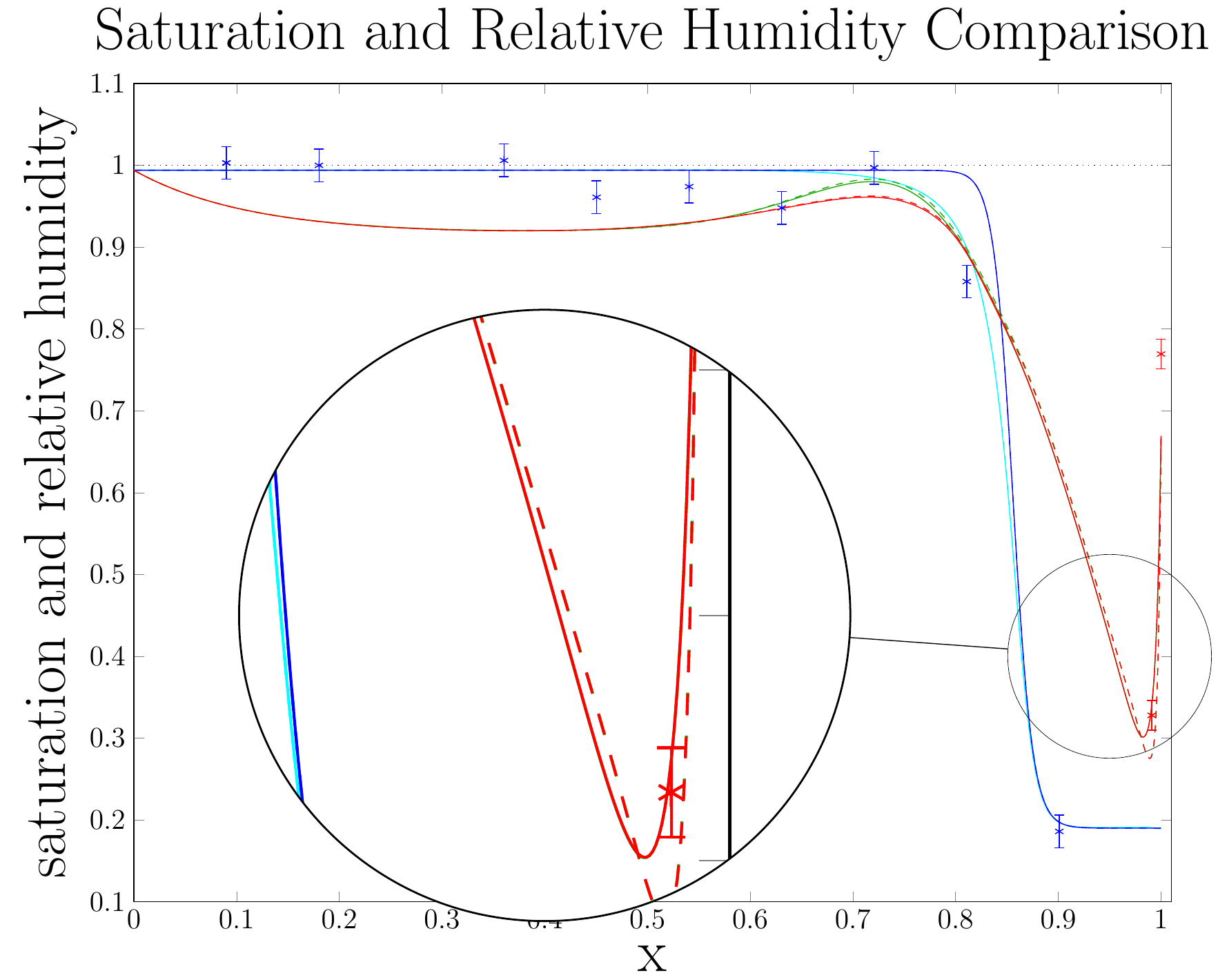}
            \label{fig:SatRelHumCompare_Square_Blowup2_t100}
            }
        \caption{Blowup comparison of relative humidity and saturation for the fully
            coupled saturation-diffusion-temperature model as compared to data from
            \cite{Smits2011}. The inset plots give a closer look at the behavior exhibited
        by these particular solutions. }
        \label{fig:SatRelHumCompare_Blowup}
\end{figure}
\renewcommand{\baselinestretch}{\normalspace}

The comparisons of the temperature solutions are shown in Figures
\ref{fig:TempCompare_Cosine} and \ref{fig:TempCompare_Square}. There is very little
difference between the models for various values of $\tau$ (or equivalently, $H_0$), so
only the curves associated with $H_0 = 10^{-5}$ are shown. Observe that the thermal
equation does a poor job capturing the extent of the diffusion near the top of the
experimental apparatus, but it does well otherwise.  Possible sources of this error come
from: (1) the terms neglected in the simplification of the thermal model, (2) the
initial condition, (3) the thermal conductivity functions (or parameters), and/or (4) the
accuracy of the sensor information. 
\linespread{1.0}
\begin{figure}[ht!]
        \centering
        \subfigure[Comparison at $t=1 \times 600$ sec]{
            \includegraphics[width=0.45\textwidth]{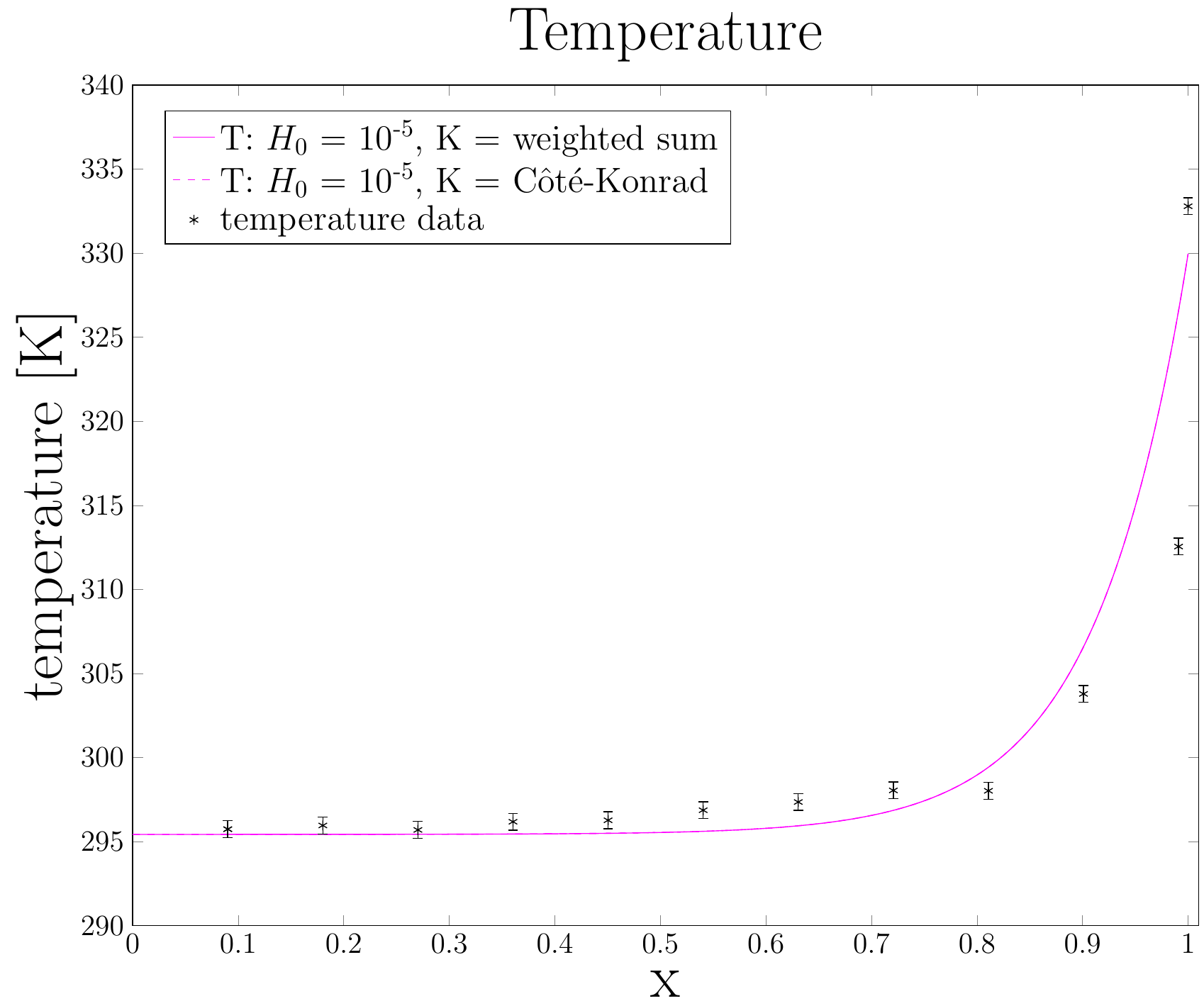}
            \label{fig:TempCompare_Cosine_t001}
            }
        \subfigure[Comparison at $t=50 \times 600$ sec]{
            \includegraphics[width=0.45\textwidth]{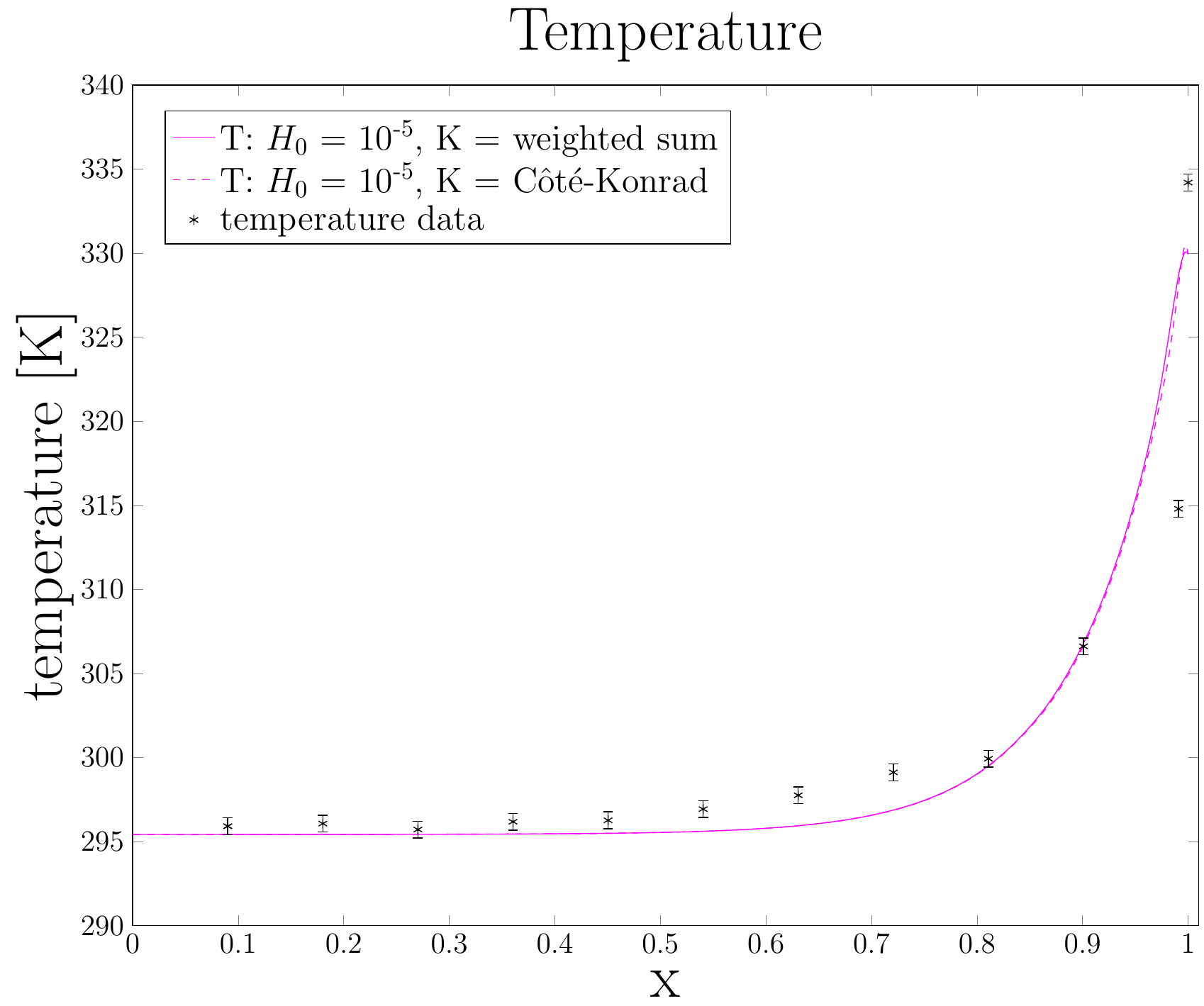}
            \label{fig:TempCompare_Cosine_t050}
        }
        \subfigure[Comparison at $t=100 \times 600$ sec]{
            \includegraphics[width=0.45\textwidth]{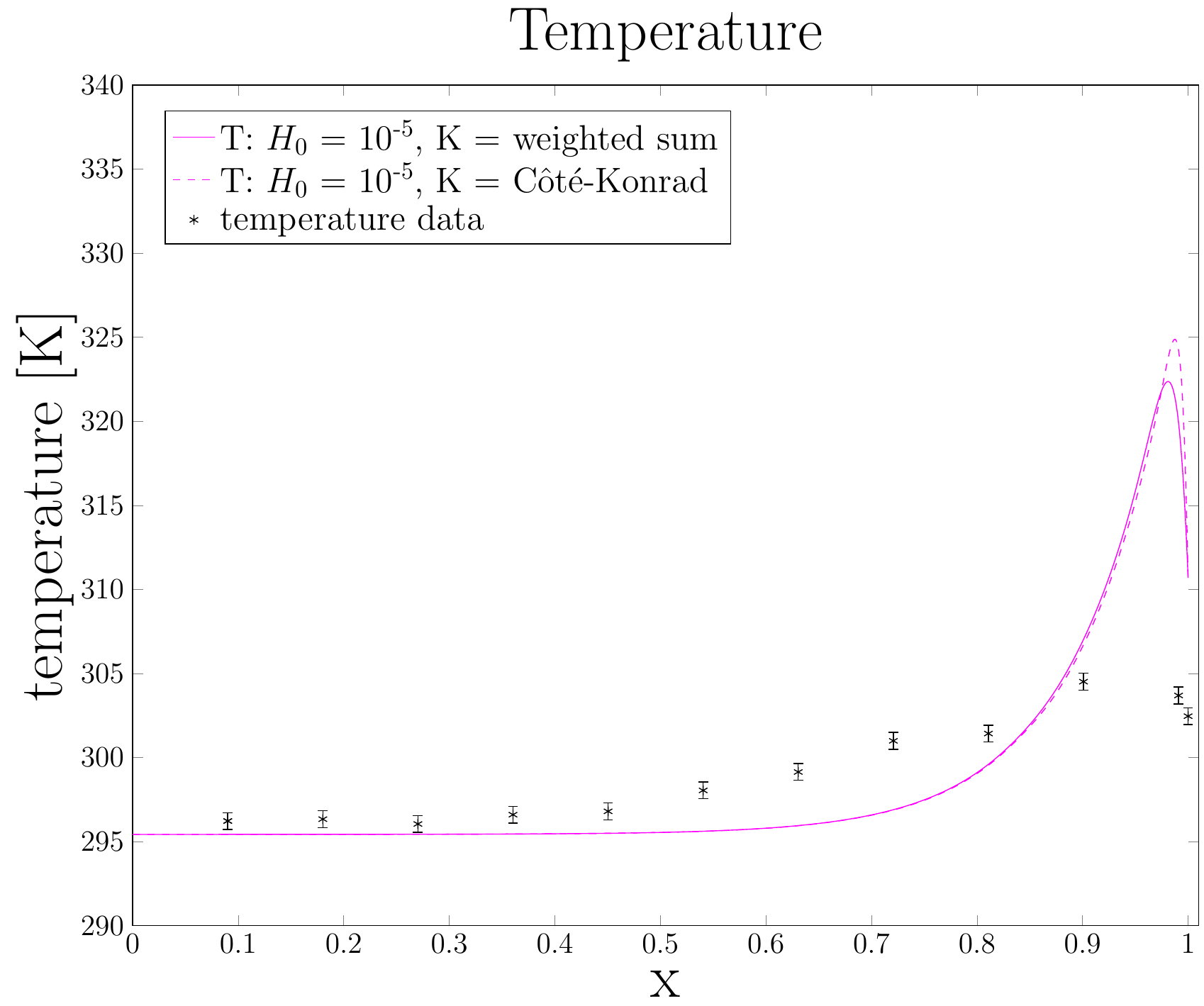}
            \label{fig:TempCompare_Cosine_t100}
            }
        \subfigure[Comparison at $t=150 \times 600$ sec]{
            \includegraphics[width=0.45\textwidth]{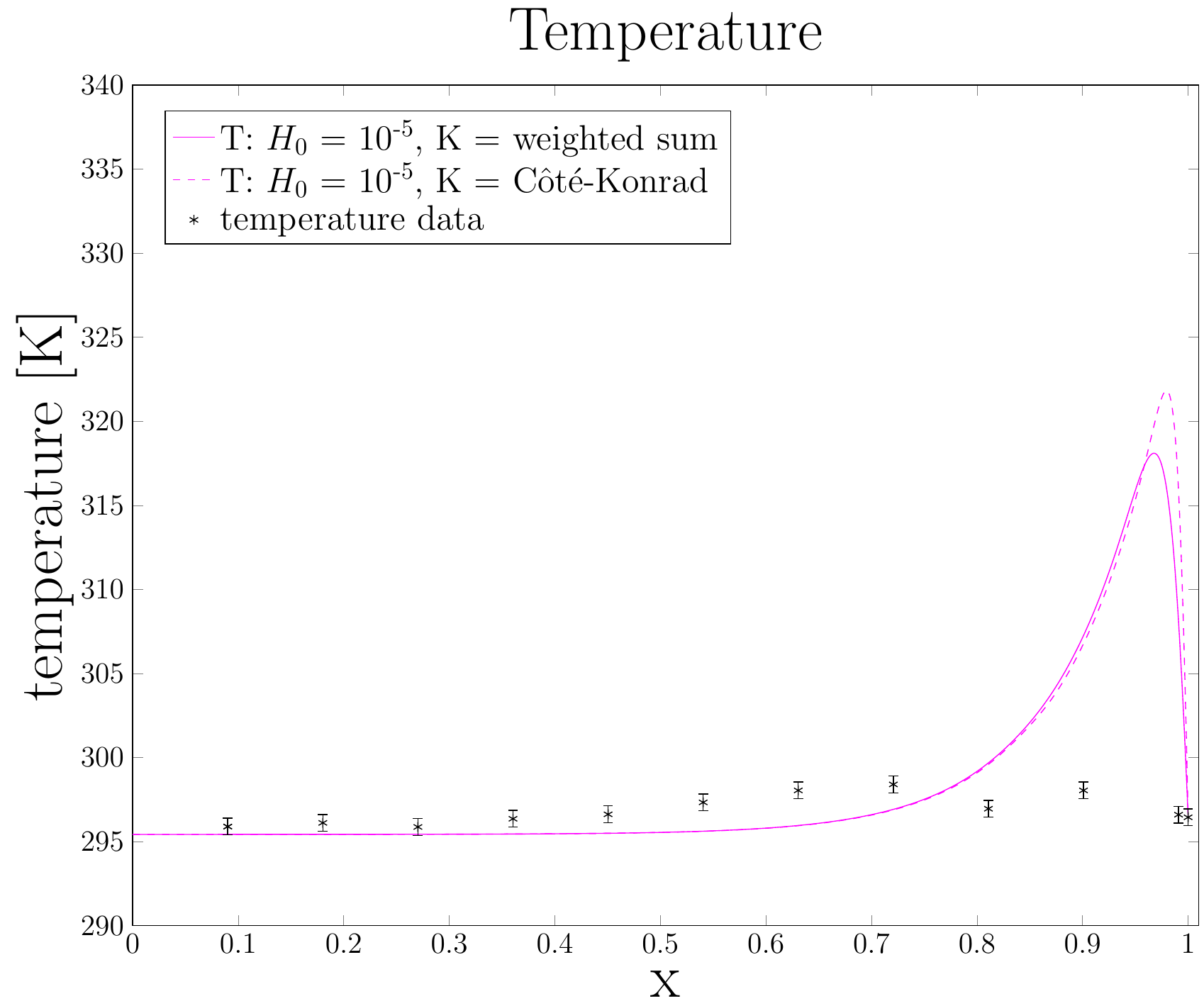}
            \label{fig:TempCompare_Cosine_t150}
        }
        \caption{Comparison of temperature solutions for the fully coupled
            saturation-diffusion-temperature model as compared to data from
            \cite{Smits2011}. Boundary conditions are taken from a sinusoidal
        approximation of boundary data.}
        \label{fig:TempCompare_Cosine}
\end{figure}
\renewcommand{\baselinestretch}{\normalspace}

\linespread{1.0}
\begin{figure}[ht!]
        \centering
        \subfigure[Comparison at $t=1 \times 600$ sec]{
            \includegraphics[width=0.45\textwidth]{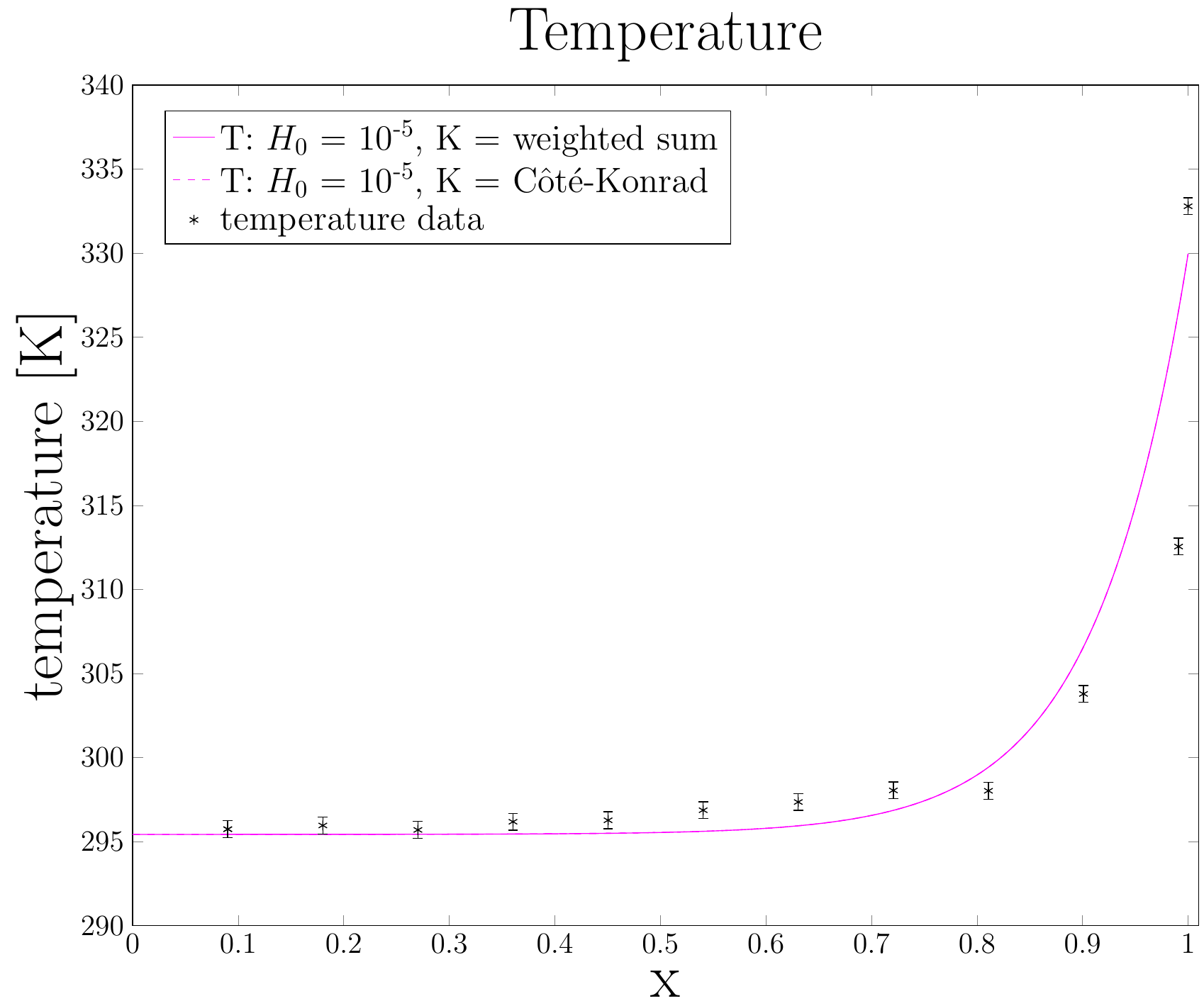}
            \label{fig:TempCompare_Square_t001}
            }
        \subfigure[Comparison at $t=50 \times 600$ sec]{
            \includegraphics[width=0.45\textwidth]{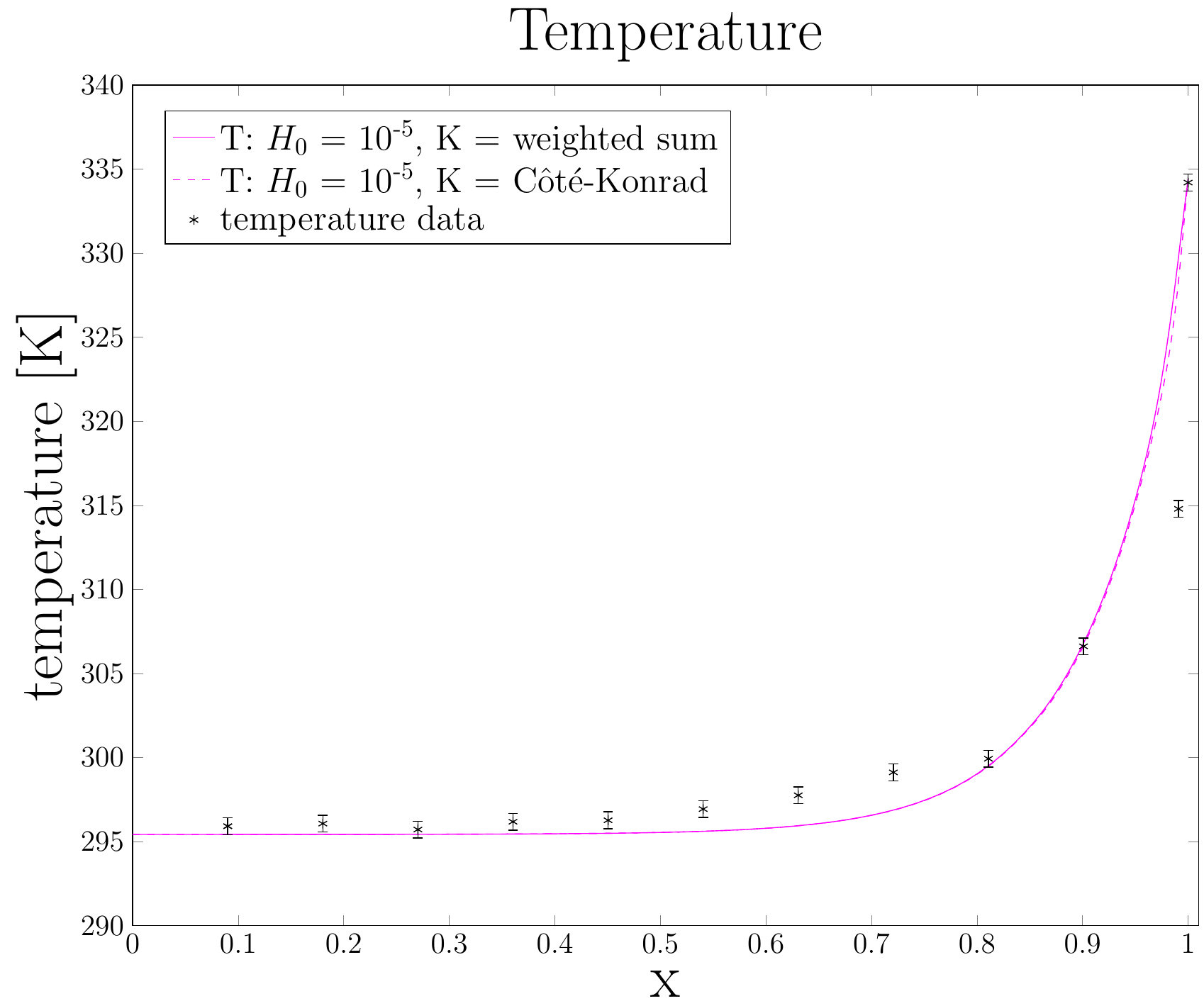}
            \label{fig:TempCompare_Square_t050}
        }
        \subfigure[Comparison at $t=100 \times 600$ sec]{
            \includegraphics[width=0.45\textwidth]{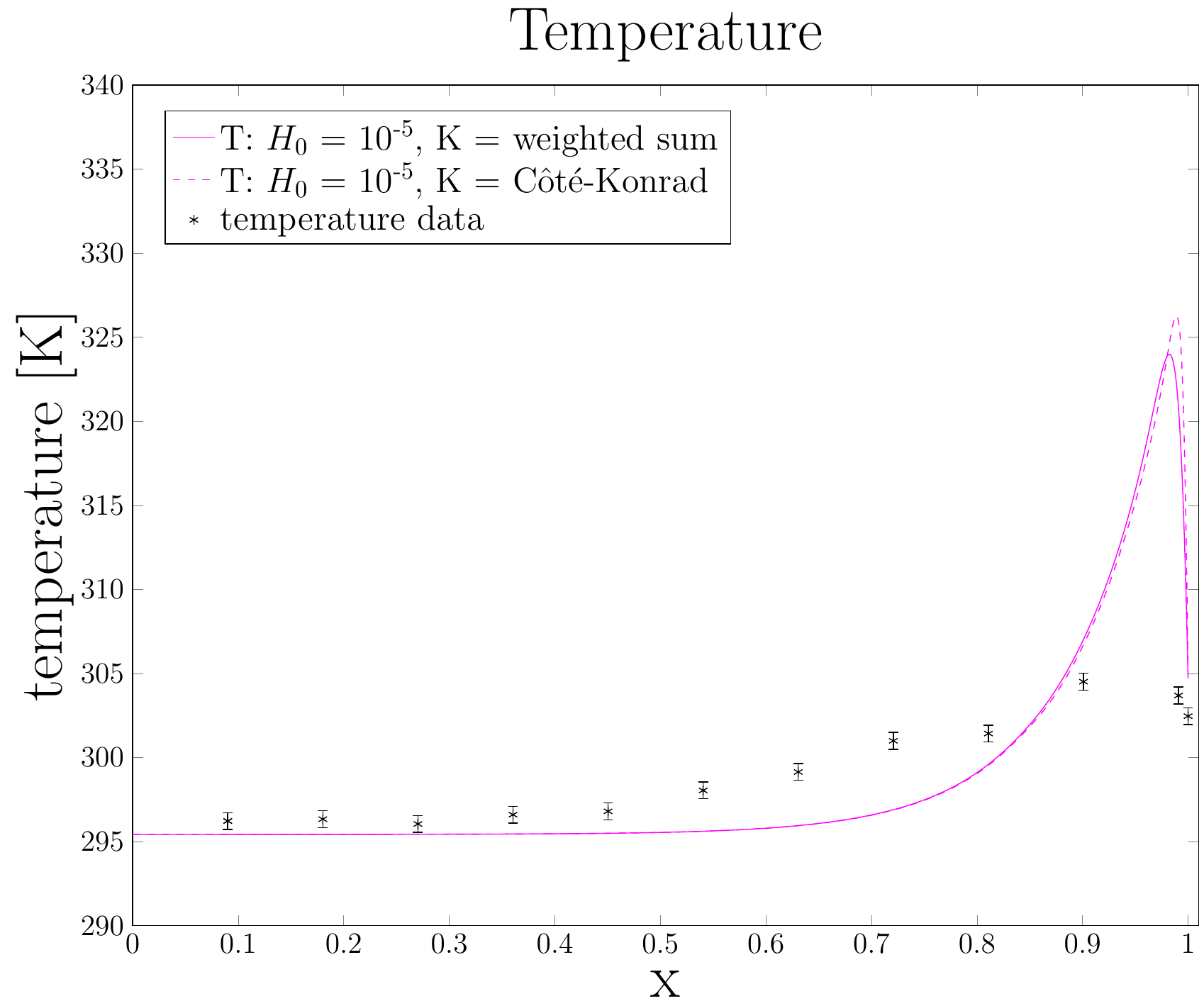}
            \label{fig:TempCompare_Square_t100}
            }
        \subfigure[Comparison at $t=150 \times 600$ sec]{
            \includegraphics[width=0.45\textwidth]{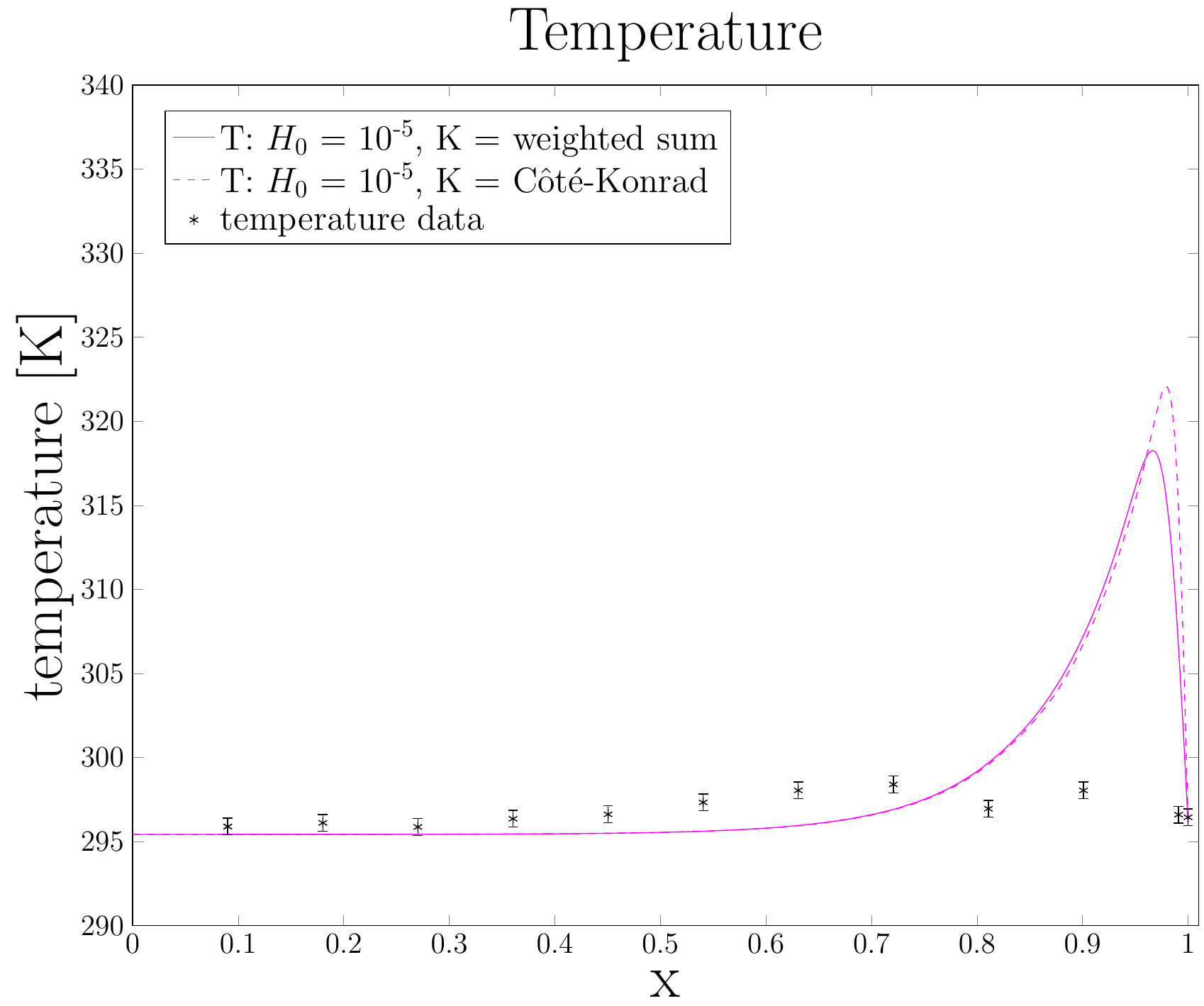}
            \label{fig:TempCompare_Square_t150}
        }
        \caption{Comparison of temperature solutions for the fully coupled
            saturation-diffusion-temperature model as compared to data from
            \cite{Smits2011}. Boundary conditions are taken from a smoothed square wave
        approximation of boundary data.}
        \label{fig:TempCompare_Square}
\end{figure}
\renewcommand{\baselinestretch}{\normalspace}

Comparison of the error estimates between the two models (Table
\ref{tab:rel_err_classical} compared to Table 
\ref{tab:rel_err_sq_wave}), we see that the present model gives slightly better
approximations as measured with this metric. Table \ref{tab:error_percent_improvement}
gives the percent improvement of the present model over the classicl model.  These values
are chosen from the tables presented herein, and as such this is a lower bound on the
percent improvement. The fact that there was an improvement in error is less important, in
the author's opinion, than the fact that the models predict nearly the same error while
(1) removing the necessity for the enhanced diffusion parameter, and (2) putting the
entire system of equations on a firm thermodynamic footing.

\linespread{1.0}
\begin{table}[ht*]
    \centering
    \begin{tabular}{|c||ccc|}
        \hline
        \multicolumn{1}{|c||}{} & \multicolumn{3}{c|}{\% Improvement} \\
        Boundary Cond. & Saturation & Rel. Humidity & Temperature \\ \hline \hline
         Sinusoidal & 1.71\%   & 22.85\%  & 16.33\% \\ \hline
         Sq. Wave   & 1.71\%   & 5.91\%   & 26.78\% \\ \hline
    \end{tabular}
    \caption{Percent improvement of the present model over the classical model using
        equation \eqref{eqn:l2_data_error_max} as the error metric.}
    \label{tab:error_percent_improvement}
\end{table}
\renewcommand{\baselinestretch}{\normalspace}

In this problem there are several parameters of interest, but the present study suggests
that variations in these parameters play little role in the overall dynamics of the
problem.  This narrows us down to only 1 fitting parameter for this problem: the
coefficient, $M$, in the rate of
evaporation term.  This was taken to best match with the fitted evaporation rate in
\cite{Smits2011} so it is expected that this value can be considered roughly constant. The classical model
consisting of Richards' equation, Philip and de Vries diffusion equation (with enhancement
fitting factors), and the de Vries heat transport equation contains at least two fitting
parameters that are calculated using a least squares statistical approach.

\section{Conclusion}
In Sections \ref{sec:numerical_saturation} - \ref{sec:numerical_coupled_sat_vap}, several
numerical results were presented indicating the consistency of the present models with the
classical mathematical models for saturation and relative humidity.  Of particular
importance is the analysis of the enhanced diffusion problem.  The arguments presented in
Section \ref{sec:numerical_vapor_diffusion} indicate that modeling vapor diffusion in
unsaturated media with the chemical potential can eliminate the necessity for a fitted
enhancement factor. Also shown within this section is a sensitivity analysis of the $\tau$
parameter for the dynamic capillary pressure term as well as the coefficient of $\grad
\rh$ that appears in the saturation equation.

In Section \ref{sec:FullyCoupledSolutions} it was shown that the fully coupled system
matches quantitatively and qualitatively to experimental data for heat and moisture
transport. There are several problems with the matching to this experimental data. First
of all, the spatial data is very coarse so getting an exact fit for the initial conditions
is difficult.  Secondly, the data is noisy so getting a reasonable fit for the boundary
conditions (especially in the initial experimental times), is difficult.  Lastly, the
Stefan nature of this problem causes numerical difficulties.  

The numerical simulations
presented herein indicate that the proposed models match both physical and experimental
expectations for a heat and moisture transport model. The model is sensitive to the choice
of thermal conductivity function and further investigation is needed to determine which
function(s) are appropriate. The slight non-physical nature of the results for certain
boundary conditions needs to be investigated.   Further studies (both
numerical and experimental) need to be performed and provide a baseline for future
research endeavors. The coupling of these three processes, all of which were derived from
a thermodynamic foundation, opens to the door to future research endeavors on coupled
processes in unsaturated soils.

\newpage
\chapter{Conclusions and Future Work}\label{ch:ConclusionsFutureWork}
Throughout this thesis it has been demonstrated that HMT and the macroscale chemcial
potential are powerful modeling tools that can be used to derive mathematical models for
rather complex phenomena in porous media. Chapter-by-chapter, the results are as follows. 

In Chapter \ref{ch:diffusion_comp}, a short discussion of different Fickian diffusion
coefficients was presented.  While this work is not {\it new} in the sense of the creation
of new theories or equations, it serves as an aid to understand the assumptions commonly
used in diffusion-related research and experimentation.  We showed that under certain
assumptions that all of the pore-scale Fickian diffusion coefficients can be assumed
constant.  In Chapter \ref{ch:HMT}, the focus turned back to porous media and the
fundamental framework for volume averaging and Hybrid Mixture Theory were presented. This
chapter serves as a reference for these topics and no new results were presented.

In Chapter \ref{ch:Exploit} we derived and generalized the primary
constitutive equations for unsaturated media.  In particular, we generalized the forms of
Darcy's and Fick's laws proposed by Bennethum \cite{Bennethum1996a}, Weinstein
\cite{weinstein}, and others. The form of Fourier's law derived is similar to that of
Bennethum \cite{Bennethum1999}, but contains terms particular to multiphase media. To the
author's knowledge, the chemical potential form of Fourier's law is new to this work. The
exact generalizations of Darcy's and Fick's laws presented here are novel to this work,
but similar terms have been proposed in other works
\cite{Bennethum1996a,Bennethum1996b,weinstein}.  What is novel to this work is the
extension of these terms to multiphase flow and the use of the macroscale chemical
potential as a dependent variable to obtain additional insight.

In Chapter \ref{ch:Transport}, a coupled system of equations for heat and moisture
transport was derived. Of particular importance are the generalizations of Richards
equation and the Philip and de Vries vapor diffusion equation.  It was demonstrated that
the enhanced diffusion model of Philip and de Vries can be re-framed in terms of the
macroscale chemical potential. This re-framing removes the necessity of the enhancement
factor proposed by Philip and de Vries. The coupling of the vapor diffusion and saturation
terms is achieved through the chemical potential. The relationship between relative
humidity and chemical potential is well known in
chemistry and thermodynamics, but to the author's knowledge it has not been previously
used in the porous media literature. Also in Chapter \ref{ch:Transport}, a generalization
of the heat transport equation given by Bennethum \cite{Bennethum1999} to multiphase media
was derived via HMT.  This new model collapses to Bennethum's model in the case of a
saturated porous medium and also suggests corrections to the classical de Vries model
proposed in 1958.

In Chapter \ref{ch:Existence_Uniqueness}, the questions of existence and uniqueness were
studied for each individual equation (while holding the other dependent variables fixed).  These
results were known for the saturation equation, but the results presented for the vapor
diffusion and thermal equation are unique to this work.  This is, of course, because these
equations are novel to this work. More work needs to be done on this front. In particular,
the uniqueness result for vapor diffusion equation is absent.  Also, the results for the
thermal equation depend on strict boundary conditions which are not necessarily met
physically. It is emphasized that the existence and uniqueness results presented herein
are only preliminary.

The numerical results in Chapter \ref{ch:Transport_Solution} are performed for validation
purposes. Whenever presenting new equations it is necessary to compare against existing
models and, when possible, experimentally obtained data. To that end, in Sections
\ref{sec:numerical_saturation}, \ref{sec:numerical_vapor_diffusion} and
\ref{sec:numerical_coupled_sat_vap}, numerical solutions to each equation were presented along with 
sensitivity analyses and comparisons to classical equations. In Section \ref{sec:FullyCoupledSolutions}, numerical
solutions to the coupled system were presented and compared to experimental data.  It was
shown that under certain parameter values the newly proposed model agrees with the
experimental data.

There are many avenues for future research left uncompleted in this work.  Relaxation of
the assumptions outlined at the beginning of Chapter \ref{ch:Transport} (Sections
\ref{sec:Assumptions} and \ref{sec:simplifying_assumptions}) provide the
initial avenues for further research.  
\begin{enumerate}
    \item It is well known that clay soils swell when wetted.  Relaxing the rigidity
        assumption on the solid phase would allow for this phenomenon. Mathematically,
        though, this creates further complications in the liquid- and gas-phase mass
        conservation equations since the divergence of the solid-phase velocity will no
        longer be zero due to the remaining terms in the solid-phase mass balance
        equation:
        \[ \epsdot{s} + \eps{s} \diver \bv{s} = \sum_{\al \ne s} \ehat{\al}{s}. \]
        The right-hand side of this equation may be zero in most cases (no dissolution or
        precipitation of solid particles), but a constitutive equation would be needed for
        $\eps{s}$ or $\epsdot{s}$.  This has further complications in that saturation, $S
        = \eps{l} / (1-\eps{s})$, now varies with solid-phase volume fraction as well as
        liquid phase volume fraction.

        The stress in the solid phase was discussed in Chapter \ref{ch:Exploit} as a
        consequence of the entropy inequality.  Upon further investigation this will give
        a generalization to the Terzhagi Stress Principle which states that the fluid
        phases can support some stress on the solid matrix. This principle will be
        necessary to express the stresses on the deforming solid.

    \item Changing the soil parameters and studying the sensitivity to empirical and
        measured relationships can be another avenue of future research.  In the present
        study the soil is
        assumed to be isotropic. For this reason it is possible to take the permeability
        as a scalar function. This is not necessarily a reasonable assumption in
        real physical problems, and one possible avenue of research is to relax this
        assumption and take a full tensor representation of the permeability. Clearly a 2-
        or 3-dimensional simulation would have to be considered in this research.
        
        Another adjustment to the soil properties is to the
        relative permeability and capillary pressure - saturation functions. The van
        Genuchten relationships were used for this work, but these are not the only
        functional forms available. 
        The Fayer-Simmons model \cite{Fayer1995}, for example, is an extension of the van
        Genuchten capillary pressure - saturation relationship that accounts for very
        small saturation content. A simple avenue for future research is to implement this
        model (and others like it) and the study to the sensitivity of the models to these
        small changes.

    \item The present model can be extended to multiple fluid phases. In this model we
        have only considered one liquid phase and one gas phase.  It is reasonable to
        assume that the gas phase is a binary ideal gas, but it is not always reasonable
        to assume that the only fluid present is a pure liquid. 
        
        One possible avenue for
        future research is to assume that the liquid phase has a dissolved contaminant. In
        that case it is not reasonable to assume that the density of the liquid is only a
        function of temperature. Furthermore, the diffusive term in the liquid phase
        equation may not be absent (depending on the type of contaminant).  
        
        Another possible case is where several liquid phases are present. In this case it
        would be necessary to include a mass conservation equation (and related Darcy-type
        constitutive law) for this other fluid phase.  This is a common case when water,
        oil, and gas are present within the porous matrix.  This type of system has seen
        recent media attention due to the practice of fracking (driving high pressure
        water and chemicals into porous rock to release oil and natural gas). 
\end{enumerate}

As mentioned previously in this chapter, an avenue for future research is to shift the
focus from modeling to analysis.  The existence and uniqueness results presented in
Chapter \ref{ch:Transport_Solution} are incomplete, and more research needs to be
completed to give a full analytic description of the behavior of these equations.  Of
particular interest is the study of the fully coupled system.  There are several papers
discussing {\it strongly coupled system} of reaction diffusion equations (e.g.
\cite{Amann1990,Le2005}).  These papers serve as a starting point to understanding the
analysis for the coupled system.

A third avenue for future research is to focus on the numerical method for solving the
coupled system.  The numerical solutions presented in Chapter \ref{ch:Transport_Solution}
were found using \texttt{Mathematica}'s \texttt{NDSolve} package.  This is a general purpose
finite difference solver designed to solve wide classes of ordinary and partial
differential equations.  Even so, there are problems with the methods used within that
package.  The foremost issue is the use of central differencing schemes. When solving
advective (or hyperbolic-type) equations it is often advantageous to use upwind-type
numerical schemes.  This is not possible with the \texttt{NDSolve} package and causes some
issues with the numerical solutions in advection-dominated simulations.

Future research for the numerical method can be approached in several ways:
\begin{enumerate}
    \item The equations of interest in this work form a system of conservation laws, and
        as such a finite volume method is likely the best choice. Peszy\'nska
        and Yi \cite{Peszynska2008} derived a cell centered finite difference method and a
        locally conservative Euler-Lagrange method based on the finite difference method
        for the saturation equation with the third-order dynamical capillary pressure
        term.  These are both similar to finite volume methods, and the methods derived
        therein can possibly serve as a basis for extension to the coupled system. This
        paper is of particular interest as the bulk of the numerical difficulties in the
        saturation equation arise as a result of the weight of the third-order term.
    \item There are other packages available to solve general classes of partial
        differential equations.  \texttt{COMSOL} (Computational Multi Physics) is one such
        package that is common amongst engineering groups.  This is a particularly
        non-mathematical approach to future research, but \texttt{COMSOL} and other
        packages can be used to test many cases of parameters and terms suggested by the
        entropy inequality.

    \item The numerical simulations used to compare to the experimental data were solved
        in one spatial dimension.  While this assumption is approximately valid for the
        experiment of interest, these equations need to be compared against
        multi-dimensional data.  One avenue of future research is to explore the numerical
        solution to the system and compare it to a two-dimensional experiment. One such
        experiment can be found in \cite{Smits2012}. This experiment is of interest since
        many of the parameters are the same as those used in the column experiment. 
\end{enumerate}

\vspace{0.2in}
\noindent {\bf Final Words} \\
The initial purpose of this work was to explore the use of the chemical potential as a
modeling tool in porous media.  It was demonstrated that the chemical potential is a
powerful modeling tool when the underlying physical processes are diffusive in nature.
The down sides to using the chemical potential are that it is indirectly measured and not
widely understood.  Within this work the main advantage to using the chemical potential
was to rewrite the gas-phase diffusion equation.  This allowed for the removal of the
enhancement factor from the classical diffusion equation.

In this work we derived three new equations that, when coupled, form a set of governing
equations for heat and moisture transport in porous media that explains previously
unexplained phenomena.  The ideas and questions proposed within this chapter set up a
research agenda for future years of scholarly work.

%
\begin{appendix}
    \chapter{Microscale Nomenclature}\label{app:pore-scale_nomenclature}
This appendix contains nomenclature for Part I: Pore-Scale modeling. While there is some
overlap in notation between the two parts, this appendix allows the notation in Chapter
\ref{ch:diffusion_comp} to stand alone. The equation references indicate the approximate first
instance of the symbol (these are hyperlinked in the digital version for ease of use).
The supporting text for these equations usually gives context and more detail.

\underline{Superscripts, Subscripts, and Other Notations}
\begin{itemize}
    \item[] $(\cd)^{\aj}$: $j^{th}$ component of $\al-$phase 
    \item[] $(\cd)^\al$: $\al-$phase
    \item[] $(\cd)^{a,b}$: difference of two quantities, $(\cd)^{a,b} =
        (\cd)^a - (\cd)^b$
    \item[] $\foten{a}$: (bold symbol) vector quantity
    \item[] $(\cd)_*$: a reference quantity or a quantity evaluated at
        a reference state
\end{itemize}

\noindent \underline{Latin Symbols}
\begin{itemize}
    \item[] $c^{g_j}$: molar concentration of the $j^{th}$ constituent in the gas phase
        [mol/length$^3$] \eqref{eqn:Ficks_molar_flux}
    \item[] $C^{\aj}$:   Mass fraction of $j^{th}$ compontnent $C^{\aj} = \rho^{\aj}/\rhoa$
        [-] \eqref{eqn:Ficks_mass_flux}
    \item[] $D$:   diffusion coefficient [length$^2$/time] \eqref{eqn:Ficks_mass_flux} -
        \eqref{eqn:Ficks_molar_flux_chempot}
    \item[] $D_\gamma$: diffusion coefficient associated with the $\gamma$-form of Fick's
        law [length$^2$/time] \eqref{eqn:Ficks_mass_flux} -
        \eqref{eqn:Ficks_molar_flux_chempot}
    \item[] $\foten{J}^{g_j}$:   flux of species $j$ in the gas phase
        [mass/(length$^2$-time)] \eqref{eqn:Ficks_mass_flux} -
        \eqref{eqn:Ficks_molar_flux_chempot}
    \item[] $m^{\aj}$:   molar mass of the $j^{th}$ constituent in phase $\al$ [mass/mol]
        \eqref{eqn:mass_molar_diffusion_coeff_comparison_v1}
    \item[] $p^{\aj}$:   partial pressure of species $j$ in phase $\al$ [force/length$^2$]
        \eqref{eqn:chempot_pore_defn}
    \item[] $p^{\al}$:   pressure of $\al$ phase [force/length$^2$] \eqref{eqn:chempot_pore_defn}
    \item[] $R$: universal gas constant [energy/(mass-temperature)] \eqref{eqn:Ficks_molar_flux_chempot}
    \item[] $R^{g_j}$:   specific gas constant for species $j$ in the
        gas phase [energy/(mass-temperature)] \eqref{eqn:Ficks_mass_flux_chempot}
    \item[] $T$:   absolute temperature \eqref{eqn:Ficks_mass_flux_chempot}
    \item[] $t$:   time \eqref{eqn:mass_balance}
    \item[] $\bv{\aj}$:   velocity of species $j$ in phase $\al$ [length/time]
        \eqref{eqn:Ficks_mass_flux}
    \item[] $\bv{\al}$:   velocity of phase $\al$ [length/time]
        \eqref{eqn:Ficks_mass_flux}
    \item[] $x^{g_j}$: molar concentration of $j^{th}$ constituent in the gas mixture [-]
        \eqref{eqn:Ficks_molar_flux}
\end{itemize}

\noindent \underline{Greek Symbols}
\begin{itemize}
    \item[] $\mu^{g_v}$:   chemical potential of water vapor in gas phase [energy/mass]
        \eqref{eqn:Ficks_mass_flux_chempot} - \eqref{eqn:Ficks_molar_flux_chempot}
    \item[] $\mu^{g_v}_*$:   chemical potential of water vapor at standard temperature
        and pressure [energy/mass] \eqref{eqn:mass_molar_diffusion_coeff_comparison_v1}
    \item[] $\rhoaj$:   mass density of species $j$ in phase $\al$
        [mass/length$^3$] \eqref{eqn:Ficks_mass_flux}
    \item[] $\rhoa$:   mass density of phase $\al$ [mass/length$^3$]
        \eqref{eqn:Ficks_mass_flux}
    \item[] $\rho_{sat}$:   mass density of water vapor under saturated
        conditions [mass/length$^3$] (paragraph before \eqref{eqn:flux_mu_c})
\end{itemize}

%
    \chapter{Macroscale Appendix} 
    \section{Nomenclature} \label{app:nomenclature}
This appendix contains nomenclature for Part II: Macroscale modeling. While there is some
overlap in notation between the two parts, this appendix allows Part II to stand alone.
This appendix also clarifies any notational discrepancies between the pore-scale and
macroscale models. The equation references indicate the approximate first instance of the
symbol (these are hyperlinked in the digital version for ease of use).  The supporting
text for these equations usually gives context and more detail.

\underline{Superscripts, Subscripts, and Other Notations}
\begin{itemize}
    \item[] $(\cd)^{\aj}$: \quad $j^{th}$ component of $\al-$phase on
        macroscale
    \item[] $(\cd)^\al$: \quad $\al-$phase on macroscale
    \item[] $\hat{(\cd)}$: \quad denotes exchange from other interface,
        phase, or component
    \item[] $(\cd)^{a,b}$: \quad difference of two quantities, $(\cd)^{a,b} =
        (\cd)^a - (\cd)^b$
    \item[] $(\cd)^j$: \quad pore scale property of component $j$
    \item[] $\foten{a}$: \quad (bold symbol) vector quantity 
    \item[] $\soten{A}$: \quad second order tensor (matrix)
    \item[] $(\cd)_0$ or $(\cd)_*$: \quad reference state
\end{itemize}

\underline{Latin Symbols}
\begin{itemize}
    \item[] $a$: \quad fitting parameter for enhanced diffusion model [-]
        \eqref{eqn:enhancement_cass}
    \item[] $b^{\aj}, b^\al$: \quad External entropy source [energy /
        (mass-time-temperature)] \eqref{eqn:micro_entropy}
    \item[] $C^{\aj}$: \quad Mass fraction of $j^{th}$ compontnent [-]
        \eqref{eqn:mass_concentration}
    \item[] $C_u^l$: \quad Coefficients of $\grad u$ terms in generalized Darcy's law
        ($u=S,T,\rh$) \eqref{eqn:Darcy_pressure_one_liquid}
    \item[] $\soten{C}^{s}$: \quad Right Cauchy-Green tensor of the solid phase $=
        (\soten{F}^{s})^T \cd \soten{F}^s$ [-] \eqref{eqn:CauchyGreen}
    \item[] $\Cbars$: \quad Modified right Cauchy-Green tensor $= (\Fbars)^T \cd \Fbars$ [-] \eqref{eqn:CauchyGreen}
    \item[] $\soten{D}^\al$: \quad  diffusivity tensor [length$^2$/time]
        \eqref{eqn:Ficks_v1}
    \item[] $\mathcal{D}$:  \quad generalized diffusivity function [length$^2$/time]
        \eqref{eqn:modified_diffusion_coefficient}
    \item[] $\bd{\al}$: \quad Rate of deformation tensor $=(\grad \bv{\al})_{sym}$
        [1/time] \eqref{eqn:entropy}
    \item[] $e^{\aj}, e^\al$: \quad energy density [energy/mass] \eqref{eqn:energy_balance_species}
    \item[] $\ehatbaj$: \quad rate of mass transfer from phase $\beta$ to component $j$ in
        phase $\al$ per unit mass density [1/time] \eqref{eqn:upscaled_cons_mass_v2}
    \item[] $\ehatba$: \quad rate of mass transfer from phase $\beta$ to phase $\al$ per
        unit mass density [1/time] \eqref{eqn:upscaled_cons_mass_summed}
    \item[] $\soten{F}^s$: \quad Deformation gradient of the solid phase [-] \eqref{eqn:DeformationGradient}
    \item[] $\Fbars$: \quad Modified deformation gradient of the solid phase [-] \eqref{eqn:DeformationGradient}
    \item[] $\foten{g}$, $\foten{g}^{\aj}, \foten{g}^\al$: \quad gravity [length/time$^2$]
        \eqref{eqn:pore_scale_mom_balance_eulerian}
        \item[] $h^{\aj}, h^{\al}$: \quad external supply of energy [energy / (mass-time)]
            \eqref{eqn:energy_balance_species}
    \item[] $\hat{\foten{i}}^{\aj}$: \quad rate of momentum gain due to interaction with
        other species within the same phase per unit mass density [force/mass] \eqref{eqn:upscaled_mom_balance_species}
    \item[] $J^s$: \quad Jacobian of the solid phase [-] \eqref{eqn:jacobian}
    \item[] $\soten{K}^\al$:  \quad hydraulic conductivity tensor for $\al$ phase
        [length/time] \eqref{eqn:mom_near_eq}
    \item[] $\soten{K}$:  \quad thermal conductivity tensor
        [energy/(mass-time-temperature)] \eqref{eqn:heat_near_eq}
    \item[] $m$: \quad van Genuchten parameter ($m=1-1/n$) [-]
        \eqref{eqn:vanGenuchten_krw}
    \item[] $n$:  \quad van Genuchten {\it pore-size distribution} parameter ($n=1/(1-m)$) [-] \eqref{eqn:vanGenuchten_krw}
    \item[] $p^\al$:  \quad classical pressure in the $\al$ phase [force/length$^2$]
        \eqref{eqn:generalpressure}
    \item[] $\pp{\al}{\beta}$:  \quad cross coupling classical pressure [force/length$^2$]
        \eqref{eqn:pressure_al_beta}
    \item[] $p_c = p^g - p^\ell$: \quad capillary pressure [force/length$^2$] \eqref{eqn:cap_pressure_wet_nonwet}
    \item[] $\overline{p}^\al$:\quad thermodynamic pressure [force/length$^2$] \eqref{eqn:pbar_total}
    \item[] $\ppbar{\al}{\beta}$:\quad cross coupling thermodynamic pressure
        [force/length$^2$] \eqref{eqn:pbar_al_beta}
    \item[] $\bq^{\al}$:\quad  heat flux for $\al$ phase [energy/(length$^2$-time)] \eqref{eqn:energy_phase}
    \item[] $\bq$: \quad total heat flux [energy/(length$^2$-time)] \eqref{eqn:total_heat_flux}
    \item[] $\bq^{\al}$: \quad Darcy flux for $\al$ phase [length/time]
        \eqref{eqn:classical_Darcy_v2}
    \item[] $\hat{Q}^{\aj}$:\quad rate of energy gain due to interaction with other species
        within the same phase per unit mass density not due to mass or momentum transfer
        [energy/(mass-time)] \eqref{eqn:energy_balance_species}
    \item[] $\hat{Q}_\al^{\beta_j}$:\quad energy transfer from phase $\al$ to constituent $j$ in
        phase $\beta$ per unit mass density not due to mass or momentum transfer
        [energy/(mass-time)] \eqref{eqn:energy_balance_species}
    \item[] $\foten{r}$:\quad microscale spatial variable [length] \eqref{eqn:local_spatial_coord}
    \item[] $\hat{r}^{\aj}$:\quad rate of mass gain due to interaction with other species within
        the same phase per unit mass density [1/time] \eqref{eqn:upscaled_cons_mass_v2}
    \item[] $R$:\quad gas constant [energy/(mol-temperature)] \eqref{eqn:ideal_gas_law}
    \item[] $R^{g_j}$:\quad  specific gas constant for $j^{th}$ constituent in gas phase
        [energy/(mass - temperature)] \eqref{eqn:diffusion_coeff_conversion}
    \item[] $\soten{R}^{\al}$:\quad resistivity tensor [time/length] \eqref{eqn:mom_near_eq}
    \item[] $S=\eps{l}/(\porosity)$:\quad  liquid saturation [-]
        \eqref{eqn:defn_saturation_v1}
    \item[] $t$:\quad  time 
    \item[] $T$:\quad  absolute temperature \eqref{eqn:Legendre}
    \item[] $\stress{\aj}$:\quad partial stress tensor for the $j^{th}$ constituent in the $\al$
        phase [force/length$^2$] \eqref{eqn:upscaled_mom_balance_species}
    \item[] $\stress{\al}$:\quad total stress tensor for $\al$ phase [force/length$^2$]
        \eqref{eqn:upscaled_mom_balance_phase}
    \item[] $\Thatbaj$:\quad rate of momentum transfer through mechanical interactions from phase
        $\beta$ to the $j^{th}$ constituent of phase $\al$ [force/length$^3$]
        \eqref{eqn:upscaled_mom_balance_species}
    \item[] $\Thatba$:\quad rate of momentum transfer through mechanical interactions from phase
        $\beta$ to phase $\al$ [force/length$^3$]
        \eqref{eqn:upscaled_mom_balance_phase}
    \item[] $\bv{\al}$:\quad  velocity of the $\al$ phase relative to a fixed coordinate
        system [length/time] \eqref{eqn:upscaled_cons_mass_summed}
    \item[] $\bv{\aj,\al} = \bv{\aj} - \bv{\al}$:\quad  diffusive velocity [length/time]
        \eqref{eqn:entropy}
    \item[] $\bv{\al,s} = \bv{\al} - \bv{s}$:\quad  velocity relative to solid phase velocity
        [length/time] \eqref{eqn:entropy}
    \item[] $\foten{w}_{\al \beta_j}$:\quad velocity of constituent $j$ at interface between
        phases $\al$ and $\beta$ [length/time] \eqref{eqn:averaging_theorem_equations}
\end{itemize}

\underline{Greek Symbols}
\begin{itemize}
    \item[] $\al$: \quad phase
    \item[] $\al$: \quad van Genuchten parameter [-] \eqref{eqn:van_Genuchten_capillary_pressure}
    \item[] $\beta$: \quad phase
    \item[] $\delta(t)$:\quad Dirac delta function [-] 
    \item[] $\delta A_{\al \beta}$: \quad Portion of the $\al \beta-$
        interface in REV \eqref{eqn:averaging_theorem_equations}
    \item[] $\eps{\al}$: \quad volumetric content of $\al$ phase per volume of REV [-] \eqref{eqn:mass_balance_averaged_terms}
    \item[] $\porosity = \eps{\ell} + \eps{g}$: \quad porosity [-]
        \eqref{eqn:porosity_defn_v2}
    \item[] $\eta^\al$: \quad specific entropy of $\al$ phase [energy / (mass -
        temperature)] \eqref{eqn:entropy_upscaled_v1}
    \item[] $\hat{\eta}^{\aj}$:\quad entropy gain due to interaction with other species within
        the same phase per unit mass density [energy/(mass-time-temperature)] \eqref{eqn:entropy_upscaled_v1}
    \item[] $\eta$: \quad enhancement factor for diffusion [-] \eqref{eqn:multiplicative_enhancement}
    \item[] $\Gamma^\al$: \quad macroscale Gibbs potential [energy/mass] \eqref{eqn:Gibbs_Duhem}
    \item[] $\gamma^\al$:\quad indicator function which is $1$ if in phase $\al$ and zero
        otherwise \eqref{eqn:averaging_theorem_equations}
    \item[] $\kappa$:\quad  permeability [length$^2$] \eqref{eqn:conductivity_permeability}
    \item[] $\hat{\Lambda}^\al$:\quad  rate of entropy production for $\al$ phase [entropy
        / time] \eqref{eqn:entropy_upscaled_v1}
    \item[] $\lambda^{\aj}$:\quad Lagrange multiplier for the continuity equation of phase
        $\al$ \eqref{eqn:entropy_modified}
    \item[] $\foten{\lambda}^{\al_N}$:\quad Lagrange multiplier for the $N^{th}$ term dependence
        relation of the components in phase $\al$ \eqref{eqn:entropy_modified}
    \item[] $\lambda = p^{g_v}/p^g$:\quad  ratio of partial pressure to bulk pressure in gas
        phase [-] \eqref{eqn:water_vapor_chem_pot_macro}
    \item[] $\mu_\al$:\quad dynamic viscosity for phase $\al$ [force-time] \eqref{eqn:conductivity_permeability}
    \item[] $\mu^{\aj}$:\quad  macroscale chemical potential of $j^{th}$ species in $\al$ phase
        [energy/mass]  \eqref{eqn:chem_pot_defn}
    \item[] $\nu_\al$:\quad kinematic viscosity for phase $\al$ \eqref{eqn:hydraulic_cond_nu}
    \item[] $\hat{\Phi}_\beta^{\aj}$: \quad Entropy transfer through mechanical
        interactions from phase $\beta$ to phase $\al$ per unit mass
        [energy/(mass-time-temperature)] \eqref{eqn:entropy_upscaled_v1}
    \item[] $\pi^\al$:\quad  wetting potential of $\al$ phase [force/length$^2$] \eqref{eqn:pi_total}
    \item[] $\pi^{\al_{(\beta)}}$:\quad  cross coupling wetting potential [force /
        length$^2$] \eqref{eqn:pi_al_beta}
    \item[] $\varphi$: \quad relative humidity [-] \eqref{eqn:mu_rho_functional_forms}
    \item[] $\psi^\al$: \quad specific Helmholtz potential of the $\al$ phase
        [energy/mass] \eqref{eqn:Legendre}
    \item[] $\rhoa$: \quad mass density of $\al$ phase (mass $\al$ per volume of $\al$) [mass /
        length$^3$] \eqref{eqn:bulk_density}
    \item[] $\rhoaj$: \quad mass density of $j^{th}$ constituent in $\al$ phase (mass $\aj$ per
        volume $\al$) [mass / length$^3$] \eqref{eqn:bulk_density}
    \item[] $\rho_{sat}$: \quad saturated vapor density [force / length$^2$] \eqref{eqn:density_rel_hum_macro}
    \item[] $\tau$: \quad tortuosity [-] \eqref{eqn:multiplicative_enhancement}
    \item[] $\tau$: \quad scaling coefficient for dynamic saturation term [-] \eqref{eqn:epsdot_near_eq}
\end{itemize}

%
    \section{Upscaled Definitions} \label{app:UpscaledDefinitions}
Definitions of bulk phase, species, and averaged variables resulting from
upscaling. Recall that an overbar indicates a mass averaged quantity and
angular brackets indicate a volume averaged quantity.
\[ \left< (\cd) \right>^\al = \frac{1}{|\delta V_\al|} \int_{\delta V} (\cd)
\gamma_\al dv \quad \text{and} \quad \overline{(\cd)}^\al =
\frac{1}{\left<\rho\right>^\al |\delta V_\al|} \int_{\delta V_\al} (\cd)
\gamma_\al dv \]
\begin{flalign}
    b^{\aj} &= \overline{b^j}^\al \\
    b^\al &= \sumj C^{\aj} b^{\aj} \\
    C^{\aj} &= \frac{\rho^{\aj}}{\rhoa} \\
    e^{\aj} &= \overline{e^j}^\al + \frac{1}{2} \overline{\bv{j} \cd
    \bv{j}}^\al - \frac{1}{2} \bv{\aj} \cd \bv{\aj} , \\
    e^\al &= \sumj C^{\aj} \left( e^{\aj} + \frac{1}{2} \bv{\aj,\al} \cd
    \bv{\aj,\al} \right) \\
    \ehatbaj &= \frac{\epsa}{|\delta V_\al|} \int_{A_{\al \beta}} \rho^j 
    \left( \foten{w}_{\al \beta_j} - \bv{j} \right) \cd \foten{n}^\al da \\
    \ehatba &= \sumj \ehatbaj \\
    \foten{g}^{\aj} &= \overline{\foten{g}^j}^\al \\
    \foten{g}^\al &= \sumj C^{\aj} \foten{g}^{\aj} \\
    h^{\aj} &= \overline{h^j}^{\al} + \overline{\foten{g}^j \cd
    \bv{j}}^{\al} - \foten{g}^{\aj} \cd \bv{\aj}, \\
    h^\al &= \sumj C^{\aj} \left( h^{\aj} + \foten{g}^{\aj} \cd
    \bv{\aj,\al} \right) \\
    \hat{\foten{i}}^{\aj} &= \epsa \rhoaj \left(
    \overline{\hat{\foten{i}}^j}^\al + \overline{\hat{r}^j \bv{j}}^\al -
    \bv{\aj} \overline{\hat{r}^j}^\al \right) \\
    \notag \bq^{\aj} &= \left< \bq^j \right>^\al + \left< \stress{j} \cd
    \bv{j} \right>^\al - \stress{\aj} \cd \bv{\aj} + \rho^{\aj} \bv{\aj}
    \left( e^{\aj} + \frac{1}{2} \bv{\aj} \cd \bv{\aj} \right) \\
    & \quad - \rhoaj \overline{\bv{j} \left( e^j + \frac{1}{2} \bv{j} \cd
    \bv{j} \right)}^\al , \\
    \bq^\al &= \sumj \left[ \bq^{\aj} + \stress{\aj} \cd \bv{\aj,\al} -
    \rhoaj \left( e^{\aj} + \frac{1}{2} \bv{\aj,\al} \cd \bv{\aj,\al}
    \right) \bv{\aj,\al} \right] \\
    \notag \hat{Q}^{\aj} &= \epsa \rhoaj \left[ \overline{\hat{Q}^j}^\al +
    \overline{\foten{i}^j \cd \bv{j}}^\al - \left(
    \overline{\hat{\foten{i}}^j}^\al + \overline{\hat{r}^j \bv{j}}^\al -
    \bv{\aj} \overline{\hat{r}^j}^\al \right)\cd \bv{\aj} \right. \\
    & \quad \left. + \overline{\hat{r}^j \left( e^j + \frac{1}{2} \bv{j}
    \cd \bv{j}\right)}^\al - \overline{\hat{r}^j}^\al \left(
    \overline{e^j}^\al + \frac{1}{2} \overline{\bv{j} \cd \bv{j}
    }^\al\right)\right] , \\
    \notag \hat{Q}_{\beta}^{\aj} &= \frac{\epsa}{|\delta V_\al|} \left\{
    \int_{A_{\al \beta}} \left[ \bq^j + \stress{j}\cd \bv{j} + \rho^j
    \left( e^j + \frac{1}{2} \bv{j} \cd \bv{j} \right)\left(\foten{w}_{\al
    \beta_j} - \bv{j} \right) \right] \cd \foten{n}^{\al} da . \right. \\
    \notag & \qquad \qquad - \left( e^{\aj} - \frac{1}{2} \bv{\aj} \cd
    \bv{\aj} \right) \int_{A_{\al \beta}} \rho^j \left( \foten{w}_{\al
    \beta_j} - \bv{j} \right) \cd \foten{n}^\al da \\
    & \qquad \qquad \left. - \bv{\aj} \int_{A_{\al \beta}} \left[ \stress{j}
    + \rho^j \bv{j} \left( \foten{w}_{\al \beta_j} - \bv{j} \right)\right]
    \cd \foten{n}^\al da \right\} \\
    \hat{Q}_\beta^\al &= \sumj \left[ \hat{Q}_\beta^{\aj} + \Thatbaj \cd
    \bv{\aj,\al} + \ehatbaj \left( e^{\aj} + \frac{1}{2} \bv{\aj,\al} \cd
    \bv{\aj,\al} \right) \right] \\
    \hat{r}^{\aj} &= \epsa \rhoaj \overline{\hat{r}^j}^\al \\
    \stress{\aj} &= \left< \stress{j} \right>^{\al} + \rho^{\aj} \bv{\aj}
    \otimes \bv{\aj} - \rhoaj \overline{\bv{j} \otimes \bv{j}}^\al \\
    \stress{\al} &= \sumj \left( \stress{\aj} - \rho^{\aj} \bv{\aj,\al}
    \otimes \bv{\aj,\al} \right) \\
    \notag \Thatbaj &= \frac{\epsa}{|\delta V_\al|} \left[ \int_{A_{\al
    \beta}} \left[ \stress{j} + \rho^j \bv{j} \left( \foten{w}_{\al
    \beta_j} - \bv{j} \right) \right] \cd \foten{n}^{\al} da \right. \\
    & \qquad \qquad \left. - \bv{\aj} \int_{A_{\al \beta}} \rho^j \left(
    \foten{w}_{\al \beta_j} - \bv{j} \right) \cd \foten{n}^\al da \right]
    \\
    \Thatba &= \sumj \left( \Thatbaj + \ehatbaj \bv{\aj,\al} \right) \\
    \bv{\aj} &= \overline{\bv{j}}^\al \\
    \bv{\al} &= C^{\aj} \bv{\aj} \\
    \eta^{\aj} &= \overline{\eta^j}^\al \\
    \eta^{\al} &= \sumj C^{\aj} \eta^{\aj} \\
    \hat{\eta}^{\aj} &= \epsa \rhoaj \left( \overline{\hat{\eta}^j}^{\al} +
    \overline{\hat{r}^j \eta^j}^\al - \overline{\hat{r}^j}^\al \eta{\aj}
    \right) \\
    \hat{\Lambda}^{\aj} &= \overline{\hat{\Lambda}^j}^\al \\
    \hat{\Lambda}^{\al} &= \sumj \hat{\Lambda}^{\aj} \\
    \psiaj &=  \overline{\psi^j}^{\al} \\
    \psia &= \sumj C^{\aj} \psiaj \\
    \rho^{\aj} &= \overline{\rho^j}^\al \\
    \rhoa &= \sumj \rhoaj 
\end{flalign}

%
    \section{Identities Needed to Obtain Inquality \ref{eqn:entropy}} \label{app:EntropyIdentities}
\begin{flalign}
    \notag \sumj \left( \frac{\epsa \rhoaj}{T} \md{\aj}{\psiaj} \right) &= \frac{\epsa \rhoaj}{T}
    \md{\aj}{\psia} + \frac{\psia}{T} \ehatba \\
    \notag & \quad + \sumj \left\{ \frac{1}{T} \bv{\aj,\al} \cd \grad \left( \epsa \rhoaj
    \psiaj \right) - \frac{\psiaj}{T} \ehatbaj \right. \\
    & \qquad \qquad \quad \left. - \frac{\psiaj}{T} \rhataj - \frac{\epsa \rhoaj}{T}
    \psiaj \diver \bv{\aj,\al} \right\} \\
    \sumj \left( \frac{\epsa \rhoaj}{T} \eta^{\aj} \md{\aj}{T} \right) &=
    \frac{\epsa \rhoa}{T} \md{\al}{T} + \sumj \frac{\epsa \rhoaj}{T} \eta^{\aj}
    \bv{\aj,\al} \cd \grad T \\
    \sumj \left( \frac{\epsa}{T} \stress{\aj} : \grad \bv{\aj} \right) &= \sumj \left\{
        \frac{\epsa}{T} \stress{\aj} : \grad \bv{\aj,\al} + \frac{\epsa}{T} \stress{\aj} :
        \grad \bv{\aj} \right\} \\
    \sumj \sumba \hat{\Psi}^{\aj}_\beta &= -\sumj \sumba \ehatbaj \eta^{\aj} \\
    \sumj \hat{Q}^{\aj} &= -\sumba \left[ \ihataj \cd \bv{\aj,\al} + \rhataj \left( \psiaj
        + T \eta^{\aj} \frac{1}{2} \left( \bv{\aj,\al} \right)^2 \right) \right] \\
    \notag \sumj \sumba \hat{Q}^{\aj}_\beta &= -\suma \sumba \left\{ \Thatba \cd \bv{\al,s} +
    \frac{1}{2} \ehatba \left( \bv{\al,s} \right)^2 \right. \\
    \notag & \quad +\left. \sumj \left[ \Thatbaj \cd \bv{\aj,\al} + \frac{1}{2} \ehatbaj
        \left( \bv{\aj,\al} \right)^2 \right] \right\} \\
    & \quad + \sumba \sumj \left\{ \ehatbaj \left( \psiaj + T \eta^{\aj} \right) \right\}
\end{flalign}

%
    \chapter{Exploitation of the Entropy Inequality -- An Abstract Perspective} \label{app:AbstractEntropy}
This short appendix is meant to give a brief and abstract description of how the
exploitation of the entropy inquality works.  By ``abstract'' we mean that we will not
assign any physical meaning to the variables.  Instead we will simply state how the
variables relate to each other and how they relate to the full set of chosen independent
variables.  The secondary purpose of this appendix is to make clear a few assumptions
related to constitutive equations that are necessary in order for the exploitation of the
entropy inequality to be successfull.  We conclude with an inequality that dictates how
the linearization of the constitutive relations must behave in order not to violate the
second law of thermodynamics. This is similar to the 1968 Nobel Prize winning analysis by
Onsager, who showed the {\it reciprocal relations} that must hold at equilibrium for
irreversible processes.

Let $S$ be the set of all independent variables for the Helmholtz potential. This set
defines the physical system of interest. Define the following sets:
\begin{itemize}
    \item $\{ x_j \} := $ the set of all variables that are neither constitutive nor
        independent.  Examples typically include $\dot{T}, \dot{\grad T}, \dot{\grad \rho^{l_j}},
        \cdots$.
    \item $\{ y_k \} := $ the set of all constitutive variables which are zero at
        equilibrium.  Examples typically include
        $\epsdot{\al}$, $\ehat{\al}{\beta_j}$, and $\hat{r}^{\aj}$.
    \item $\{ \tilde{y}_\kappa \} :=$ the set of all constitutive variables which are not
        zero at equilibrium.  Examples typically include $\Thatbaj$, and $\ihataj$.
    \item $\{ z_l \} := $ the set of all variables that are zero at equilibrium. Examples
        typically include $\grad T$, $\bv{\al,s}$, $\bv{\aj,\al}$, and $\bd{\al}$.
\end{itemize}
Since each $z_l$ is an independent variable it is clear that $\{ z_l \} \subset S$.
Furthermore we observe that $\{ x_j \} \cap S = \emptyset.$ The constitutive variables, on
the other hand, are known to be functions of variables in $S$ and as such the statement that
$\left(  \{ y_k \} \cup \{ \tilde{y}_\kappa\} \right) \cap S = \emptyset$ is easily
misinterpreted. It is a true statement that $\left( \{
y_k \} \cup \{ \tilde{y}_\kappa \}
\right) \cap S = \emptyset$, and it is correct not to choose constitutive variables as
independent variables. The confusion is in the fact that $y_k = y_k(S)$ for all $k$ (and
for all $\kappa$).

The rate of entropy generation, $\hat{\Lambda}$, {\it can} be written as a linear
combination of the variables from $\{ x_j \}$, $\{ y_k \}$, and $\{ z_l \}$, where the
coefficients are functions of variables from $S$. That is,
\begin{flalign}
    0 \le \suma \hat{\Lambda}^\al = \hat{\Lambda} &= \sum_j \left( x_j X_j \right) +
    \sum_k \left( y_k Y_k \right) + \sum_l \left( z_l Z_l \right) \ge 0
    \label{eqn:entropy_full}
\end{flalign}
where
\begin{flalign}
    X_j = X_j(S,\tilde{Y}_\kappa (S)), \quad Y_k = Y_k(S,\tilde{Y}_\kappa (S)), \quad
    \text{and} \quad Z_l = Z_l(S,\tilde{Y}_\kappa (S)).
\end{flalign}
This is not the only way to algebraically rearrange $\hat{\Lambda}$, but this is what is
commonly done during the exploitation process.

We now use inequality \eqref{eqn:entropy_full} to derive equations that hold for all time,
at equilibrium, and near equilibrium.

\section{Results that Hold For All Time}
We have {\it no control} over the variables $x_j$ since they are neither constitutive nor
independent.  This means that they could be positive or negative, large or small.
In order for the second law of thermodynamics to hold for all time, the coefficients $X_j$
must therefore be zero for all time.  This implies that 
\begin{flalign}
    \hat{\Lambda} &= \sum_k \left( y_k Y_k \right) +
    \sum_l \left( z_l Z_l \right) \ge 0.
    \label{eqn:entropy_simp1}
\end{flalign}

\section{Equilibrium Results}
The definition of equilibrium is {\it the state at which all of the variables $\{ y_k \}$
are zero}.  This definition is based on physical intuition and will vary depending on the
system of interest.  From thermodynamics, the rate of entropy generation must be minimized
at equilibrium.  This implies that the gradient of the entropy generation function must be
the zero vector (as understood with $S$ as the independent variables for the gradient).
\begin{flalign}
    0 = \pd{\hat{\Lambda}}{s_i} \Big|_{eq} \, \text{ for all } i.
\end{flalign}
Taking this partial derivative of the right-hand side of \eqref{eqn:entropy_simp1} we see
that
\begin{flalign}
    0 = \pd{\hat{\Lambda}}{s_i} \Big|_{eq} &= \left[ \sum_k \left( y_k \pd{Y_k}{s_i} \right) + \sum_k \left(
    Y_k \pd{y_k}{s_i} \right) + \sum_l \left( z_l \pd{Z_l}{s_i} \right) + \sum_l \left(
    Z_l \pd{z_l}{s_i} \right) \right]_{eq}.
\end{flalign}
Since $z_l |_{eq} = 0 = y_k |_{eq}$ for all $l,k$ and $\pd{z_l}{s_i} = \delta_{il}$ (since
$z_l \in S \, \forall l$) we get
\begin{flalign}
    0 = \pd{\hat{\Lambda}}{s_i} \Big|_{eq} &= \left[ \sum_k \left(
        Y_k \pd{y_k}{s_i} \right) + Z_i \delta_{il} \right]_{eq}.
\end{flalign}

At this point we make an assumption that greatly affects the constitutive variables.
\begin{description}
    \item[Assumption:] At equilibrium we must have
        \begin{flalign}
            \pd{y_k}{z_l} \Big|_{eq} = 0 \text{ for all } l,k
        \end{flalign}
\end{description}
Under this assumption it is clear that $Z_l = 0$ at equilibrium for all $l$.  Notice that
this says nothing about when $s_i \not \in \{ z_l \}$. From this argument, each equation
$Z_l = 0$ gives a constraint on some of the variables in $S$.  The assumption made can be
viewed as a further restriction on the constitutive variables, but it is not clear whether
this assumption is physical.


\section{Near Equilibrium Results}
For the near equilibrium results we consider two types of variables: variables that are zero
at equilibrium and constitutive variables. In each case we linearize about the equilibrium
state.  A typical linearization result for variables which are zero at equilibrium is
\begin{flalign}
    \notag Z_l |_{n.eq} &= \underbrace{(Z_l)|_{eq}}_{=0} + \sum_m \pd{Z_l}{z_m} \Big|_{eq}
    z_m + \sum_n \pd{Z_l}{y_n} \Big|_{eq} y_n + \cdots \\
    &= \sum_m \pd{Z_l}{z_m} \Big|_{eq} z_m + \sum_n \pd{Z_l}{y_n}  \Big|_{eq} y_n + \cdots.
    \label{eqn:linearization1}
\end{flalign}
The value of $(Z_l)|_{eq}$ is zero by the above arguments, and the partial derivatives are
now functions of all of the other variables which are not zero at equilibrium: 
\begin{flalign}
    C_{lm} &:= \pd{Z_l}{z_m} \Big|_{eq} = \pd{Z_l}{z_m} \Big|_{eq}(\xi_1, \xi_2, \dots),
    \\
    D_{ln} &:= \pd{Z_l}{y_n} \Big|_{eq} = \pd{Z_l}{y_n} \Big|_{eq}(\xi_1, \xi_2,
    \dots)
\end{flalign}
for $\xi_n \in S \setminus \{ z_l \}.$ Written more simply
\begin{flalign}
    Z_l |_{n.eq} = C_{lm} z_m + D_{ln} y_n
\end{flalign}
where the summations are implicit over repeated indices.

For the constitutive variables we do a similar linearization, but note that the
equilibrium state is not necessarily zero.  Therefore,
\begin{flalign}
    Y_k |_{n.eq} = (Y_K)|_{eq} + \sum_p \pd{Y_k}{y_p} \Big|_{eq} y_p + \sum_q
    \pd{Y_k}{z_q} \Big|_{eq} z_q + \cdots.
    \label{eqn:linearization2}
\end{flalign}
Making similar definitions as before,
\begin{flalign}
    E_{kp} &:= \pd{Y_k}{y_p} \Big|_{eq} \\
    F_{kq} &:= \pd{Y_k}{z_q} \Big|_{eq},
\end{flalign}
the linearization result for the constitutive equations is
\begin{flalign}
    Y_k |_{n.eq} = Y_k|_{eq} + E_{kp} y_p + F_{kq} z_q
\end{flalign}
where the summations are implicit over repeated indices.

The trouble here is that we must have some information about the equilibrium state of the
constitutive variable. The presumption that this is zero may be non-physical.  An example
of this is the capillary pressure relationship derived in multiphase unsaturated media.  
\begin{flalign}
    p_c = (p_c)|_{eq} + \tau \epsdot{l},
\end{flalign}
where $p_c$ is the capillary pressure, $\tau = \partial p_c / \partial \epsdot{l}$, and
the equilibrium capillary pressure is given as a function of saturation $p_c |_{eq} =
p_c(S)$ via the van Genuchten approximation. This equilibrium constitutive equation is
known not to be zero.  In other systems the issue may be more subtle, but in any case
one needs to have some information (whether from experiments or from other theory) to
define the equilibrium state of the constitutive variable.

\section{Linearization and Entropy}
Consider again equation \eqref{eqn:entropy_simp1}, but now substitute the linearized
results into $Y_k$ and $Z_l$
\begin{flalign}
    \notag 0 \le \hat{\Lambda} &= \left( y_k Y_k \right) + \left( z_l Z_l \right) \\
    \notag &= y_k \left( Y_k|_{eq} + E_{kp} y_p + F_{kq} z_q \right) + z_l \left( C_{lm}
    z_m + D_{ln} y_n \right) \\
    &= y_k Y_k|_{eq} + y_k E_{kp} y_p + y_k F_{kq} z_q + z_l C_{lm} z_m + z_l D_{ln} y_n
\end{flalign}
(summations are again taken over repeated indices). Recognizing the quadratic terms as
matrix products and rewriting in block matrix form gives
\begin{flalign}
    0 &\le  y_k Y_k|_{eq} + \begin{pmatrix} \foten{y} \\ \foten{z} \end{pmatrix}^T
        \begin{pmatrix} \soten{E} & \soten{F} \\ \soten{D} & \soten{C} \end{pmatrix}
            \begin{pmatrix} \foten{y} \\ \foten{z} \end{pmatrix}
     = y_k Y_k|_{eq} + \boldsymbol{\zeta}^T
    \soten{\mathcal{A}} \boldsymbol{\zeta} \label{eqn:entropy_block_form}
\end{flalign}

Notice that in the absence of constitutive variables \eqref{eqn:entropy_block_form}
simplifies to 
\begin{flalign}
    0 \le \foten{z}^T \soten{C} \foten{z}.
\end{flalign}
Simply stated this means that $\soten{C}$ must be positive semidefinite in order for the
second law of thermodynamics to hold.  This is Onsager's Nobel Prize winning result.
Recall that $Z_l|_{n.eq} = C_{lm} z_m$. If we take, for example, $z_m = \grad T$ and
observe that $Z_l|_{n.eq}$ is minus the heat flux near equilibrium (the coefficient in the
entropy inequality associated with $\grad T$ is minus the heat flux), then the $l-m$ entry
in $\soten{C}$ is the heat flux tensor.  Onsager's result dictates the positivity of the
heat flux tensor and the linearized result give Fourier's law: $-\foten{q} = \soten{K}
\grad T$.  In other words, there is no accident that many physical ``laws'' take the same
form as Fourier's law; they are a result of the non-negativity of $\soten{C}$ and the
entropy inequality. This is a simple example, but it should help to elucidate the problem that
arises when constitutive equations are introduced.

Returning to equation \eqref{eqn:entropy_block_form} we see that it is not immediately
obvious that $\soten{\mathcal{A}}$ needs to be positive semidefinite. In fact, the only
way that we can guarantee that $\soten{\mathcal{A}}$ has this property is if $y_k
Y_k|_{eq} \le 0$.  That is, there must be a physical restriction on $y_k Y_k|_{eq}$ that,
when violated, one perceives nonphysical results. 

To give a physical example of this we return to the capillary pressure example.
In this case we have $y_k = \epsdot{l}$ and $Y_k|_{eq} =
p_c(S)$ where we are ignoring all other constitutive equations (or we are taking
$Y_r|_{eq}=0$ for all $r\ne k$). The restriction derived
herein states that $p_c(S) \epsdot{l} \le 0$ for all time. The time derivative can clearly
take either sign, but what this seems to be indicating is that in drainage (when
$\epsdot{l} < 0$) the equilibrium capillary pressure must be positive, and in imbibition
(when $\epsdot{l} > 0$) the equilibrium capillary pressure must be negative. This is a bit
contradictory since ``drainage'' and ``imbibition'' are non-equilibrium phenomena, and as
such it is not possible to measure the {\it equilibrium} capillary pressure at these
states. This leaves us with a conundrum: Is there a fundamental misinterpretation of the
capillary pressure in this example, or is there is a deep-seated flaw in the exploitation
of the entropy inequality.

    \chapter{Summary of Entropy Inequality Results}\label{app:EntropyResults}
The following is a concise collection of the results derived from the entropy inequality.
This appendix is to be used for reference when building the macroscale models.

\section{Results that Hold For All Time}
\begin{itemize}
    \item Helmholtz potential and entropy are conjugate variables (equation
        \eqref{eqn:Helmholtz_entropy_conjugate})
        \begin{flalign}
            \pd{\psia}{T} = -\eta^{\al}.
        \end{flalign}

    \item Lagrange multiplier for fluid phase (equation \eqref{eqn:lagrange_betaj}) 
        \begin{flalign}
            \lambda^{\beta_j} = \suma \frac{\epsa \rhoa}{\eps{\beta}}
            \pd{\psia}{\rho^{\beta_j}}.
        \end{flalign}

    \item Lagrange multiplier for the dependence of the diffusive velocities on the
        $N^{th}$ species (equation \eqref{eqn:lagrange_lamhataN})
        \begin{flalign}
            \lamhataN &= -\frac{1}{\rhoa} \sumj \left[ \stress{\aj} + \left( \rhoaj
            \lamaj \right) \soten{I} \right] + \psia \soten{I}.
        \end{flalign}

    \item Solid phase pressure (equation \eqref{eqn:solid_pressure} )
        \begin{flalign}
            p^s = -\frac{1}{3} tr \left( \stress{s} \right) &= -\frac{J^s}{\eps{s}} \suma
            \left( \eps{\al} \rho^{\al} \pd{\psi^{\al}}{J^s} \right).
        \end{flalign}

    \item Solid phase stress (equation \eqref{eqn:solid_stress})
        \begin{flalign}
            \stress{s} = -p^s \soten{I} + \stress{s}_{e} + \frac{\eps{l}}{\eps{s}}
            \stress{l}_{h} + \frac{\eps{g}}{\eps{s}} \stress{g}_{h} 
        \end{flalign}
        where
        \begin{subequations}
            \begin{flalign}
                \stress{s}_{e} &= 2\left( \rho^{s} \Fbars \cd \pd{\psi^{s}}{\Cbars} \cd \left(
                \Fbars \right)^T -\frac{1}{3} \rho^{s} \pd{\psi^{s}}{\Cbars} : \Cbars \soten{I}
                \right), \\
                \stress{l}_{h} &= 2\left( \rho^{l} \Fbars \cd \pd{\psi^{l}}{\Cbars} \cd \left(
                \Fbars \right)^T -\frac{1}{3} \rho^{l} \pd{\psi^{l}}{\Cbars} : \Cbars \soten{I}
                \right), \\
                \stress{g}_{h} &= 2\left( \rho^{g} \Fbars \cd \pd{\psi^{g}}{\Cbars} \cd \left(
                \Fbars \right)^T -\frac{1}{3} \rho^{g} \pd{\psi^{g}}{\Cbars} : \Cbars \soten{I}
                \right).
            \end{flalign}
        \end{subequations}
\end{itemize}

\section{Equilibrium Results}
\begin{itemize}
    \item Fluid pressures (equations \eqref{eqn:generalpressure},
        \eqref{eqn:pressure_al_beta}, \eqref{eqn:pressure_sum_al_beta}, and
        \eqref{eqn:pbar_al_beta} - \eqref{eqn:three_pressures})

        \begin{flalign}
            p^{\beta} = -\frac{1}{3} tr \left( \stress{\beta} \right) &= \sumj \suma
            \left( \frac{\eps{\al} \rho^{\al} \rho^{\beta_j}}{\eps{\beta}}
            \pd{\psi^{\al}}{\rho^{\beta_j}} \right) \\
            \pp{\al}{\beta} &= \sumj \left( \frac{\eps{\al} \rho^{\al}
            \rho^{\beta_j}}{\eps{\beta}} \left. \pd{\psia}{\rho^{\beta_j}}
            \right|_{\eps{\al}, \rho^{\al_k}, \eps{\beta}, \rho^{\beta_m}} \right) \\
            p^{\beta} &= \suma \pp{\al}{\beta} \\
            \ppbar{\al}{\beta} &:= -\eps{\al} \rho^{\al} \left.
            \pd{\psia}{\eps{\beta}} \right|_{\eps{\al}, \eps{\al} \rho^{\al_k},
            \eps{\beta} \rho^{\beta_k}} \\
            \ppi{\al}{\beta} &:= \eps{\al} \rho^{\al} \left.
            \pd{\psia}{\eps{\beta}} \right|_{\eps{\al}, \rho^{\al_k},
            \rho^{\beta_k}} \\
            \pp{\al}{\beta} &= \ppbar{\al}{\beta} + \ppi{\al}{\beta} \\
            \overline{p}^{\beta} &:= \suma \ppbar{\al}{\beta} \\
            \pi^{\beta} &:= \suma \ppi{\al}{\beta} \\
            p^{\beta} &= \overline{p}^{\beta} + \pi^{\beta}. 
        \end{flalign}

%

    \item Momemtum transfer between phases (equation
        \eqref{eqn:phase_mom_transfer_equilib})
        \begin{flalign}
            \notag -\left( \That{s}{\beta} + \That{\gamma}{\beta} \right) &= \left(
            \eps{\beta} \rho^{\beta} \pd{\psi^{\beta}}{\eps{\beta}} - p^{\beta} \right)
            \grad \eps{\beta} + \eps{\beta} \rho^{\beta} \pd{\psi^{\beta}}{\eps{\gamma}}
            \grad \eps{\gamma}  \\
            \notag & \quad + \eps{\beta} \rho^{\beta} \sumjj \pd{\psi^{\beta}}{C^{s_j}}
            \grad C^{s_j}+ \eps{\beta} \rho^\beta \pd{\psi^\beta}{\epsdot{l}} \grad
            \epsdot{\beta} + \eps{\beta} \rho^\beta \pd{\psi^\beta}{\epsdot{\gamma}} \\
            \notag & \quad - \sumj \left[ \left( \eps{\gamma} \rho^{\gamma}
                \pd{\psi^{\gamma}}{\rho^{\beta_j}} + \eps{s} \rho^s
                \pd{\psi^s}{\rho^{\beta_j}} \right) \grad \rho^{\beta_j} \right]+
                \eps{\beta} \rho^{\beta} \sumj \pd{\psi^{\beta}}{\rho^{\gamma_j}} \grad
                \rho^{\gamma_j} \\ 
           & \quad + \eps{\beta} \rho^{\beta} \pd{\psi^{\beta}}{J^s} \grad J^s +
           \eps{\beta} \rho^{\beta} \pd{\psi^{\beta}}{\Cbars} : \grad \left( \Cbars
           \right), 
       \end{flalign}

   \item Momentum transfer between constiuents (equation \eqref{eqn:species_mom_transfer_equilib})
       \begin{flalign}
           \sumba \Thatbaj + \ihat^{\aj} &= -\grad \left( \epsa \rhoaj \psiaj \right) +
           \lamaj \grad \left( \epsa \rhoaj \right) + \psia \grad \left( \epsa \rhoaj
           \right).
       \end{flalign}

   \item Partial heat flux (equation \eqref{eqn:eq_partial_heat_flux})
       \begin{flalign}
           \suma \epsa \bq^\al = \foten{0} 
       \end{flalign}

   \item Chemical potential definition (equations \eqref{eqn:chem_pot_defn} and
       \eqref{eqn:chem_pot_defn2})
       \begin{flalign}
           \mu^{\beta_j} &= \left. \pd{\psi_T}{\left( \eps{\beta} \rho^{\beta_j} \right)}
           \right|_{\epsa, \eps{\beta}, \rho^{\al_k}, \rho^{\beta_m}} = \left. \suma
           \pd{\left( \epsa \rhoa \psia \right)}{\left( \eps{\beta} \rho^{\beta_j}
       \right)} \right|_{\epsa, \eps{\beta}, \rho^{\al_k}, \rho^{\beta_m}} \\
       &= \psi^\beta + \lambda^{\beta_j}
   \end{flalign}

   \item Mass transfer (equation \eqref{eqn:mass_transfer_linearized})
       \begin{flalign}
           \mu^{l_j} \Big|_{eq} = \mu^{g_j} \Big|_{eq}
       \end{flalign}
\end{itemize}

\section{Near Equilibrium Results}
\begin{itemize}
    \item Momentum transfer between phases (equation \eqref{eqn:mom_near_eq})
        \begin{flalign}
            \left( \sum_{\al \neq \beta} \That{\al}{\beta} \right)_{near} = \left(
            \sum_{\al \neq \beta} \That{\al}{\beta} \right)_{eq} - \left(
            \eps{\beta}\right)^2 \soten{R}^{\beta} \cd \bv{\beta,s}.
        \end{flalign}

    \item Momentum transfer between constituents (equation
        \eqref{eqn:mom_species_near_eq})
        \begin{flalign}
            \left( \sumba \Thatbaj + \ihat^{\aj} \right)_{near} = \left( \sumba \Thatbaj +
            \ihat^{\aj} \right)_{eq} - \epsa \rhoaj \soten{R}^{\aj} \cd \bv{\aj,\al}.
        \end{flalign}

    \item Partial heat flux  (equation \eqref{eqn:heat_near_eq})
        \begin{flalign}
            \left( \suma \epsa \bq^\al \right) &= -\soten{K} \cd \grad T
        \end{flalign}
\end{itemize}

\section{Constitutive Equations}
\begin{itemize}
    \item Darcy's law
        \begin{itemize}
            \item Pressure formulation (equation \eqref{eqn:Darcy})
                \begin{flalign}
                    \notag & \eps{\beta} \soten{R}^{\beta} \cd \left( \eps{\beta}
                    \bv{\beta,s} \right) \\
                    \notag &= -\eps{\beta} \grad p^{\beta} -
                    \ppi{\beta}{\beta} \grad \eps{\beta} - \ppi{\beta}{\gamma}\grad
                    \eps{\gamma} + \eps{\beta} \rho^{\beta} \foten{g} \\
                    \notag & \quad + \sumj \left[ \left( \eps{\gamma} \rho^{\gamma}
                    \pd{\psi^{\gamma}}{\rho^{\beta_j}} + \eps{s} \rho^s
                    \pd{\psi^s}{\rho^{\beta_j}} \right) \grad \rho^{\beta_j} \right] -
                    \eps{\beta} \rho^{\beta} \sumj \pd{\psi^{\beta}}{\rho^{\gamma_j}} \grad
                    \rho^{\gamma_j} \\ 
                    & \quad - \eps{\beta} \rho^{\beta} \pd{\psi^{\beta}}{J^s} \grad J^s -
                    \eps{\beta} \rho^{\beta} \pd{\psi^{\beta}}{\Cbars} : \grad \left(
                    \Cbars \right) + \diver \left( \Foten{\nu}^{\beta} : \bd{\beta}
                    \right).
                \end{flalign}

            \item Chemical potential formulation (equation \eqref{eqn:Darcy_ChemPot})
                \begin{flalign}
                    \notag & \soten{R} \cd \left( \eps{\beta} \bv{\beta,s} \right) \\
                    & \quad = -\sumj \left( \rho^{\beta_j} \grad \mu^{\beta_j} \right) -
                    \rho^{\beta} \eta^{\beta} \grad T + \rho^{\beta} \foten{g} +
                    \frac{1}{\eps{\beta}} \diver \left( \Foten{\nu}^{\beta} : \bd{\beta}
                    \right) 
                \end{flalign}

        \end{itemize}

    \item Fick's law (equation \eqref{eqn:Ficks_Law})
        \begin{flalign}
            \epsa \rhoaj \soten{R}^{\aj} \cd \bv{\aj,\al} = - \rhoaj \grad
            \muaj + \rhoaj \foten{g}.
        \end{flalign}

    \item Total heat flux (equation \eqref{eqn:Fourier_chem_pot})
        \begin{flalign}
            \bq{} &= -\soten{K} \cd \grad T - \sum_{\al=l,g} \left[ \left( \sumj \left(
                \rho^{\aj} \mu^{\aj} \right) + \rhoa T \eta^\al \right) \bv{\al,s}
                \right].
            \end{flalign}
\end{itemize}

%
    \chapter{Dimensional Quantities}\label{app:dimensional_quantities}
This appendix contains typical values for the quantities found in the macroscale heat and
moisture transport model. Note that since many of the tables are large so some are turned
sideways and some are bumped to different pages by default.

\linespread{1.0}
\begin{table}[ht*]
    \centering
    \begin{tabular}{|c|c|c|c|c|}
        \hline
        Symbol      & Quantity                  & Dimensions & \multicolumn{2}{|c|}{reference value} \\
                    &                           &            & liquid
                    (water)         & gas (air) \\ \hline \hline
        $\epsa$     & volume fraction           & $-$ & $-$ & $-$ \\ 
        $\rhoa$    & density                    & $M L^{-3}$ & $1000 \frac{kg}{m^3}$ & $1 \frac{kg}{m^3}$ \\
        $\rho^{g_v}$    & vapor density         & $M L^{-3}$ & $-$  & $2 \times 10^{-2} \frac{kg}{m^3}$ \\
        $\rho^{g_a}$    & dry air density       & $M L^{-3}$ & $-$  & $1.2 \frac{kg}{m^3}$ \\
        $T$         & temperature               & $K$ & $298.15K$ & $298.15K$ \\
        $R^{g_a}$   & gas constant (air)        & $L^2 t^{-2} K^{-1}$ & $-$ & $286.9 \frac{J}{kg \cd K}$\\
        $R^{g_v}$   & gas constant (vapor)      & $L^2 t^{-2} K^{-1}$ & $-$ & $461.5 \frac{J}{kg \cd K}$ \\
        $\soten{D}^\al$ & diffusion coefficient & $L^2 t^{-1}$ & $2 \times 10^{-9} \frac{m^2}{s}$ & $2.5 \times 10^{-5} \frac{m^2}{s}$ \\ 
        $\mu^{g_v}$ & chem. potential (vapor)        & $L^2 t^{-2}$ & $-$ & $-1.27 \times 10^7 \frac{J}{kg}$ \\
        $\mu^{g_a}$ & chem. potential (air)        & $L^2 t^{-2}$ & $-$ & $1.47 \times 10^7 \frac{J}{kg}$ \\
        $\foten{g}$ & gravity                   & $L t^{-2}$ & $9.81 \frac{m}{s^2}$ & $9.81 \frac{m}{s^2}$\\
        $\kappa$    & permeability              & $L^2$ & (see Tab.  \ref{tab:typical_permeability}) & (see Tab.  \ref{tab:typical_permeability})\\
        $\mu_\al$   & dynamic viscosity         & $M L^{-1} t^{-1}$ & $10^{-3} Pa\cd s$ & $10^{-5} Pa\cd s$\\
        $\eta^\al$  & specific entropy          & $L^2 t^{-2} K^{-1}$ & $3886 \frac{J}{kg \cd K}$ & $6519 \frac{J}{kg \cd K}$ \\
        $M$         & evaporation coefficient   & $t L^{-2}$ & & \\
        \hline
    \end{tabular}
    \caption{Dimensional quantities}
    \label{tab:mass_units}
\end{table}

\begin{sidewaystable}
    \begin{center}
        \begin{tabular}[ht*]{|*{14}{c|}}
            \hline
            $K \, [m/s]$ (water) & $10^0$ & $10^{-1}$ & $10^{-2}$ & $10^{-3}$ & $10^{-4}$ &
            $10^{-5}$ & $10^{-6}$ & $10^{-7}$ & $10^{-8}$ & $10^{-9}$ & $10^{-10}$ & $10^{-11}$ &
            $10^{-12}$ \\ \hline
            $K \, [m/s]$ (air) & $10^{-1}$ & $10^{-2}$ & $10^{-3}$ & $10^{-4}$ & $10^{-5}$
            & $10^{-6}$ & $10^{-7}$ & $10^{-8}$ & $10^{-9}$ & $10^{-10}$ & $10^{-11}$ &
            $10^{-12}$ & $10^{-13}$ \\ \hline
            $\kappa \, [m^2]$ & $10^{-7}$ & $10^{-8}$ & $10^{-9}$ & $10^{-10}$ &
            $10^{-11}$ & $10^{-12}$ & $10^{-13}$ & $10^{-14}$ & $10^{-15}$ & $10^{-16}$ &
            $10^{-17}$ & $10^{-18}$ & $10^{-19}$ \\ \hline \hline
            Permeability &  \multicolumn{4}{|c|}{pervious} & \multicolumn{4}{|c|}{semipervious} &
            \multicolumn{5}{|c|}{impervious} \\ \hline
            Aquifer & \multicolumn{5}{|c|}{good} & \multicolumn{4}{|c|}{poor} &
            \multicolumn{4}{|c|}{none} \\ \hline
            Sand and Gravel & \multicolumn{2}{|c|}{clean gravel} &
            \multicolumn{3}{|c|}{clean sand} & \multicolumn{4}{|c|}{fine sand} &
            \multicolumn{4}{|c|}{-} \\ \hline
            Clay and Organic & \multicolumn{4}{|c|}{-} & \multicolumn{2}{|c|}{peat} &
            \multicolumn{3}{|c|}{stratified clay} & \multicolumn{4}{|c|}{unweathered clay} \\
            \hline
            Rocks & \multicolumn{4}{|c|}{-} & \multicolumn{3}{|c|}{oil rocks} &
            \multicolumn{2}{|c|}{sandstone} & \multicolumn{2}{|c|}{limestone} &
            \multicolumn{2}{|c|}{granite} \\ \hline \hline
            $K \, [cm/s]$ (water) & $10^2$ & $10^{1}$ & $10^{0}$ & $10^{-1}$ & $10^{-2}$ &
            $10^{-3}$ & $10^{-4}$ & $10^{-5}$ & $10^{-6}$ & $10^{-7}$ & $10^{-8}$ & $10^{-9}$ &
            $10^{-10}$ \\ \hline
            $K \, [cm/s]$ (air) & $10^{1}$ & $10^{0}$ & $10^{-1}$ & $10^{-2}$ & $10^{-3}$ &
            $10^{-4}$ & $10^{-5}$ & $10^{-6}$ & $10^{-7}$ & $10^{-8}$ & $10^{-9}$ & $10^{-10}$ &
            $10^{-11}$ \\ \hline
            $\kappa \, [cm^2]$ & $10^{-3}$ & $10^{-4}$ & $10^{-5}$ & $10^{-6}$ & $10^{-7}$
            & $10^{-8}$ & $10^{-9}$ & $10^{-10}$ & $10^{-11}$ & $10^{-12}$ & $10^{-13}$ &
            $10^{-14}$ & $10^{-15}$ \\ \hline
        \end{tabular}
    \end{center}
    \caption{Typical values of hydraulic conductivity ($K$) for water and air, and
        associated values for permeability ($\kappa$).  Note that $K = \kappa \rho g /
        \mu$ where $\rho^g = 1kg/m^3, \, \rho^l = 1000kg/m^3 \, \mu_l = 10^{-3} Pa \cdot
        s$, and $\mu_g = 10^{-5} Pa \cdot s$. Modified from Bear pg. 136 \cite{Bear1988}}
    \label{tab:typical_permeability}
\end{sidewaystable}
\renewcommand{\baselinestretch}{\normalspace}

\end{appendix}
%
%
%
%


\begin{thebibliography}{10}

\bibitem{Alt2012}
Hans~Wilhelm Alt.
\newblock {An Abstract Existence Theorem for Parabolic Systems}.
\newblock {\em Communications on Pure and Applied Analysis}, 11(09):2079--2123,
  2012.

\bibitem{Alt1983}
Hans~Wilhelm Alt and Stephan Luckhaus.
\newblock {Quasilinear Elliptic-Parabolic Differential Equations}.
\newblock {\em Mathematische Zeitschrift}, 183:311--341, August 1983.

\bibitem{Amann1990}
H~Amann.
\newblock {Dynamic theory of quasilinear parabolic equations. II.
  Reaction-diffusion systems}.
\newblock {\em Differential Integral Equations}, 3(1):13--75, 1990.

\bibitem{Atkins2010}
Peter~William Atkins and Julio~De Paula.
\newblock {\em {Atkins' Physical chemistry}}.
\newblock Oxford University Press, 2010.

\bibitem{Bear1988}
Jacob Bear.
\newblock {\em {Dynamics of fluids in porous media}}.
\newblock Courier Dover Publications, 1988.

\bibitem{Benkhalifa1995}
M~Benkhalifa and G~Arnaud.
\newblock {The viscous air flow pattern in the Stefan diffusion tube}.
\newblock {\em Transport in porous media}, pages 15--36, 1995.

\bibitem{Bennethum1994}
Lynn Bennethum.
\newblock {\em {Multiscale Hybrid Mixture Theory for swelling systems with
  interfaces}}.
\newblock PhD thesis, Purdue University, 1994.

\bibitem{Bennethum1997}
Lynn~Schreyer Bennethum.
\newblock {Modified Darcy's law, Terzaghi's effective stress principle and
  Fick's law for swelling clay soils}.
\newblock {\em Computers and Geotechnics}, 20(3-4):245--266, 1997.

\bibitem{Bennethum2007}
Lynn~Schreyer Bennethum.
\newblock {Theory of flow and deformation of swelling porous materials at the
  macroscale}.
\newblock {\em Computers and Geotechnics}, 34(4):267--278, July 2007.

\bibitem{bennethum}
Lynn~Schreyer Bennethum.
\newblock {\em {Notes for Introduction to Continuum Mechanics}}.
\newblock 2011.

\bibitem{Bennethum1996}
Lynn~Schreyer Bennethum and John~H. Cushman.
\newblock {Clarifying Mixture Theory and the Macroscale Chemical Potential for
  Porous Media}.
\newblock {\em International Journal of Engineering Science},
  34(14):1611--1621, 1996.

\bibitem{Bennethum1996b}
Lynn~Schreyer Bennethum and John~H Cushman.
\newblock {Multiscae, Hybrid Mixture Theory for Swelling Systems - II:
  Constitutive Theory}.
\newblock {\em International Journal of Engineering Science}, 34(2):147--169,
  1996.

\bibitem{Bennethum1996a}
Lynn~Schreyer Bennethum and John~H Cushman.
\newblock {Multiscale, hybrid mixture theory for swelling systems - I: balance
  laws}.
\newblock {\em International Journal of Engineering Science}, 34(2):125--145,
  1996.

\bibitem{Bennethum1999}
Lynn~Schreyer Bennethum and John~H. Cushman.
\newblock {Coupled Solvent and Heat Transport of a Mixture of Swelling Porous
  Particles and Fluids : Single Time-Scale Problem}.
\newblock {\em Sciences-New York}, pages 211--244, 1999.

\bibitem{Murad2000}
Lynn~Schreyer Bennethum, Marcio~A. Murad, and John~H. Cushman.
\newblock {Macroscale Thermodynamics and the Chemical Potential for Swelling
  Porous Media}.
\newblock {\em Transport in Porous Media}, 39(2):187--225, 2000.

\bibitem{Bennethum2004}
Lynn~Schreyer Bennethum and Tessa~F. Weinstein.
\newblock {Three Pressures in Porous Media}.
\newblock {\em Transport in Porous Media}, 54:1--34, 2004.

\bibitem{Berentsen2006}
C.W.J. Berentsen, S.~M. Hassanizadeh, A~Bezuijen, and O~Oung.
\newblock {Modeling of Two-Phase Flow in Porous Media Including Non-Equilibrium
  Capillary Pressure Effects}.
\newblock In {\em Computational Methods in Water Resources XVI}, 2006.

\bibitem{Bird2007}
Robert~Byron Bird, Warren~E. Stewart, and Edwin~N. Lightfoot.
\newblock {\em {Transport Phenomena 2ed.}}
\newblock J. Wiley, 2007.

\bibitem{Bixler1985}
N.E. Bixler.
\newblock {NORIA -- A finite element computer program for analyzing water,
  vapor, air, and energy transport in porous media, Rep SAND84-2057, UC-70}.
\newblock Technical report, Sandia Natl. Lab., Albuquerque, N.M., 1985.

\bibitem{Bottero2011}
S.~Bottero, S.~Majid Hassanizadeh, and P.~J. Kleingeld.
\newblock {From Local Measurements to an Upscaled Capillary
  Pressure–Saturation Curve}.
\newblock {\em Transport in Porous Media}, 88(2):271--291, March 2011.

\bibitem{Callen1985}
Herbert~B. Callen.
\newblock {\em {Thermodynamics and an introduction to thermostatistics}}.
\newblock Wiley, 1985.

\bibitem{Camassel2005}
B~Camassel, N~Sghaier, M~Prat, and S~Bennasrallah.
\newblock {Evaporation in a capillary tube of square cross-section: application
  to ion transport}.
\newblock {\em Chemical Engineering Science}, 60(3):815--826, February 2005.

\bibitem{Campbell1985}
Gaylon~S. Campbell.
\newblock {\em {Soil physics with BASIC: transport models for soil-plant
  systems}}.
\newblock Elsevier, 1985.

\bibitem{Cass1984}
A~Cass, GS~Campbell, and TL~Jones.
\newblock {Enhancement of Thermal Water Vapor Diffusion In Soil}.
\newblock {\em Soil Science Society of America}, 48(1):25--32, 1984.

\bibitem{Celia2009}
Michael~a Celia and Jan~M Nordbotten.
\newblock {Practical modeling approaches for geological storage of carbon
  dioxide.}
\newblock {\em Ground Water}, 47(5):627--38, 2009.

\bibitem{Coleman1963}
Bernard~D. Coleman and Walter Noll.
\newblock {The thermodynamics of elastic materials with heat conduction and
  viscosity}.
\newblock {\em Archive for Rational Mechanics and Analysis}, 13(1):167--178,
  December 1963.

\bibitem{Crank1979}
John Crank.
\newblock {\em {The Mathematics of Diffusion}}.
\newblock Clarendon Press, 1979.

\bibitem{Cushman1985}
J.H. Cushman.
\newblock {Multiphase transport based on compact distributions}.
\newblock {\em Acta Applicandae Mathematicae}, 3(3):239--254, 1985.

\bibitem{Cushman2002}
John~H. Cushman, Bill~X. Hu, and Lynn~Schreyer Bennethum.
\newblock {A primer on upscaling tools for porous media}.
\newblock {\em Advances in Water Resources}, 25(8-12):1043--1067, August 2002.

\bibitem{Darcy1856}
H.~Darcy.
\newblock {\em {Les fontaines publiques de la ville de Dijon}}.
\newblock Dalmont, 1856.

\bibitem{deVries1958}
D.~A. DeVries.
\newblock {Simultaneous Transfer of Heat and Moisture in Porous Media}.
\newblock {\em Transactions, American Geophysical Union}, 39(5):909--916, 1958.

\bibitem{Dormieux2006}
Luc Dormieux, Djimedo Kondo, and Frans-Josef Ulm.
\newblock {\em {Microporomechanics}}.
\newblock John Wiley \& Sons, 2006.

\bibitem{Eringen2003}
A.~Cemal Eringen.
\newblock {Note On Darcy’s Law}.
\newblock {\em Journal of Applied Physics}, 94(2):1282, 2003.

\bibitem{Evans2010}
Lawrence~C. Evans.
\newblock {\em {Partial differential equations}}.
\newblock AMS Bookstore, 2010.

\bibitem{Fayer1995}
MJ~Fayer and C.S. Simmons.
\newblock {Modified soil water retention functions for all matric suctions}.
\newblock {\em Water Resources Research}, 31(5):1233--1238, 1995.

\bibitem{Fick1855}
A~Fick.
\newblock {Ueber Diffusion}.
\newblock {\em Annalen der Physik}, 2006.

\bibitem{Folland1999}
Gerald~B. Folland.
\newblock {\em {Real Analysis -- Modern Techniques and Their Applications}}.
\newblock John Wiley \& Sons, 2 edition, 1999.

\bibitem{Gray1977}
W~G Gray and P~C~Y Lee.
\newblock {On the Theorems for Local Volume Averaging of Multiphase Systems}.
\newblock {\em International Journal of Multiphase Flow}, 3:333--340, 1977.

\bibitem{Gray2007a}
W.~G. Gray and C.~T. Miller.
\newblock {Consistent thermodynamic formulations for multiscale hydrologic
  systems: Fluid pressures}.
\newblock {\em Water Resources Research}, 43(9):W09408, September 2007.

\bibitem{Gray1998}
William~G. Gray and S.~Majid Hassanizadeh.
\newblock {Macroscale contimuum mechanics for multiphase porous-media flow
  includeing phases, interfaces, common lines, and common points}.
\newblock {\em Advances in Water Resources}, 21:261--281, 1998.

\bibitem{Gray2007}
William~G. Gray and Bernhard~A Schrefler.
\newblock {Analysis of the solid phase stress tensor in multiphase porous
  media}.
\newblock {\em International Journal for Numerical and Analytical Methods in
  Geomechanics}, (July 2006):541--581, 2007.

\bibitem{Gray1979}
M~Hassanizadeh and W~G Gray.
\newblock {General conservation equations for multi-phase systems: 1. Averaging
  procedure}.
\newblock {\em Advances in Water Resources}, 2:131--144, 1979.

\bibitem{Hassanizadeh1986}
S.~Majid Hassanizadeh.
\newblock {Derivation of basic equations of mass transport in porous media,
  Part 1. Macroscopic balance laws}.
\newblock {\em Advances in Water Resources}, 9(4):196--206, December 1986.

\bibitem{Hassanizadeh1986a}
S.~Majid Hassanizadeh.
\newblock {Derivation of basic equations of mass transport in porous media,
  Part 2. Generalized Darcy's and Fick's laws}.
\newblock {\em Advances in Water Resources}, 9(4):207--222, December 1986.

\bibitem{Hassanizadeh2001}
S.~Majid Hassanizadeh and A.~Y. Beliaev.
\newblock {A Theoretical Model of Hysteresis and Dynamic Effects in the
  Capillary Relation for Two-phase Flow in Porous Media}.
\newblock {\em Transport in Porous Media}, (1):487--510, 2001.

\bibitem{Hassanizadeh2002}
S.~Majid Hassanizadeh, M.A. Celia, and H.K. Dahle.
\newblock {Dynamic effect in the capillary pressure-saturation relationship and
  its impacts on unsaturated flow}.
\newblock {\em Vadose Zone Journal}, 1(1):38, 2002.

\bibitem{Joekar-Niasar2007}
V.~Joekar-Niasar, S.~M. Hassanizadeh, and a.~Leijnse.
\newblock {Insights into the Relationships Among Capillary Pressure,
  Saturation, Interfacial Area and Relative Permeability Using Pore-Network
  Modeling}.
\newblock {\em Transport in Porous Media}, 74(2):201--219, December 2007.

\bibitem{Joekar-Niasar2010}
V.~Joekar-Niasar, S.~Majid Hassanizadeh, and H.~K. Dahle.
\newblock {Non-equilibrium effects in capillarity and interfacial area in
  two-phase flow: dynamic pore-network modelling}.
\newblock {\em Journal of Fluid Mechanics}, 655:38--71, July 2010.

\bibitem{Joekar-Niasar2012}
Vahid Joekar-Niasar and S.~Majid Hassanizadeh.
\newblock {Uniqueness of Specific Interfacial Area–Capillary
  Pressure–Saturation Relationship Under Non-Equilibrium Conditions in
  Two-Phase Porous Media Flow}.
\newblock {\em Transport in Porous Media}, 94(2):465--486, February 2012.

\bibitem{Kerkhof1997}
Piet J. a.~M. Kerkhof.
\newblock {New light on some old problems: Revisiting the Stefan tube, Graham's
  law, and the Bosanquet equation}.
\newblock {\em Industrial \& Engineering Chemistry Research},
  5885(1879):915--922, 1997.

\bibitem{Kleinfelter2007}
Natalie Kleinfelter, Moongyu Park, and John~H Cushman.
\newblock {Mixture theory and unsaturated flow in swelling soils}.
\newblock {\em Transport in Porous Media}, 68(1):69--89, 2007.

\bibitem{Le2005}
Dung Le and TT~Nguyen.
\newblock {Global Existence for a Class of Triangular Parabolic Systems on
  Domains of Arbitrary Dimension}.
\newblock {\em Proceedings of the American Mathematical Society},
  133(7):1985--1992, 2005.

\bibitem{LeVeque1992}
Randall~J. LeVeque.
\newblock {\em {Numerical methods for conservation laws}}.
\newblock Birkh\"{a}user, 1992.

\bibitem{LeVeque2007}
Randall~J. LeVeque.
\newblock {\em {Finite difference methods for ordinary and partial differential
  equations ...}}
\newblock Society for Industrial and Applied Mathematics, 2007.

\bibitem{McQuarrie1997}
Donald~Allan McQuarrie and John~Douglas Simon.
\newblock {\em {Physical chemistry: a molecular approach}}.
\newblock University Science Books, 1997.

\bibitem{Mikelic2010}
Andro Mikeli\'{c}.
\newblock {A global existence result for the equations describing unsaturated
  flow in porous media with dynamic capillary pressure}.
\newblock {\em Journal of Differential Equations}, 248(6):1561--1577, March
  2010.

\bibitem{Nordbotten2008}
J.~M. Nordbotten, M.~a. Celia, H.~K. Dahle, and S.~M. Hassanizadeh.
\newblock {On the definition of macroscale pressure for multiphase flow in
  porous media}.
\newblock {\em Water Resources Research}, 44(6):1--8, June 2008.

\bibitem{Oden2009}
John~Tinsley Oden and Leszek Demkowicz.
\newblock {\em {Applied functional analysis}}.
\newblock CRC Press, 2009.

\bibitem{Peszynska2008}
M~Peszynska and SY~Yi.
\newblock {Numerical methods for unsaturated flow with dynamic capillary
  pressure in heterogeneous porous media}.
\newblock {\em Int. J. Numer. Anal. Model}, 5:126--149, 2008.

\bibitem{deVries1957}
J.~R. Philip and D.~A. DeVries.
\newblock {Moisture Movement in Porous Materials under Temperature Gradients}.
\newblock {\em Transactions, American Geophysical Union}, 38(2):222--232, 1957.

\bibitem{Pinder2006}
George~F. Pinder and Michael~A. Celia.
\newblock {\em {Subsurface Hydrology}}.
\newblock John Wiley \& Sons, 2006.

\bibitem{Pop2004}
I.S. Pop, F.~Radu, and P.~Knabner.
\newblock {Mixed finite elements for the Richards’ equation: linearization
  procedure}.
\newblock {\em Journal of Computational and Applied Mathematics},
  168(1-2):365--373, July 2004.

\bibitem{Prunty2002a}
Lyle Prunty.
\newblock {Spatial distribution of heat of wetting in porous media}.
\newblock {\em ASAE Annual International Meeting/CIGR XVth World}, 2002.

\bibitem{Richards1931}
L.A. Richards.
\newblock {Capillary Conduction of Liquids Through Porous Mediums}.
\newblock {\em Journal of Applied Physics}, 1:318--331, 1931.

\bibitem{Rincon2005}
M.A. Rincon, J.~L\'{\i}maco, and I.~{Shih Liu}.
\newblock {Existence and Uniqueness of Solutions of a Nonlinear Heat Equation}.
\newblock {\em TEMA - Tend\^{e}ncias em Matem\'{a}tica Aplicada e
  Computacional}, 6(2):1--11, August 2005.

\bibitem{Saito2006}
Hirotaka Saito, Jiri \v{S}imůnek, and Binayak~P. Mohanty.
\newblock {Numerical Analysis of Coupled Water, Vapor, and Heat Transport in
  the Vadose Zone}.
\newblock {\em Vadose Zone Journal}, 5(2):784, 2006.

\bibitem{Sakai2009}
Masaru Sakai, Nobuo Toride, and Jiř\'{\i} \v{S}imůnek.
\newblock {Water and Vapor Movement with Condensation and Evaporation in a
  Sandy Column}.
\newblock {\em Soil Science Society of America Journal}, 73(3):707, 2009.

\bibitem{Schreyer-Bennethum2012}
Lynn Schreyer-Bennethum.
\newblock {Macroscopic Flow Potentials in Swelling Porous Media}.
\newblock {\em Transport in Porous Media}, 94(1):47--68, April 2012.

\bibitem{Shahraeeni2010}
Ebrahim Shahraeeni and Dani Or.
\newblock {Pore-scale analysis of evaporation and condensation dynamics in
  porous media.}
\newblock {\em Langmuir : the ACS journal of surfaces and colloids},
  26(17):13924--36, September 2010.

\bibitem{Shahraeeni2012a}
Ebrahim Shahraeeni and Dani Or.
\newblock {Pore scale mechanisms for enhanced vapor transport through
  partially-saturated porous media}.
\newblock {\em Water Resources Research}, 2012.

\bibitem{Shahraeeni2012b}
Ebrahim Shahraeeni and Dani Or.
\newblock {Pore scale mechanisms for enhanced vapor transport through partially
  saturated porous media}.
\newblock {\em Water Resources Research}, pages 1--35, 2012.

\bibitem{Shokri2009a}
N.~Shokri, P.~Lehmann, and D.~Or.
\newblock {Critical evaluation of enhancement factors for vapor transport
  through unsaturated porous media}.
\newblock {\em Water Resources Research}, 45(10):1--9, October 2009.

\bibitem{Silverman1999}
T.S. Silverman.
\newblock {\em {A pore-scale experiment to evaluate enhanced vapor diffusion in
  porous media}}.
\newblock PhD thesis, New Mexico Institute of Mining and Technology, 1999.

\bibitem{Slatery1967}
John~C Slattery.
\newblock {Flow of viscoelastic fluids through porous media}.
\newblock {\em AIChE Journal}, 13(6):1066--1071, 1967.

\bibitem{Smits2011}
Kathleen~M. Smits, Abdullah Cihan, Toshihiro Sakaki, and Tissa~H.
  Illangasekare.
\newblock {Evaporation from soils under thermal boundary conditions:
  Experimental and modeling investigation to compare equilibrium- and
  nonequilibrium-based approaches}.
\newblock {\em Water Resources Research}, 47(5):1--14, May 2011.

\bibitem{Smits2012}
Kathleen~M. Smits, Viet~V. Ngo, Abdullah Cihan, Toshihiro Sakaki, and Tissa~H.
  Illangasekare.
\newblock {An evaluation of models of bare soil evaporation formulated with
  different land surface boundary conditions and assumptions}.
\newblock {\em Water Resources Research}, 48(12):W12526, December 2012.

\bibitem{Smits2010a}
Kathleen~M. Smits, Toshihiro Sakaki, Anuchit Limsuwat, and Tissa~H.
  Illangasekare.
\newblock {Thermal Conductivity of Sands under Varying Moisture and Porosity in
  Drainage Wetting Cycles}.
\newblock {\em Vadose Zone Journal}, 9(1):172, 2010.

\bibitem{VanGenuchten1980}
M.~Th. van Genuchten.
\newblock {A Closed-form Equation for Predicting the Hydraulic Conductivity of
  Unsaturated Soils}.
\newblock {\em Soil Science of America}, 44:892--989, 1980.

\bibitem{Webb1998}
SW~Webb.
\newblock {Review of Enhanced Vapor Diffusion in Porous Media}.
\newblock {\em Sandia National Laboratories, Albuquerque, New}, 1998.

\bibitem{weinstein}
Tessa~F. Weinstein.
\newblock {\em {Three-Phase Hybrid Mixture Theory for Swelling Drug Delivery
  Systems}}.
\newblock PhD thesis, University of Colorado Denver, 2005.

\bibitem{Whitaker1967}
Stephen Whitaker.
\newblock {Diffusion and dispersion in porous media}.
\newblock {\em AIChE Journal}, 13(3):420--427, 1967.

\bibitem{Whitaker1969}
Stephen. Whitaker.
\newblock {Advances in theory of fluid motion in porous media}.
\newblock {\em Industrial \& Engineering Chemistry}, 61(12):14--28, 1969.

\bibitem{Whitaker1991a}
Stephen Whitaker.
\newblock {Coupled Transport in Multiphase Systems: A Theory of Drying}.
\newblock {\em Advances in Heat Transfer}, 31:1--104, 1991.

\bibitem{Whitaker1991}
Stephen Whitaker.
\newblock {Role of the species momentum equation in the analysis of the Stefan
  diffusion tube}.
\newblock {\em Industrial \& Engineering Chemistry Research}, pages 978--983,
  1991.

\bibitem{Wojciechowski2011}
Keith~J. Wojciechowski.
\newblock {\em {Analysis and Numerical Solution of Nonlinear Volterra Partial
  Integrodifferential Equations Modeling Swelling Porous Materials}}.
\newblock PhD thesis, University of Colorado Denver, 2011.

\bibitem{WolframNDSolve2012}
Wolfram.
\newblock {Mathematica 9 Documentation -- NDSolve reference.wolfram.com
  /mathematica /ref /NDSolve.html}, 2012.

\bibitem{Zach}
E.~C. Zachmanoglou and Dale~W. Thoe.
\newblock {\em {Introduction to partial differential equations with
  applications}}.
\newblock Courier Dover Publications, 1986.

\end{thebibliography}

\end{document}